\documentclass[A4,12pt]{article}
\usepackage{multirow}
\usepackage{eqparbox}
\usepackage{arydshln}
\usepackage{wrapfig}
\usepackage[caption=false,font=footnotesize]{subfig}
\usepackage[titletoc,toc,title]{appendix}

\usepackage{fullpage}
\usepackage{color}
\usepackage{url}
\usepackage{graphicx}
\usepackage{epstopdf}
\usepackage{xspace}
\usepackage[cmex10]{amsmath}
\usepackage{amssymb}
\usepackage{amsthm}
\usepackage{algorithmicx}
\usepackage[noend]{algpseudocode}
\usepackage{listings}
\usepackage{subfig}
\usepackage{tabularx}
\usepackage{pgfplots}
\pgfplotsset{compat=newest}
\usepackage{pgfplotstable}

\usepackage{fancyvrb}
\usepackage{ifthen}
\usepackage{tikz}
\usepackage[caption=false,font=footnotesize]{subfig}

\usepackage{titlesec}
\titleformat{\paragraph}
{\normalfont\normalsize\bfseries}{\theparagraph}{1em}{}
\titlespacing*{\paragraph}
{0pt}{3.25ex plus 1ex minus .2ex}{1.5ex plus .2ex}
\newcounter{tempEquationCounter} 
\newcounter{thisEquationNumber}
\newenvironment{floatEq}
{\setcounter{thisEquationNumber}{\value{equation}}\addtocounter{equation}{1}
\begin{figure*}[!t]
\normalsize\setcounter{tempEquationCounter}{\value{equation}}
\setcounter{equation}{\value{thisEquationNumber}}
}
{\setcounter{equation}{\value{tempEquationCounter}}
\hrulefill\vspace*{4pt}
\end{figure*}
}
\newcommand{\comm}[1]{} 

\makeatletter
\newcounter{subsubparagraph}[subparagraph]
\renewcommand\thesubsubparagraph{%
  \thesubparagraph.\@arabic\c@subsubparagraph}
\newcommand\subsubparagraph{%
  \@startsection{subsubparagraph}    
    {6}                              
    {\parindent}                     
    {3.25ex \@plus 1ex \@minus .2ex} 
    {-1em}                           
    {\normalfont\normalsize\bfseries}}
\newcommand\l@subsubparagraph{\@dottedtocline{6}{10em}{5em}}
\newcommand{\subsubparagraphmark}[1]{}
\makeatother

\newcommand{\FINDATE}{31.08.2016} 
\newcommand{\DNUM}{D2.4}
\newcommand{\DNAME}{Report on the final prototype of programming abstractions for energy-efficient inter-process communication}
\newcommand{\DFMTNAME}{Report on the final prototype \\ of programming  abstractions for energy-efficient \\ inter-process communication}
\newcommand{\DSHORTNAME}{Report on the final prototype of programming abstractions}

\usepackage{siunitx}
\usepackage{diagbox}
\usepackage{makecell}
\newtheorem{lemma-uit}{Lemma}

\newcommand\UB{\mathit{UB}}

\algrenewcommand\alglinenumber[1]{\scriptsize #1:}

\algnewcommand{\LComment}[1]{\Statex  \(\triangleright\) #1 \hfill~}

\algdef{SE}[DOWHILE]{Do}{doWhile}{\algorithmicdo}[1]{\algorithmicwhile\ #1}%

\algnewcommand\EMPTY{\textbf{EMPTY}}

\algblockdefx[NAME]{Start}{End}%
	[1]{\textbf{Struct} #1:}%
	{EndStruct}
\algnotext[NAME]{End}

\algblockdefx[Name]{START}{END}%
	[1]{#1}%
	{Ending}
\algnotext[Name]{END}

\usepackage{verbatim}  

\usepackage[T1]{fontenc}
\usepackage[english]{babel}
\usepackage{url,xspace}
\usepackage{stmaryrd}
\usepackage{paralist}
\usepackage{amsthm,amssymb,amsmath}
\DeclareMathAlphabet{\mathcal}{OMS}{cmsy}{m}{n}
\usepackage{fixltx2e}
\usepackage{array}

\usepackage{tikz}

\usetikzlibrary{patterns}
\usetikzlibrary{calc,positioning}
\usetikzlibrary{shapes}
\usetikzlibrary{matrix}
\usetikzlibrary{decorations}
\usetikzlibrary{decorations.pathmorphing}
\usetikzlibrary{decorations.markings}
\usetikzlibrary{decorations.pathreplacing}
\usetikzlibrary{backgrounds}

\usepackage[linesnumbered,vlined,noend,ruled,nofillcomment]{algorithm2e}

\usepackage{ifthen}
\usepackage{arrayjobx}

\usepackage{dsfont}

\usepackage{mathtools}

\graphicspath{{energy-ds-chalmers/pdfs/},{figs-chalmers/markov/},{figs-chalmers/aggregate/},{figs-uit/},{figs-uit-model/},{figs-uit-model/diagrams/}}

\begin{document}

\thispagestyle{empty}

\vspace{-3cm}
\begin{center}
\textbf{SEVENTH FRAMEWORK PROGRAMME}\\
\textbf{THEME ICT-2013.3.4}\\
Advanced Computing, Embedded and Control Systems
\end{center}
\bigskip

\begin{center}
\includegraphics[width=\textwidth]{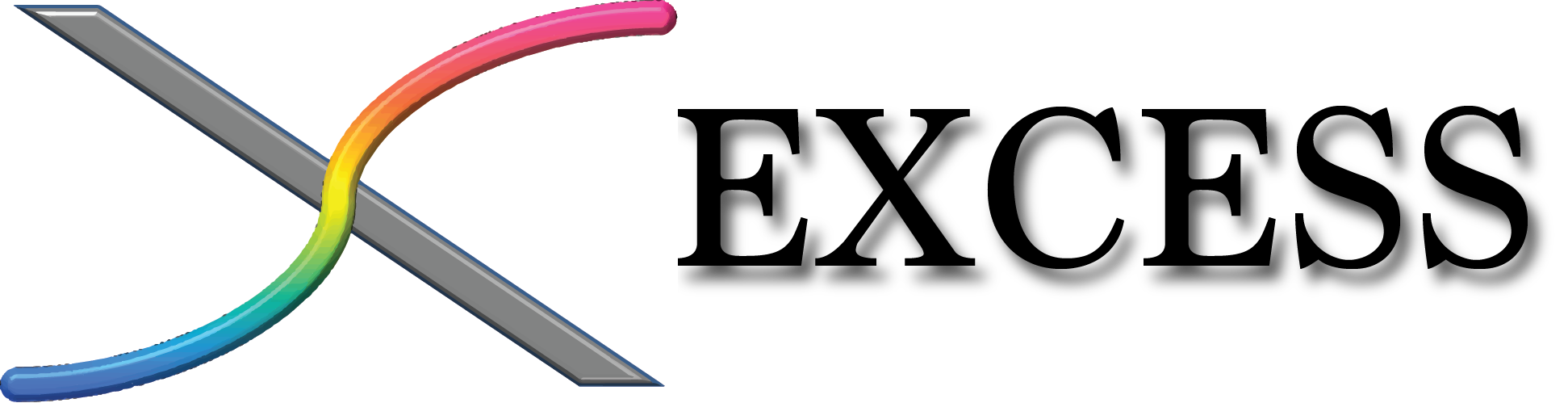}
\end{center}
\bigskip

\begin{center}
Execution Models for Energy-Efficient
Computing Systems\\
Project ID: 611183
\end{center}
\bigskip

\begin{center}
\Large
\textbf{\DNUM} \\
\textbf{\DNAME}
\end{center}
\bigskip

\begin{center}
\large
Phuong Ha, Vi Tran, Ibrahim Umar,
Aras Atalar, Anders Gidenstam,
Paul Renaud-Goud, Philippas Tsigas,
Ivan Walulya

\end{center}

\vfill

\begin{center}
\includegraphics[width=3cm]{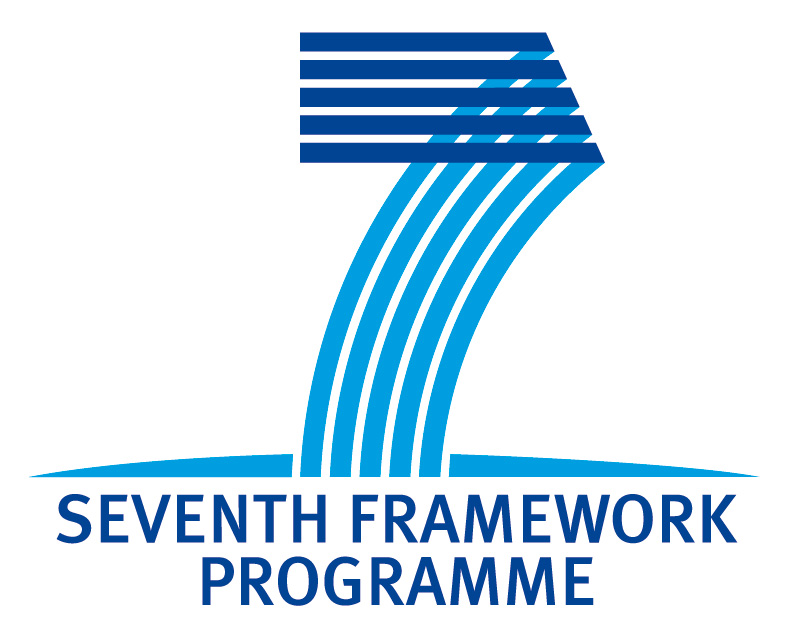}\\
Date of preparation (latest version): \FINDATE \\
Copyright\copyright\ 2013 -- 2016 The EXCESS Consortium  $^\ast$ \\
{\scriptsize $\ ^\ast$ Some sections in this report have been published, see copyright notices at the beginning of the sections.}\\
\hrulefill \\
The opinions of the authors expressed in this document do not
necessarily reflect the official opinion of EXCESS partners or of
the European Commission.
\end{center}

\newpage

\section*{DOCUMENT INFORMATION}

\vspace{1cm}

\begin{center}
\begin{tabular}{ll}
\textbf{Deliverable Number} & \DNUM \\
\textbf{Deliverable Name} & \begin{minipage}{10cm}{\DFMTNAME}\end{minipage} \\
\textbf{Authors}
& Phuong Ha \\
& Vi Tran \\
& Ibrahim Umar \\
& Aras Atalar\\
& Anders Gidenstam \\
& Paul Renaud-Goud\\
& Philippas Tsigas \\
& Ivan Walulya \\

\textbf{Responsible Author} & Phuong Ha\\
& e-mail: \url{phuong.hoai.ha@uit.no} \\
& Phone: +47 776 44032 \\
\textbf{Keywords} & High Performance Computing; \\
& Energy Efficiency \\
\textbf{WP/Task} & WP2/Task 2.1, 2.2, 2.3, 2.4 \\
\textbf{Nature} & R \\
\textbf{Dissemination Level} & PU \\
\textbf{Planned Date} &  31.08.2016\\
\textbf{Final Version Date} & 31.08.2016 \\
\textbf{Reviewed by} &  \\
\textbf{MGT Board Approval} & \\ 
\end{tabular}
\end{center}

\newpage

\section*{DOCUMENT HISTORY}

\vspace{1cm}

\begin{center}
\begin{tabular}{llll}
\textbf{Partner} & 
\textbf{Date} & 
\textbf{Comment} & 
\textbf{Version} \\
UiT (P.Ha, V.Tran) & 01.07.2016 & Deliverable skeleton & 0.1 \\
Chalmers (A.\ Atalar) & 19.07.2016 & Input - energy model and energy evaluation & 0.2 \\
Chalmers (I.\ Walulya) & 20.07.2016 & Input - implementation of streaming aggregation& 0.2 \\
UiT (V.Tran) & 03.08.2016 & Input - energy and power model & 0.3 \\
Chalmers (I.\ Walulya) & 22.08.2016 & Input revise & 0.4 \\
UiT (I.Umar, V.Tran) & 29.08.2016 & Input revise  & 0.5 \\
\end{tabular}
\end{center}

\newpage

\begin{abstract}
Work package 2 (WP2) aims to develop libraries for energy-efficient inter-process communication and data sharing on the EXCESS platforms. 
The Deliverable D2.4 reports on the final prototype of programming abstractions for energy-efficient inter-process communication. Section \ref{sec:introduction} is the updated overview of the prototype of programming abstraction and devised power/energy models. The Section 2-6 contain the latest results of the four studies: 
\begin{itemize} 
\item GreenBST, a energy-efficient and concurrent search tree (cf. Section \ref{sec:Concurrent-Data-Structures})
\item Customization methodology for implementation of streaming aggregation in embedded systems (cf. Section \ref{sec:stream})
\item Energy Model on CPU for Lock-free Data-structures in Dynamic Environments (cf. Section \ref{sec:chalmers-lock-free})
\item A General and Validated Energy Complexity Model for Multithreaded Algorithms (cf. Section \ref{sec:uit-energy-model})
\end{itemize}
\end{abstract}

\newpage

\section*{Executive Summary}
Work package 2 (WP2) investigate and model the trade-offs between energy consumption and performance of data structures and algorithms for inter-process communication. WP2 also provides concurrent data structures and algorithms that support energy-efficient massive parallelism while minimizing inter-component communication.
 
The main achievements of Deliverable D2.4 are summarized as follows.
\begin{itemize}
\item
We have described the cache-oblivious abstraction that is used in developing our energy-efficient and concurrent data structures. We also present in the same section a detailed description of GreenBST, an energy-efficient concurrent search tree that was 
briefly described in D2.3. Also in this deliverable, GreenBST is tested with new state-of-the-art concurrent search trees that are not included in D2.3. The latest experimental results showed that GreenBST is more energy efficient and has 
higher throughput for both the concurrent search- and update- intensive workloads than the state-of-the-art. We also have implemented GreenBST for Myriad2 platform and have conducted an experimental evaluation using the implementation.
\item We present a methodology for the customization of streaming aggregation implemented in modern low power embedded devices. The methodology is based on design space exploration and provides a set of customized implementations that can be used by developers to perform trade-offs between throughput, latency, memory and energy consumption. We compare the proposed embedded system implementations of the streaming aggregation operator with the corresponding HPC and GPGPU implementations in terms of performance per watt. Our results show that the implementations based on low power embedded systems provide up to 54 and 14 times higher performance per watt than the corresponding Intel Xeon and Radeon HD 6450 implementations, respectively. 
\item We present an energy model on CPU for lock-free data-structures in dynamic environments. Lock-free data structures are based on retry loops and are called by 
application-specific routines. In D2.3, we illustrate 
the performance impacting factors and the model that we use to cover a subset 
of the lock-free structures that we consider here. In the former study, the analysis is built upon properties that arise only when the sizes of the retry loops and the application-specific work are constant. 
In this work, we introduce two new frameworks that can be used to the 
capture the performance of a wider set of lock-free data structures (\textit{i.e.}\xspace the 
size of retry loops follow a probability distribution) in dynamic environments 
(\textit{i.e.}\xspace the size of application specific follows a probability distribution). 
These analyses allow us to estimate the energy consumption of an extensive set of 
lock-free data structures that are used under various access patterns.
\item We introduces a new general energy model ICE for analyzing the energy complexity of a wide range of multi-threaded algorithms. Compared to the EPEM model reported in D2.3, this model proposed using Ideal Cache memory model to compute I/O complexity of the algorithms. Besides a case study of SpMV to demonstrate how to apply the ICE model to find energy complexity of parallel algorithms, Deliverable D2.4 also reports a case study to apply the ICE model to Dense Matrix Multiplication (matmul). The model is then validated with both data-intensive (i.e., SpMV) and computation-intensive (i.e., matmul) algorithms according to three aspects: different algorithms, different input types/sizes and different platforms. In order to make the reading flow easy to follow, we include in this report a complete study of ICE model along with latest results.
\end{itemize}





\newpage

\tableofcontents

\newpage


\section{Introduction} 
\label{sec:introduction}

D2.4 reports the final prototype of programming abstraction based on the results from Task 2.1 to 2.4, including: i) the latest results of Task 2.1 on investigating and modeling the trade-off between energy and performance of concurrent data structures and algorithms \cite{HaTUTGRWA:2014} ii) the improved results of Task 2.2 on providing essential concurrent data structures and algorithms for inter-process communication \cite{HaTUAGRT15} and iii) the additional results of Task 2.3 on developing novel concurrent data structures and Task 2.4 on memory-access algorithms that are locality- and heterogeneity-aware \cite{HaTUAGRTW16}. 
The detailed studies (including their motivation, contributions and current results) of D2.4 are introduced in the followings subsections.
\subsection{Energy-efficient and Concurrent Data Structures and Algorithms}
Like other fundamental abstractions for energy-efficient computing,
search trees need to support both high concurrency and 
fine-grained data locality.
However, existing locality-aware search trees such as ones based on the van 
Emde Boas layout (vEB-based trees), poorly support {\em concurrent} (update) 
operations while existing highly-concurrent search 
trees such as the non-blocking 
binary search trees do not consider data locality.

We present GreenBST, a
practical energy-efficient concurrent search tree
that supports fine-grained data locality as
vEB-based trees do, but unlike vEB-based trees, 
GreenBST supports high concurrency.
GreenBST is a $k$-ary leaf-oriented tree of GNodes where each
GNode is a fixed size tree-container with the van Emde Boas layout.
As a result, GreenBST minimizes data transfer between memory 
levels while supporting highly concurrent (update) operations. 
Our experimental evaluation using the recent implementation of non-blocking binary search trees, 
highly concurrent B-trees, conventional vEB trees, as well as the 
portably scalable concurrent trees shows that GreenBST is efficient:
its energy efficiency (in operations/Joule) and throughput (in operations/second) are up to 65\% and 69\% 
higher, respectively, than the other trees on a high performance computing (HPC) platform 
(Intel Xeon), an embedded platform (ARM), and an accelerator platform (Intel Xeon Phi).
The results also provide insights into how to develop energy-efficient data structures in general.

\subsection{Customization methodology for implementation of streaming aggregation in embedded systems}
Streaming aggregation is a fundamental operation in the area of stream processing and its implementation provides various challenges. Data flow management is traditionally performed by high performance computing systems. However, nowadays there is a trend of implementing streaming operators in low power embedded devices, due to the fact that they often provide increased performance per watt in comparison with traditional high performance systems. In this work, we present a methodology for the customization of streaming aggregation implemented in modern low power embedded devices. The methodology is based on design space exploration and provides a set of customized implementations that can be used by developers to perform trade-offs between throughput, latency, memory and energy consumption. We compare the proposed embedded system implementations of the streaming aggregation operator with the corresponding HPC and GPGPU implementations in terms of performance per watt. Our results show that the implementations based on low power embedded systems provide up to 54 and 14 times higher performance per watt than the corresponding Intel Xeon and Radeon HD 6450 implementations, respectively. 

\subsection{Energy Model on CPU for Lock-free Data-structures in Dynamic Environments}
In this section, we firstly consider the modeling and the analysis of the performance 
of lock-free data structures. Then, we combine the perfomance analysis with our power 
model that is introduced in D2.1~\cite{EXCESS:D2.1} and D2.3~\cite{EXCESS:D2.3} to estimate 
the energy efficiency of lock-free data structures that are used in various settings. 
 
Lock-free data structures are based on retry loops and are called by application-specific routines. 
In contrast to the model and analysis provided in D2.3, we consider 
here the lock-free data structures in dynamic environments. The size of each of the retry loops, 
and the size of the application routines invoked in between, are not constant 
but may change dynamically. 

We present two analytical frameworks for calculating the performance of lock-free
data structures. The new frameworks follow two different approaches. The first framework, 
the simplest one, is based on queuing theory.
It introduces an average-based approach that facilitates a more coarse-grained analysis, 
with the benefit of being ignorant of size distributions. Because of this 
independence from the distribution nature it covers a set of complicated designs. 
The second approach, instantiated with an exponential distribution for the size 
of the application routines, uses Markov chains, and is tighter because it constructs 
stochastically the execution, step by step.

Both frameworks provide a performance estimate which is close to what we observe 
in practice. We have validated our analysis on (i) several fundamental lock-free 
data structures such as stacks, queues, deques and counters, some of them employing 
dynamic helping mechanisms, and (ii) synthetic tests covering a wide range of 
possible lock-free designs. We show the applicability of our results by introducing 
new back-off mechanisms, tested in application contexts, and by designing an efficient 
memory management scheme that typical lock-free algorithms can utilize. Finally, we
reveal how these results can be used to obtain the energy consumption of the lock-free data 
structures.

\label{energy-models}


\subsection{A General and Validated Energy Complexity Model for Multithreaded Algorithms}
Like time complexity models that have significantly contributed to the analysis and development of fast algorithms, energy complexity models for parallel algorithms are desired as crucial means to develop energy efficient algorithms for ubiquitous multicore platforms. Ideal energy complexity models should be validated on real multicore platforms and applicable to a wide range of parallel algorithms. However, existing energy complexity models for parallel algorithms are either theoretical without model validation or algorithm-specific without ability to analyze energy complexity for a wide-range of parallel algorithms.  

This paper presents a new general validated energy complexity model for parallel (multithreaded) algorithms. The new model abstracts away possible multicore platforms by their static and dynamic energy of computational operations and data access, and derives the energy complexity of a given algorithm from its {\em work}, {\em span} and {\em I/O} complexity. 
The new model is validated by different sparse matrix vector multiplication (SpMV) algorithms and dense matrix multiplication (matmul) algorithms running on high performance computing (HPC) platforms (e.g., Intel Xeon and Xeon Phi). The new energy complexity model is able to characterize and compare the energy consumption of SpMV and matmul kernels according to three aspects: different algorithms, different input matrix types and different platforms. The prediction of the new model regarding which algorithm consumes more energy with different inputs on different platforms, is confirmed by the experimental results. In order to improve the usability and accuracy of the new model for a wide range of platforms, the platform parameters of ICE model are provided for eleven platforms including HPC, accelerator and embedded platforms.
\label{uit-energy-motiv}

\newpage
\section{Libraries of Energy-efficient and Concurrent Data Structures}  \label{sec:Concurrent-Data-Structures}
In this section, we describe the cache-oblivious abstraction that is used in developing our energy-efficient 
and concurrent data structures. The inclusion of the cache-oblivious abstraction that is previously
described in the D2.2 is intended to help the readers to fully understand the methodology that 
is used for promoting energy-efficiency in data structures (cf. Section \ref{sec:wb-method}). 
The section continues with the detailed description of GreenBST, an energy-efficient concurrent search tree (cf.
Section \ref{sec:search-trees} and \ref{sec:hGBST}). In contrast to the D2.3, GreenBST in this
deliverable is presented with more details, emphasizing on its complete structure and concurrency control. 
The section concludes with the experimental results of the 
developed libraries of concurrent data structure (cf. Section \ref{sec:evaluation}). 
We add several state-of-the-art trees that are not included in D2.3 in the energy efficiency and 
throughput comparison of the concurrent data structure libraries.

\subsection{Cache-oblivious Abstraction}  \label{sec:wb-method}
Energy efficiency is one of the most important factors in 
designing high performance systems.
As a result, data must be organized and accessed  in an energy-efficient manner through novel fundamental data structures and algorithms that strive for the energy limit.
Unlike conventional locality-aware algorithms that only concern
about whether the data is on-chip (e.g., cache) or not (e.g., DRAM),
new energy-efficient data structures and algorithms must consider data locality
in finer-granularity: \textit{where on chip the data is}. Dally \cite{Dally11} 
predicted that for chips using the 10nm technology, the energy required between accessing data in nearby
on-chip memory and accessing data across the chip will differ as much as 75x 
(2pJ versus 150pJ), whereas the energy required 
between accessing the on-chip data
and accessing the off-chip data will only differ by 2x (150pJ versus 300pJ). Therefore, 
in order to construct energy efficient software systems, data
structures and algorithms must support not only high parallelism but also fine-grained data locality~\cite{Dally11}.


In order to devise locality-aware algorithms, 
we need theoretical execution models that promote data locality. 
One example of such models is the the cache-oblivious (CO) models 
\cite{Frigo:1999:CA:795665.796479}, 
which enable the analysis of data transfer between two levels of the memory hierarchy.
CO models are using the same analysis as 
the widely known I/O models \cite{AggarwalV88} 
except in CO models an optimal replacement is assumed.
Lower data transfer complexity implies better data locality 
and higher energy efficiency as energy consumption caused by data 
transfer dominates the total energy consumption \cite{Dally11}.
These models require the knowledge of the algorithm and some parameters
of the architecture to be known beforehand, hence they are
white-box methods.

The cache-oblivious (CO) models (cf. Section \ref{sec:cache-oblivious}) support not only fine-grained data locality but also portability. A CO algorithm that is optimized for 2-level memory, is asymptotically optimized for unknown multilevel memory (e.g., register, L1C, L2C, ..., LLC, memory), enabling fine-grained data locality (e.g., minimizing data movement between L1C and L2C). As cache sizes and block sizes in the CO models are unknown, CO algorithms are expected to be portable across different systems. For example, the memory transfer cost of an algorithm (e.g., how many data blocks need to be transferred between two level of memory), which is analyzed using the CO model, will be applicable on both HPC machines and embedded platforms (e.g., Myriad1/2 platforms), irrespective of the variations in the hardware parameters such as memory hierarchy, specifications and sizes. The performance portability is useful for analyzing the data movement and energy consumption of an algorithm in a platform-independent manner.

The memory transfer cost of an algorithm obtained using the CO model can be regarded as a first piece of information that can enable software designers to rapidly analyze the performance and energy consumption of their algorithms. After all, memory transfer is one of the parameters that dominate the total energy consumption. As for the next step, the transfer cost can be fed directly into the energy model of a specific platform to get a good approximation on the energy consumption of the algorithm on the platform. 

Algorithms and data structures analyzed using the 
cache-oblivious models \cite{Frigo:1999:CA:795665.796479} are 
found to be cache-efficient and disk-efficient \cite{Brodal:2004aa, Demaine:2002aa}, making them suitable
for improving energy efficiency in modern high performance systems. Nowadays, multilevel memory
hierarchies in commodity systems are becoming more prominent as modern CPUs tend to have at least 3 level of caches and disks start to incorporate hybrid-SSD cache memories. With minimal effort, cache-oblivious algorithms
are expected to be always locality-optimized irrespective of 
variations in memory hierarchies, enabling less data transfers
between memory levels that directly translate into runtime energy savings.

Since their inception, cache-oblivious models have been extensively used 
for designing locality-aware fundamental algorithms and data 
structures \cite{Brodal:2004aa, Demaine:2002aa, Fagerberg:2008aa}. Among those algorithms
are scanning algorithms (e.g., traversals, aggregates, and array reversals), divide 
and conquer algorithms (e.g., median selection, and matrix multiplication), and sorting
algorithms (e.g., mergesort and funnel-sort \cite{Frigo:1999:CA:795665.796479}). Several static data structures (e.g., static search trees, and funnels) 
and dynamic data structures (e.g., ordered files, b-trees, priority queues, and linked-list) have
been also analyzed using the cache-oblivious models. Performance of the said cache-oblivious 
algorithms and data structures have been reported similar to or sometimes better than the performance of their traditional  cache-aware counterparts.




\subsubsection{I/O model.} \label{sec:iomodel}
The I/O\footnote{The term "I/O" is from now on used a shorthand for block I/O operations} model was introduced
by Aggarwal and Vitter \cite{AggarwalV88}. In their seminal paper, Aggarwal and Vitter
postulated that the memory hierarchy consists of two levels, an internal memory with size
$M$ (e.g., DRAM) and an external storage of infinite size (e.g., disks). Data is transferred 
in $B$-sized blocks between those two levels of memory and the CPU can only access 
data that are available in the internal memory. In the I/O model, an algorithm's time complexity 
is assumed to be dominated by how many block transfers are required, 
as loading data from disk to memory 
takes much more time than processing the data.

For this I/O model, B-tree \cite{Bayer:1972aa} is an optimal search tree \cite{CormenSRL01}.
B-trees and its concurrent variants \cite{BraginskyP12, Comer79, Graefe:2010:SBL:1806907.1806908,
Graefe:2011:MBT:2185841.2185842}
are optimized for a known memory block size $B$ (e.g., page size) to minimize the
number of memory blocks accessed by the CPU during a search, thereby improving data
locality. The I/O transfer complexity of B-tree is $O(\log_B N)$, the optimal.

However, the I/O model has its drawbacks. 
Firstly, to use this model, an algorithm has to know the $B$ and $M$ (memory size) 
parameters in advance.
The problem is that these parameters are sometimes unknown (e.g., when memory is shared with other applications)
and most importantly not portable between different platforms.
Secondly, in reality there are different block sizes at different levels of the memory hierarchy
that can be used in the design of locality-aware data layout for search trees. 
For example in \cite{KimCSSNKLBD10, Sewall:2011aa}, 
Intel engineers have come out with very fast search trees by crafting  
a platform-dependent data layout based on the register size, SIMD width, cache line size,
and page size. 

Existing B-trees limit spatial locality optimization to the memory level
with block size $B$, leaving access to other memory levels with different block size
unoptimized.
For example a traditional B-tree that is optimized for searching data in disks (i.e.,
$B$ is page size), where each node is an array of sorted keys, 
is optimal for transfers between a disk and RAM.
However, data transfers between RAM and last level cache (LLC) 
are no longer optimal.
For searching a key inside each $B$-sized block in RAM, the transfer complexity 
is $\Theta (\log (B/L))$ transfers between RAM and LLC, 
where $L$ is the cache line size.
Note that a search with optimal cache line transfers of $O(\log_L B)$ is achievable
by using the van Emde Boas layout \cite{BrodalFJ02}. This layout has been proved 
to be optimal for search using the cache-oblivious model \cite{Frigo:1999:CA:795665.796479}. 

\subsubsection{Cache-oblivious model} \label{sec:cache-oblivious}
The cache-oblivious model was introduced 
by Frigo et al. in \cite{Frigo:1999:CA:795665.796479},
which is similar to the I/O model except that the block size $B$ and memory size
$M$ are unknown.
Using the same
analysis of the Aggarwal and Vitter's two-level I/O model, an algorithm is categorized as 
\textit{cache-oblivious} if it has no variables that need to be tuned with respect to 
hardware parameters, 
such as cache size and cache-line length in order to achieve optimality, 
assuming that I/Os are performed by an optimal off-line cache replacement strategy.

If a cache-oblivious algorithm is optimal for arbitrary two-level memory, the
algorithm is also optimal for any adjacent pair of available levels of the memory hierarchy. 
Therefore without knowing anything about memory level hierarchy and the size of each level, a cache-oblivious
algorithm can automatically adapt to multiple levels of the memory hierarchy.
In \cite{Brodal:2004aa}, cache-oblivious algorithms were reported performing better 
on multiple levels of memory hierarchy and 
more robust despite changes in memory size parameters compared to the cache-aware
algorithms. 

One simple example is that in the cache-oblivious model, B-tree is no longer optimal
because of the unknown $B$. Instead, the van Emde Boas (vEB) layout-based trees that are 
described by Bender  
\cite{BenderDF05, BenderFFFKN07, BenderFGK05} and Brodal, 
\cite{BrodalFJ02}, are optimal.
We would like to refer the readers to \cite{Brodal:2004aa, Frigo:1999:CA:795665.796479} 
for a more comprehensive overview of the I/O model and cache-oblivious model.

We provide some of the examples of cache-oblivious algorithms 
and cache oblivious data structures in the following texts.

\subsubsection{Cache-oblivious Algorithms}  \label{sec:COM-coalg}

\paragraph{Scanning algorithms and their derivatives} 
One example of a naive cache-oblivious (CO) algorithm is the 
\textit{linear scanning} of an $N$ element array that requires $\Theta(N/B)$ I/Os or transfers.
Bentley's \textit{array reversal algorithm} and Blum's \textit{linear time selection algorithm}
are primarily based on the scanning algorithm, therefore they
also perform in $\Theta(N/B)$ I/Os \cite{Brodal:2004aa, Demaine:2002aa}.

\paragraph{Divide and conquer algorithms.} Another example of CO algorithms
in divide and conquer algorithms is the matrix operation algorithms.
Frigo et al. proved that \textit{transposition} of an $n \times m$ matrix 
was optimally solved in $\mathcal{O}(mn/B)$ I/Os and 
the \textit{multiplication} of an $m\times n$-matrix and an $n \times p$-matrix 
was solved using $\mathcal{O}((mn + np + mp)/B + mnp/(B\sqrt{M}))$ 
I/Os, where $M$ is the memory size \cite{Frigo:1999:CA:795665.796479}. As for square matrices (e.g., $N \times N$), using the Strassen's algorithm
and the cache-oblivious model,
the required I/O bound has been proved to be $O(N^2/B + N^{\lg7}/B\sqrt{M})$.

\begin{figure}[t]
\centering \scalebox{0.7}{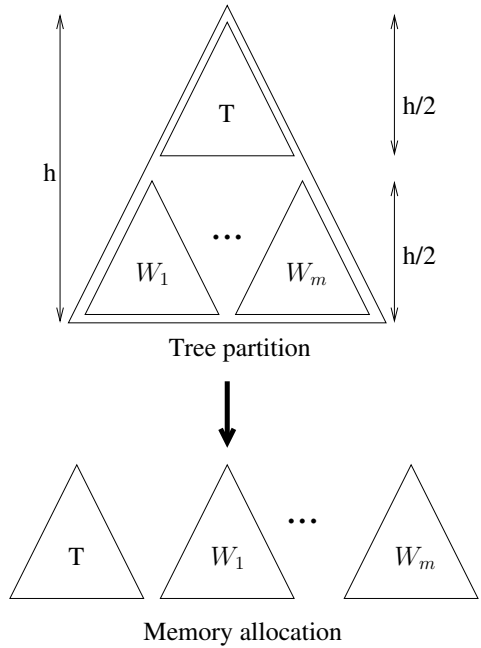}
\caption{Static van Emde Boas (vEB) layout: a tree of height $h$ is recursively split at height $h/2$. 
The top subtree $T$ of height $h/2$ and $m=2^{h/2}$ bottom subtrees $W_1;W_2; \ldots ;W_m$ 
of height $h/2$ are located in contiguous memory locations where T is located before  
$W_1;W_2;\ldots;W_m$.}\label{fig:vEB}
\end{figure}

\paragraph{Sorting algorithms.} Demaine gave two examples of cache-oblivious 
sorting algorithm in his brief survey paper \cite{Demaine:2002aa}, namely the
\textit{mergesort} and \textit{funnelsort} \cite{Frigo:1999:CA:795665.796479}. 
In the same text he also wrote that both
sorting algorithms achieved the optimal $\Theta(\frac{N}{B} \log_2 \frac{N}{B})$ I/Os,
matching those in the original analysis of Aggarwal and Vitter \cite{AggarwalV88}.

\subsubsection{Cache-oblivious Data Structures}  \label{sec:COM-cods}

\paragraph{Static data structures}

\begin{figure}[t]
\centering  \includegraphics[width=\columnwidth]{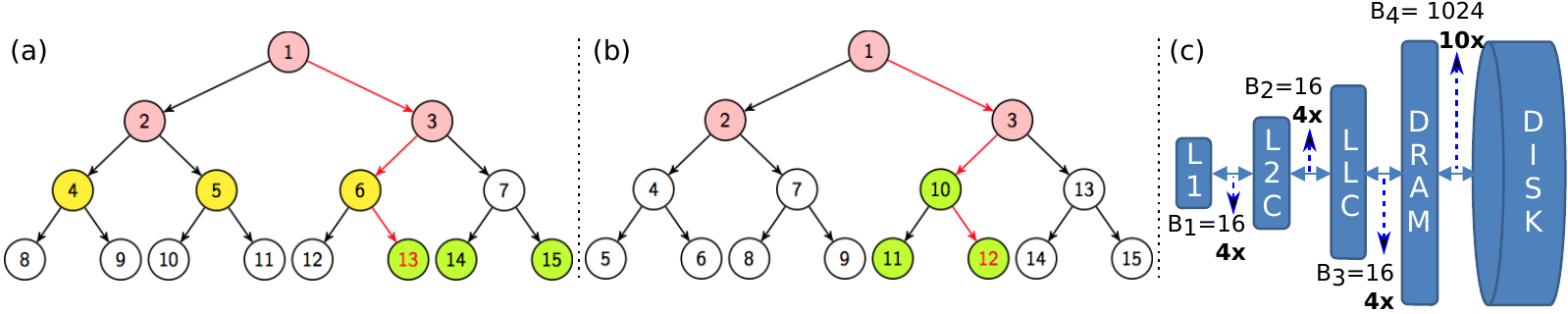}
\caption{Illustration of the required data block transfer in searching for (a) key 13 in BFS tree and 
(b) key 12 in vEB tree, where a node's value is {\em its address in the physical memory}. Note
that in (b), adjacent nodes are grouped together (e.g., (1,2,3) and (10,11,12)) because of the 
{\em recursive} tree building. The similarly colored nodes indicates a single block transfer $B$.
An example of multi-level memory is shown in (c), where $B_x$ is the {\em block transfer} size 
$B$ between levels of memory.}\label{fig:vEB-bfs}
\end{figure}

One of the examples of cache-oblivious (CO) static data structures is the \textit{CO search trees} that 
can be achieved using the van Emde Boas (vEB) layout \cite{Prokop99, vanEmdeBoas:1975:POF:1382429.1382477}. 
The vEB-based trees recursively arrange related data in contiguous memory locations, 
minimizing data transfer between any two adjacent levels of the memory hierarchy (cf. Figure \ref{fig:vEB}). 

Figure \ref{fig:vEB-bfs} illustrates the vEB layout, where the size $B$ of memory blocks transferred  
between 2-level memory in the I/O model \cite{AggarwalV88} is 3 (cf. Section \ref{sec:iomodel}). Traversing a complete binary tree with the Breadth First Search layout (or BFS tree for short) (cf. Figure \ref{fig:vEB-bfs}a) with height 4 will need three memory transfers to locate the key at leaf-node 13. The first two levels with
three nodes $(1, 2, 3)$ fit within a single block transfer while the next two levels need to be loaded
in two separate block transfers that contain nodes $(4, 5, 6)$ and nodes $(13, 14, 15)$, respectively. Generally, the number of memory transfers for a BFS tree of size $N$  is $(\log_2N-\log_2B) = \log_2 N/B \approx \log_2 N$ for $N \gg B$. 

For a vEB tree with the same height, the required memory transfers is only 
two. As shown in Figure \ref{fig:vEB-bfs}b, locating the key in leaf-node 12 requires only a transfer of nodes $(1, 2, 3)$ 
followed by a transfer of nodes $(10, 11, 12)$. Generally,  the memory transfer complexity for searching for a key in a tree of size $N$ is now reduced to $\frac{\log_2N}{\log_2B} = \log_B N$, simply by using 
an efficient tree layout so that nearby nodes are located in adjacent memory locations. 
If $B=1024$, searching a BFS tree for a key at a leaf requires 10x (or $\log_2 B$) more I/Os than searching a vEB tree with the same size $N$ where $N \gg B$. 

On commodity machines with multi-level memory, 
the vEB layout is even more efficient. 
So far the vEB layout is shown to have $\log_2B$ less I/Os for two-level memory. 
In a typical machine having three 
levels of cache (with cache line size of 64B), a RAM (with page size of 4KB) and a disk, 
searching a vEB tree can achieve up to 640x less I/Os than searching a BFS tree, assuming the node size is 4 bytes (Figure \ref{fig:vEB-bfs}c).

\paragraph{Dynamic data structures.}
In a standard \textit{linked-list} structure supporting traversals, insertions and deletions, 
the best-known cache-oblivious solution was $\mathcal{O}((\lg^2 N)/B)$ I/Os for updates
and $\mathcal{O}(K/B)$ for traversing $K$ elements in the list  \cite{Demaine:2002aa}.

The first cache-oblivious \textit{priority queue} was due to Arge et al. 
\cite{Arge:2002:CPQ:509907.509950} and it supports inserts and delete-min operations in $\mathcal{O}( {^1/_B} \log_{M/B} {^N/_B})$ I/Os.

The vEB layout in static cache-oblivious search tree has inspired many cache-oblivious \textit{dynamic search trees} such as 
cache-oblivious B-trees \cite{BenderDF05, BenderFFFKN07, BenderFGK05} and cache-oblivious binary trees \cite{BrodalFJ02}.
All of these search tree implementations have been proved having the optimal bounds of
$\mathcal{O}(\log_B N)$ in searches and require amortized $\mathcal{O}(\log_B N)$ I/Os for updates. 

However, vEB-based trees poorly support {\em concurrent} update operations. 
Inserting or deleting a node may result in 
relocating a large part of the tree in order to maintain 
the vEB layout (cf. Section \ref{subsec:staticveb}). Bender et al. 
\cite{BenderFGK05} discussed the problem
and provided important theoretical designs of concurrent vEB-based B-trees.
Nevertheless, we have found that the theoretical designs are not very efficient in practice 
due to the actual overhead of maintaining necessary pointers as well as their large memory footprint.

\subsubsection{New Relaxed Cache-oblivious Model}  \label{sec:COM-rco}

We observe that is unnecessary to keep a vEB-based tree in a contiguous block of memory whose size is greater
than some upper bound. In fact, allocating a contiguous block of memory for a 
vEB-based tree does not guarantee a contiguous block of
\textit{physical memory}. Modern OSes and systems utilize
different sizes of continuous physical memory blocks,  for example, in the form
of pages and cache-lines. A contiguous block in virtual
memory might be translated into several blocks with gaps
in RAM; also, a page might be cached by several cache lines with gaps at any level of cache. 
This is one of the motivations for the new relaxed cache oblivious model proposed. 

We define {\em relaxed cache oblivious} algorithms to be cache-oblivious (CO)
algorithms with the restriction that an upper bound $\UB$ on the unknown memory
block size $B$ is known in advance. 
As long as an upper bound on all the block
sizes of multilevel memory is known, the new relaxed CO model maintains the key
feature of the original CO model \cite{Frigo:1999:CA:795665.796479}. 
First, temporal locality is exploited perfectly as there are no constraints on cache size
$M$ in the model. As a result, an optimal offline cache replacement policy can be
assumed. In practice, the Least Recently Used (LRU) policy with memory of size 
$(1+\epsilon)M$, where $\epsilon>0$, is nearly as good as the optimal replacement policy
with memory of size $M$ \cite{Sleator:1985:AEL:2786.2793}.
Second, analysis for a simple two-level memory
are applicable for an unknown multilevel memory (e.g., registers, L1/L2/L3 caches
and memory). Namely, an algorithm that is optimal in terms of data movement for a 
simple two-level memory is asymptotically optimal for an unknown multilevel memory. 
This feature enables algorithm designs that can utilize fine-grained data locality 
in the multilevel memory hierarchy of modern architectures. 

The upper bound on the contiguous block size can be obtained easily from 
any system (e.g., page-size or any values greater than that), which is platform-independent. In fact, the search performance in  
the new relaxed cache oblivious model is resilient to different upper bound values (cf. Lemma \ref{lem:dynamic_vEB_search} in Section \ref{sec:concurrentvEB}).

\subsubsection{New Concurrency-aware van Emde Boas Layout} \label{sec:concurrentvEB}

\begin{figure}[t] 
\centering 
\scalebox{0.6}{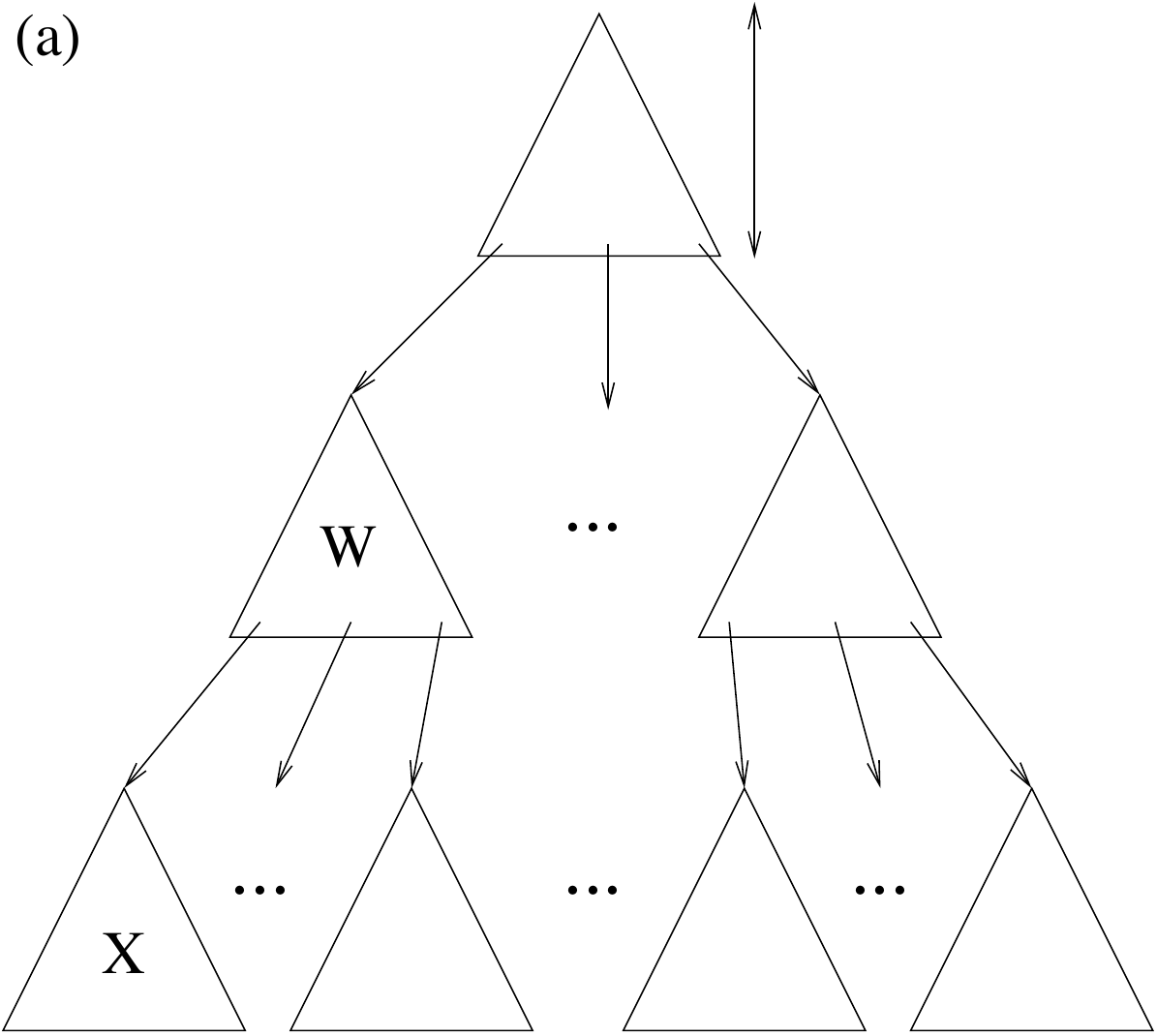}
\scalebox{0.6}{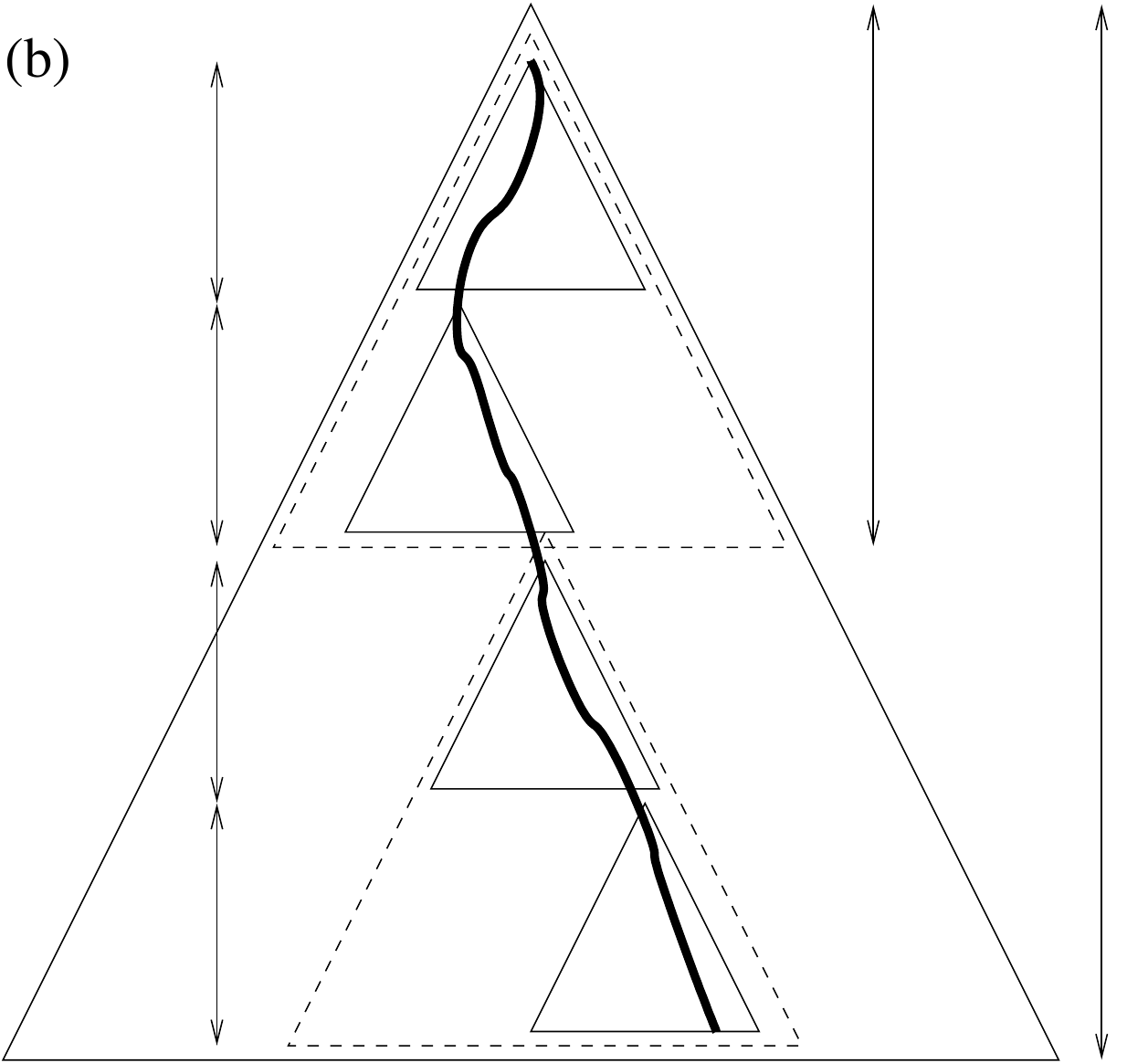}
\caption{{\em(a)} New concurrency-aware vEB layout. {\em(b)} 
Search using concurrency-aware vEB layout.}\label{fig:search_complexity}\label{fig:dynamicVEB}
\end{figure}

We propose improvements to the conventional van Emde Boas (vEB) layout
to support high performance and high concurrency, which results in new {\em concurrency-aware} dynamic vEB layout. 
We first define the following notations that will be used to elaborate on the improvements:
\begin{itemize}

\item $b_i$ (unknown): block size in terms of the number of nodes at level $i$ of the memory
hierarchy (like $B$ in the I/O model \cite{AggarwalV88}), which is unknown as in the cache-oblivious 
model \cite{Frigo:1999:CA:795665.796479}. When the specific level $i$ of the memory
hierarchy is irrelevant, we use notation $B$ instead of $b_i$ in order to be
consistent with the I/O model.

\item $\UB$ (known): the upper bound (in terms of the number of nodes) on the
block size $b_i$ of all levels $i$ of the memory hierarchy.

\item {\em $\Delta$Node}: the largest recursive subtree of a van Emde Boas-based 
search tree that contains at most $\UB$ nodes (cf. dashed triangles of height $2^L$ in
Figure \ref{fig:search_complexity}b). $\Delta$Node is a fixed-size tree-container
with the vEB layout.

\item "level of detail" $k$ is a partition of the tree into recursive subtrees of 
height at most $2^k$. 

\item Let $L$ be the level of detail of $\Delta$Node. Let $H$ be the height
of a $\Delta$Node, we have $H = 2^L$. For simplicity, we assume $H = \log_2
(\UB+1)$.

\item $N, T$: size and height of the whole tree in terms of basic nodes (not in
terms of $\Delta$Nodes).

\end{itemize}  

\subparagraph{Conventional van Emde Boas (vEB) layout.} \label{subsec:staticveb}

The conventional van Emde Boas (vEB) layout has been introduced in
cache-oblivious data structures \cite{BenderDF05, BenderFFFKN07, BenderFGK05, BrodalFJ02,
Frigo:1999:CA:795665.796479}. Figure \ref{fig:vEB} illustrates the vEB layout.
Suppose we have a complete binary tree with height $h$. For simplicity, we
assume $h$ is a power of 2, i.e., $h=2^k, k \in \mathbb{N}$.
The tree is recursively laid out in the memory as follows. The tree is
conceptually split between nodes of height $h/2$ and $h/2+1$, resulting in a top
subtree $T$ and $m_1 = 2^{h/2}$ bottom subtrees $W_1, W_2, \cdots, W_{m_1}$ of
height $h/2$. The $(m_1 +1)$ top and bottom subtrees are then located in
contiguous memory locations where $T$ is located before $W_1, W_2, \cdots,
W_{m_1}$. Each of the subtrees of height $h/2$ is then laid out similarly to $(m_2 +
1)$ subtrees of height $h/4$, where $m_2 = 2^{h/4}$. The process continues until
each subtree contains only one node, i.e., the finest {\em level of detail}, 0.

The main feature of the vEB layout is that the cost of any search in this layout
is $O(\log_B N)$ memory transfers, where $N$ is the tree size and $B$ is the {\em
unknown} memory block size in the cache-oblivious model \cite{Frigo:1999:CA:795665.796479}. Namely, its search is cache-oblivious. 
The search cost is the
optimal and matches the search bound of B-trees that requires the memory block
size $B$ to be {\em known in advance}. Moreover, at any level of detail, each subtree
in the vEB layout is stored in a contiguous block of memory.

Although the conventional vEB layout is helpful for utilizing data locality, it poorly
supports concurrent update operations. Inserting (or deleting) a node at
position $i$ in the contiguous block storing the tree may restructure a large
part of the tree. For example, inserting
new nodes in the full subtree $W_1$ (a leaf subtree) in Figure \ref{fig:vEB} will  affect the other
subtrees $W_2, W_3, \cdots, W_m$ by rebalancing existing nodes between $W_1$ 
and the subtrees in order to have
space for new nodes. Even worse, we will need to allocate a new contiguous block of
memory for the whole tree if the previously allocated block of memory for
the tree runs out of space \cite{BrodalFJ02}. Note that we cannot use dynamic
node allocation via pointers since at {\em any} level of detail, each subtree in the vEB layout must be stored in a {\em
contiguous} block of memory.

\subparagraph{Concurrency-aware vEB layout.} \label{sec:relaxed-veb}

In order to make the vEB layout suitable for highly concurrent data structures
with update operations, we introduce a novel {\em concurrency-aware} dynamic vEB layout. Our key
idea is that if we know an upper bound $\UB$ on the unknown memory block size
$B$, we can support dynamic node allocation via pointers while maintaining the
optimal search cost of $O(\log_B N)$ memory transfers without knowing $B$ (cf.
Lemma \ref{lem:dynamic_vEB_search}). The assumption on known upper bound $\UB$ is supported 
by the fact that in practice it is unnecessary to keep the vEB layout in
a contiguous block of memory whose size is greater than some upper bound.
         
Figure \ref{fig:dynamicVEB}a illustrates the new concurrency-aware vEB layout based on the
relaxed cache oblivious model. Let $L$ be the coarsest level of detail such that
every recursive subtree contains at most $\UB$  nodes. Namely, let $H$ and $S$ 
be the height and size of such a subtree then $H=2^L$ and $S=2^H - 1 < \UB$.
The tree is recursively
partitioned into level of detail $L$ where each subtree represented by a
triangle in Figure \ref{fig:dynamicVEB}a,  is stored in a contiguous memory block
of size $\UB$. Unlike the conventional vEB, the subtrees at level of detail $L$
are linked to each other using pointers, namely each subtree at level of detail
$k > L$ is not stored in a contiguous block of memory.  Intuitively, since $\UB$
is an upper bound on the unknown memory block size $B$, storing a subtree at
level of detail $k > L$ in a contiguous memory block of size greater than $\UB$,
does not reduce the number of memory transfers, provided there is perfect alignment. 
For example, in Figure
\ref{fig:dynamicVEB}a, traveling from a subtree $W$ at level of detail $L$, which
is stored in a contiguous memory block of size $\UB$, to its child subtree $X$ at
the same level of detail will result in at least two memory transfers: one for
$W$ and one for $X$. Therefore, it is unnecessary to store both $W$ and $X$ in a
contiguous memory block of size $2\UB$. As a result, the memory transfer cost for search operations in
the new concurrency-aware vEB layout is intuitively the same as that of the conventional vEB layout (cf. Lemma \ref{lem:dynamic_vEB_search}) while the concurrency-aware vEB supports
high concurrency with update operations.

\begin{lemma-uit}
For any upper bound $\UB$ of the {\em unknown} memory block size $B$, a search in a complete binary tree with the new concurrency-aware vEB layout achieves the optimal memory transfer 
$O(\log_B N)$, where $N$ and $B$ are the tree size and 
the {\em unknown} memory block size in the cache-oblivious model
\cite{Frigo:1999:CA:795665.796479}, respectively.
\label{lem:search_mem} \label{lem:dynamic_vEB_search}
\end{lemma-uit}
\begin{proof} (Sketch)
Figure \ref{fig:search_complexity}b illustrates the proof.  
Let $k$ be the coarsest level of detail such that every recursive subtree 
contains at most $B$ nodes. Since $B \leq \UB$, $k \leq L$, where $L$ is 
the coarsest level of detail at which every recursive subtree ($\Delta$Nodes) 
contains at most $\UB$ nodes. 
That means there are at most $2^{L-k}$ subtrees
along the search path in a $\Delta$Node and no subtree of depth $2^k$ is split 
due to the boundary of $\Delta$Nodes. Namely, triangles of height $2^k$ fit
within a dashed triangle of height $2^L$ in Figure \ref{fig:search_complexity}b. 

Because at any level of detail $i \leq L$ in the concurrency-aware vEB layout, a recursive subtree of depth $2^i$ is stored in a contiguous block 
of memory, each subtree of depth $2^k$ {\em within} a $\Delta$Node is stored in
at most 2 memory blocks of size $B$ (depending on the starting location of
the subtree in memory). Since every subtree of depth $2^k$ fits in a
$\Delta$Node (i.e.,
no subtree is stored across two $\Delta$Nodes), every subtree of depth $2^k$ is 
stored in at most 2 memory blocks of size $B$.

Since the tree has height $T$, $\lceil T / 2^k \rceil$ subtrees of depth $2^k$ 
are traversed in a search and thereby at most  $2  \lceil T / 2^k \rceil$ memory 
blocks are transferred. 

Since a subtree of height $2^{k+1}$ contains more than $B$ nodes, 
$2^{k+1} \geq \log_2 (B + 1)$, or $2^{k} \geq \frac{1}{2} \log_2 (B+ 1)$. 

We have $2^{T-1} \leq N \leq 2^T$ since the tree is a {\em complete} binary tree. 
This implies $ \log_2 N \leq T \leq \log_2 N +1$.  

Therefore, the number of memory blocks transferred in a search is 
$2  \lceil T / 2^k \rceil \leq 4 \lceil \frac{\log_2 N + 1}{\log_2 (B + 1)} \rceil 
= 4 \lceil \log_{B+1} N + \log_{B+1} 2\rceil$ $= O(\log_B N)$, where $N \geq 2$.
\end{proof}

Unlike the conventional vEB layout, the new concurrency-aware vEB layout can solve the concurrency problems that might arise if the whole tree structure must be placed in a contiguous memory allocation. For example, when a conventional vEB layout tree is full, all of the tree structure must be re-allocated into a new bigger contiguous memory; and as a result, the whole tree must be locked to ensure correct concurrent search and update operations. The concurrency-aware vEB layout supports dynamic node allocation and new containers of size $\UB$ can be appended as needed to the existing tree structure whenever the tree is full. Therefore, in the concurrency-aware vEB layout, fine-grained locks can be use as the synchronization mechanism for concurrent tree operations. 

A library of novel locality-aware and energy efficient concurrent search trees based on the new concurrency-aware vEB layout is presented in Section~\ref{sec:search-trees}. The practical information on how to use the library is available in Appendix~\ref{sec:tree-library}.

\subsection{GreenBST} \label{sec:intro}\label{sec:search-trees}
\fbox{
\begin{minipage}{0.96\textwidth}
		\small	
		\textbf{Copyright Notice:} Most material in Section \ref{sec:search-trees}, \ref{sec:hGBST}, \ref{sec:evaluation},  \ref{sec:perfeval}, and \ref{sec:conclusions} is based on the following article \cite{Umar2016}:
	
		Ibrahim Umar, Otto Anshus, and Phuong Ha. Greenbst: Energy-efficient concurrent search tree. In Proceedings of Euro-Par 2016: Parallel Processing: 22nd International Conference on Parallel and Distributed Computing, pages 502--517, 2016. DOI:  10.1007/978-3-319-43659-3\_37
\end{minipage}
}
\newline

Recent researches have suggested that the energy consumption of future 
computing systems will be dominated by the cost of data movement~\cite{Dally11, Vi1, Vi2}. 
It is predicted that for 10nm technology chips, 
the energy required between accessing data in nearby on-chip memory and 
accessing data across the chip, will differ as much as 75$\times$ (2pJ versus 150pJ), 
whereas the energy required between accessing on-chip data and accessing 
off-chip data will only differ 2$\times$ (150pJ versus 300pJ)~\cite{Dally11}. Therefore, in order to 
construct energy-efficient software systems, data structures and 
algorithms must not only be concerned with whether the data 
is on-chip (e.g., in cache) or not (e.g., in DRAM), but 
must consider also data locality in {\em finer-granularity}: where the data is located on the chip.

Concurrent trees are fundamental data structures that are widely used in different contexts such as load-balancing \cite{DellaS00, HaPT07, ShavitA96} and searching \cite{Afek:2012:CPC:2427873.2427875, BronsonCCO10, Brown:2011:NKS:2183536.2183551, Crain:2012:SBS:2145816.2145837, DiceSS2006, EllenFRB10}. 
Concurrent search trees are crucial data structures that are 
widely used as a backend in many important systems such as 
databases (e.g., SQLite \cite{sqlite}), filesystems (e.g., 
Btrfs \cite{Rodeh:2008:BSC:1326542.1326544}), 
and schedulers (e.g., Linux's Completely Fair Scheduler (CFS)), 
among others.
These important systems can access and organize 
data in a more energy efficient manner by adopting the 
energy-efficient concurrent search trees as their backend structures.

Devising fine-grained data locality layout 
for concurrent search trees is challenging, mainly 
because of the trade-offs needed: (i) a platform-specific locality optimization 
might not be {\em portable} (i.e., not work on different platforms while 
there are big interests of concurrent data structures for unconventional platforms~\cite{Ha:2010aa,Ha:2012aa}),
(ii) the usage of transactional memory~\cite{Herlihy:1993aa,Ha:2009aa} and 
multi-word synchronization~\cite{Ha:2005aa,Ha:2003aa,Larsson:2004aa} complicates locality 
because each core in a CPU needs to consistently 
track read and write operations that are performed by the other cores, and (iii) 
fine-grained locality-aware layouts (e.g., van Emde Boas layout) poorly support concurrent update operations. 
Some of the fine-grained locality-aware search trees such as
Intel Fast \cite{KimCSSNKLBD10} and Palm \cite{Sewall:2011aa}
are optimized for a specific platform. Concurrent B-trees (e.g., B-link tree 
\cite{Lehman:1981:ELC:319628.319663}) only perform well if their $B$ size is optimal.
Highly concurrent search trees such as non-blocking concurrent 
search trees~\cite{EllenFRB10,Natarajan:2014:FCL:2692916.2555256} 
and Software Transactional Memory (STM)-based search 
trees~\cite{Afek:2012:CPC:2427873.2427875,Crain:2012:SBS:2145816.2145837}, 
however, do not take into account fine-grained data locality.

Fine-grained data locality for {\em sequential} search trees 
can be theoretically achieved using the van Emde Boas (vEB) 
layout \cite{Prokop99,vanEmdeBoas:1975:POF:1382429.1382477},
which is analyzed using cache-oblivious (CO) models \cite{Frigo:1999:CA:795665.796479}.
An algorithm is categorized as {\em cache-oblivious} for a two-level memory hierarchy if it has no variables that 
need to be tuned with respect to cache size 
and cache-line length, in order to optimize its data transfer complexity, 
assuming that the optimal off-line cache replacement strategy is used.
If a {\em cache-oblivious} algorithm is optimal for an arbitrary two-level memory,
the algorithm is also asymptotically optimal for any adjacent pair 
of available levels of the memory hierarchy 
\cite{Brodal:2004aa}.
Therefore, cache-oblivious algorithms
are expected to be locality-optimized irrespective of 
variations in memory hierarchies, enabling less data transfer
between memory levels and thereby saving energy.
 
However, the throughput of a vEB-based tree when doing
{\em concurrent} updates is lower compared to when it is 
doing {\em sequential} updates. 
Inserting or deleting a node may result in 
relocating a large part of the tree in order to maintain 
the vEB layout. Solutions to this problem have been 
proposed~\cite{BenderFGK05}.
The first proposed solution's structure 
requires each node to have parent-child pointers.
Update operations may result in updating the pointers. Pointers 
will also increase the tree memory footprint.
The second proposed solution uses the exponential tree algorithm~\cite{548472}.
Although the exponential tree is an important theoretical
breakthrough, it is complex~\cite{CormenSRL01}. 
The exponential tree grows exponentially in size, which not only complicates 
maintaining its inter-node pointers, but also exponentially 
increases the tree's memory footprint.
Recently, we have proposed a {\em concurrency-aware vEB layout}~\cite{Umar:2015, deltatreeTR2013}, which has a higher throughput
when doing {\em concurrent} updates compared to when it is 
doing {\em sequential} updates.
In the same study, we have proposed DeltaTree,
a B+tree that uses the concurrency-aware vEB layout. 
We have documented that the concurrency-aware vEB layout 
can improve DeltaTree's
{\em concurrent} search and update throughput over 
a concurrent B+tree~\cite{Umar:2015}. 

\begin{figure}[t]
\centering 
\includegraphics[width=0.4\linewidth]{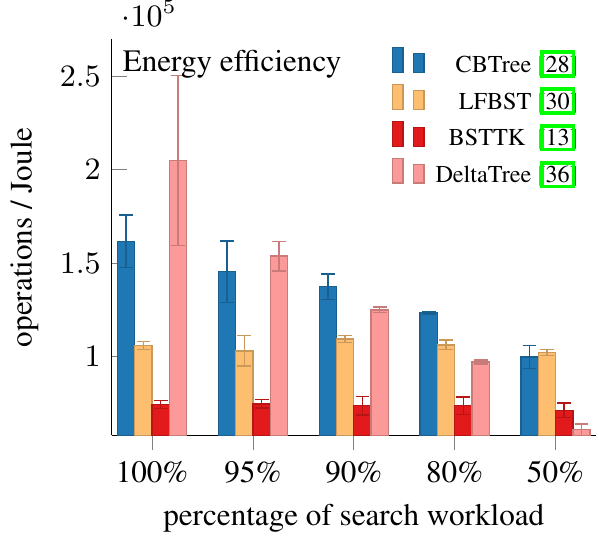}
\caption{Result of 5 million tree operations of decreasing search percentage workloads using 12 cores (1 CPU). DeltaTree's energy efficiency and throughput are lower than the other concurrent search trees after 95\% search workload on a dual Intel Xeon E5-2650Lv3 CPU system with 64GB RAM.}\label{fig:intro-comp}
\end{figure}

Nevertheless, we find 
DeltaTree's throughput and energy efficiency are lower
than the state-of-the-art concurrent search 
trees (e.g., the portably scalable
search tree~\cite{David:2015:ACS:2694344.2694359}) 
for the update-intensive workloads (cf. Figure \ref{fig:intro-comp}).
Our investigation reveals that the cost of 
DeltaTree's runtime maintenance 
(i.e., rebalancing the nodes) dominates the execution time.
However, reducing the frequency of the runtime maintenance 
lowers DeltaTree's energy efficiency and 
throughput for the search-intensive workloads,
because DeltaTree nodes will then be sparsely populated and frequently
imbalanced.
Note that DeltaTree energy efficiency and throughput are already
optimized for the search intensive 
workloads~\cite{Umar:2015,GBST-PPOPP2016}.  

In this section, we present \textit{GreenBST}, an energy-efficient 
concurrent search tree that is more energy efficient  
and has higher throughput for both the concurrent 
search- and update-intensive workloads than the other concurrent search 
trees (cf. Table \ref{tbl:algos}). GreenBST applies 
two significant improvements on 
DeltaTree in order to lower the cost of the tree runtime maintenance and reduce the tree memory footprint. 
First, unlike DeltaTree, GreenBST rebalances incrementally 
(i.e., fine-grained node rebalancing). 
In DeltaTree, the rebalance procedure has to 
rebalance {\em all} the keys within a node
and the frequency of rebalancing cannot be lowered as they are necessary
to keep DeltaTree in good shape (i.e., keeping DeltaTree's height low and its 
nodes are densely populated).
Incremental rebalance makes the overall cost 
of each rebalance in GreenBST lower than DeltaTree.
Second, we reduce the tree memory footprint by 
using a different layout for GreenBST's leaf nodes ({\em heterogeneous} layout). 
Reduction in the memory footprint also reduces GreenBST's data transfer, which 
consequently increases the tree's energy efficiency and throughput
in both update- and search- intensive workloads. We will show that with these
improvements, GreenBST can become up to 195\% more energy efficient than DeltaTree
(cf. Section \ref{sec:evaluation}).
 
We evaluate GreenBST's energy efficiency (in operations/Joule) and throughput (in operations/second) 
against six prominent concurrent search trees (cf. Table \ref{tbl:algos}) using a parallel micro-benchmarks 
{\em Synchrobench}~\cite{Gramoli:2015:MYE:2688500.2688501} and STAMP database 
benchmark {\em Vacation}~\cite{4636089} (cf. Section \ref{sec:evaluation}). 
We present memory and cache profile data to provide insights into what make GreenBST energy efficient
(cf. Section \ref{sec:evaluation}). We also provide insights into what are the key ingredients for 
developing energy-efficient data structures in general (cf. Section \ref{sec:perfeval}).

\begin{table}[!t]
\begin{center}
\scriptsize
\caption{List of the evaluated concurrent search tree algorithms.}
\begin{tabular}{c l l p{4.2cm} l p{1.5cm} p{1.5cm}}
\hline
\bf \# & \bf Algorithm & \bf Ref & \bf Description & \bf	Synchronization& \bf Code \mbox{authors} & \bf Data structure \\
\hline
1	&	SVEB	& \cite{BrodalFJ02}    & {\em Conventional} vEB layout search tree & global mutex & U. Aarhus & binary-tree	\\
2	&	CBTree	&\cite{Lehman:1981:ELC:319628.319663}	& Concurrent B-tree (B-link tree)&  lock-based & U. Troms\o & b+tree \\				
3	&	Citrus	&  \cite{Arbel:2014:CUR:2611462.2611471}	&	RCU-based search tree  & lock-based &  Technion &	binary tree	\\
4	&	LFBST	&\cite{Natarajan:2014:FCL:2692916.2555256} & Non-blocking binary search tree & lock free	&	UT Dallas & binary tree\\
5	&	BSTTK 	&\cite{David:2015:ACS:2694344.2694359} 	& Portably scalable concurrent search tree & lock-based & EPFL	& binary tree \\
6	&	DeltaTree	& \cite{Umar:2015} &	Locality aware concurrent search tree	&	lock-based	&	U. Troms\o 		&	b+tree\\	
7	&	\bf GreenBST&  - &	Improved locality aware concurrent search tree & lock-based &  this paper &	b+tree	\\
\hline
\end{tabular}
\end{center}
\label{tbl:algos}
\end{table}

\paragraph*{\bf Our contributions.}
Our contributions are threefold:
\begin{enumerate} \itemsep1pt \parskip0pt \parsep0pt

  \item We have devised a new {\em portable fine-grained locality-aware} concurrent search trees, 
  \textit{GreenBST} (cf. Section \ref{sec:hGBST}). GreenBST are based on our proposed concurrency-aware 
  vEB layout~\cite{Umar:2015} with the two improvements, namely the incremental node 
  rebalance and the heterogeneous node layouts. 
      
  \item We have evaluated GreenBST throughput (in operations/second) and energy efficiency 
  (in operations/Joule) with six prominent concurrent search trees (cf. Table \ref{tbl:algos}) on
  three different platforms (cf. Section \ref{sec:evaluation}).
  We show that compared to the state of the art concurrent search trees, GreenBST has the best
  energy efficiency and throughput across different platforms for most of the concurrent search- and update- intensive workloads.

  GreenBST code and evaluation benchmarks are available at: {\small \url{https://github.com/uit-agc/GreenBST}}.

  \item We have provided insights into how to develop energy-efficient data structures in general (cf. Section \ref{sec:perfeval}). 
  
\end{enumerate}

\subsection{GreenBST design overview}\label{sec:overview}\label{sec:hGBST}
We devise GreenBST based on the concurrency-aware vEB layout~\cite{Umar:2015} (cf. Section \ref{sec:relaxed-veb}), 
based on the idea that the layout has the same 
data transfer efficiency between two memory levels as the {\em conventional} 
sequential vEB layout~\cite{Prokop99,vanEmdeBoas:1975:POF:1382429.1382477}.
Therefore, theoretically, we can use the concurrency-aware layout within a {\em concurrent} 
search tree to 
minimize data movements between memory levels, 
which can eventually be a basis of an energy-efficient concurrent search tree.

GreenBST and DeltaTree is designed by devising three major
strategies, namely it uses a common GNode map instead of pointers or 
arithmetic-based implicit BST (i.e., a node's successor 
memory address is calculated {\em on the fly}) for node traversals, crafting an
efficient inter-node connection, and using balanced layouts.
In addition to the shared common traits with DeltaTree, 
GreenBST also employs two new major strategies:
(i) GreenBST uses incremental GNode rebalance and (ii) GreenBST 
uses heterogeneous GNode layouts.

\begin{figure}[!t] \centering 
\includegraphics[width=0.45\textwidth]{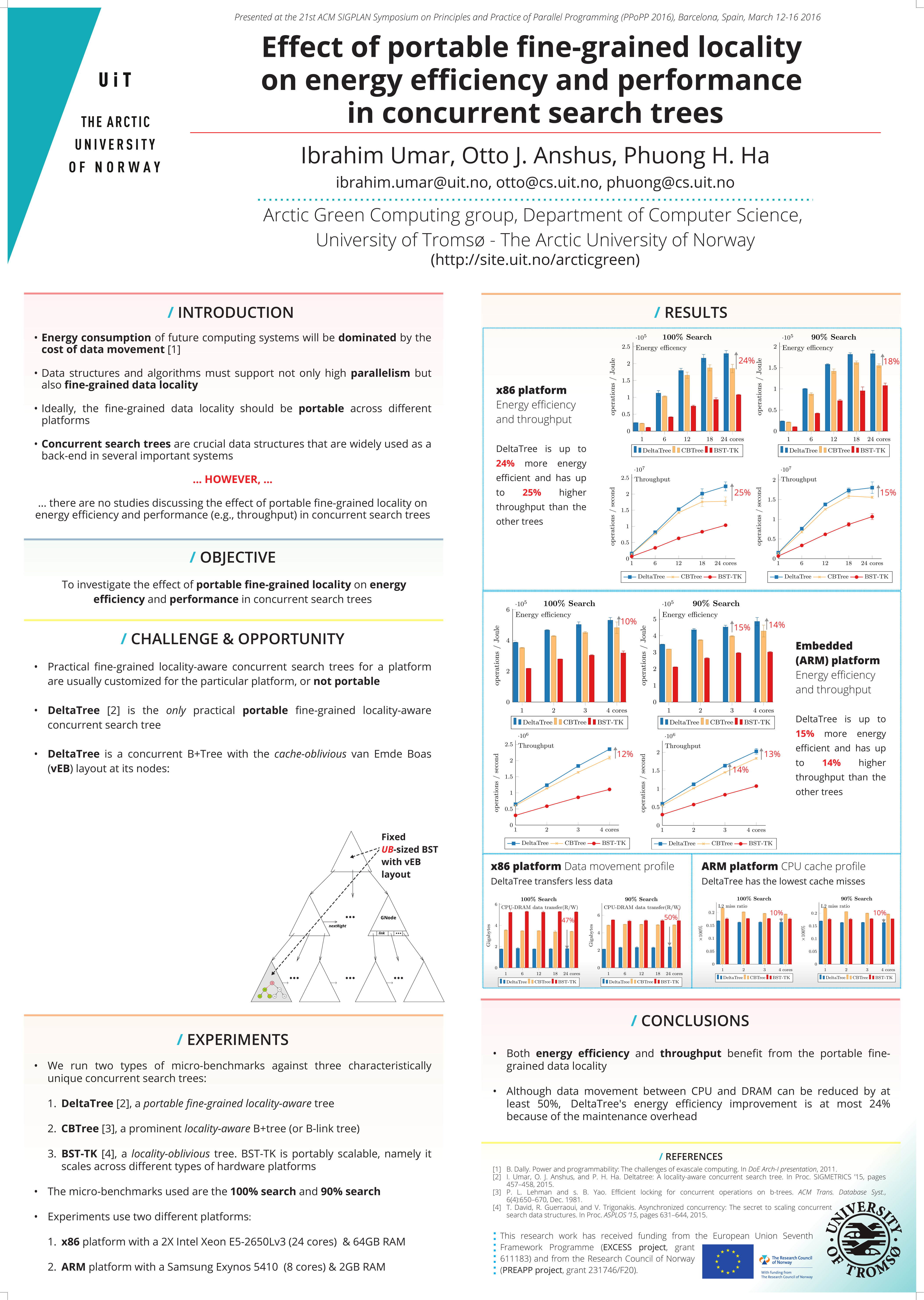}
\caption{Illustration of the GreenBST layout.}
\label{fig:gbst-layout}
\end{figure}

\subsubsection{Data structures.}

GreenBST is a collection of GNodes where each GNode
consists of $\UB$ internal \textbf{nodes} that hold the tree keys and
a $^1/_2\UB$ \textbf{link} array that links the GNode internal leaf 
nodes to another GNode's root node (cf. Figure~\ref{fig:gbst-layout}). 
The chain of GNodes formed
a B+tree (to avoid confusion, from this point onward, we refer to the "fat" nodes
of GreenBST as GNode and the GNode's
internal tree nodes as {\em internal nodes} or {\em nodes}).
Each GNode also contains  a lock (\textbf{locked}); a {\bf rev} counter that
is used for optimistic concurrency~\cite{Kung:1981aa};
\textbf{nextRight} variable, which is a pointer that points to the GNode's right sibling; 
and \textbf{highKey} variable, which contains the lowest key member of the right sibling GNode.
These last four variables are used for GreenBST concurrency control.
 
\subsubsection{Cache-resident map instead of pointers or arithmetic implicit array.}  \label{sec:mapdesc}

\begin{figure}[!t] \centering 
\includegraphics[width=0.95\textwidth]{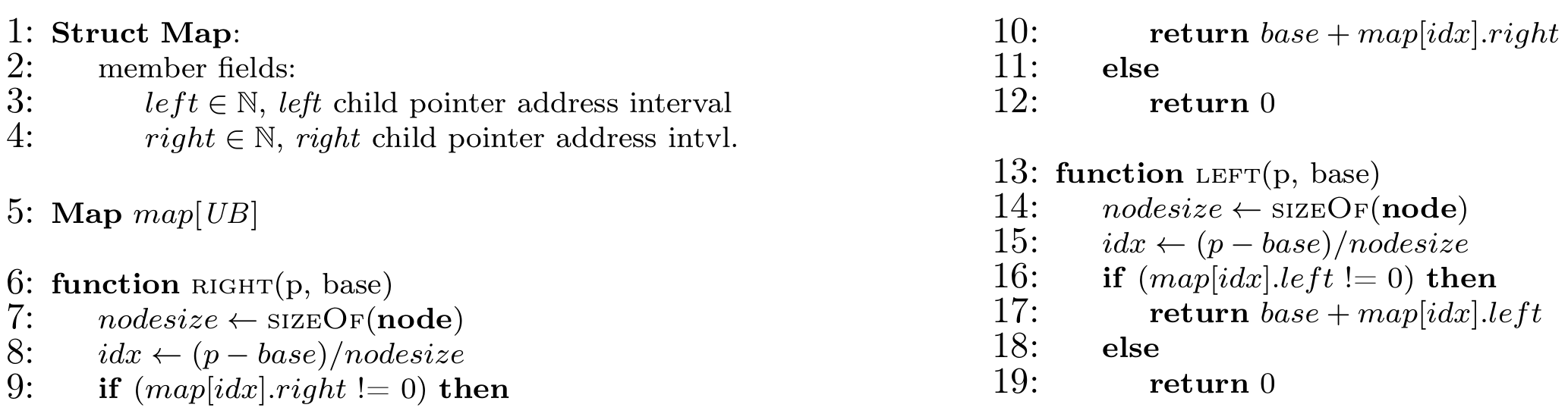}
\caption{Map structure and the \textit{mapping} functions.}
\label{lst:map-func}
\end{figure}

GreenBST does not use pointers to link between 
its internal nodes, instead it uses a single map-based implicit BST array.
This approach is unique to the concurrency-aware vEB layout as it benefits from the
usage of the fixed-size GNodes. The usage of pointers and arithmetic-based implicit array
in cache-oblivious (CO) trees has been previously studied~\cite{BrodalFJ02} and
both are found to have weaknesses. Pointer-based CO tree search operation
is slow, mainly because of overheads in every data transfer between memory 
(although CO tree can minimize data transfers, the inclusion of pointers
can lower the amount of meaningful data (e.g., keys) in each block 
transfer). The implicit array
that uses arithmetic calculation for every node traversal may increase the cost
of computation, especially if the tree is big.

The cache-resident-maps technique emulates BST's
(left and right) child traversals inside a GNode using a combination of
a cache-resident GNode \textit{map} structure and \textsc{left} and \textsc{right}
functions (cf. Figure \ref{lst:map-func}).
The
\textsc{left} and \textsc{right} functions, given an arbitrary node $v$ and
its GNode's root memory addresses, return the addresses 
of the left and right child nodes of $v$, or $0$ if $v$ has no children 
(i.e., $v$ is an internal leaf node of a GNode). 
The \textsc{left} and \textsc{right} operations throughout GreenBST share
a common cache-resident \textit{map} instance (cf. 
Figure \ref{lst:map-func}, line 5). All GNodes use 
the same fixed-size vEB layout, so
only one \textit{map} instance with size
$\UB$ is needed for all traversing operations. This makes GreenBST's 
memory footprint small and keeps the frequently used \textit{map} instance in cache.

Note that the mapping approach does not induce memory fragmentation.
This is because the mapping approach applies only for each GNode,
and \textit{map} is only used to point to internal nodes within a GNode. 
GNode layout uses a contiguous memory block of fixed size $\UB$ 
and {\em update} operations can only change the values of GNode internal nodes 
(e.g., from EMPTY to a key value in the case of insertion), 
but cannot change GNode's memory layout.

\subsubsection{Inter-GNode connection.}  \label{sec:interdesc}

\begin{figure}[!t] \centering 
\includegraphics[width=0.95\textwidth]{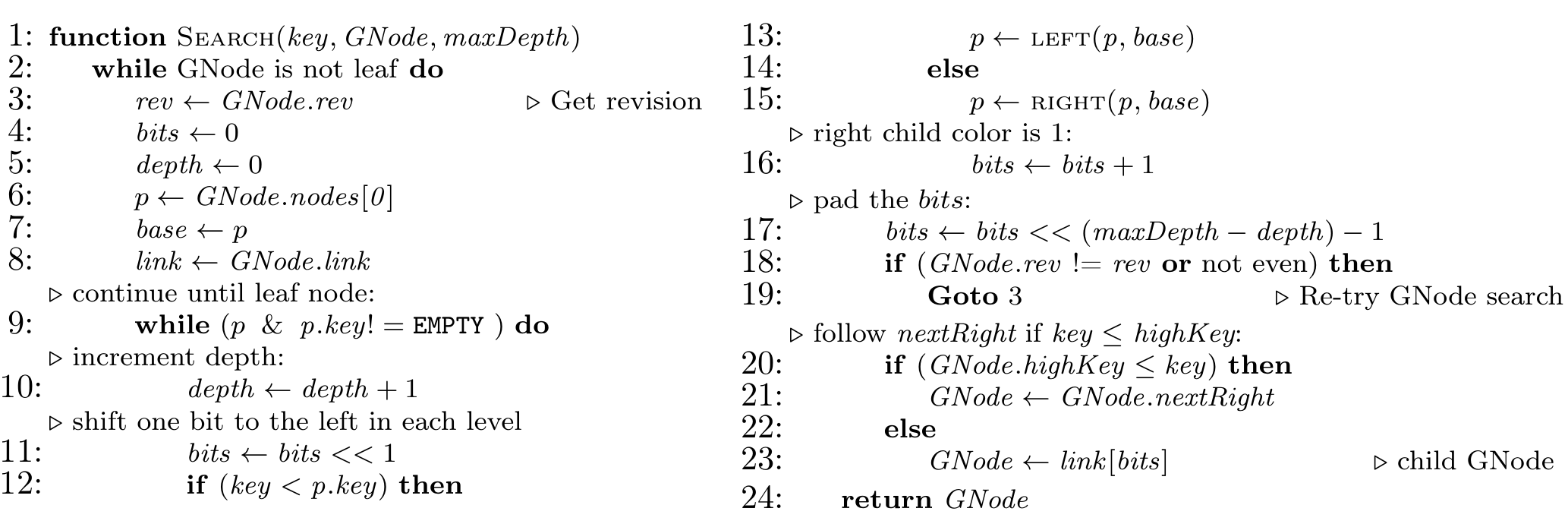}
\caption{Search within pointer-less GNode. This function will return the {\em leaf} GNode containing
the searched key. From there, an implicit array search using \textsc{left} and \textsc{right} functions
is adequate to pinpoint the key location. The search operations are utilizing both the 
\textbf{nextRight} pointers and \textbf{highKey}
variables to handle concurrent search even during GNode split.}
\label{lst:scannode-func}
\end{figure}

To enable traversing from a GNode to its child GNodes, 
we develop a new inter-GNode connection mechanism. 
We logically assign binary values to GNode's internal edges so that each path from 
GNode root to an internal leaf node is represented by a unique bit-sequence. 
The bit-sequence is then used as an index in a  {\bf link} array containing pointers to child GNodes. 
As GNode's internal node has only left and right edges, we assign 0 and 1 to the left and right edges, respectively. 
The maximum size of the bit representation is GNode's height or $\log(\UB)$ bits. 
We allocate a link pointer array whose size is half $\UB$ length. 
The algorithm in Figure \ref{lst:scannode-func} explains how the inter-GNode connection works in a pointer-less search function.

\subsubsection{Balanced and concurrent tree.} \label{sec:balanced}
GreenBST adopts the concurrent algorithms of B-link tree
that provides lock-free search operations and 
adopts the B+tree structure for its high-level structure~\cite{Lehman:1981:ELC:319628.319663}. However,
unlike B-link tree, GreenBST is an in-memory tree and uses optimistic concurrency
to handle lock-free concurrent search operations even in the occurrences of
the unique "in-place" GNodes maintenance operations.

Similar to B-link tree, GreenBST  {\em insert} operations build the tree from the bottom up,
but unlike B-link tree, GreenBST insert operation can trigger
{\em rebalance} operation, a unique GreenBST feature 
to maintain GNode's small height.

Function \textsc{rebalance($T_i$)}
is responsible for rebalancing a
GNode $T_i$ after an insertion.
If a new node $v$ is inserted at the {\em last level} node of a GNode, 
that GNode is rebalanced to a complete BST. A rebalance operation sets all
GNode leaves node height to $\lfloor\log N\rfloor + 1$,
where $N$ is the count of the GNode's internal nodes
and $N\leq\UB$.
Note that this is the default rebalance strategy used by DeltaTree,
the incremental rebalance used by GreenBST is explained
further in this section.

The {\em delete} operation in GreenBST 
simply marks the requested key ($v$) as deleted. This function 
fails if $v$ does not exist in the tree or $v$ is already marked.
GreenBST does not employ merge operation between GNodes 
as node reclamation is done by the rebalance and split operations. 
The offline memory reclamation techniques used in the B-link tree 
\cite{Lehman:1981:ELC:319628.319663} can be deployed to merge nearly empty 
GNodes in the case where delete operations are the majority.
Our new search trees aim at workloads dominated by search operations.  

GreenBST concurrency control uses locks and  
\textbf{nextRight} and \textbf{highKey} variables
to coordinate between search and update operations~\cite{Lehman:1981:ELC:319628.319663}
in addition to \textbf{rev} variable that is used for the search's optimistic concurrency.
When a GNode needs to be maintained by either rebalance or split operations,
the GNode's {\bf rev} counter is incremented by one before the operation
starts. The GNode 
counter is incremented by one again after the maintenance operation finishes.
Note that
all maintenance procedures happen when the lock is still held by the insert
operation and therefore, only one operation may update {\bf rev} counter
and maintain a GNode at a time.
The usage of {\bf rev} counter is to prevent search from returning a wrong
key because of 
the "in-place" GNode maintenance operation. Advanced locking techniques \cite{HaPT07_JSS, KarlinLMO91, LimA94} can also be used.

The {\em search} operation in GreenBST uses a combination of function
\textsc{Search} (cf. Figure 
\ref{lst:scannode-func}) and an implicit tree traversal using a map. 
Function \textsc{Search} 
traverses the tree from the internal root node of the root GNode down to a leaf GNode, 
at which the search is handed over to the 
implicit tree traversal to find the searched key within the leaf GNode. 
GreenBST {\em search} operation
does not wait nor use lock, even in the occurrence
of the concurrent updates.

GreenBST {\em search} uses optimistic concurrency~\cite{Kung:1981aa} to ensure the operation 
always returns the correct answer even if it arrives at
a GNode that is undergoing the in-place
maintenance operation (i.e., {\em rebalance} and {\em split}).
First, before starting to traverse a GNode, a search operation records the GNode {\bf rev} counter. 
Before following a link to a child GNode or returning a key,
the search operation re-checks again the counter. If the current 
counter value is an odd number or if it is not equal to 
the recorded value, the search operation needs to retry search as this
indicates that GNodes are being or have been maintained. 

\subsubsection{Incremental Rebalance.} 
As explained earlier, the rebalance in DeltaTree always
involves $\UB$ keys, which eventually makes insertions require
amortized $\mathcal{O}(\UB)$ time. GreenBST borrows
the incremental rebalance idea similar to the conventional vEB layout~\cite{BrodalFJ02}
that has the amortized $\mathcal{O}((\log^2\UB)/(1-\Gamma_1))$ time if
used in GreenBST.
However, unlike the conventional vEB layout that 
might have to rebalance the whole tree, 
we only apply the incremental rebalance to GNodes. 
To explain the idea, we denote {\em density($w$)} as the 
ratio of number of keys inside a subtree rooted at $w$ divided by
the number of maximum keys that a subtree rooted at $w$ can hold.
For example, a 
subtree with root $w$ that is located three levels away from an 
internal leaf of a GNode can hold at most $2^3-1$ keys. 
If the subtree only contains 3 keys, then
{\em density($w$)} $= ^3/_7 = 0.42$.
We also denote a {\em density threshold} 
$0<\Gamma_1<\Gamma_2<...<\Gamma_H=1$,
where $H$ is the GNode's height.
The main idea is: after a new key is inserted at an 
internal leaf position $v$,
we find the nearest ancestor $w$ of $v$ where
{\em density($w$)} $\leq\Gamma_{\mathit{depth}(w)}$ and 
{\em depth($w$)} is the level where $w$ resides, counted from
the root of the GNode. If that $w$ is found, 
we rebalance the subtree rooted at $w$.

\subsubsection{Heterogeneous GNodes.} 
We aim to reduce the overhead of rebalancing and 
lower the GreenBST height 
with the usage of different layouts for the leaf GNodes.
All DeltaTree's GNodes use the leaf-oriented
BST layout, hence DeltaTree uses {\em homogeneous} GNodes. 
Unlike DeltaTree, leaf GNodes in GreenBST use 
the internal tree layout instead of the external (or leaf-oriented) tree layout.
GreenBST uses {\em heterogeneous} GNodes as there are two difference 
GNode layouts used.
In the internal tree layout, keys are located in all nodes of a tree, 
while in the external tree layout, 
keys are only located in the leaf nodes. The reasoning behind this choice is
although leaf-oriented GNodes layout is required for 
inter-GNode connection 
(i.e., between parent- and child- GNodes),
leaf GNodes do not have any children and therefore, 
do not need to adopt same structure as the other GNodes.

\subsection{GreenBST experiments}\label{sec:evaluation}
\begin{table}[!t]
\centering
\scriptsize
\caption{We use 4 different benchmark platforms to evaluate the trees' energy efficiency and performance.}
\begin{tabular}{| c || p{3.5cm} | p{3cm} | p{3cm} | p{3.2cm} |}
\hline
\bf Name & 	\bf HPC &  		\bf		ARM & 			\bf	MIC  & \bf Myriad2\\
\hline
\hline
\bf System &	Intel Haswell-EP &	Samsung Exynos5 Octa &	Intel Knights Corner & Movidius Myriad2\\
\hline
\bf Processors&	2x Intel Xeon E5-2650L v3 &	1x Samsung Exynos 5410		&1x Xeon Phi 31S1P & 1x Myriad2 SoC \\
\hline
\bf \# cores&	24 (without hyperthreading)&	$-$ 4x Cortex A15 cores \newline $-$ 4x Cortex A7 cores	& 57 (without hyperthreading) &	$-$ 1x LeonOS core \newline $-$ 1x LeonRT core \newline $-$ 12x Shave cores \\
\hline
\bf Core clock&	2.5 GHz& $-$ 1.6 GHz (A15 cores) \newline $-$ 1.2 GHz (A7 cores) &	1.1 GHz &  600 MHz\\
\hline
\bf L1 cache	& 32/32 KB I/D	& 32/32 KB I/D	& 32/32 KB I/D		& $-$ LeonOS (32/32 KB I/D) \newline  $-$ LeonRT (4/4 KB I/D) \newline $-$ Shave (2/1 KB I/D)\\
\hline
\bf L2 cache	& 256 KB	& $-$ 2 MB (shared, A15 cores) \newline $-$ 512 KB (shared, A7 cores)  & 512 KB & $-$ 256 KB (LeonOS) \newline  $-$ 32 KB (LeonRT) \newline $-$ 256 KB (shared, Shave)\\
\hline
\bf L3 cache	& 30 MB (shared)&	-	& - & 2MB "CMX" (shared) \\
\hline
\bf Interconnect&	8 GT/s Quick Path Interconnect (QPI)&	CoreLink Cache Coherent \newline Interconnect (CCI) 400&	5 GT/s Ring Bus Interconnect & 400 GB/sec Interconnect\\
\hline
\bf Memory&	64 GB DDR3&	2 GB LPDDR3&	6 GB GDDR5 & 128 MB LPDDR II\\
\hline
\bf OS&		Centos 7.1 (3.10.0-229 kernel)	&Ubuntu 14.04 (3.4.103 kernel)	& Xeon Phi uOS (2.6.38.8+mpss3.5) & RTEMS (MDK 15.02.0)\\
\hline
\bf Compiler &	GNU GCC 4.8.3 & GNU GCC 4.8.2 & Intel C Compiler 15.0.2 & Movidius MDK 15.02.0\\
\hline
\end{tabular}
\label{tbl:platforms}
\end{table}

We run several different benchmarks to evaluate GreenBST throughput and energy efficiency. 
We combine the benchmark results
with the last level cache (LLC) and memory profiles of the trees
to draw a conclusion of whether GreenBST 
improved fine-grained data locality layout (i.e., heterogeneous layout) and concurrency
(i.e., lower overall cost of runtime maintenance) over DeltaTree
are able to
make GreenBST the most energy-efficient tree across different platforms. 
In addition, we would like to also conclude whether GreenBST improvements over 
DeltaTree are useful to increase GreenBST's energy efficiency 
when processing the update-intensive workloads. Note that we are not collecting
the computation profiles (e.g., Mflops/second) because all the tree operations are 
data-intensive instead of compute-intensive.

We conduct an experiment on GreenBST and several 
prominent concurrent search trees (cf. Table \ref{tbl:algos}) 
using parallel micro-benchmark that is 
based on Synchrobench~\cite{Gramoli:2015:MYE:2688500.2688501} (cf. Figure \ref{fig:combEval}). 
The trees' LLC and memory profiles during the micro-benchmarks are collected and presented in Figure \ref{fig:combEval}d and \ref{fig:combEval}e, respectively. 
To investigate GreenBST behavior in real-world applications, we implement GreenBST and CBTree as the backend structures in the STAMP database benchmark \texttt{\small Vacation}~\cite{4636089}, alongside the \texttt{\small Vacation}'s original backend structure red-black tree (rbtree) (cf. Figure \ref{fig:vacation}).

All the experimental benchmarks are conducted on an Intel high performance computing ({\bf HPC}) platform 
with 24 core 2$\times$ Intel Xeon E5-2650Lv3 CPU and 64GB of RAM, an {\bf ARM} embedded platform 
with an 8 core Samsung Exynos 5410 CPU and 2GB of RAM (Odroid XU+E), 
an accelerator platform based on the Intel Xeon Phi 31S1P with 57 cores and 6GB of RAM ({\bf MIC} platform),
and a specialized computing platform ({\bf Myriad2} platform). The detailed specifications for the testing 
platforms can be found in Table~\ref{tbl:platforms}.
For the parallel micro-benchmark, the trees are pre-initialized with several initial keys
before running 5 million operations of 100\% (search-intensive) and 50\% searches (update-intensive), respectively.
The initial keys given to both the ARM and MIC platforms are 
$2^{22}$ keys and to the HPC platform are $2^{23}$ keys.
All experiments are 
repeated at least 5 times to guarantee consistent results.

{\em Energy efficiency metrics} (in operations/Joule) 
are the energy consumption divided by the number of operations and {\em throughput metrics} 
(in operations/second) are the number of operations 
divided by the maximum time for the threads to finish the whole operations.
Energy metrics are collected from the on-board power measurement 
on the ARM platform, Intel RAPL interface on the HPC platform, and micras sysfs 
interface (i.e., \texttt{/sys/class/micras/power}) on the MIC platform.

\begin{figure}[!hpt]
\centering
\includegraphics[width=\linewidth]{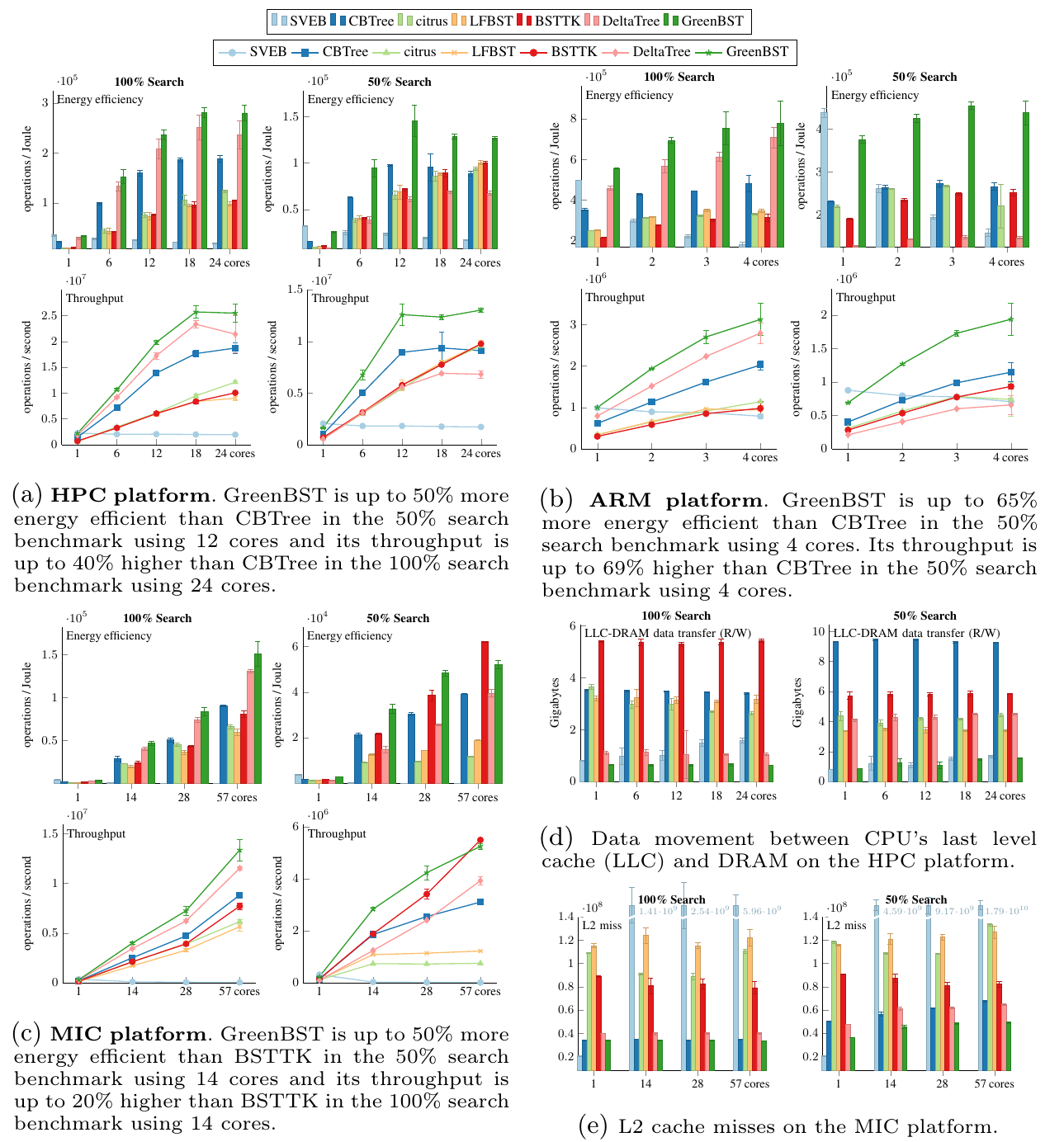}
\caption{(a,b,c) Energy efficiency and throughput comparison of the trees.
On the HPC platform, DeltaTree and GreenBST energy efficiency and throughput decreases in the 50\% search benchmark using 18 and 24 cores (i.e., with 2 chips) because of the coherence overheads between two CPUs (cf. Section \ref{sec:multicpu-issue}). In the 50\% search benchmark using 57 cores (MIC platform), BSTTK energy efficiency and throughput
beats GreenBST by 20\% because of the coherence overheads in the MIC platform (cf. Section \ref{sec:multicpu-issue}).
(d) LLC-DRAM data movements on the HPC platform, collected from the CPU counters using Intel PCM. (e) L2 cache miss counter on the MIC platform, collected using PAPI library.}
\label{fig:combEval}
\end{figure}

\begin{table}[!t]
\centering
\caption{The tree memory footprint after $2^{23}$ integer keys insertion on the HPC platform.}\label{tbl:size}
\scriptsize
\begin{tabular}{|l|c|c|c|c|c|c|c|}
\hline
\diaghead{\theadfont Diag Column Head II}
{ }{Tree name} & SVEB &CBTree & citrus & LFBST & BSTTK & DeltaTree & \bf GreenBST \\
\hline
Memory used (in {\bf GB})& $0.1$ & $0.4$ & $0.8$ & $0.7$ & $1.0$ & $0.6$ & $\bf 0.4$\\
\hline
\end{tabular}
\end{table}

\subsubsection{Experimental results on HPC, ARM, and MIC platforms} Based on the results in Figure \ref{fig:combEval} and \ref{fig:vacation}, GreenBST's energy efficiency and throughput are 
the highest compared to DeltaTree and the other trees.
Because of its {\em incremental rebalance}, GreenBST outperforms DeltaTree
(and the other trees) in the update-intensive workloads. With
its {\em heterogeneous layout}, GreenBST is able to 
outperform DeltaTree in the search-intensive workloads.
GreenBST energy efficiency and throughput
are up to 195\% higher than that of DeltaTree for the update intensive benchmark
and up to 20\% higher for the search intensive benchmark (cf. Figure \ref{fig:combEval}b).
Compared to the other trees, GreenBST energy efficiency and throughput 
are up to 65\% and 69\% higher, respectively. Note that CBTree (B-link tree) 
is a highly-concurrent B-tree variant that it's still used as a backend in popular database systems such as PostgreSQL.

The reason behind GreenBST good results is that
GreenBST's data transfer (cf. Figure \ref{fig:combEval}d)  and LLC misses (cf. Figure \ref{fig:combEval}e)
are among the lowest of all the trees. We would like to emphasize that even 
GreenBST memory footprint is the same to that CBTree (cf. Table \ref{tbl:size}), 
GreenBST data transfer is significantly lower than CBTree's.
These facts prove that the combination of 
locality-aware layout and the optimizations that GreenBST has over 
DeltaTree are beneficial
to both fine-grained locality and concurrency, which are the key ingredients 
of an energy-efficient concurrent search tree.

\begin{figure}[!t]
\centering
\resizebox{\textwidth}{!}{\includegraphics{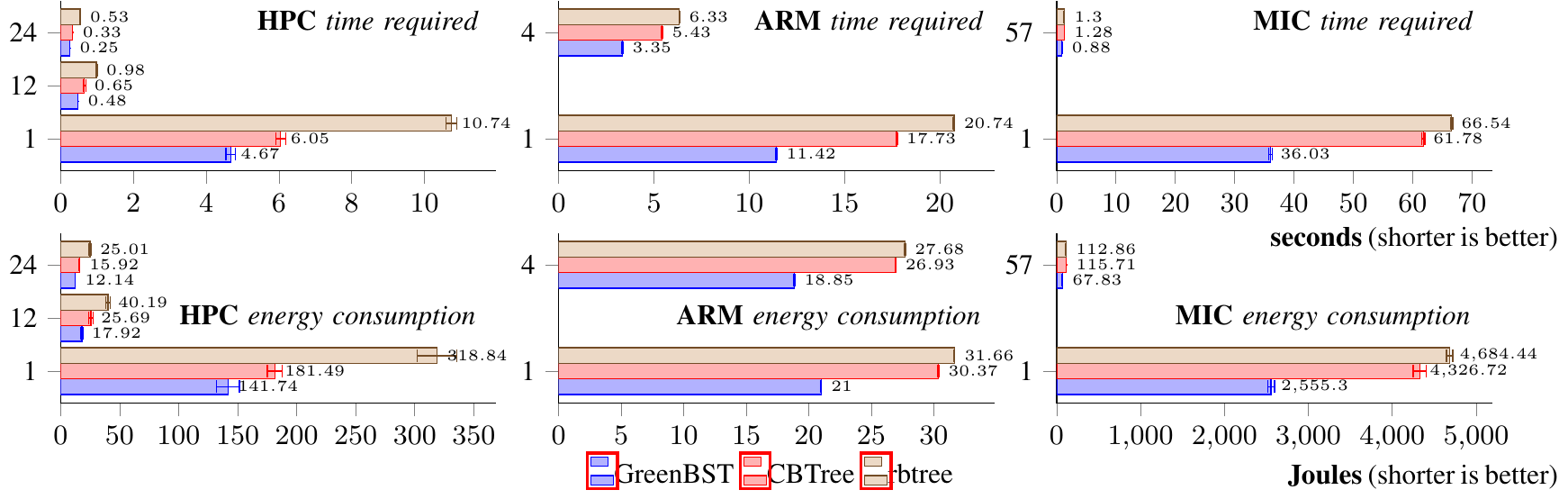}}
\caption{GreenBST energy efficiency and throughput against CBTree and STAMP's built-in red-black tree (rbtree) for the vacation benchmark.
At best, GreenBST consumes 41\% less energy and requires 42\% less time than CBTree (in the 57 clients benchmark on the MIC platform).
\label{fig:vacation}\label{fig:vacation-arm}\label{fig:vacation-mic}
}
\end{figure}

\subsubsection{Experimental results on Myriad2 platform}

We have implemented DeltaTree and GreenBST that work on the Myriad2 platform by crafting a new concurrency control for the trees.  A new concurrency is required because Myriad2 platform does not support atomic operations and has a limited number of usable hardware mutexes. Therefore, to circumvent these limitations, we create a new concurrency control scheme that works similarly to a ticket lock mechanism. In this scheme, we utilize LeonRT processor as a lock manager for the shaves. With LeonRT acting as a lock manager, all shaves need to request a DeltaNode or a GNode lock from LeonRT before it can lock the DeltaNode or GNode for update and maintenance operation. Our locking technique implementation uses only a shared array structure with $2\times\mathit{sv}$ size, where $\mathit{sv}$ is the number of active shaves. For low latency lock operations, we put this lock structure in the Myriad2's CMX memory.
All other DeltaTree and GreenBST structures are unchanged (e.g., the tree itself) and placed in the DDR memory.

We tested our
GreenBST and DeltaTree implementations on Myriad2 against the concurrent B+tree (B-link tree) \cite{Lehman:1981:ELC:319628.319663}.
The B-link tree implementation (CBTree) also utilized the same locking technique and memory placement strategy as GreenBST and DeltaTree.

Figure \ref{fig:ene-myriad2} shows that the energy efficiency of GreenBST 
is up to 4$\times$ better than that of CBTree in the 100\% search using 12 shaves on the Myriad2 platform. In terms of 
throughput on the Myriad2 platform, Figure \ref{fig:res-myriad2} indicates that 
GreenBST has up to 4$\times$ more throughput than CBTree in the 100\% search case
when using all available 12 shaves.

\begin{figure}[!t] 
\centering 
\includegraphics[width=0.45\linewidth]{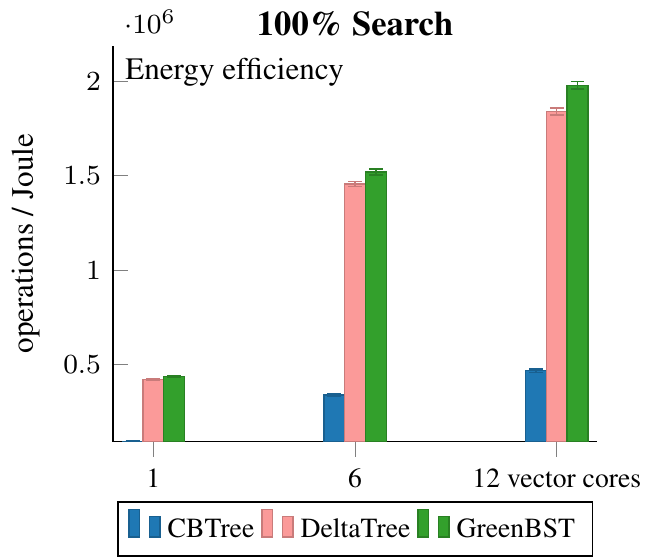}
\caption{Energy comparison using 
$2^{20}$ initial values on an Myriad2 platform. DeltaTree is up to 4$\times$ more
energy efficient than CBTree in 100\% search operation with 12 shaves.}\label{fig:ene-myriad2}
\end{figure}

\begin{figure}[!t]
\footnotesize 
\centering 
\includegraphics[width=0.45\linewidth]{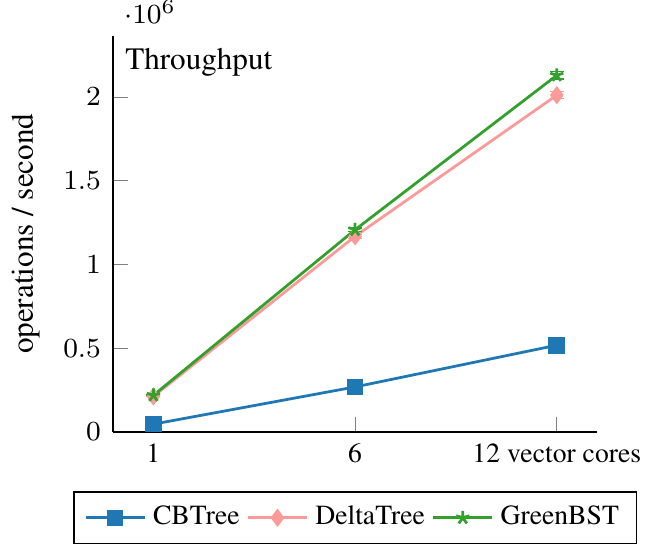}
\caption{Throughput comparison using 
$2^{20}$ initial values on an Myriad2 platform. DeltaTree is up to 4$\times$ faster 
than CBTree in 100\% search operation with 12 shaves.}\label{fig:res-myriad2}
\end{figure}

\subsection{Discussions} \label{sec:perfeval}

Some of the benchmark results show that besides data movements,
efficient concurrency control is also necessary in order to produce energy-efficient
data structures.
For example, the conventional vEB tree (SVEB) always transferred the smallest amount of
data between memory and the CPU, but unfortunately, its energy efficiency 
and throughput
failed to scale when using 2 or more cores.
SVEB is not designed for concurrent operations and an inefficient concurrency control 
(a global mutex) has to be implemented in order to include the tree in this study
(note that we are unable to use a more fine-grained concurrency because SVEB 
uses recursive layout in a contiguous memory block).
Therefore, even if SVEB has the smallest amount of data transfer 
during the micro-benchmarks, 
the concurrent cores have to spend a lot of time waiting and competing for 
a lock.
This is inefficient as a CPU core still consumes power (e.g., static power) 
even when it is waiting (idle).

Finally, an important lesson that we have learned is that minimizing overheads
in locality-aware data structures can reduce the structure's energy consumption. 
One of the main differences between DeltaTree and GreenBST is that DeltaTree 
uses the homogeneous (leaf-oriented) layout, while GreenBST 
does not. Leaf-oriented GNodes increases
DeltaTree's memory footprint by 50\% as compared 
to GreenBST (cf. Figure \ref{fig:combEval}e) and
has caused higher data transfer between LLC and DRAM (cf. Figure \ref{fig:combEval}d).
Bigger leaf size also increases maintenance cost for each leaf
GNode, because there are more data that need to be arranged in every rebalance or split 
operation, which leads to lower update concurrency.
Therefore, DeltaTree energy efficiency and throughput are lower than GreenBST.

\paragraph*{Inter-CPU and many-core coherence issue} \label{sec:multicpu-issue}
Our experimental analysis has revealed that multi-CPU and 
many-core cache coherence, if triggered, 
can degrade concurrent update throughput and energy efficiency of 
the locality-aware trees. Figure \ref{fig:combEval}a shows the "dips" in 
GreenBST's 50\% update energy efficiency and throughput on the HPC platform
(i.e., in the {\it 50\% update/18 cores} and {\it 50\% update/24 cores} cases). 
Figure \ref{fig:combEval}c also
shows that BSTTK beats GreenBST in the {\it 50\% update/57 cores} case
on the MIC platform.

Using the CPU performance counters, we have found 
that the GreenBST concurrent updates frequently triggered the 
inter-CPU coherency mechanism.
In the HPC platform, coherency mechanism causes heavy bandwidth saturation in 
the CPU interconnect. In the MIC platform, 
it causes most of the L2 data cache misses  
to be serviced from other cores and saturates the platform's 
bidirectional ring interconnect. These facts highlight the challenge 
faced by the locality-aware concurrent search tree: because 
of its locality awareness (i.e.,
related data are kept nearby and often re-used), the tree concurrent 
update operations might trigger 
heavy interconnect traffic on the multi-CPU platforms. The coherency
mechanisms increase the total number of 
data transfer and the platform's energy consumption.

\subsection{Conclusions}\label{sec:conclusions}
The results presented in this paper not only show that GreenBST 
is an energy-efficient concurrent search tree, but also provide 
an important insight into how to develop energy efficient data 
structures in general. On single core systems, having locality-aware 
data structures that can lower data movement has been 
demonstrated to be good enough to increase energy-efficiency.
However, on multi-CPU and many cores systems, 
data-structures' locality-awareness alone is not enough and 
good concurrency and multi-CPU cache strategy are needed. Otherwise, the 
energy overhead of "waiting/idling" CPUs or multi-CPU coherency mechanism 
can exceed the energy saving obtained by fewer data movements.

\newpage
\section{Customization methodology for implementation of streaming aggregation in embedded systems}\label{sec:stream}
\fbox{
	\begin{minipage}{0.96\textwidth}
		\small	
		\textbf{Copyright Notice:} Most material in this section is based on the following article:
		Lazaros Papadopoulos, Dimitrios Soudris , Ivan Walulya, Philippas Tsigas : Customization methodology for implementation of streaming aggregation in embedded systems. 	\emph{Journal of Systems Architecture}, May 2016\cite{Papadopoulos201648}. DOI:  10.1016/j.sysarc.2016.04.013
	\end{minipage}
}
\subsection{Introduction}
\label{introduction}
Efficient real-time processing of data streams produced by modern interconnected systems is a critical challenge. In the past, low-latency streaming was mostly associated with network operators and financial institutions. Processing of millions of events such as phone calls, text messages, data traffic over a network and extracting useful information is important for guaranteeing high Quality of Service. Stream processing applications that handle traditional streams of data were mostly implemented by using Stream Processing Engines (SPEs) running on high performance computing systems.

However, nowadays digital data come from various sources, such as sensors from interconnected city infrastructures, mobile cameras and wearable devices. In the deviced-driven world of Internet of Things, there is a need in many cases for processing data on-the-fly, in order to detect events while they are occurring. These data-in-motion come in the form of live streams and should be gathered, processed and analyzed as quickly as possible, as they are being produced continuously. Low-power embedded devices or embedded micro-servers \cite{x-gene2} are expected not only to monitor continuous streams of data, but also to detect patterns through advanced analytics and enable proactive actions. Applying analytics to these streams of data before the data is stored for post-event analysis (data-at-rest) enables new service capabilities and opportunities. 

Streaming aggregation is a fundamental operator in the area of stream processing. It is used to extract information from data streams through data summarization. Aggregation is the task of summarizing attribute values of subsets of tuples from one or more streams. A number of tuples are grouped and aggregations are computed on their attributes in real-time fashion. High frequency trading in stock markets (e.g. continuously calculating the average number of each stock over a certain time window), real time network monitoring (e.g. computing the average network traffic over a time window) are examples of data stream processing, where streaming aggregation along with other operators is used to extract information from streams of tuples. 

Streaming aggregation performance is affected a lot by the cost of data transfer. So far, streaming aggregation scenarios have been implemented and evaluated in various architectures, such as GPUs, Nehalem and Cell processors \cite{sa_on_parallel}. Indeed, there is a trend to utilize low power embedded platforms on running computational demanding applications in order to achieve high performance per watt \cite{arm-cortexA8}\cite{appleTV}\cite{tibidabo}\cite{perf-per-watt}. 

Modern embedded systems provide different characteristics and features (such as memory hierarchy, data movement options, OS support, etc.) depending on the application domain that they target. The impact of each one of these features on performance and energy consumption of the whole system, when running a specific application, is often hard to predict at design time. Even if it is safe to assume in some cases that the utilization of a specific feature will improve or deteriorate the value of a specific metric in a particular context, it is hard to quantify the impact without testing. This problem becomes even harder when developers attempt to improve more than one metric simultaneously. A similar problem is the porting of an application running on a specific system to another with different specifications. The application usually need to be customized in the new platform differently, in order to provide improved performance and energy efficiency. The typical solution followed by developers is to try to optimize the implementation of the application on the embedded platform in an ad-hoc manner, which is a time consuming process that may yield suboptimal results. Therefore, there is a need for a systematic customization approach: Exploration can assist the effective tuning of the application and platform design options, in order to satisfy the design constraints and achieve the optimization goals. 

Towards this end, in this work, we propose a semi-automatic step-by-step exploration methodology for the customization of streaming aggregation implemented in embedded systems. The methodology is based i) on the identification of the parameters of the streaming aggregation operator that affect the evaluation metrics and ii) on the identification of the embedded platform specification features that affect the evaluation metrics when executing streaming aggregation. These parameters compose a design space. The methodology provides a set of implementation solutions. For each solution, the application and the platform parameters have different values. In other words, each customized streaming aggregation implementation is tuned differently, so it provides different results for each evaluation metric. Developers can perform trade-offs between metrics, by selecting different customized implementations. Thus, instead of evaluating solutions in ad-hoc manner, the proposed approach provides a systematic way to explore the design space. 

The main contributions of this work are summarized as follows:
\begin{enumerate}[i.]
	\item We present a methodology for efficient customization of streaming aggregation implementation in embedded systems.
	\item We show that streaming aggregation implemented on embedded devices yields significantly higher performance per watt in comparison with corresponding HPC and general purpose GPU (GPGPU) implementations.
\end{enumerate}
Finally, based on the experimental results of the demonstration of the methodology, we draw interesting conclusions on how each one of the application and platform parameters (i.e. design options) affects each one of the evaluation metrics. The methodology is demonstrated in two streaming aggregation scenarios implemented in four embedded platforms with different specifications: Myriad1, Myriad2, Freescale I.MX.6 Quad and Exynos 5 octa. The evaluation metrics are throughput, memory footprint, latency, energy consumption and scalability.

%
\subsection{Related Work}
\label{related}
Stream processing on various high performance architectures has been studied in the past extensively. Many works focus on the parallelization of stream processing \cite{adaptive1}, \cite{streamcloud}, \cite{heter}. They describe how the stream processing operators should be assigned to partitions to increase parallelism. The authors in \cite{sa_lock_free} describe another way of improving the performance of streaming aggregation: They propose lock-free data structures for the implementation of streaming aggregation on multicore architectures. The evaluation has been conducted on a 6-core Xeon processor and the results show improved scalability.  

With respect to stream processing engines (SPEs), Aurora and Borealis  \cite{borealis} are among the most well known ones. Several works that focus on the evaluation of stream processing operators on specific parallel architectures can be found in the literature. For example, an evaluation on heterogeneous architectures composed of CPU and a GPU accelerator is presented in \cite{heter}. The authors of \cite{sa_on_parallel} evaluate streaming aggregation implementations on Core 2 Quad, Nvidia GTX GPU and on Cell Broadband Engine architectures. The aggregation model used in this work is more complex, since it focuses on timestamp-based tuple processing. 

There exists several works that describe the usage of low power embedded processors to run server workloads. More specifically, many works propose the integration of low-power ARM processors in servers \cite{arm-cortexA8} \cite{appleTV},  or present energy-efficient clusters built with mobile processors \cite{tibidabo}.

In the area of embedded systems stream processing, several works focus on compilers that orchestrate parallelism, while they handle resource and timing constraints efficiently \cite{compiler}. A programming language for stream processing in embedded systems has been proposed in \cite{prog_lang}. These works are complementary to ours: The conclusions we drive from this work could assist the implementation of efficient compilers and development frameworks for stream programming. 

Design space exploration in embedded systems is another area related with the present work. Exploration methodologies have been proposed for tuning at system architecture level \cite{Pareto},  for customization of dynamic data structures \cite{ddtr} and of dynamic memory management optimization \cite {dmm}. These customization approaches are complementary to the one proposed in the present work. Performance and energy consumption of streaming aggregation implementation could improve with effective customization of data structures or of the dynamic memory management of the system.

%
\subsection{Streaming Aggregation}
In this Section we provide a description of the streaming aggregation operator and we analyze the design challenges of implementing a streaming aggregation scenario on an embedded platform.

\begin{figure}
	\centering
	\includegraphics[width=0.75\textwidth, page=1]{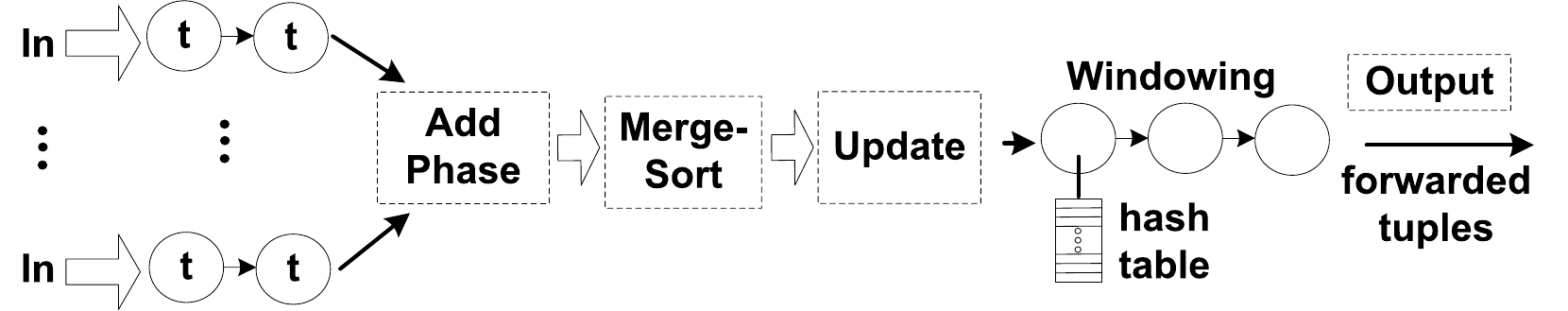}
	\caption{Time-based streaming aggregation scenario phases.}
	\label{SA_generic}
\end{figure}

\subsubsection{Streaming Aggregation description}
Streaming aggregation is a very common operator in the area of stream processing. It is used to group a set of inbound tuples and compute aggregations on their attributes, similarly to the \textit{group-by} SQL statement. In the context of this work, we discuss two aggregation scenarios: \textit{multiway time-based with sliding windows} and \textit{count-based with tumbling windows}.

\paragraph{Multiway time-based streaming aggregation}
In multiway aggregation, multiple streams of incoming tuples, which are stored in queues, are combined into one stream, through a merge operator and their tuples are sorted given their timestamp attribute. It consists of 4 phases, as presented in Fig. \ref{SA_generic}: 

\begin{enumerate}
	\item \textit{Add}: Incoming tuples are fetched from each input stream.
	\item \textit{Merge-Sort}: The tuples are merged and sorted, by the \textit{merge} operator. 
	\item \textit{Update}: Each tuple is assigned to the windows that it contributes to. 
	\item \textit{Output}: Tuples with the computed aggregated value are forwarded. 
\end{enumerate}

During the \textit{Add} phase tuples from each input stream are fetched and forwarded to the Merge-Sort phase. Since the incoming tuples are stored in a queue, they are forwarded in a FIFO manner. 

\textit{Merge-Sort} operation is used to combine streams that were sorted on a given attribute into a single stream, whose tuples are also ordered on the same attribute. In the context of this work, the tuples are sorted in timestamp order. 

\textit{Merge} and \textit{Sort} are tightly coupled operations in streaming aggregation scenarios since they share the same resource (i.e. the incoming  dequeued tuples) and they can be considered a single primitive operation. Merge-Sort phase ensures deterministic processing of the incoming tuples. A tuple is ready to be processed and forwarded to the next phase, if at least one tuple with an equal or higher timestamp has been received at each input stream. 

In the \textit{Update} phase the windowing operation is taking place and each single tuple is assigned to the window that it contributes to. In the context of this work, the aggregated values are computed over sliding windows, which have two attributes: \textit{size} and \textit{advance}. As an example, a window with \textit{size} 5 time units and \textit{advance} 2 time units, covers periods: [0, 5), [2, 7), [4, 9), etc. A tuple with timestamp 3, would contribute to windows [0, 5) and [2, 7). 

In the \textit{Output} phase, the aggregated value is calculated for all windows in which no more incoming tuples are expected to contribute (i.e. completed windows). The deterministic processing of tuples that took place in the earlier phases (more specifically during the \textit{Add} and \textit{Merge-Sort} phases), ensures that the aggregated value will be calculated only for completed windows. A new tuple is created for each aggregated value and it is forwarded, as a result of the aggregation operator. 




Multiway time-based streaming aggregation provides pipeline parallelism, which can be exploited by assigning each phase on a different processing element (PE).  However, performance relies not only on the exploitation of parallelism or on the computational power that the system provides, but also on the efficient data transfer between the phases. The sorted tuples of the  Merge-Sort phase are used by the Update phase to be assigned to the windows that each one contributes to. The Update phase provides to the Output phase information on the windows in which the last tuples contributed to. Thus, the Output phase identifies the completed windows and calculates the aggregated value for each one. The utilization of efficient means of forwarding the information from one phase to another, affects both performance and energy consumption. The same applies to the way by which memory accesses on shared data are synchronized. Other important implementation issues that should be taken into account are the size of the queues in which the inbound tuples are stored (input queues) and the memory allocation of both the queues and the data structure in which the windows are stored. 

\paragraph{Count-based streaming aggregation}
In count-based aggregation, the window size is determined by the number of tuples buffered, instead of the time passed. Our case study considers fixed size windows and aggregation takes place periodically, i.e. when a specific number of tuples is received. Every time an aggregation is completed, all currently stored tuples are evicted and the next window is initially empty (tumbling window). 

\begin{figure}
	\centering
	\includegraphics[width=0.49\textwidth, page=2]{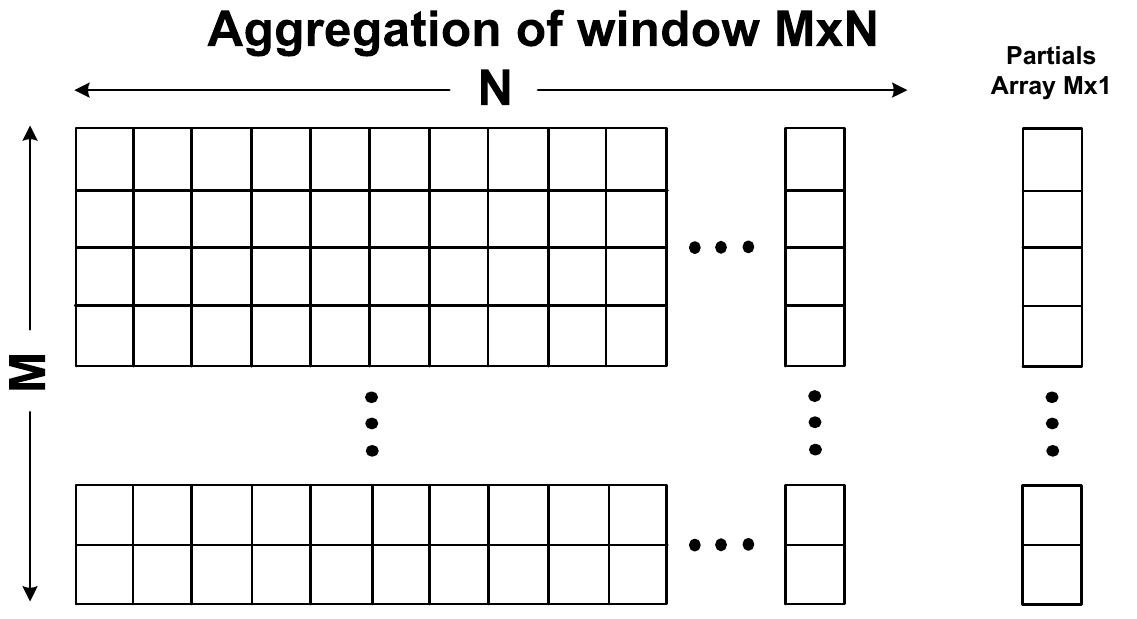}
	\caption{Window and partials array data structures used in the count-based streaming aggregation scenario.}
	\label{count-based}
\end{figure}

To implement the count-based aggregation scenario, we followed an approach based on \cite{sa_on_parallel}. The time intervals between aggregations are based on the number of tuples stored in the window and results of a specific window may depend on results of the previous one. Thus, an extra data structure is needed to store the partially aggregated results of the last window, which may be used in the following aggregation. 

Figure \ref{count-based}, shows the data structures used in the count-based scenario: A \textit{M}x\textit{N} window and the partials array, with 1x\textit{M} entries. \textit{M} is the maximum number of input streams and \textit{N} is the window width. When it is not possible to compute the aggregated value of \textit{N} tuples for a specific input stream before the current window is forwarded, the partially aggregated result is stored in partials array. This result is used by the following window to compute the aggregated value of \textit{N} tuples for the specific input stream. The output is a single tuple that it is produced by a query executed in the \textit{M} aggregated values. 

Apparently, count-based streaming aggregation provides data parallelism. Each window row can be assigned to a different processing element (PE) to compute the aggregated value of each input stream in parallel. Similarly to the time-based scenario, data transfer overhead, memory allocation issues and the window size affect the performance and the energy consumption of the operator. The embedded systems provide various solutions and each one has different impact on each evaluation metric. The design options for all the aforementioned implementation issues compose a design space that it is described in the following Section. 

\subsection{Customization Methodology}
\label{Sec: custom.}
In this Section, we first present the design space for the streaming aggregation customization and then we describe the proposed methodology. 

\subsubsection{Design Space}

\begin{figure*}[!t]
	\centering
	\includegraphics[width=0.97\textwidth, page=3]{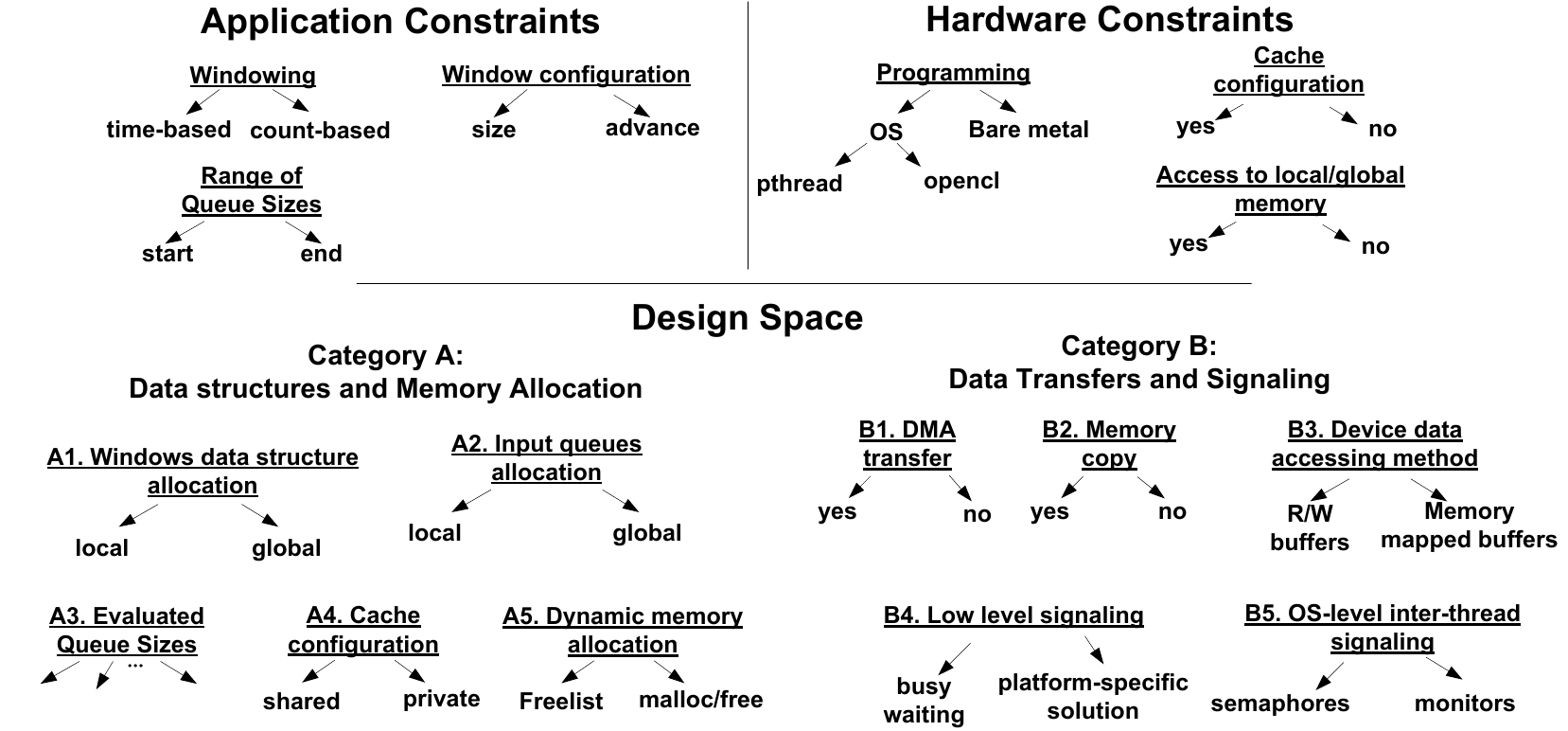}
	\caption{Constraints and Design space for streaming aggregation.}
	\label{design-space}
\end{figure*}

The design space of the streaming aggregation implementation is presented as a set of decision trees, grouped into two categories (Fig. \ref{design-space}):
\begin{itemize}
	\item
	\textit{Category A} consists of decision trees that refer to memory configuration and allocation. Cache configuration options (private cache for each core or shared cache for all cores) are depicted in decision tree \textit{A4}. \textit{A5} is related with the dynamic memory allocation that can be based on freelists or in \textit{malloc}/\textit{free} system calls. 
	\item
	In \textit{category B} are assigned decision trees related to data movement and means by which accesses to shared resources are synchronized. The first three decision trees refer to different ways that data can be copied from global to local memories, or from one local memory to another (depending on the embedded system's memory hierarchy). Decision trees \textit{B4} and \textit{B5} are about synchronization between PEs, when accessing shared buffers. At low level, synchronization can be accomplished by spinning on shared variables (i.e. busy waiting) or by using other platform specific solutions. In platforms that run OS and support POSIX threads developers can utilized semaphores or monitors.
\end{itemize}

Apparently, not all design options are applicable in any context. Fig. \ref{design-space} shows the application and the hardware constraints that affect which decision trees or leaves are applicable in each specific context. The  constraints are used to prune the decision trees and leaves that yield implementations which do not adhere to developer's requirements or they are not supported by the embedded platform. 

\begin{table}[]
	\centering
	\caption{Decision trees or leaves disabled for each application and hardware constraint.}
	\label{constraints}
	\begin{tabular}{ll}
		\hline\noalign{\smallskip}
		\multirow{2}{*}{App./Hw constraint} & Decision tree/    \\
		& leaf disabled \\
		\hline\noalign{\smallskip}
		Windowing(tuple-based) 			    & A2, A3, A5, B4  \\
		Window configuration                & may disable A1(local)  \\
		Programming(bare metal)             & B3 and B5        \\
		Programming(pthread)             	& B1, B3, B4       \\
		Programming(OpenCL)             	& B1, B2, B4, B5     \\
		Cache config.(no)       			& A4     \\
		Access to local/global(no)      	& A1, A2              
	\end{tabular}
\end{table}

Table \ref{constraints} summarizes the design options that are disabled, due to application and hardware constraints. As an example, if the embedded platform runs an OS, access to DMA and to low-level signaling mechanisms are most likely handled by the OS directly, so these design options are not exposed to developers. \textit{Window configuration} constraint may force the allocation of the data structures in a global memory. All constraints are provided manually. Constraints that prune non-compatible design space options "convert" the platform-independent design space into platform-dependent. Thus, they make the customization approach applicable in different contexts and in various embedded platforms. 

After the pruning, valid customized streaming aggregation implementations are instantiated from the remaining decision tree leaves of the design space. In other words, the implementations that will finally be explored are the ones that are produced by combining the remaining leaves to create consistent implementations. Each one of these combinations is a valid customized solution that should be evaluated. All combinations of the remaining tree leaves are evaluated by brute-force exploration.

\subsubsection{Methodology description}

The exploration methodology consists of two steps and it is presented in Fig. \ref{methodology}. The inputs of the methodology are the application and hardware constraints. The output is a streaming aggregation implementation with customized software and hardware parameters. 

\begin{figure} [t]
	\centering
	\includegraphics[width=0.89\textwidth, page=4]{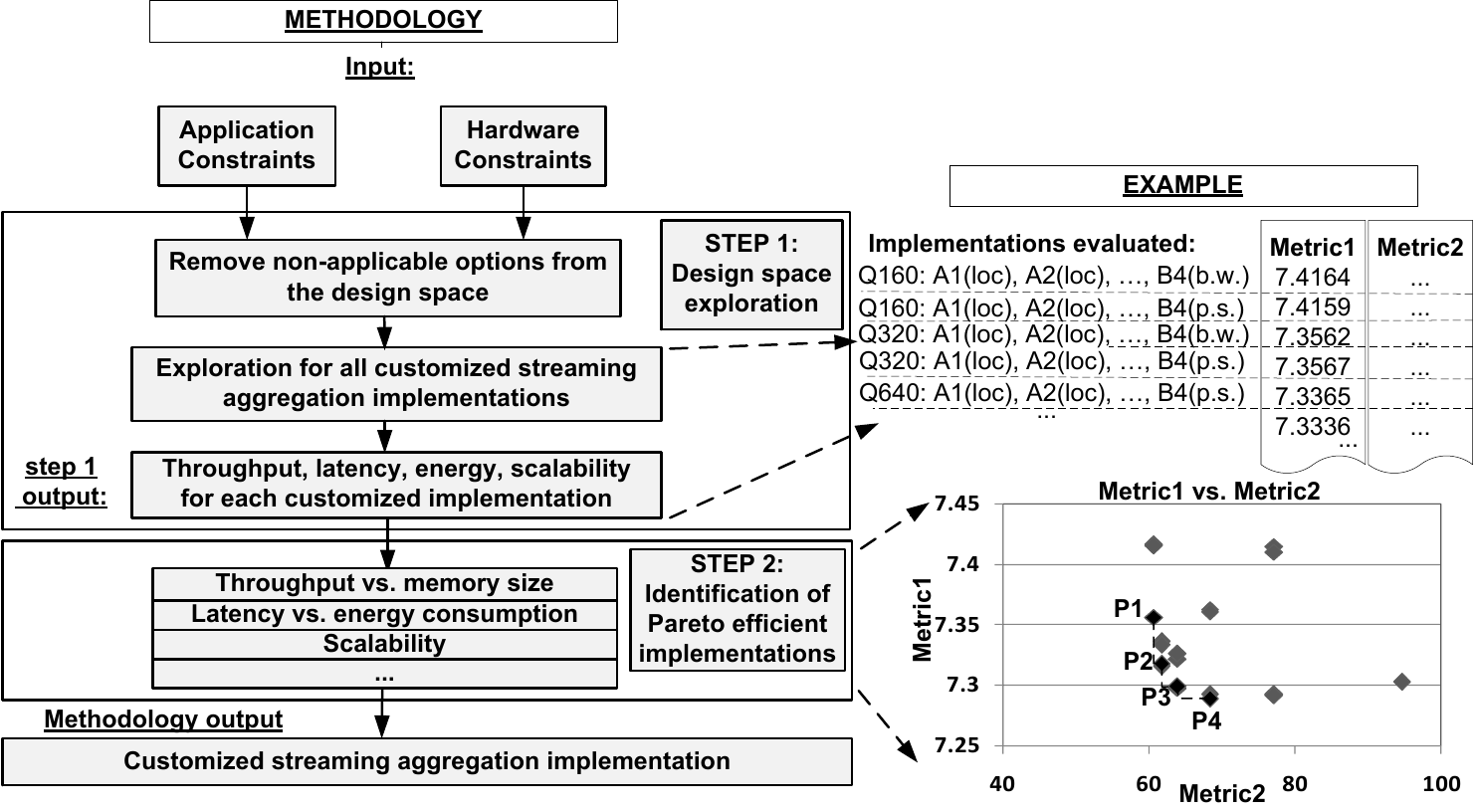}
	\caption{Customization methodology.}
	\label{methodology}
\end{figure}

The first step of the methodology aims at the pruning of the design space and the implementation of the design space exploration. First, the non-applicable options are removed from the design space due to the application and hardware constraints. Then, the streaming aggregation is executed once for each different combination of the decision tree leaves of the design space. For each customization, throughput, latency, memory size and energy consumption results are gathered. Scalability is another metric that can be evaluated, in case there is a relatively large number of PEs available. In the second step, the Pareto efficient implementations are identified. The trade-offs that can be performed by customization of the streaming aggregation on an embedded platform are presented in the form of Pareto curves. Developers can select the implementation that is most efficient according to the optimization target. 

The tool flow that supports the methodology consists of a set of bash shell scripts that handle the first step of the methodology. For the second phase, the design space pruning and the exploration are performed automatically, provided that the hardware constraints are set manually. All performance results are collected automatically. However, power (which is used to calculate energy consumption) is measured manually, since it is usually based on platform-specific hardware instrumentation. Also, the tool flow integrates a script that calculates the Pareto curve for each requested pair of metrics. 

Finally, it is important to state that most design options are normally provided as functions, macros, or compiler directives from either the platform SDK, or from the POSIX/OpenCL libraries. Therefore, it should not require significant programming effort by developers to switch between the design options presented in Fig. \ref{design-space}. Although the number of available implementations in some cases is increased, the systematic methodology we propose guarantees that all Pareto efficient implementations can be identified. 
%
\subsection{Demonstration of the Methodology}

In this Section we first provide a short description of the embedded architectures that we used for demonstration of the methodology. Then, we present the experimental setup and the evaluation results, which are discussed in the last subsection. 

\subsubsection{Platforms description}
Myriad embedded processors are designed by Movidius Ltd. \cite{movidius}. They target computer vision and data streaming applications. Myriad architectures are utilized in the context of Project Tango, which aims at the design of mobile devices capable of creating a 3D model of the environment around them \cite{tango}. They belong to the family of low power mobile processors and provide increased performance per watt \cite{perf-per-watt}. 

\begin{figure}
	\centering
	\includegraphics[width=0.65\textwidth, page=5]{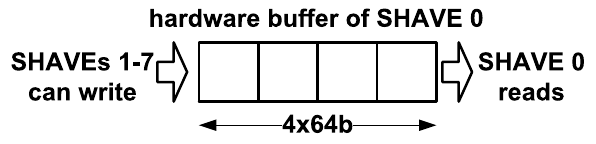}
	\caption{Myriad1 hardware buffers.}
	\label{sf}
\end{figure}

Myriad1 architecture is designed at 65nm. It  integrates 8 VLIW processing cores named Streaming Hybrid Architecture Vector Engine (SHAVEs) operating at 180MHz and a LEON3 processor that controls the data flow, handles interrupts, etc.. More technical information about Myriad1 can be found in \cite{myriad1}. A local DMA engine is available for each SHAVE. Additionally, Myriad1 provides a set of hardware buffers for direct communication between the SHAVE cores. Each SHAVE has its own hardware buffer and they are accessed in FIFO manner. The size of each one is 4x64 bit words.  As shown in Fig. \ref{sf}, each SHAVE can push data into the buffer of any other SHAVE and it can read data only from its own buffer. A SHAVE writes to the tail of another buffer and the owner of the buffer can read from the head. An interesting feature of the Myriad1 hardware buffers is the fact that when a SHAVE tries to write to a full FIFO or read from its own FIFO that happens to be empty, it stalls and enters a low energy mode. We take advantage of this, in order to propose energy efficient streaming aggregation implementations on Myriad1 platform. 

Myriad2 is designed at 28nm \cite{myriad2}. In contrast with Myriad1, Myriad2 integrates 12 SHAVE cores operating at 504MHz, along with two independent LEON4 processors: LEON-RT targeting job management and LEON-OS suitable for running RTEMS/Linux, etc.. Myriad2 provides a single top-level DMA engine and the hardware buffers size is 16x64 words. 

Regarding the memory specifications, Myriad1 provides 1MB local memory with unified address space that it is named Connection Matrix (CMX). Each 128KB are directly linked to each SHAVE processor providing local storage for data and instruction code. Therefore, the CMX memory can be seen as a group of 8 memory "slices", with each slice being connected to each one of the 8 SHAVEs. Each SHAVE accesses its own CMX slice more efficiently in comparison with the rest CMX slices. Myriad2 CMX memory is 2MB and each slice is 128KB. Also, Myriad2 provides 1KB L1 and 256KB L2 cache. Finally, both platforms provide a global DDR memory of 64MB. 

\begin{figure}[!t]
	\centering
	\subfloat[Implementation of time-based aggregation on Myriad.\label{myriad-time-mem-alloc}]{
		\includegraphics[width=0.85\textwidth, page=6]{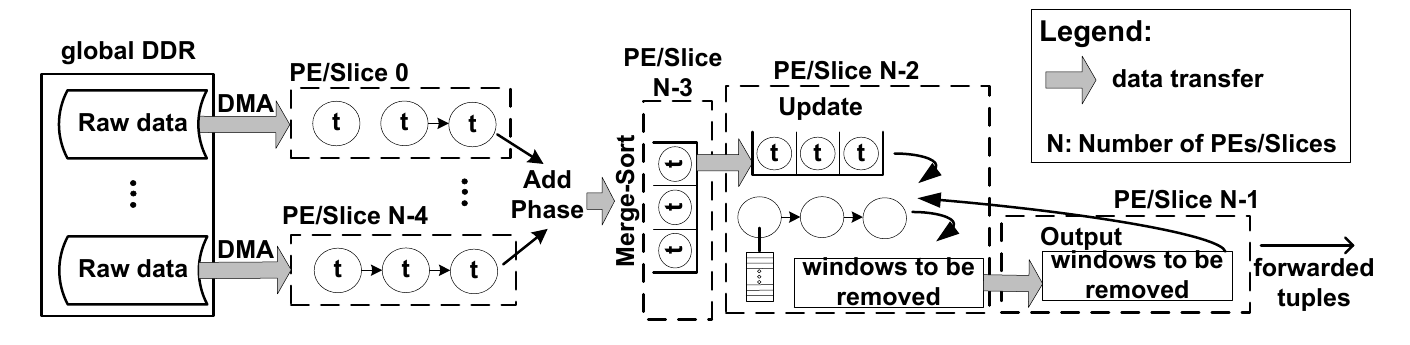}
	}
	\hfill
	\subfloat[Implementation of count-based aggregation on Myriad.\label{myriad-count-mem-alloc}]{
		\includegraphics[width=0.85\textwidth, page=8]{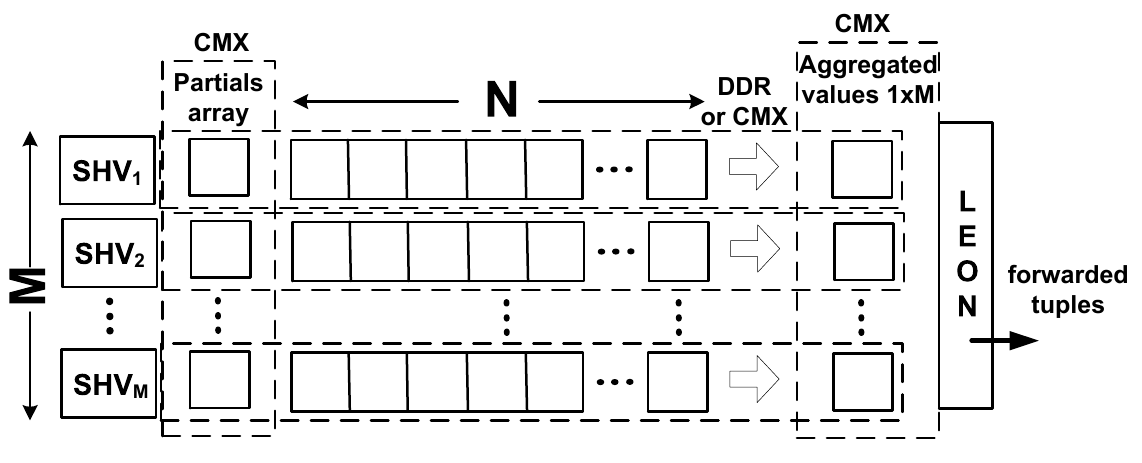}
	}
	\caption{Implementation of time-based and count-based streaming aggregation on Myriad.}
	\label{mem-alloc}
\end{figure}

Concerning the memory allocation of the time-based streaming aggregation data structures, the incoming streams of raw data (produced by sensors, cameras, etc.) are placed in DDR memory. Each input queue is handled by a different SHAVE and it is placed in its local slice. Each SHAVE that handles an input queue fetches chunks of raw data in its own memory slice, by using DMA transfers. Then, it converts the raw data into tuples and stores them in its own input queue. The windows are stored in a linked list data structure, which is allocated in the CMX slice of the SHAVE core that handles the Update phase. Memory allocation and other implementation details are displayed in Fig. \ref{myriad-time-mem-alloc}. Regarding the count-based aggregation scenario that uses a \textit{M}x\textit{N} window, each one of the \textit{M} SHAVEs continuously fetches raw data that correspond to \textit{N} tuples from DDR to CMX. However, if \textit{N} is very large and tuples cannot be stored and processed in CMX, they are placed and aggregated in DDR. Each SHAVE computes the aggregated value of \textit{N} tuples and forwards the result to LEON, which produces the output tuple that corresponds to the specific window. The implementation diagram in Fig. \ref{myriad-count-mem-alloc}.

Freescale I.MX 6 Quad integrates four ARM Cortex A9 cores that operate at 1GHz \cite{freescale}. It belongs to a family of multicore ARM-based platforms that target single board computers and run Linux-based OS. It provides 1GB RAM and two cache memory levels. On I.MX.6 the raw data are placed in data files. Chunks of raw data are fetched in RAM using \textit{freed()} function. Then, tuples are created and placed in the input queues to be forwarded to the subsequent streaming aggregation phases. 

Exynos 5 octa is an ARM-based platform that targets mobile computers. It is designed at 28nm by SAMSUNG and it is based on big.LITTLE architecture \cite{exynos}. It integrates two ARM clusters: 4 Cortex-A15 and 4 Cortex-A7 cores. Exynos 5 integrates a PowerVR SGX544 GPU that supports OpenCL1.1. It includes 3 processing cores running at 533MHz. The evaluation board integrating Exynos is the Odroid-XU that provides 2GB DDR3 RAM  \cite{odroid}. In the context of this work, we used PowerVR GPU to perform aggregation in the count-based streaming scenario, implemented in OpenCL.

\subsubsection{Experimental Setup}

\begin{table*}[]
	\centering
	\caption{Hardware constraints for Myriad1, Myriad2, I.MX.6 Quad and Exynos for both scenarios.}
	\label{hw-constraints}
	\small
	\begin{tabular}{c|ccc|ccc}
		\hline\noalign{\smallskip}
		& \multicolumn{3}{|c|}{Time-based aggregation} & \multicolumn{3}{|c}{Count-based aggregation} \\
		\hline\noalign{\smallskip}
		& Myriad1       & Myriad2       & I.MX.6     & Myriad1        & Myriad2        & Exynos    \\
		\hline\noalign{\smallskip}
		windowing                                                            & time          & time          & time       & count          & count          & count     \\
		\hline\noalign{\smallskip}
		programming                                                          & bare metal    & bare metal    & pthread    & bare metal     & bare metal     & OpenCL    \\
		\hline\noalign{\smallskip}
		cache config.                                                        & no            & yes           & no         & no             & yes            & no        \\
		\hline\noalign{\smallskip}
		\begin{tabular}[c]{@{}c@{}}access local/global mem.\end{tabular} & yes           & yes           & no         & yes            & yes            & yes      
	\end{tabular}
\end{table*}

The dataset we used to demonstrate the proposed methodology has been collected from the online audio distribution platform SoundCloud \cite{soundcloud}. It consists of a subset of approximately 40,000 users that exchanged comments between 2007 and 2013. The incoming tuples contain the following attributes: \textit{timestamp}, \textit{user\_id}, \textit{song\_id} and \textit{comment}. The aggregation function forwards the id of the user with the largest number of comments in each window. In the time-based aggregation scenario the window is sliding, while in the count-based, the window is tumbling, so the aggregated value is calculated over the last \textit{M}x\textit{N} tuples. 

The aggregation operator is implemented entirely in C. Throughput is measured as tuples processed per second, while latency as the timestamp difference between an output tuple with the aggregated value and the latest input tuple that produced it. The energy consumption results on I.MX.6 were obtained based on hardware instrumentation using a Watts Up PRO meter device and following a setup similar to methods proposed in the literature \cite{power1}\cite{power2}. In Myriad2 power was measured though the MV198 power measurement board integrated on Myriad2 evaluation board. In Myriad1 power was estimated, based on moviSim simulator provided by Movidius MDK. In Exynos it is measured based on power sensors that are provided by Odroid-XU-e evaluation board \cite{odroid}. All the values presented are the average of 10 executions, by elimination of the outliers. Each single experiment is executed from 30 seconds up to one minute.

The time-based aggregation scenario, which is actually a pipeline, is demonstrated in Myriad and I.MX.6 Quad platforms. The count-based scenario, that provides increased data parallelism, is demonstrated in Myriad and in Exynos embedded GPU. As stated earlier, Myriad1 provides 8 PEs. In time-based aggregation, each one of the merge-sort, update and output phases is assigned to a single PE. Each one of the remaining 5 PEs handles a single input queue. In Myriad2, which integrates 12 PEs, the input queues are 9. In I.MX.6 Quad that provides 4 PEs, we assigned each phase on single PE and the remaining PE handles 5 input queues. 

The hardware constraints of the evaluation boards are presented in Table \ref{hw-constraints}. The experiments we performed are the following: In the time-based aggregation scenario, in I.MX.6 we implemented the methodology using a single window configuration. However, for Myriad1 and Myriad2, we present results for two different scenarios: in the first one the window configuration (i.e. the window \textit{size} and \textit{advance} values) are set, so that the maximum memory size of the windows data structure is small enough to fit in the local memory. In the second experiment, the windows data structure can only fit in the global memory. Thus, we study how the memory allocation of the windows data structure affects the evaluation metrics. In the count-based scenario, the aggregation is performed in parallel by the accelerator of each platform: The SHAVEs in Myriad and the GPU in Exynos. 

The output of the methodology is a set of Pareto points for throughput vs. memory size and latency vs. energy consumption. In time-based scenario, we present results for scalability for Myriad1 and Myriad2. The implementations that are evaluated for scalability are the ones that were found to be Pareto efficient in latency vs. energy consumption evaluation. 

\subsubsection {Time-based aggregation results}
In the time-based scenario, we evaluate each implementation for a number of queue sizes. The queue sizes we select are the ones that provide latency below a fixed threshold. Therefore, we first measure latency for a range queue sizes and select the size values which provide latency below the threshold. Then, we proceed to the implementation of the methodology. 48 implementations are evaluated in Myriad and 4 in I.MX.6 Quad. The number of implementations that are evaluated can be reduced by selecting a smaller number of queue size values. (However, in this case fewer Pareto points may be identified). 

\paragraph{Demonstration on Myriad1}

\begin{figure}[t]
	\centering
	\subfloat[Windows list in local mem..\label{m1-lat-qs-won}]{
		\includegraphics[width=0.47\textwidth, page=1]{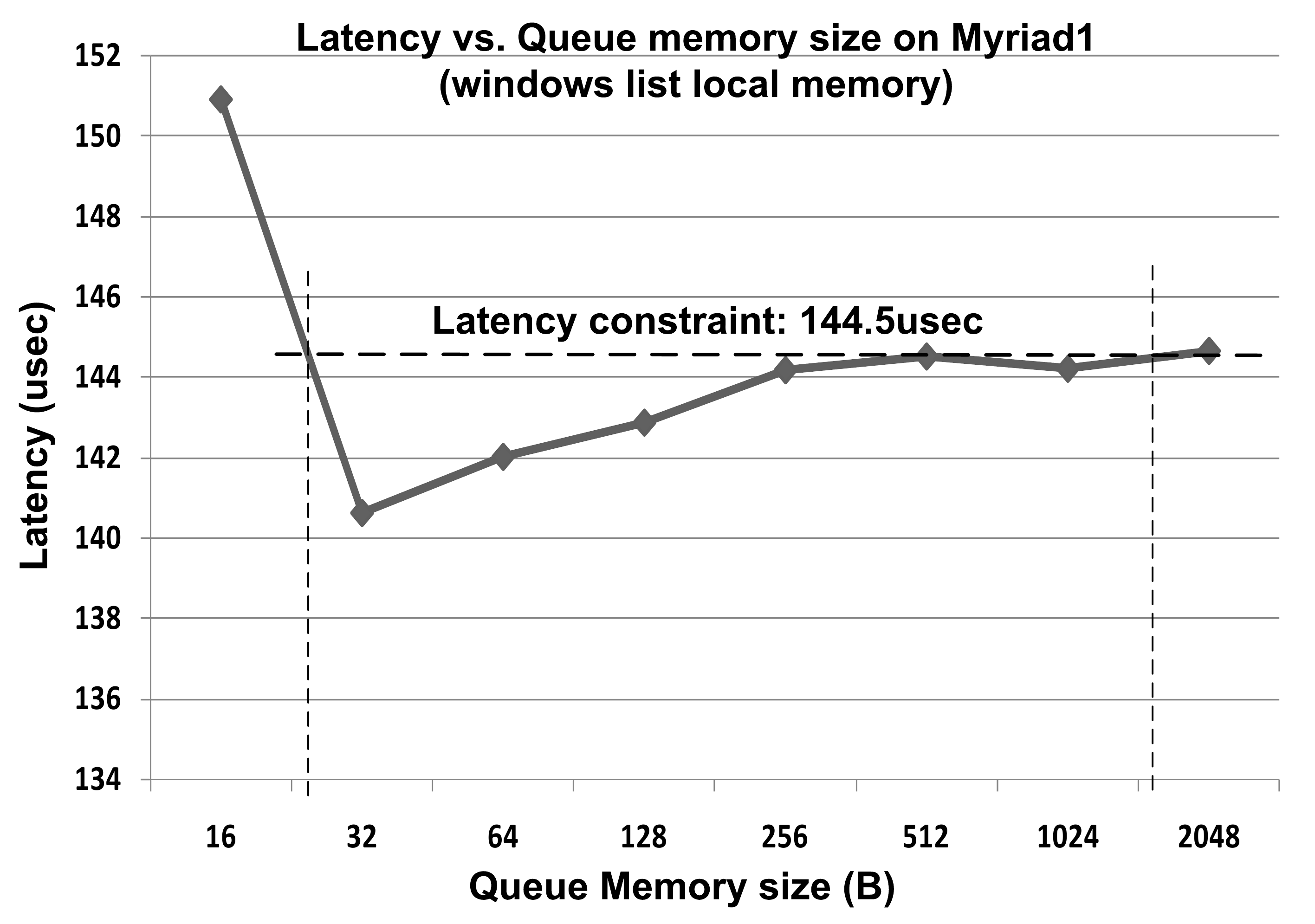}
	}
	\subfloat[Windows list in global mem..\label{m1-lat-qs-woff}]{
		\includegraphics[width=0.47\textwidth, page=2]{./figs-chalmers/stream/exp.pdf}
	}
	\caption{Latency vs. Queue size on Myriad1.}
	\label{m1-lat-qs}
\end{figure}

\begin{figure*}[!t]
	\subfloat[][Throughput vs. memory footprint  \\(Windows in local memory)\label{m1-won-thr-mem}]{
		\includegraphics[width=0.47\textwidth, page=3]{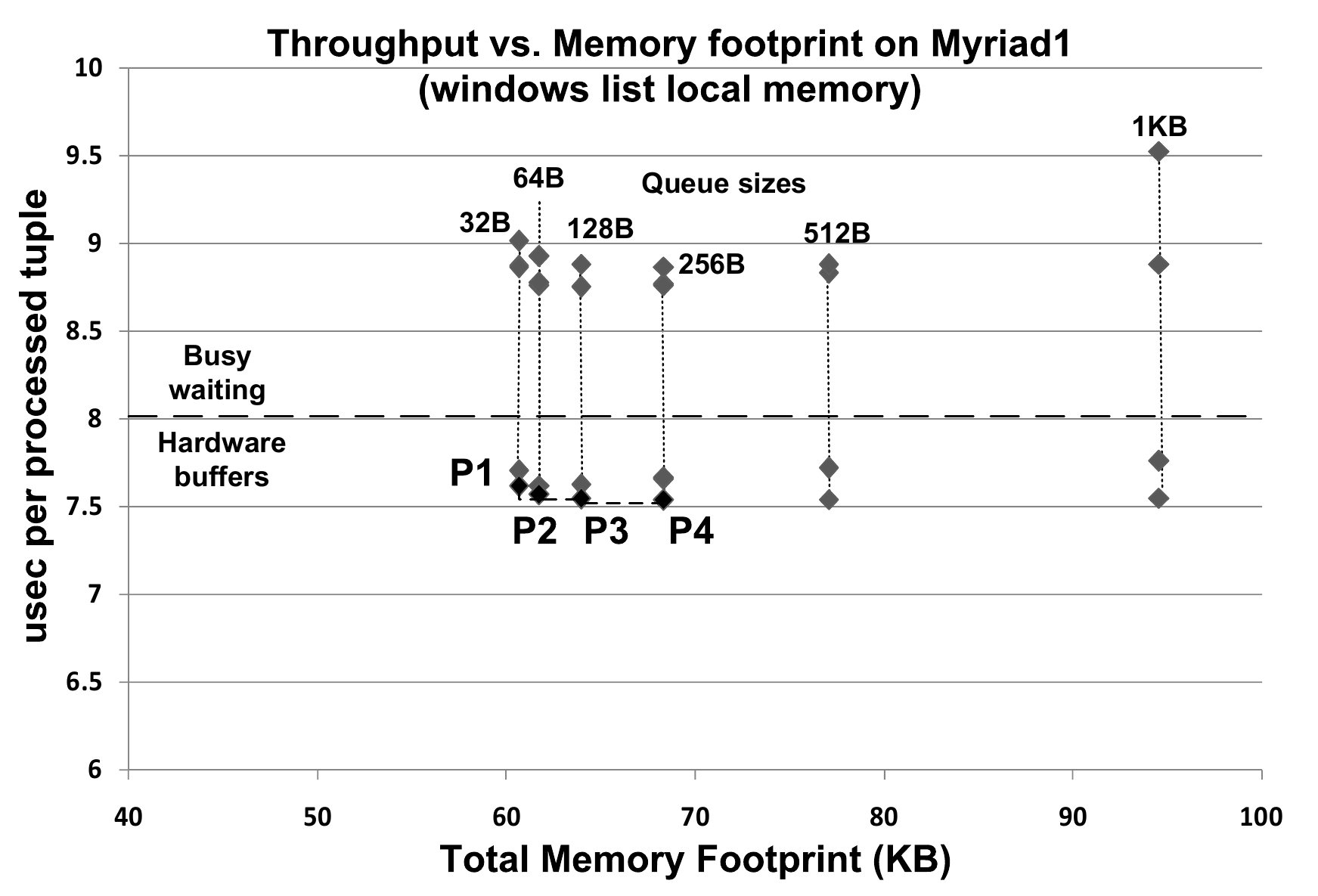}
	}
	\subfloat[][Latency vs. energy consumption \\(Windows in local memory)\label{m1-won-lat-en}]{
		\includegraphics[width=0.47\textwidth, page=4]{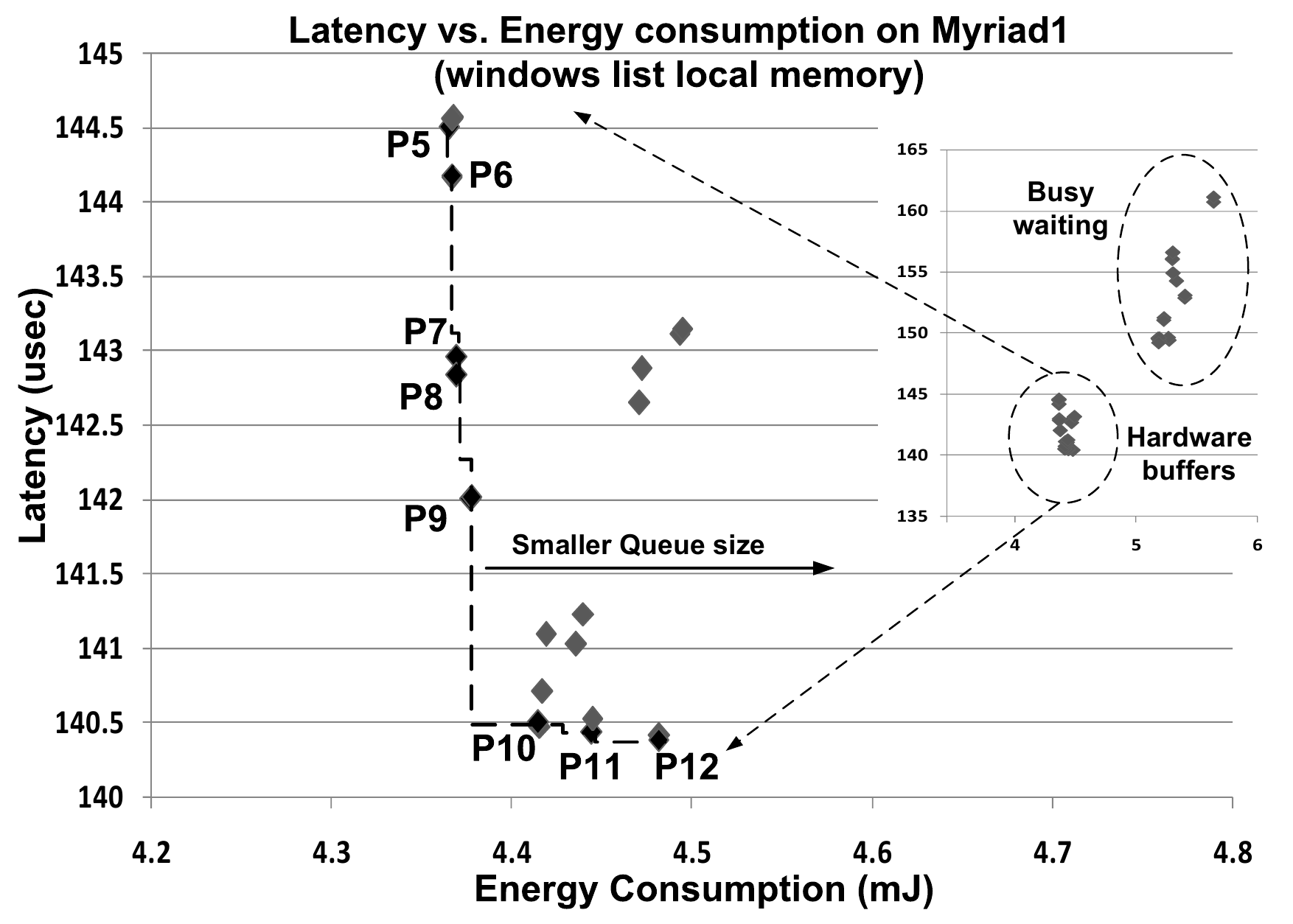}
	}
	\hfill
	\subfloat[Scalability (Windows in local memory) \label{m1-won-scal}]{
		\includegraphics[width=0.47\textwidth, page=5]{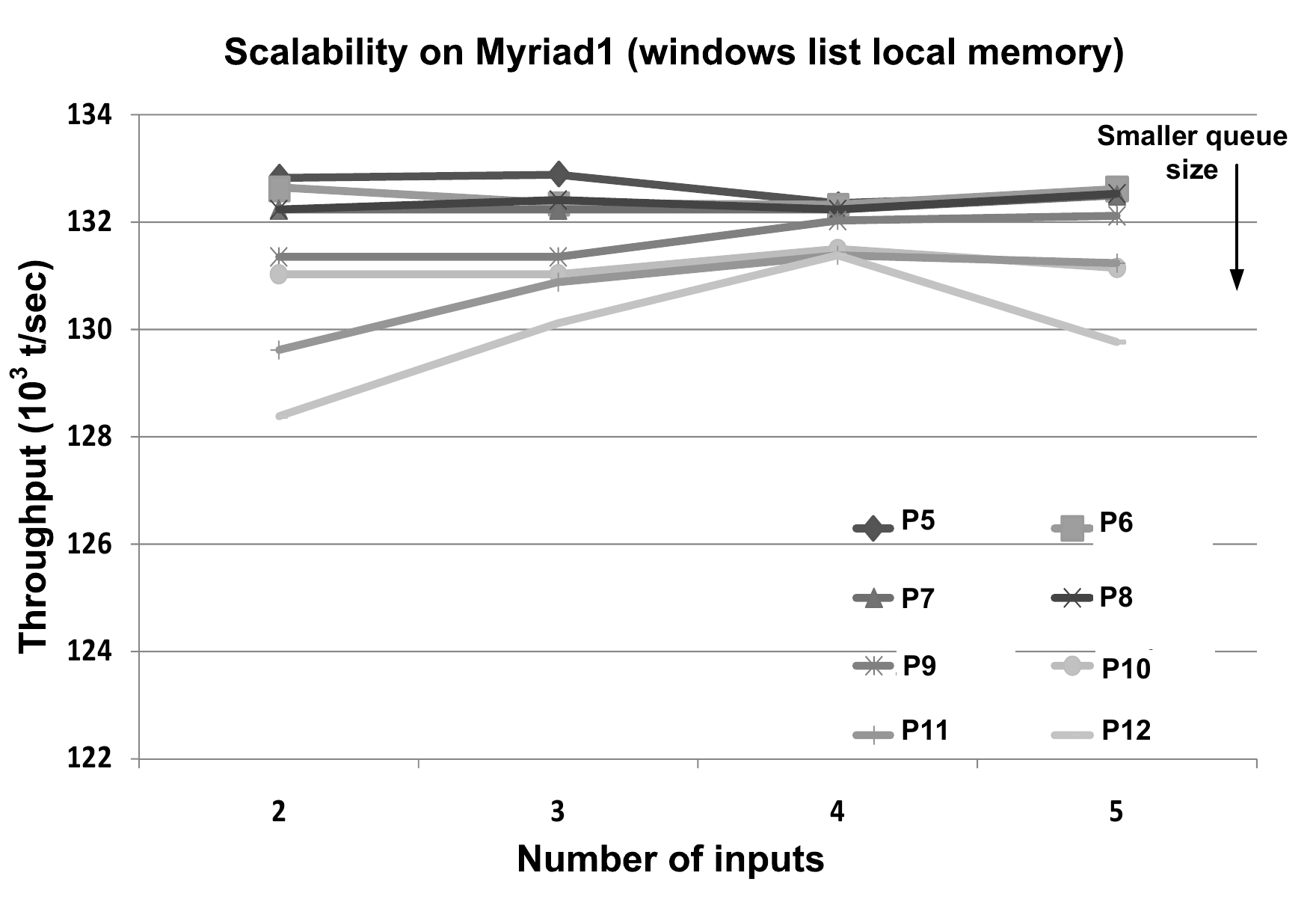}
	}
	\subfloat[][Throughput vs. memory footprint\\ (Windows in global memory)\label{m1-woff-thr-mem}]{
		\includegraphics[width=0.47\textwidth, page=6]{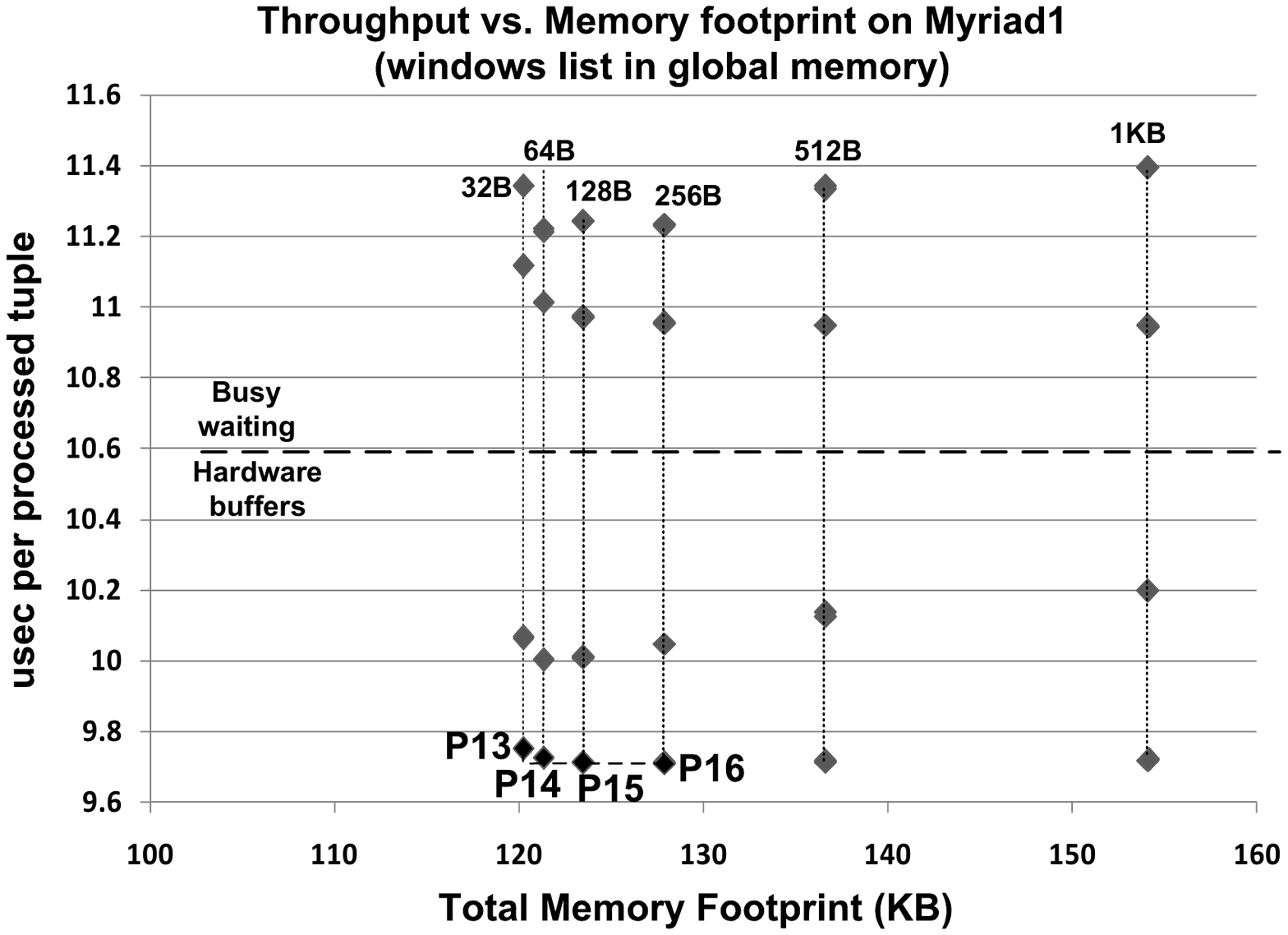}
	}
	\caption{Evaluation of time-based streaming aggregation implementations on Myriad1.}
	\label{m1-results}
\end{figure*}

\begin{figure*}[!t]
	\subfloat[][Latency vs. energy consumption\\ (Windows in global memory)\label{m1-woff-lat-en}]{
		\includegraphics[width=0.47\textwidth, page=7]{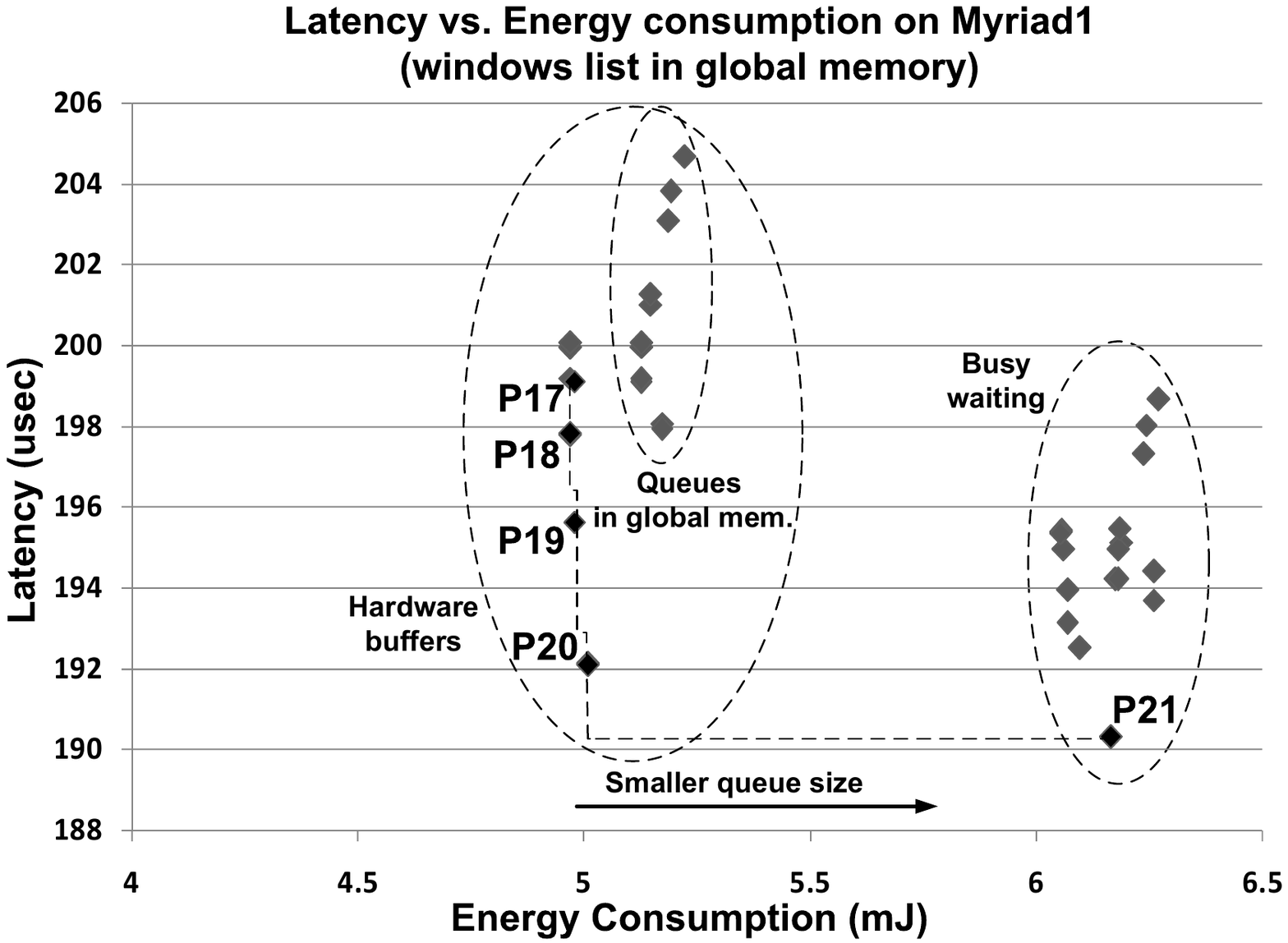}
	}
	\subfloat[Scalability (Windows in global memory)\label{m1-woff-scal}]{
		\includegraphics[width=0.47\textwidth, page=8]{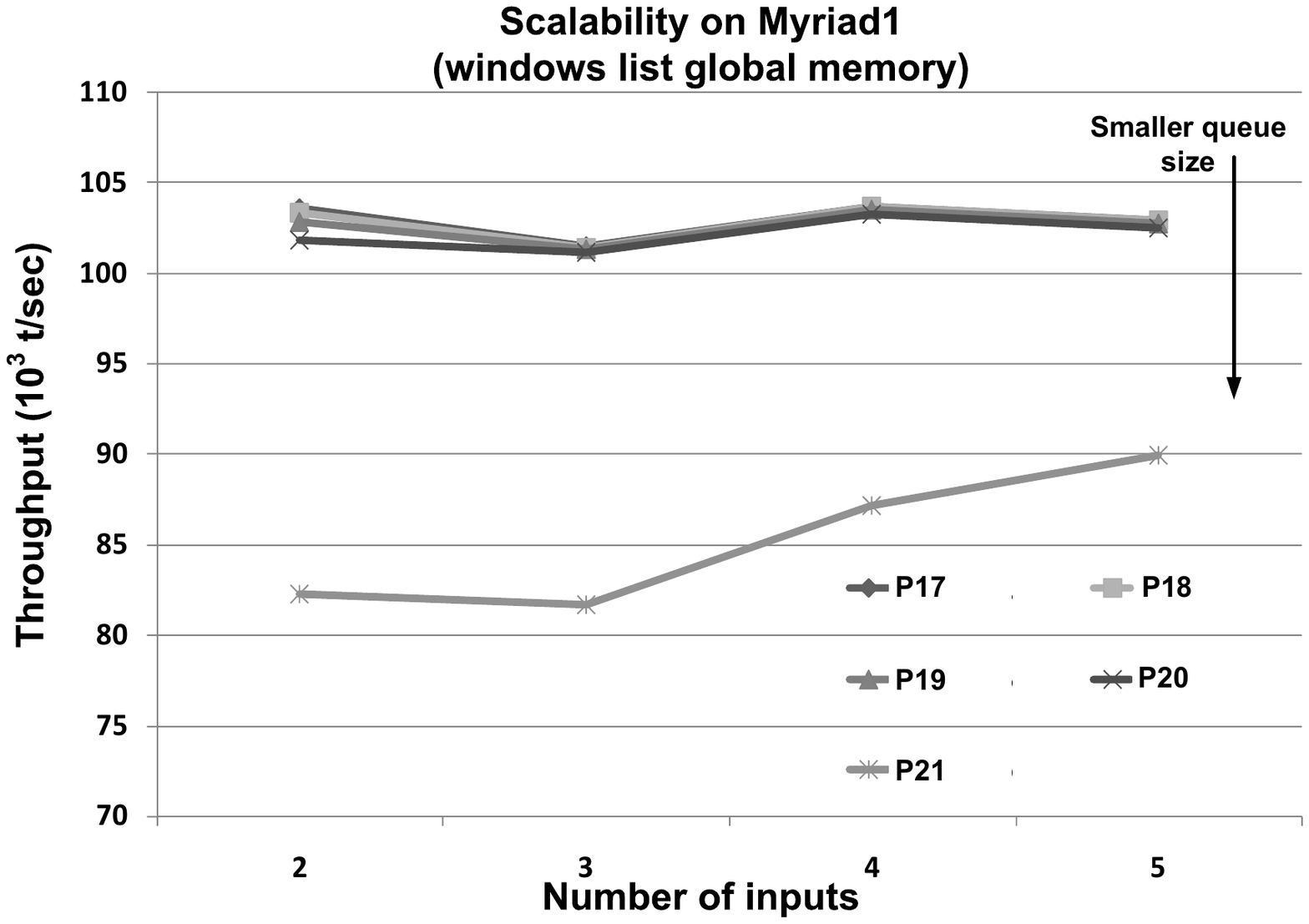}
	}
	\caption{Evaluation of time-based streaming aggregation implementations on Myriad1.}
	\label{m1-results}
\end{figure*}
\begin{table*}[!t]
	\centering
	\caption{Myriad1 Pareto efficient points description. B4(p.s.) (i.e. platform specific) refers to Myriad hardware buffers. }
	\label{m1-pareto}
	\scriptsize
	\begin{tabular}{ll|ll|ll}
		\hline\noalign{\smallskip} 
		Pareto & Description               & Pareto & Description               & Pareto & Description               \\
		\hline\noalign{\smallskip} 
		P1     & A1(l), A2(l), A3(32B),  & P8     & A1(l), A2(l), A3(128B), & P15    &A1(l), A2(l), A3(128B), \\
		& A5(fl), B2(yes), B4(p.s.) &        & A5(fl), B2(yes), B4(b.w.)  &        & A5(fl), B2(yes), B4(b.w.)  \\
		P2     & A1(l), A2(l), A3(64B),  & P9     & A1(l), A2(l), A3(64B),  & P16    & A1(on), A2(on), A3(256B), \\ 
		& A5(fl), B2(yes), B4(p.s.) &        & A4(fl), B1(yes), B4(p.s.) &        & A5(fl), B2(yes), B4(p.s.) \\ \cdashline{5-6}
		P3     & A1(l), A2(l), A3(128B), & P10    & A1(l), A2(l), A3(64B),  & P17    & A1(l), A2(l), A3(256B), \\
		& A5(fl), B2(yes), B4(p.s.) &        & A5(fl), B2(yes), B4(p.s.) &        & A5(fl), B1(yes), B4(p.s.) \\
		P4     & A1(l), A2(l), A3(256B), & P11    & A1(l), A2(l), A3(64B),  & P18    & A1(l), A2(l), A3(128B), \\ 
		& A5(fl), B2(yes), B4(p.s.) &        & A5(fl), B1(yes), B4(p.s.) &        & A5(fl), B2(yes), B4(p.s.) \\ \cdashline{1-2}
		P5     & A1(l), A2(l), A3(512B), & P12    & A1(l), A2(l), A3(32B),  & P19    & A1(l), A2(l), A3(64B),  \\ 
		& A5(fl), B2(yes), B4(p.s.) &        & A5(fl), B2(yes), B4(p.s.) &        & A5(fl), B2(yes), B4(p.s.) \\ \cdashline{3-4}
		P6     & A1(l), A2(l), A3(256B), & P13    & A1(l), A2(l), A3(32B),  & P20    & A1(l), A2(l), A3(32B),  \\
		& A5(fl), B2(yes), B4(p.s.) &        & A5(fl), B2(yes), B4(p.s.) &        & A5(fl), B2(yes), B4(p.s.) \\
		P7     & A1(l), A2(l), A3(128B), & P14    & A1(l), A2(l), A3(64B),  & P21    & A1(l), A2(l), A3(32B),  \\
		& A5(fl), B1(yes), B4(p.s.) &        & A5(fl), B2(yes), B4(p.s.) &        & A5(fl), B2(yes), B4(b.w.)
	\end{tabular}
\end{table*}

In the first experiment in Myriad1 the window size and advance values are configured so that the windows data structure can fit in the local memory. Assuming latency constraint of 144.5usec, the range of queue sizes that we evaluate are from 32B to 1024B (Fig. \ref{m1-lat-qs-won}). 

The results for throughput vs. memory evaluation are displayed in Fig. \ref{m1-won-thr-mem}. We notice that the Pareto points can be divided in two categories: The ones with performance lower than 8.0usec/tuple that correspond to implementations that utilize busy waiting and the rest ones that utilize the Myriad hardware buffers. (In both axes, the lower the values, the higher the efficiency). 4 Pareto efficient points are identified, which are described in Table \ref{m1-pareto}. All Pareto efficient customized implementations can be used to perform trade-offs between throughput and memory: throughput can increase up to 1.02\% and maximum memory size can drop up to 11.2\% by selecting P4 and P1 solutions respectively. 

Pareto points of latency vs. energy can be grouped into the same categories: The ones that exploit busy waiting and the rest that utilize hardware buffers. The later are more efficient both in terms of latency and energy consumption. 8 Pareto points can be identified that can be used to perform trade-offs between the aforementioned metrics: up to 2.85\% lower latency (P12) and up to 2.6\% lower energy consumption (P5). 

Finally, scalability evaluation of the Pareto points of latency vs. energy is shown in Fig. \ref{m1-won-scal}. Throughput remains almost constant for all implementations or increases with the number of inputs. The only exception is P12, in which the queues have very small size (32B). 

In the second experiment, we assume latency threshold to be 202usec (Fig. \ref{m1-lat-qs-woff}). We notice in both Fig. \ref{m1-woff-thr-mem} and Fig. \ref{m1-woff-lat-en} that throughput is lower and latency higher in comparison with the previous experiment, since in this one the windows are placed in the global memory. The Pareto efficient points demonstrated in Fig. \ref{m1-woff-thr-mem} can be used to perform trade-offs between throughput and memory size (up to 0.5\% for throughput by selecting P16 and up to 5.9\% in memory size by selecting P13). In Fig. \ref{m1-woff-lat-en}, we notice that Pareto point P21 is the most efficient in terms of latency (4.45\% lower in comparison with P17), while P17 implementation is the most energy efficient (19.3\% lower consumption than P21). In the scalability evaluation of Fig. \ref{m1-woff-scal}, it is shown that all implementations provide high throughput that it is affected by the number of inputs only slightly, apart from P21 that utilizes busy-waiting and yields much lower throughput in comparison with the rest of the implementations. 

\begin{figure}[t]
	\centering
	\subfloat[Windows list in local mem..\label{m2-lat-qs-won}]{
		\includegraphics[width=0.47\textwidth, page=9]{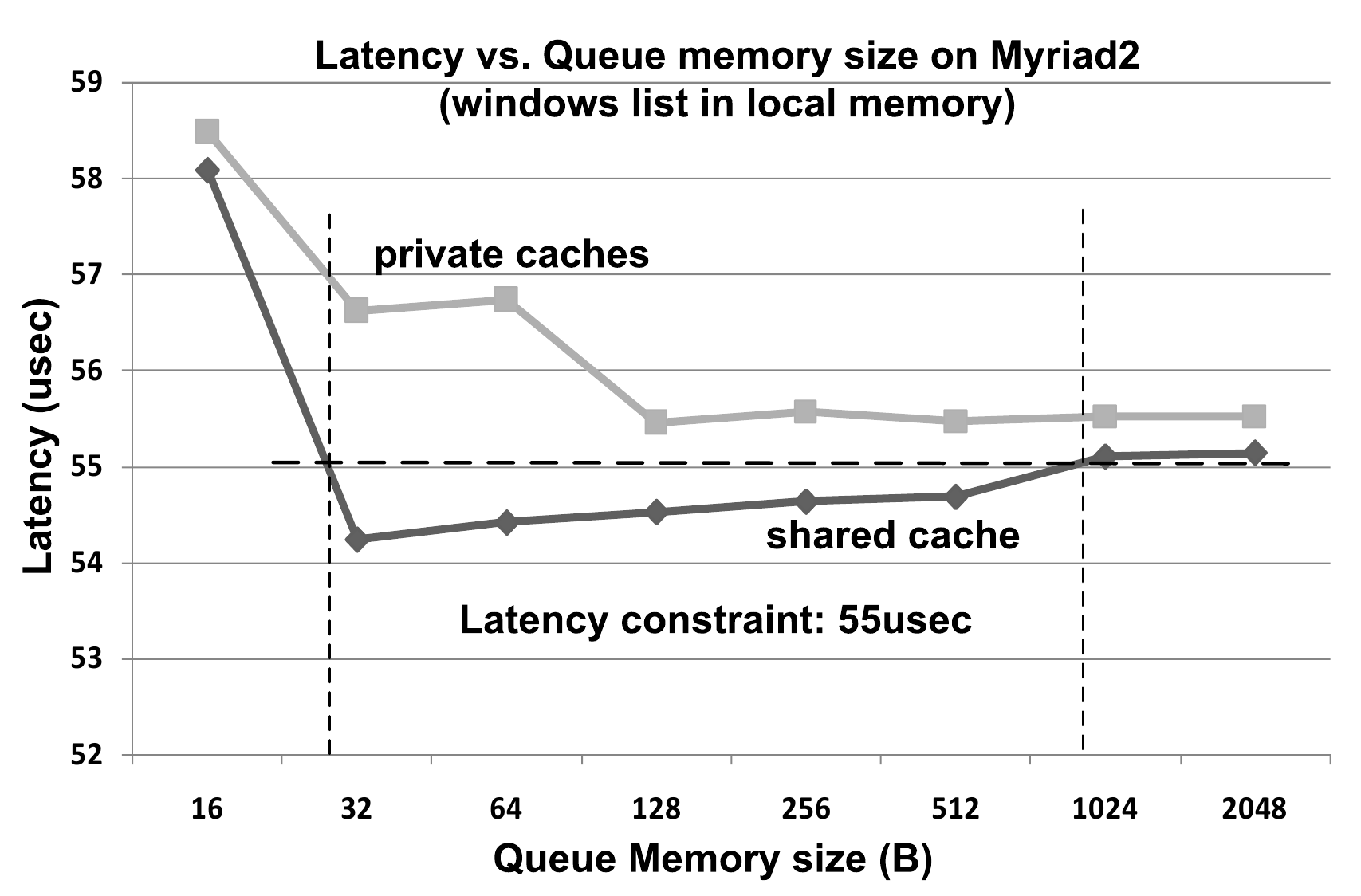}
	}
	\subfloat[Windows list in global mem..\label{m2-lat-qs-woff}]{
		\includegraphics[width=0.47\textwidth, page=10]{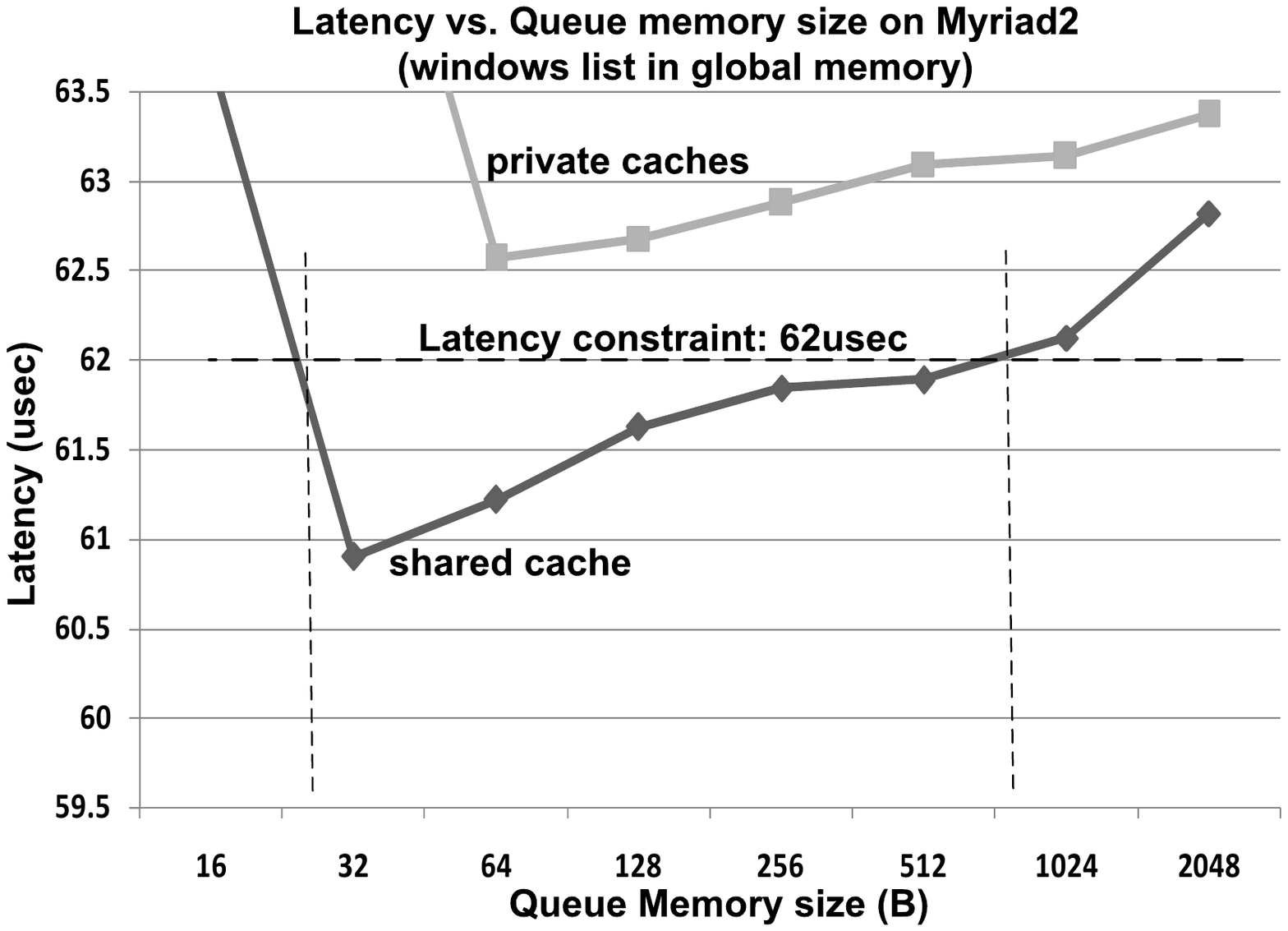}
	}
	\caption{Latency vs. Queue size on Myriad2.}
	\label{m2-lat-qs}
\end{figure}

\paragraph{Demonstration on Myriad2}

\begin{figure*}[!t]
	\subfloat[][Throughput vs. memory footprint  \\(Windows in  local memory)\label{m2-won-thr-mem}]{
		\includegraphics[width=0.32\textwidth, page=11]{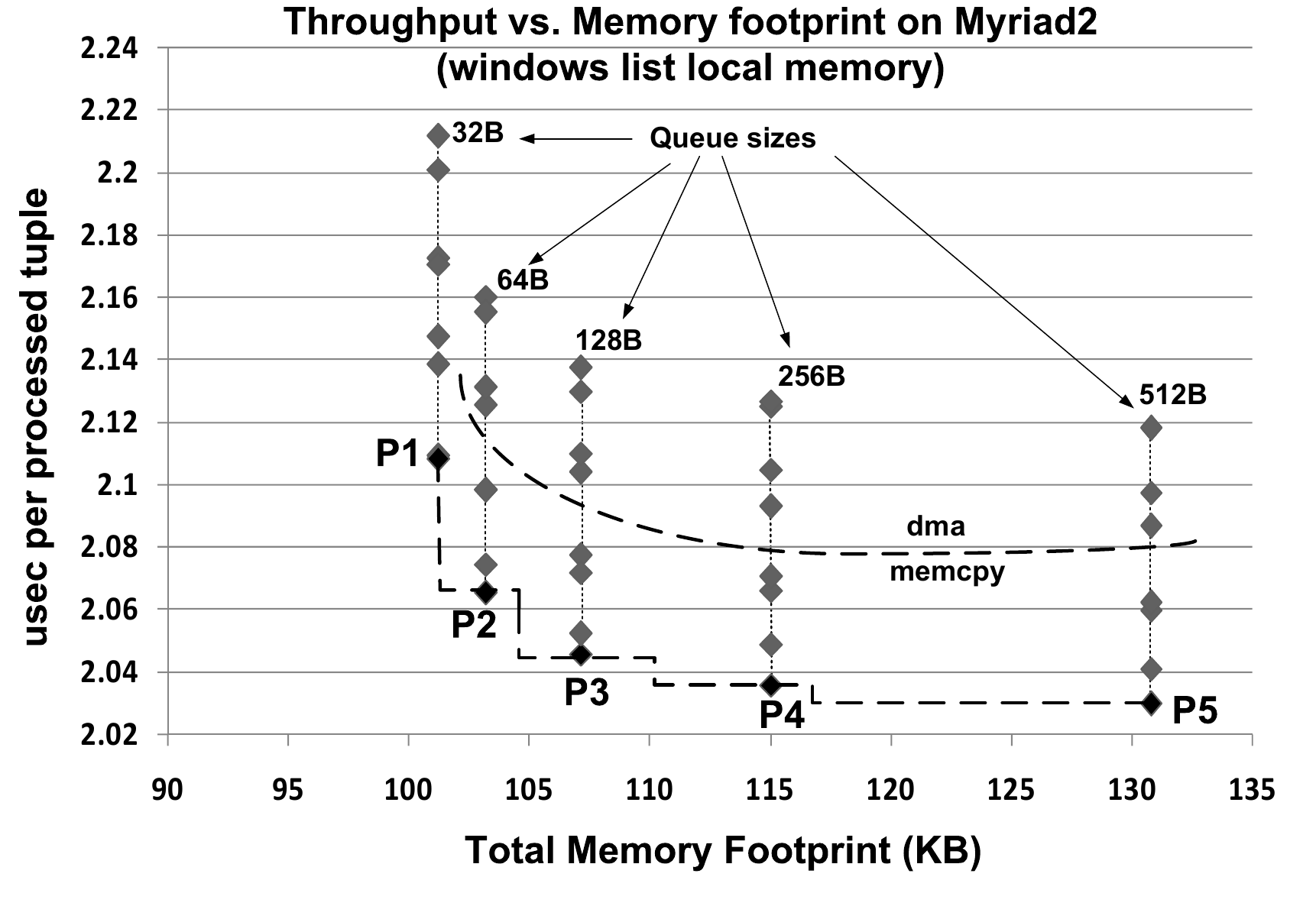}
	}
	\subfloat[][Latency vs. energy consumption \\(Windows in local memory) \label{m2-won-lat-en}]{
		\includegraphics[width=0.32\textwidth, page=12]{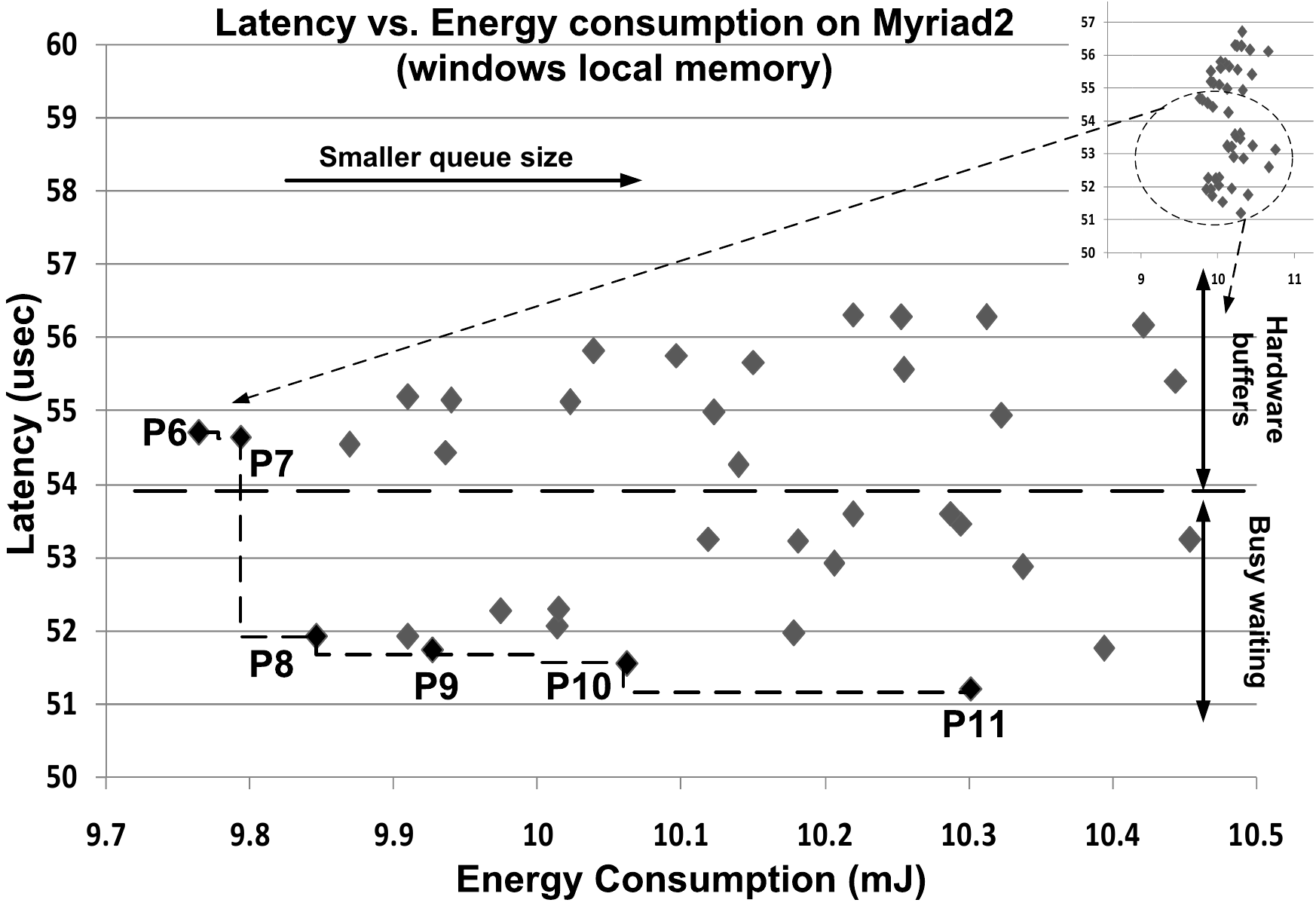}
	}
	\subfloat[Scalability (Windows in local memory) \label{m2-won-scal}]{
		\includegraphics[width=0.32\textwidth, page=13]{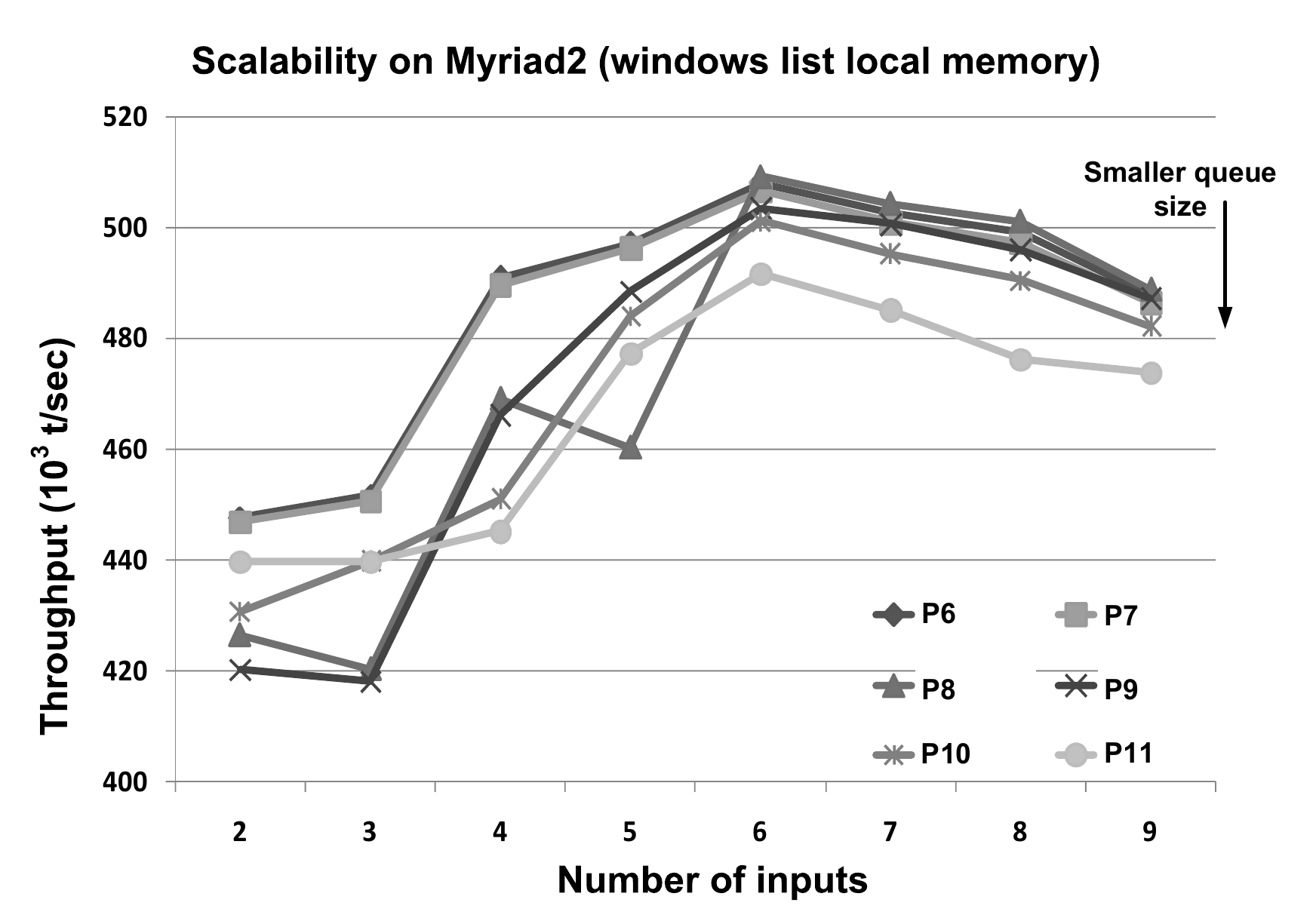}
	}
	\hfill
	\subfloat[][Throughput vs. memory footprint \\ (Windows in global memory)\label{m2-woff-thr-mem}]{
		\includegraphics[width=0.32\textwidth, page=14]{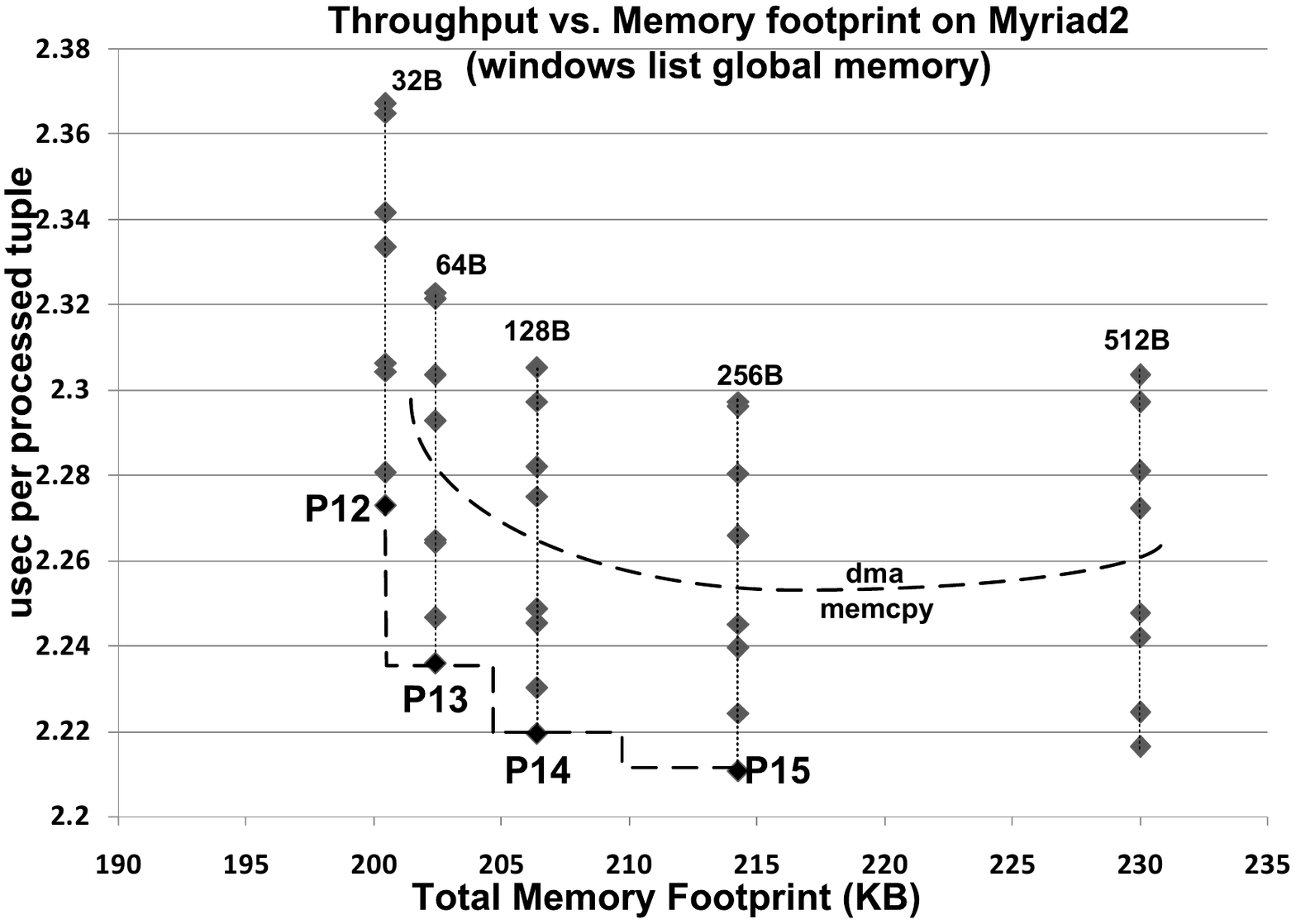}
	}
	\subfloat[][Latency vs. energy consumption \\(Windows in global memory)\label{m2-woff-lat-en}]{
		\includegraphics[width=0.32\textwidth, page=15]{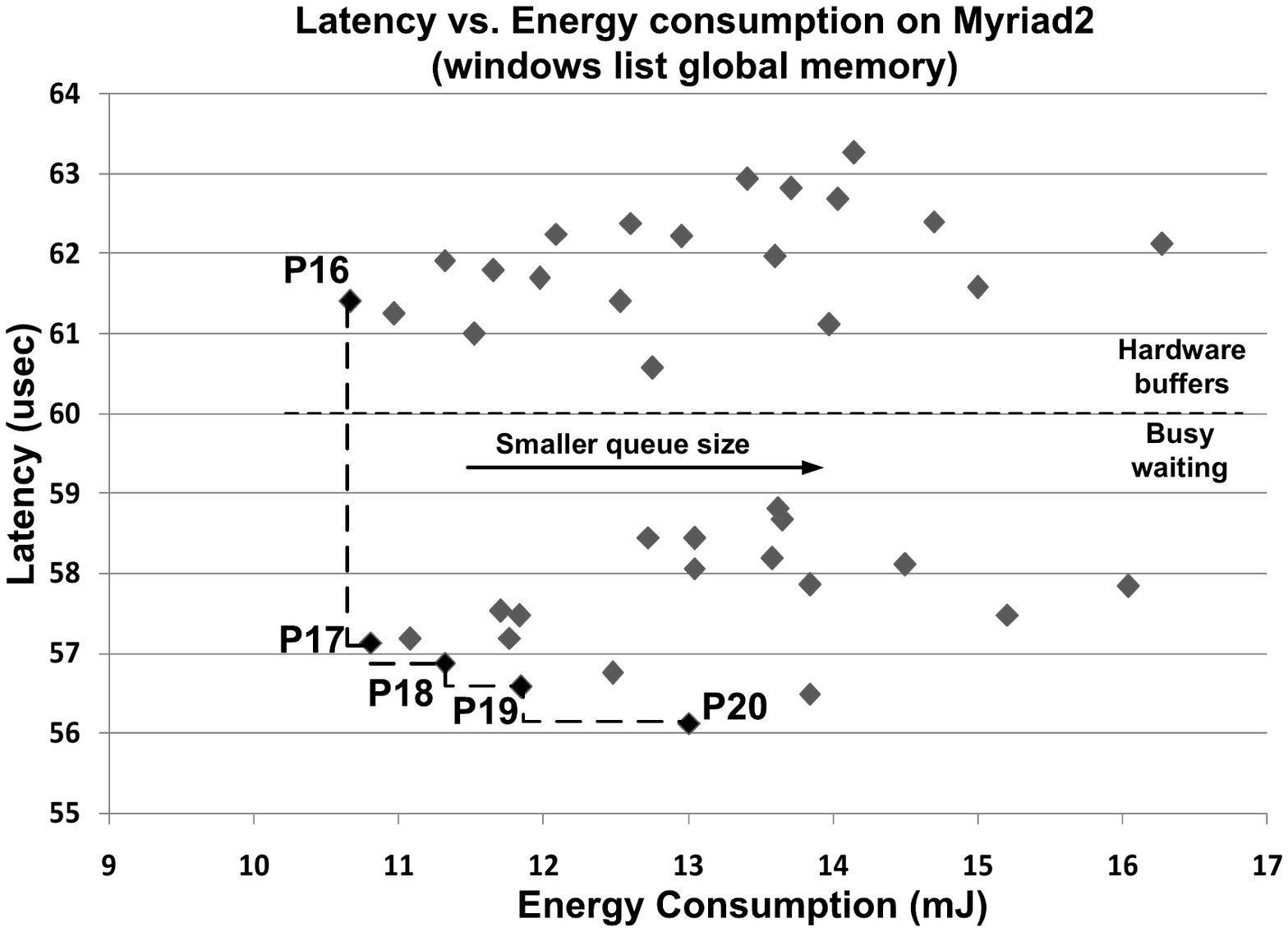}
	}
	\subfloat[Scalability (Windows in global memory)\label{m2-woff-scal}]{
		\includegraphics[width=0.32\textwidth, page=16]{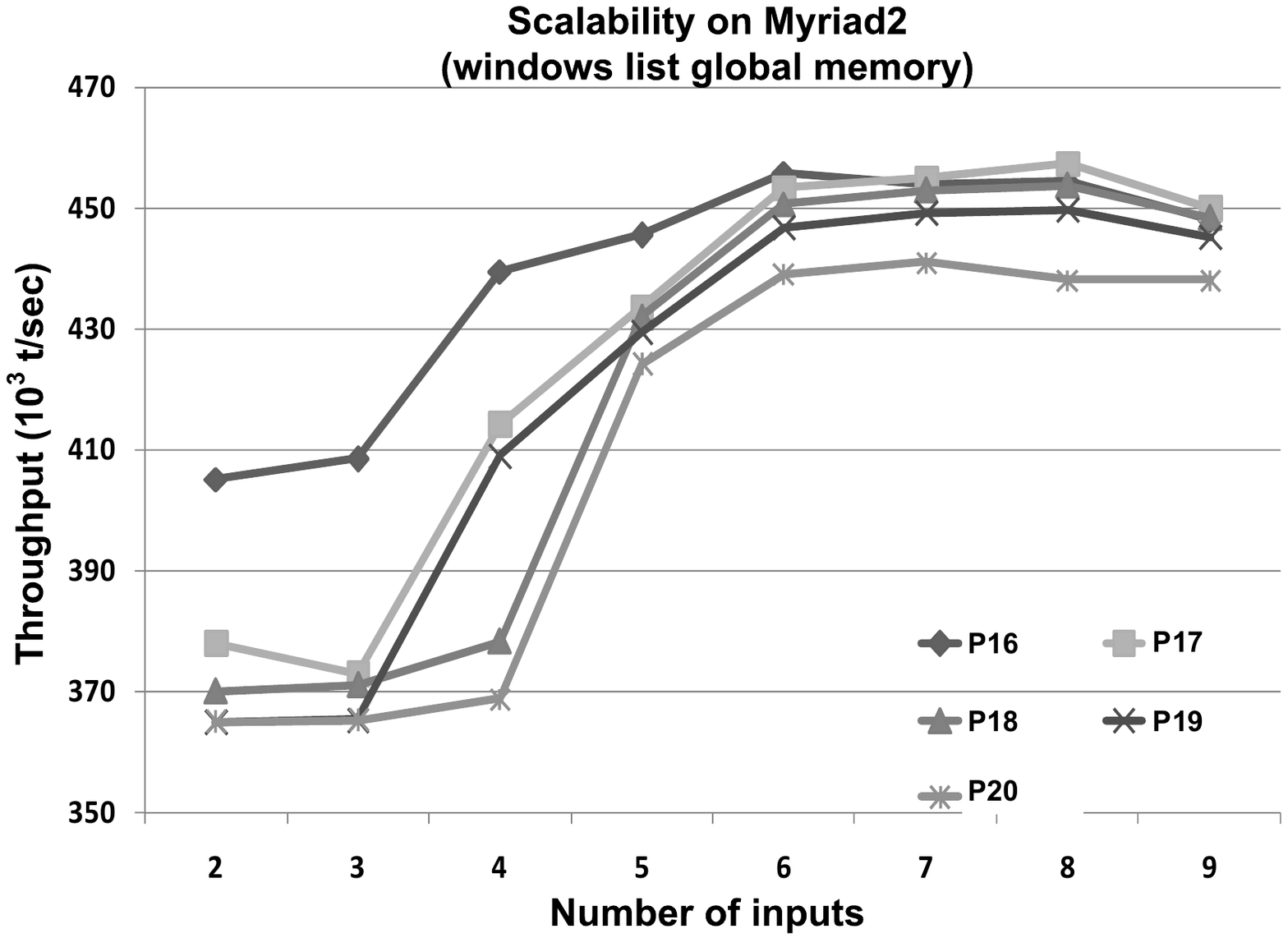}
	}
	\caption{Evaluation of time-based streaming aggregation implementations on Myriad2.}
	\label{m2-results}
\end{figure*}

\begin{table*}[!t]
	\centering
	\caption{Myriad2 Pareto efficient points description. B4(p.s.) (i.e. platform specific) refers to Myriad hardware buffers}
	
	\label{m2-pareto}
	\scriptsize
	\begin{tabular}{ll|ll|ll}
		\hline\noalign{\smallskip} 
		Par. & Description               & Par. & Description               & Par. & Description               \\
		\hline\noalign{\smallskip} 
		P1     & A1(l), A2(l), A3(32B),  & P8     & A1(l), A2(l), A3(512B), & P15    & A1(l), A2(l), A3(256B), \\ 
		& A4(s), A5(fl), B2(y), B4(b.w.) &        & A4(s), A5(fl), B2(y), B4(b.w.) &        & A4(s), A5(fl), B2(y), B4(p.s.) \\ \cdashline{5-6}
		P2     & A1(l), A2(l), A3(64B),  & P9     & A1(l), A2(l), A3(128B), & P16    & A1(l), A2(l), A3(256B), \\
		& A4(s), A5(fl), B2(yes), B4(p.s.) &        & A4(s), A5(fl), B1(y), B4(b.w.) &        & A4(s), A5(fl), B2(y), B4(p.s.) \\
		P3     & A1(l), A2(l), A3(128B), & P10    & A1(l), A2(l), A3(64B),  & P17    & A1(l), A2(l), A3(512B), \\
		& A4(s), A5(fl), B2(y), B4(p.s.) &        & A4(s), A5(fl), B2(y), B4(b.w.) &        & A4(s), A5(fl), B1(y), B4(b.w.) \\
		P4     & A1(l), A2(l), A3(256B), & P11    & A1(l), A2(l), A3(32B),  & P18    & A1(l), A2(l), A3(128B), \\ 
		& A4(s), A5(fl), B2(y), B4(p.s.) &        & A4(s), A5(fl), B1(y), B4(b.w.) &        & A4(s), A5(fl), B2(y), B4(b.w.) \\ \cdashline{3-4}
		P5     & A1(l), A2(l), A3(512B), & P12    & A1(l), A2(l), A3(32B),  & P19    & A1(l), A2(l), A3(64B),  \\ 
		& A4(s), A5(fl), B2(y), B4(p.s.) &        & A4(s), A5(fl), B2(y), B4(b.w.) &        & A4(s), A5(fl), B2(y), B4(b.w.) \\ \cdashline{1-2}
		P6     & A1(l), A2(l), A3(512B), & P13    & A1(l), A2(l), A3(64B),  & P20    & A1(l), A2(l), A3(32B),  \\
		& A4(s), A5(fl), B2(y), B4(p.s.) &        & A4(s), A5(fl), B2(y), B4(p.s.) &        & A4(s), A5(fl), B2(y), B4(b.w.) \\
		P7     & A1(l), A2(l), A3(256B), & P14    & A1(l), A2(l), A3(128B), &        &                           \\
		& A4(s), A5(fl), B1(y), B4(p.s.) &        &A4(s), A5(fl), B2(y), B4(p.s.) &        &                          
	\end{tabular}
\end{table*}

\begin{figure*}[t]
	\def\tabularxcolumn#1{m{#1}}
	\begin{tabularx}{\linewidth}{@{}cXX@{}}
		\begin{tabular}{cc}
			\subfloat[Latency vs. Queue size\label{imx6-lat-qs}]{\includegraphics[width=0.52\textwidth, page=17]{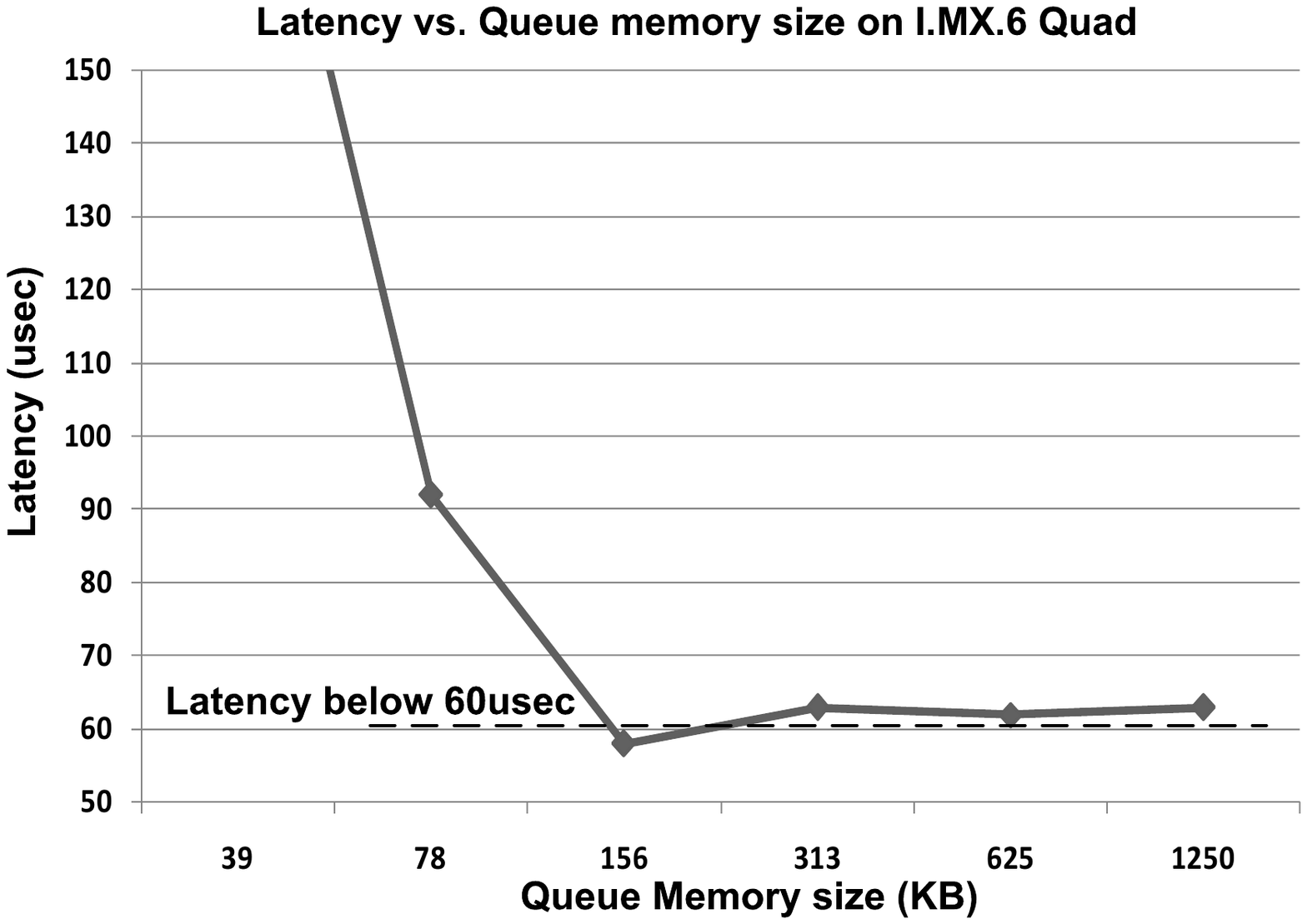}} 
			& \subfloat[Throughput evaluation\label{imx6-thr-mem}]{\includegraphics[width=0.52\textwidth, page=18]{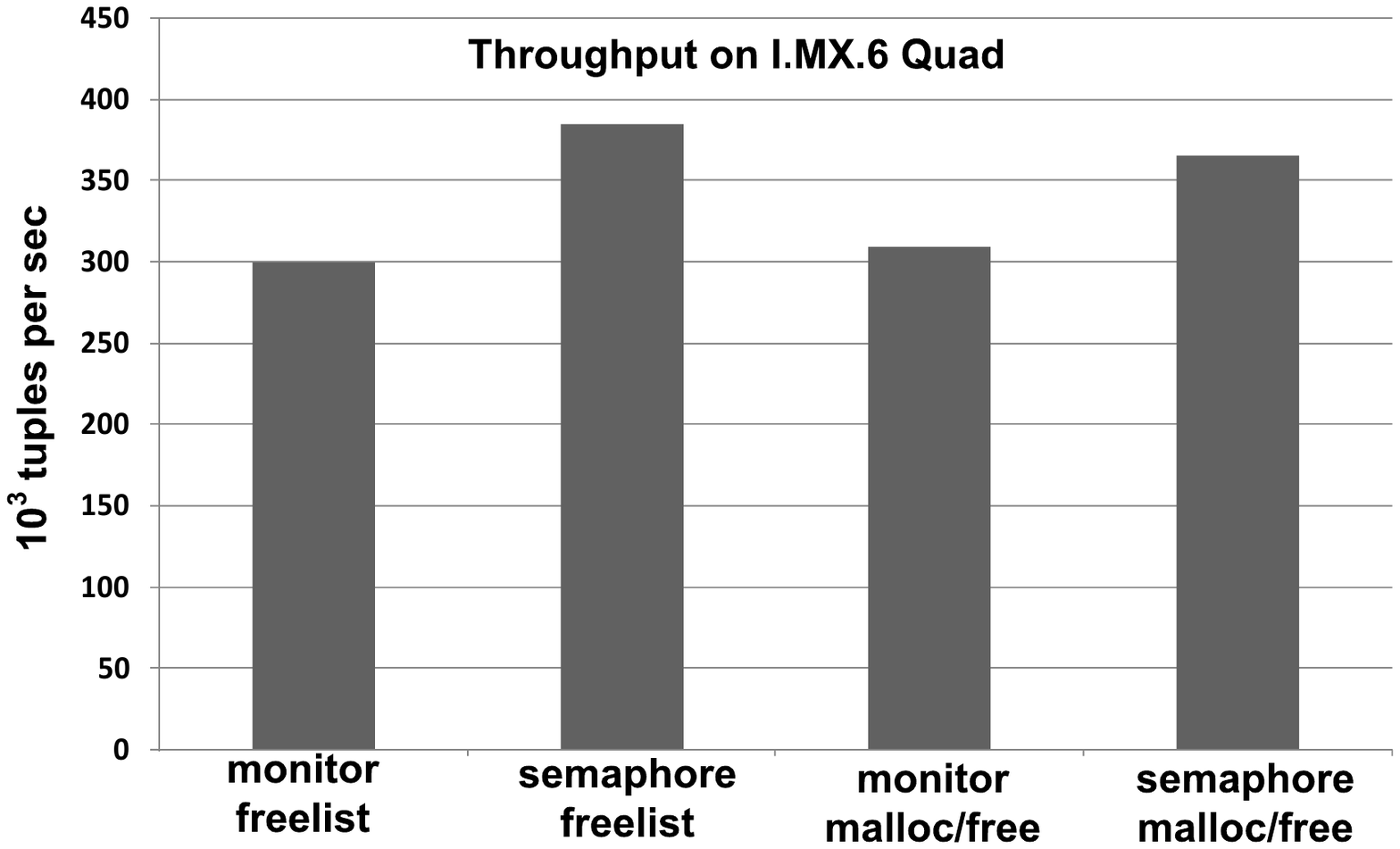}}\\
			\subfloat[Latency vs. energy consumption\label{imx6-lat-en}]{\includegraphics[width=0.52\textwidth, page=19]{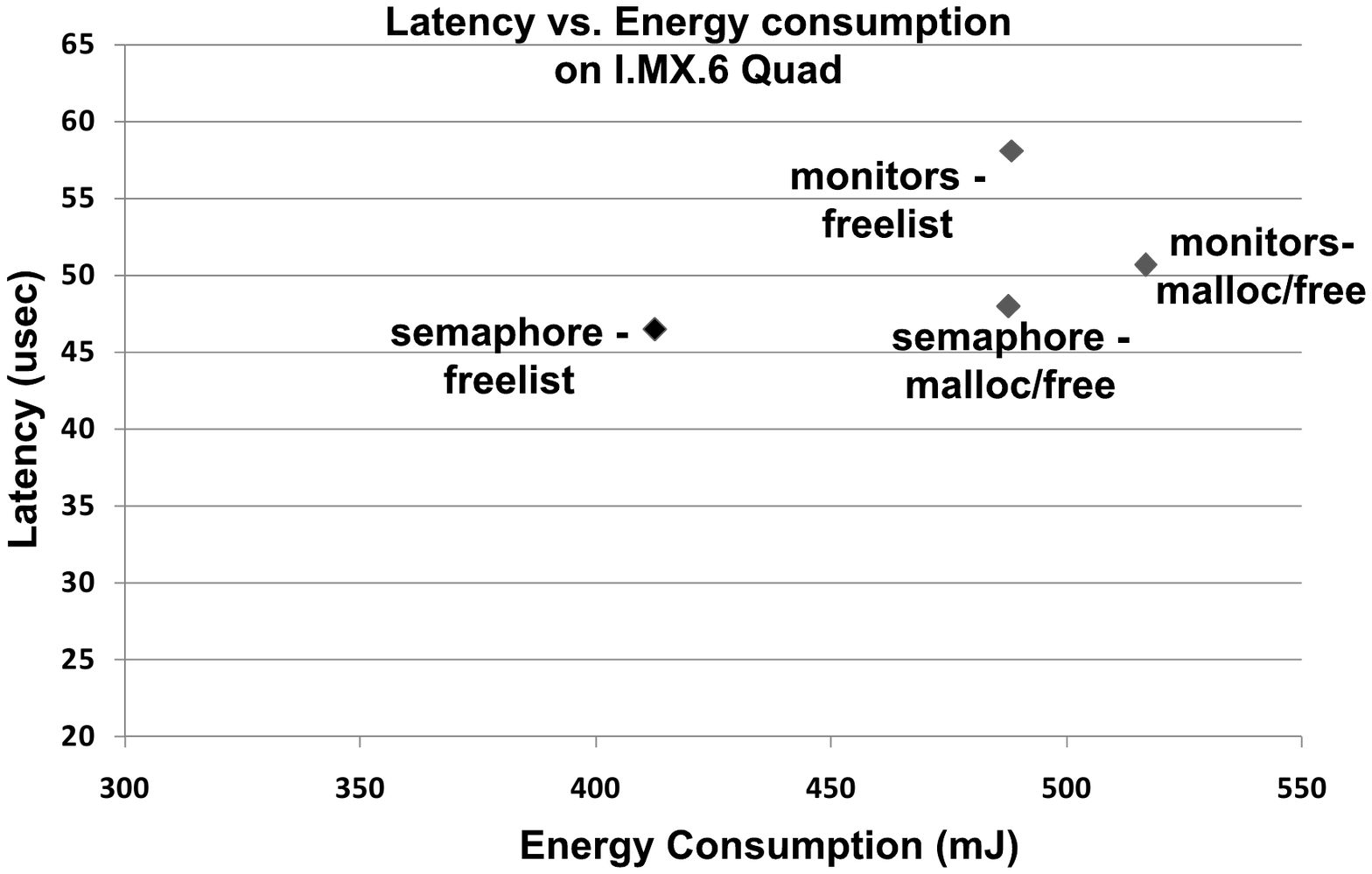}}
		\end{tabular}
	\end{tabularx}
	\caption{Evaluation results of time-based streaming aggregation implementations on I.MX.6 Quad.}
	\label{imx6-results}
\end{figure*}


Fig. \ref{m2-lat-qs-won} and Fig. \ref{m2-lat-qs-woff} show latency vs. queue sizes on Myriad2 for two different cache configurations, shared and private (decision tree \textit{A4} in Fig. \ref{design-space}). We notice that shared cache provides lower latency than private in both cases, up to 4.2\%. Therefore, all implementations that utilize private cache are pruned and they are not evaluated in step 1 of the methodology. 

In the first experiment in Myriad2, the windows data structure is placed in the local memory. Latency constraint is assumed to be at 55usec and therefore queue sizes from 32B to 512B will be evaluated (Fig. \ref{m2-lat-qs-won}). 

Throughput vs. memory footprint results of the methodology are shown in Fig. \ref{m2-won-thr-mem}. Implementations based on \textit{memcpy} provide higher performance than the ones based on dma transfers between the CMX slices. The 5 Pareto efficient points that are identified provide trade-offs up to 3.7\% for throughput (P5) and up to 22.5\% for memory footprint (P1). 

Latency vs. energy results are displayed in Fig. \ref{m2-won-lat-en}. The Pareto points can be grouped into 2 categories: the ones that utilize busy waiting synchronization scheme and the rest ones that are based on hardware buffers. The 6 Pareto efficient points can be used to perform trade-offs between latency and energy (up to 6.37\% for latency by selecting implementation P11 and 5.2\% for energy consumption, by selecting P6). 

With respect to scalability in Fig. \ref{m2-won-scal}, we notice that throughput for all implementations increases up to 6 inputs and then it drops slightly. As in Myriad1 experiments, implementations with lower queue size tend to provide lower throughput. 

In the second experiment, in which the windows data structure is placed in global memory due to its increased memory size, latency constraint is set to 62usec (Fig. \ref{m2-lat-qs-woff}) and throughput vs. memory footprint results are presented in Fig. \ref{m2-woff-thr-mem}. 4 Pareto efficient points have been identified that provide throughput vs. memory size trade-offs (up to 6.4\% for throughput and up to 3.07\% for latency). Correspondingly, the 5 Pareto efficient points in latency vs. energy consumption evaluation displayed in Fig. \ref{m2-woff-lat-en} can be used for performing trade-offs, up to 8.59\% for latency (P20) and 18\% for energy (P16). Scalability results in Fig. \ref{m2-woff-scal} are slightly different from the ones in the previous experiment. Implementations scale up to 8 inputs and most of them tend to provide slightly lower throughput when 9 inputs are used. 

\paragraph{Demonstration on I.MX.6 Quad} 
Few customized implementations exist for I.MX.6, since the operating system handles many design options. In the I.MX.6 Quad experiment latency threshold has been set to 60usec and a single effective queue size has been found: 156KB  (Fig. \ref{imx6-lat-qs}). 4 customized implementations have been evaluated and throughput results are shown in Fig. \ref{imx6-thr-mem}, while latency vs. energy results are displayed in Fig. \ref{imx6-lat-en}. We notice that the most efficient implementation in terms of both throughput, latency and energy is the one that utilizes semaphores for synchronization, along with freelist-based memory management.

\subsubsection{Count-based aggregation results}
In the count-based scenario, we evaluate each implementation for different window sizes. The selected values are provided to the first step of the methodology. 24 different implementations are evaluated in each platform.

\begin{figure*}[t]
	\def\tabularxcolumn#1{m{#1}}
	\begin{tabularx}{\linewidth}{@{}cXX@{}}
		\begin{tabular}{cc}
			\subfloat[Throughput evaluation on Myriad1\label{s2-m1-thr}]{
				\includegraphics[width=0.450\textwidth, page=21]{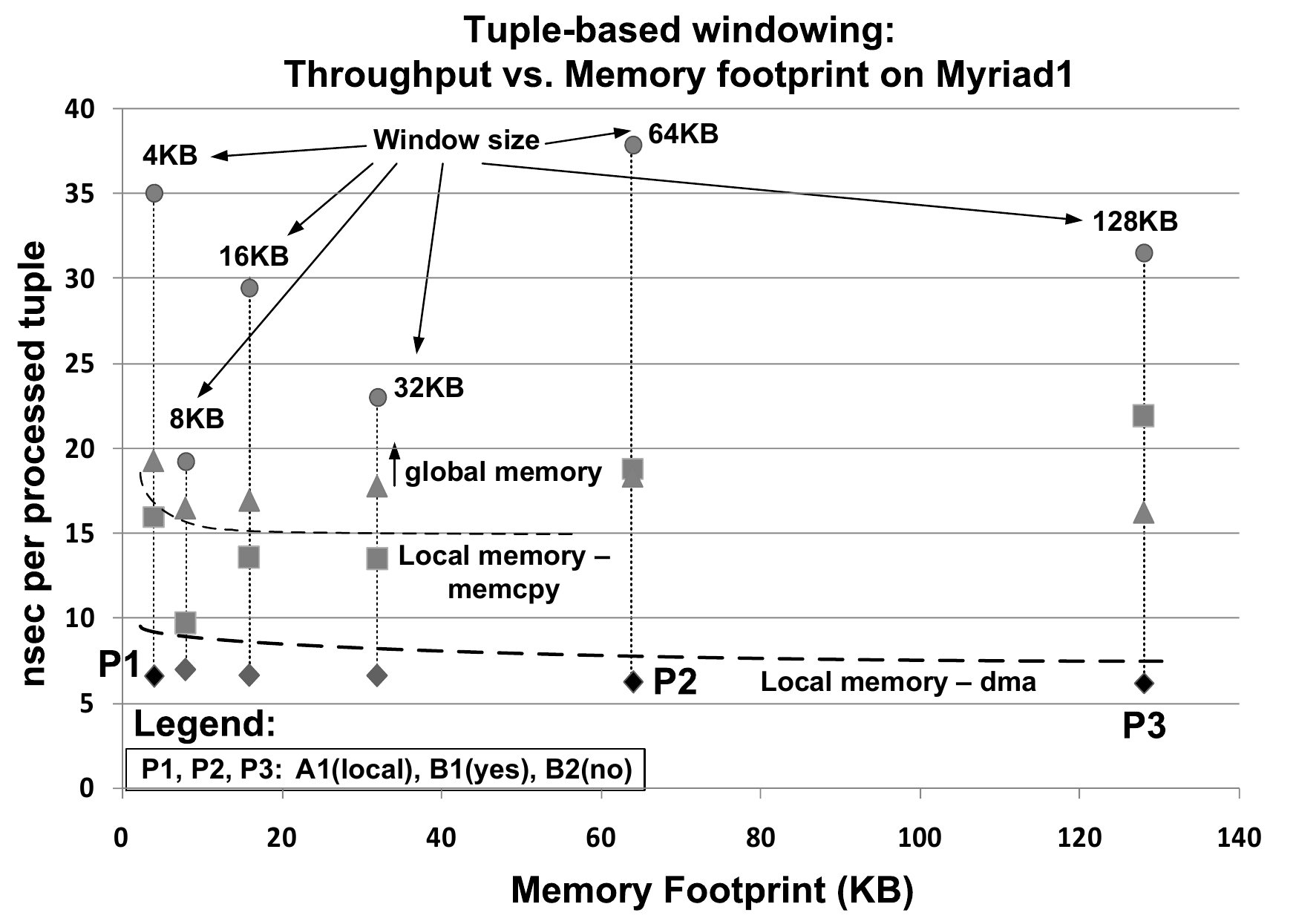}
			}
			& \subfloat[Latency vs. energy consumption on Myriad1\label{s2-m1-lat}]{
				\includegraphics[width=0.450\textwidth, page=22]{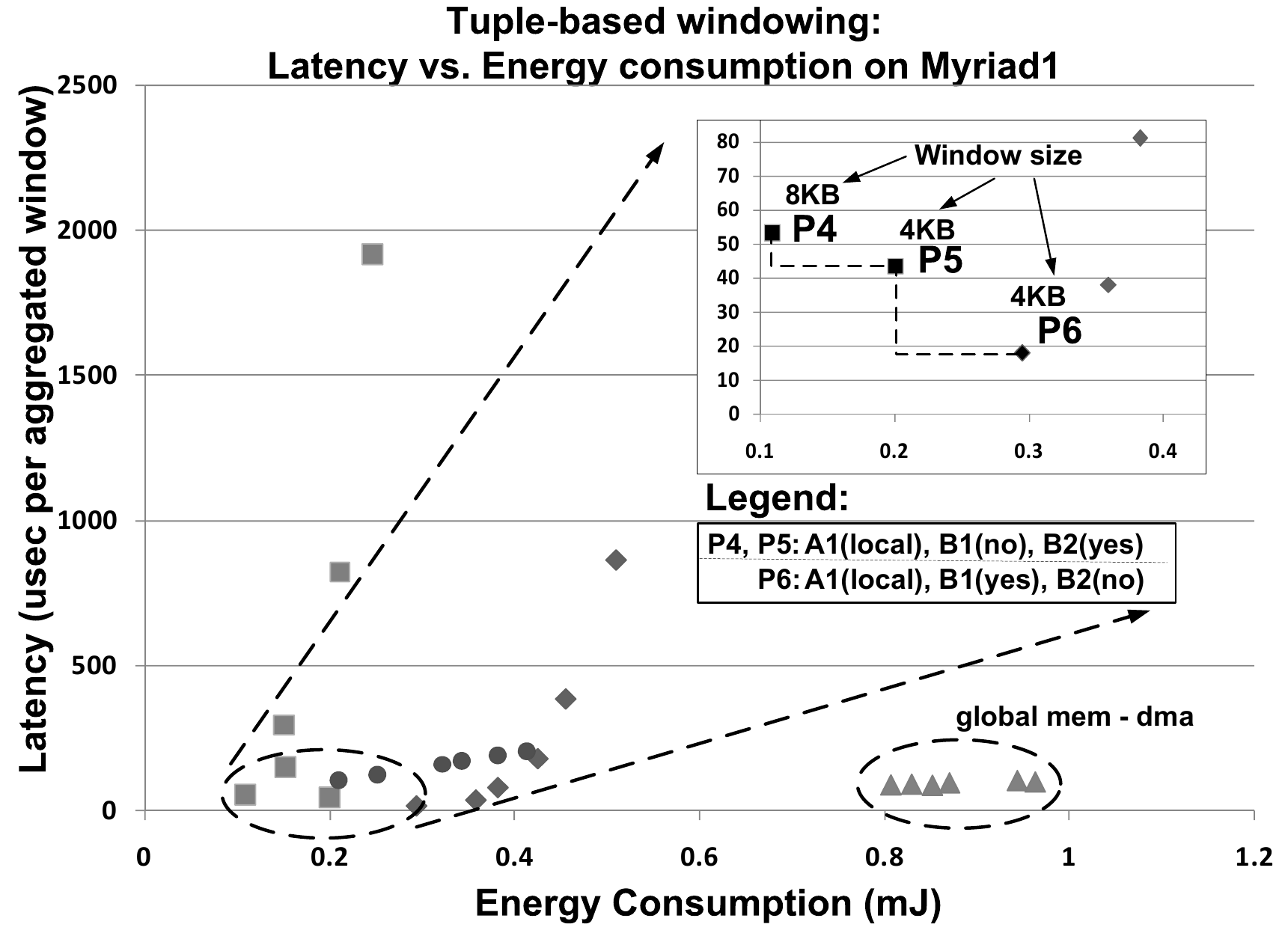}
			}\\
			\subfloat[Throughput evaluation on Myriad2\label{s2-m2-thr}]{
				\includegraphics[width=0.450\textwidth, page=24]{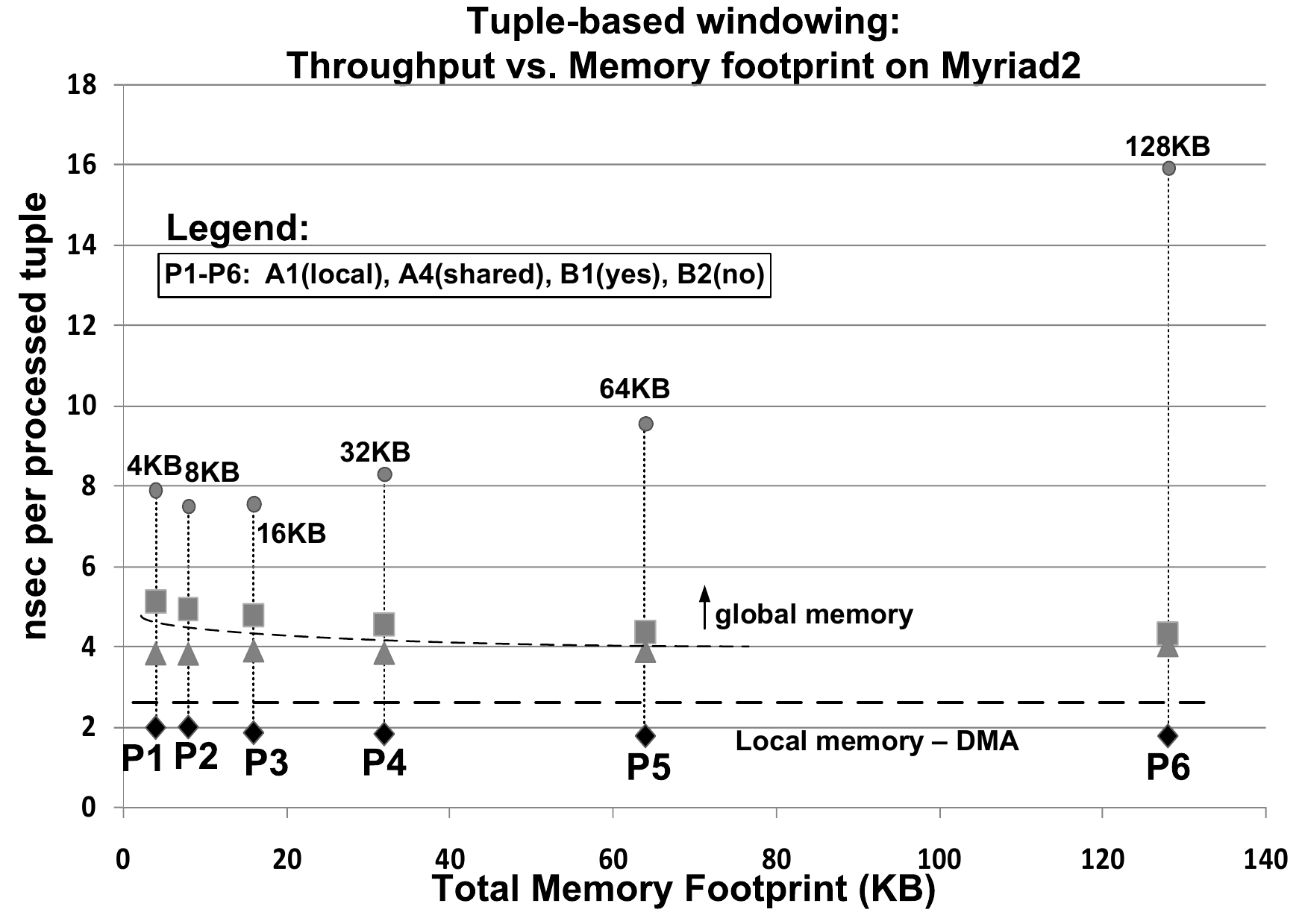}
			}	
			&\subfloat[Latency vs. energy consumption on Myriad2\label{s2-m2-lat}]{
			\includegraphics[width=0.450\textwidth, page=25]{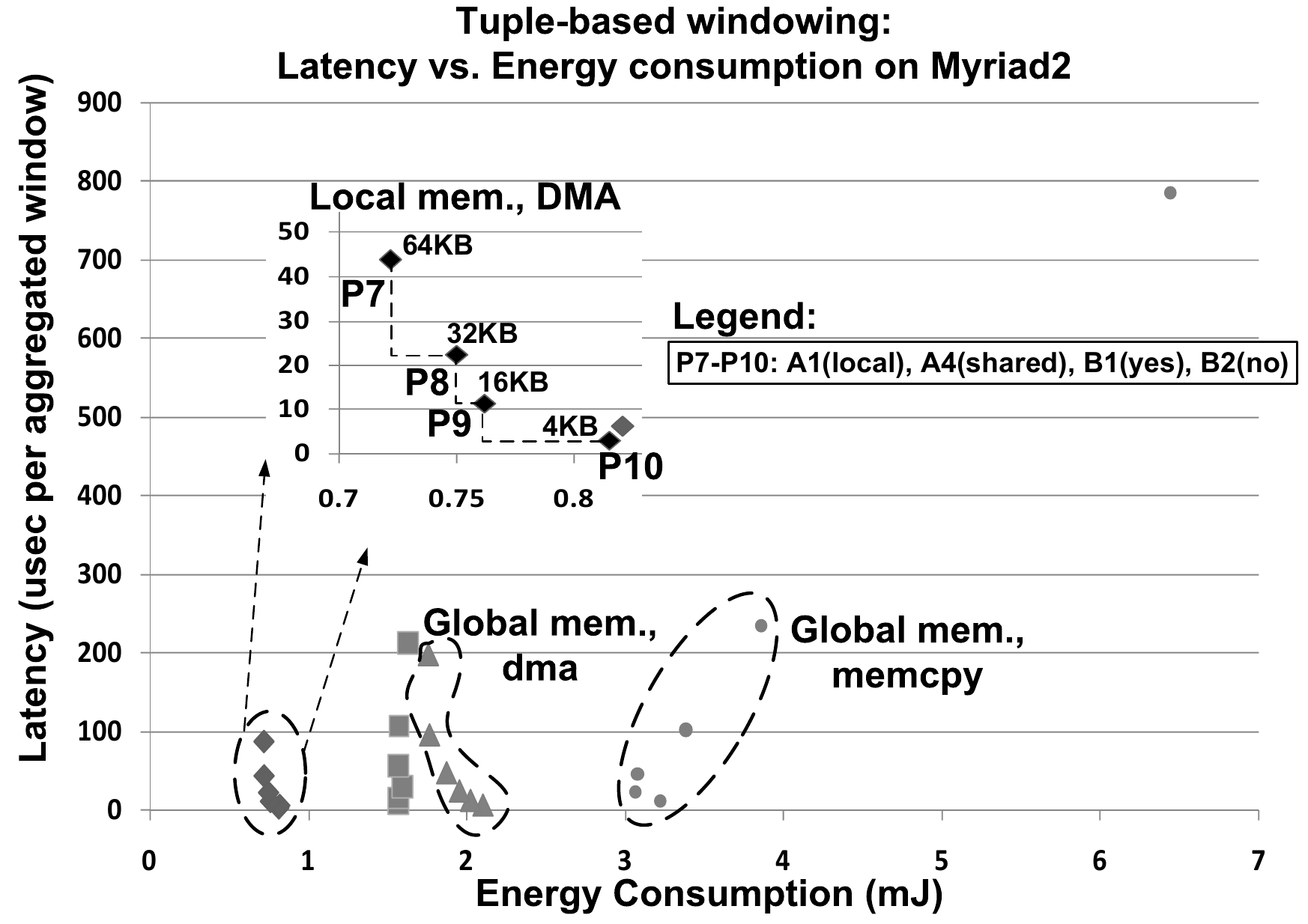}
			}\\
		\end{tabular}
	\end{tabularx}
	\caption{Evaluation results of count-based streaming aggregation implementations.}
	\label{s2}
\end{figure*}

\begin{figure*}[!t]
		\subfloat[Throughput evaluation on Exynos\label{s2-ex-thr}]{
			\includegraphics[width=0.45\textwidth, page=27]{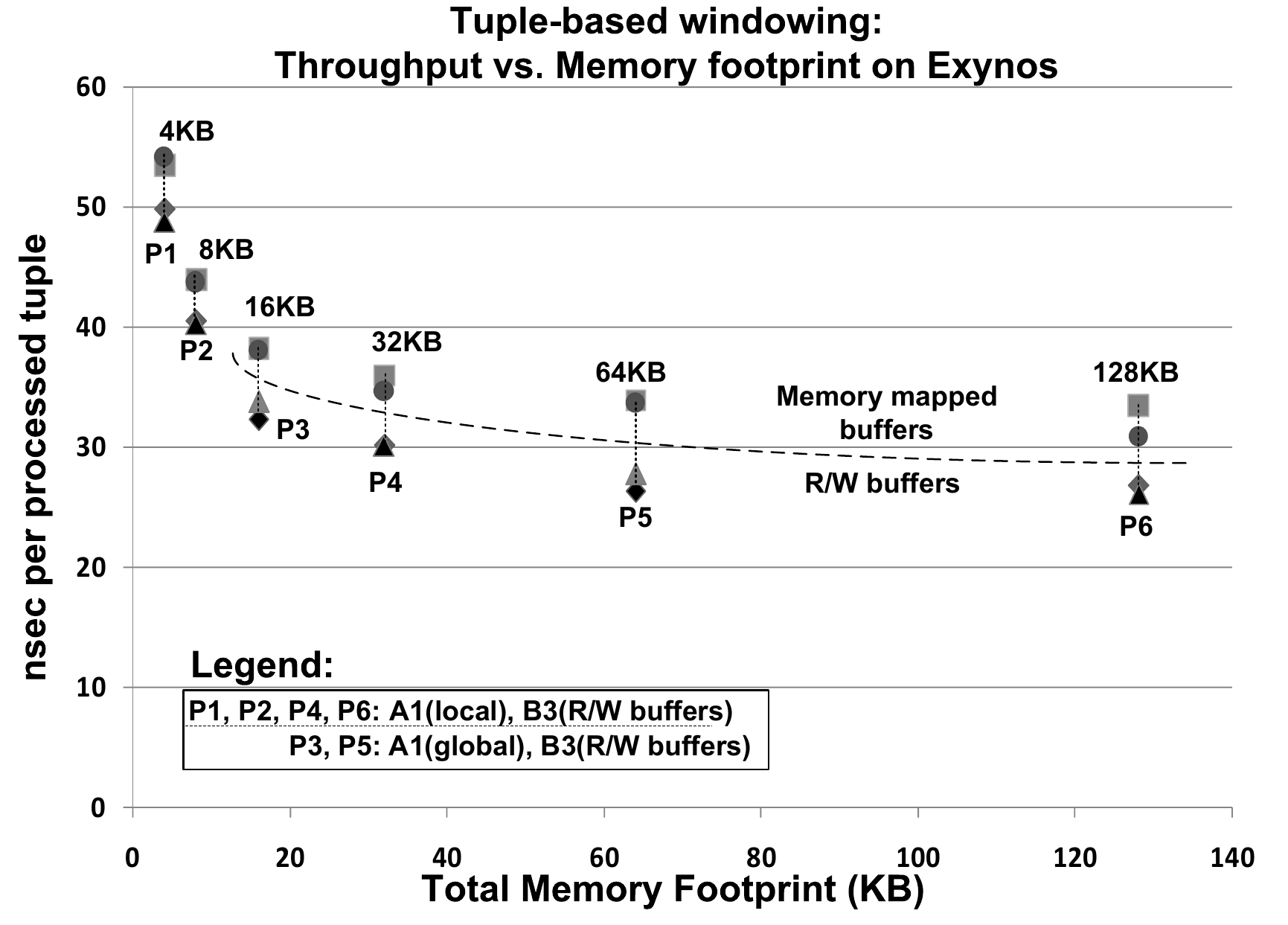}
		}
	\subfloat[Latency vs. energy consumption on Exynos\label{s2-ex-lat}]{
		\includegraphics[width=0.45\textwidth, page=28]{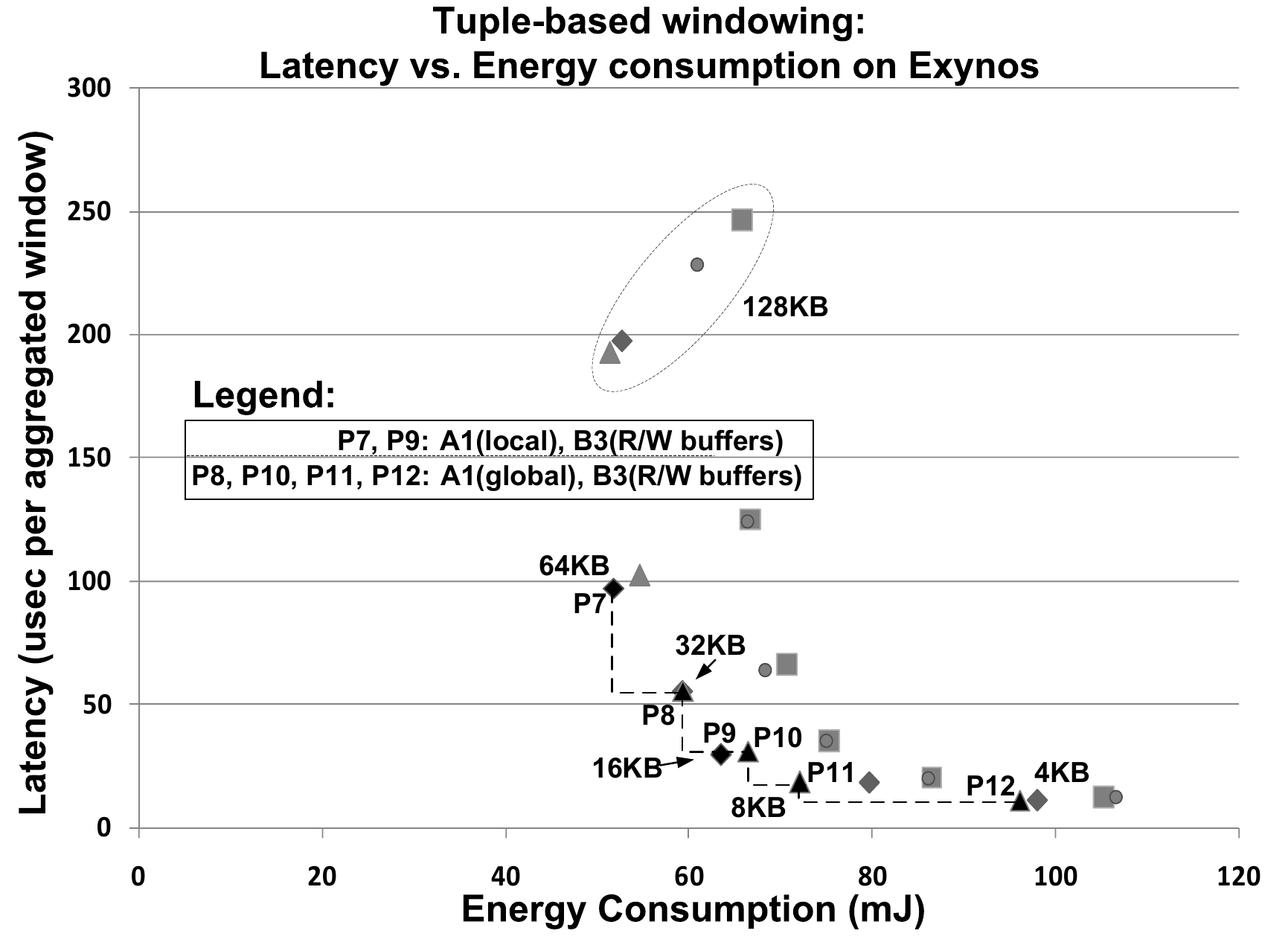}
	}
	\caption{Evaluation results of count-based streaming aggregation implementations.}
	\label{s21}
\end{figure*}


\paragraph{Demonstration on Myriad1}
Fig. \ref{s2-m1-thr} shows throughput vs. memory footprint on Myriad1. Implementations that process tuples in local memory and transfer data from global to local memory through DMA provide higher throughput. For instance, at 4KB window size, P1 provides 58\% higher throughput than the implementation that uses \textit{memcpy} for data transfer. 

Latency vs. energy consumption results are presented in Fig. \ref{s2-m1-lat}. We notice that smaller windows provide lower latency. Also, transferring tuples in local memories provides lower latency than processing windows in Myriad1 global memory. 3 Pareto points are identified that provide trade-offs between latency and energy consumption.

\paragraph{Demonstration on Myriad2}
\begin{figure}
	\centering
	\includegraphics[width=0.65\textwidth, page=23]{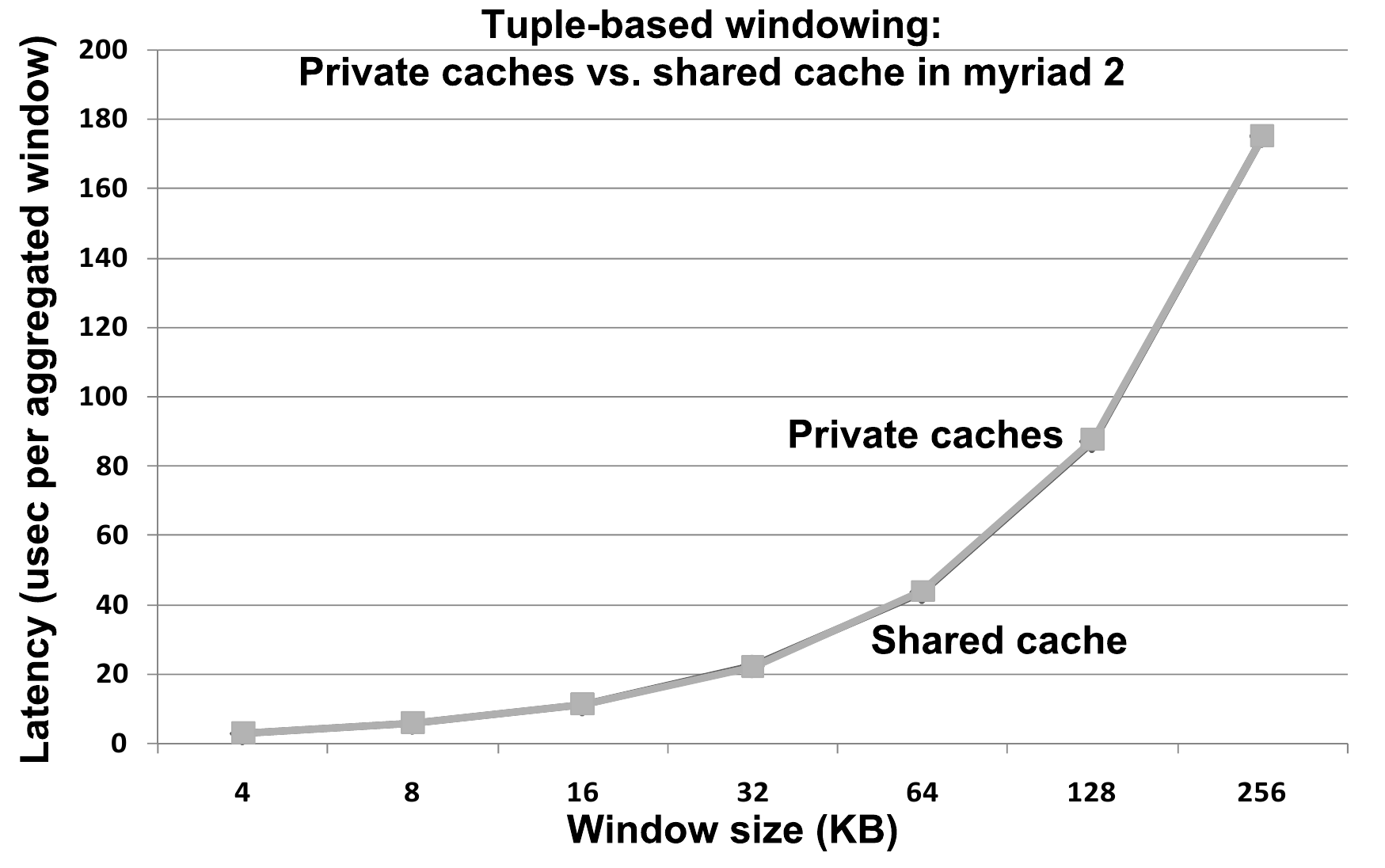}
	\caption{Latency vs. window size on Myriad2 for count-based streaming aggregation.}
	\label{m2-lat-win-size}
\end{figure}

In Myriad2, we first evaluate latency vs. window size for two different cache configurations. As shown in Fig. \ref{m2-lat-win-size}, utilization of shared cache provides slightly lower latency than private caches (less than 1\%). Therefore, implementations that utilize private caches are pruned and the design space is reduced. 

As in Myriad1, implementations that provide higher throughput are the ones in which tuples are transferred through DMA and processed in local memory. Fig. \ref{s2-m2-thr} shows that throughput increases up to 59\% using the aforementioned implementation, in comparison with the implementation in which tuples are processed in global memory, with window size 64KB. Also, we notice that larger windows provide slightly higher throughput. For instance, increasing window size from 4KB to 128KB, yields throughput increase about 10\% (P1 to P6). 

Implementations that utilize local memory and DMA transfers provide both low latency and energy efficiency, as shown in Fig. \ref{s2-m2-lat}. Processing in global or in local memory affects both latency and energy consumption results. For instance, tuples in local memory and utilization of DMA with 4KB window size provides 31.4\% lower energy consumption than the corresponding implementation with tuple processing in global memory. 

\paragraph{Demonstration on Exynos 5}
Throughput vs. memory footprint results are displayed in Fig. \ref{s2-ex-thr}. Larger window sizes provide higher throughput. Implementations that utilize R/W buffers yield higher performance than corresponding implementations with memory mapped data buffers: up to 21\% for 64KB window size. 

Regarding latency vs. energy consumption, displayed in Fig. \ref{s2-ex-lat}, 6 Pareto points are identified. Smaller window sizes provide lower latency, but higher energy consumption, due to the increased rate of data transfers. Utilization of R/W buffers is more efficient than memory mapped ones, both in terms of latency and energy consumption. Due to the relatively small buffer size, the overhead of utilizing R/W buffers is also small.

\subsubsection{Performance per watt evaluation} 
One of the goals of this work is to compare performance per watt of streaming aggregation mapped on low power embedded platforms with the corresponding results on an HPC CPU and a GPGPU. In this subsection, we first provide details on the implementation of the operator on the aforementioned platforms and then we present the evaluation results. 

We implemented the time-based streaming aggregation scenario on an Intel Xeon E5 CPU with 8 cores operating at 3.4GHz, with 16GB RAM, running Ubuntu Linux 12.04. Compiler is gcc v.4.9.2 and optimization flag is \textit{-O3}. Power consumption was measured through hardware instrumentation and refers to dynamic CPU Power. Throughput and latency were measured similarly to the embedded implementations. Data transfer was based on \textit{memcpy()} operations and synchronization based on semaphores.

The results are presented on Table \ref{s1-perf-per-watt}. The values for Myriad1, Myriad2 and I.MX.6 correspond to the implementation that provides the best results for each specific metric. To ensure fair comparison, all values for all platforms utilize 5 input queues. Performance per watt is calculated as number of tuples forwarded per second, per watt. 

In Table \ref{s1-perf-per-watt}, we notice that in terms of performance, latency on Intel Xeon is 62.3\% lower than in Myriad2, while it is 3.8 and 9.3 times lower than in I.MX.6 and Myriad1, respectively. In terms of throughput, Xeon provides more than two times higher throughput than Myriad2, 2.8 than I.MX.6 and 8.3 times higher than Myriad1. The high performance of Intel Xeon is related with the higher computational power it provides and the fact that it operates in much higher frequency than the embedded architectures. However, in terms of performance per watt, embedded platforms outperform Intel Xeon. Because the Myriad processors consume very low power, they achieve higher performance per watt: 54 times higher in Myriad2, while in Myriad1 it is 20 times higher. Finally, I.MX.6 provides 24 times higher performance per watt in comparison with Intel Xeon. 

Count-based aggregation scenario was implemented in OpenCL 1.1 and evaluated in AMD Radeon HD 6450 general purpose GPU \cite{hd6450}. The host runs Ubuntu Linux 12.04 with gcc v.4.9.2. Throughput and latency were measured similarly to the corresponding embedded implementations, while power consumption is estimated based on GPU's specifications. Device data accessing is based on R/W buffers. 

The results are presented in Table \ref{s2-perf-per-watt}. Embedded platforms provide lower throughput and higher latency than Radeon GPGPU. However, both Myriad boards yield higher performance per watt than GPGPU, due to the very low power that they require. More specifically, Myriad2 provides about 14 higher performance per watt, while Myriad1 7 times.

\subsubsection{Discussion of Experimental Results}

\begin{table} [!t]
	\centering
	\caption{Time-based streaming aggregation: Comparison between latency, throughput and performance per watt on embedded and Intel Xeon architectures.}
	\label{s1-perf-per-watt}
	\begin{tabular}{l|l|l|l}
		\hline\noalign{\smallskip} 
		& \multirow{2}{*}{Latency (usec)} & Throughput                  & \multirow{2}{*}{(t/sec)/watt} \\
		&                                 & (t/sec) &                             \\
		\hline\noalign{\smallskip} 
		Myriad1 &   140.38                              &  132,622                            &    379,041                         \\
		\hline\noalign{\smallskip} 
		Myriad2 &    39.8                             &      497,154                       &        1,004,766                     \\
		\hline\noalign{\smallskip} 
		I.MX.6  &    58                             &        384,952                     &      446,787                       \\
		\hline\noalign{\smallskip} 
		Xeon    &   15                              &        1,105,221                     &     18,412                       
	\end{tabular}
\end{table}

\begin{table} [!t]
	\centering
	\caption{Count-based streaming aggregation: Comparison between latency, throughput and performance per watt on embedded and Radeon HD 6450.}
	\label{s2-perf-per-watt}
	\begin{tabular}{l|l|l|l}
		\hline\noalign{\smallskip} 
		& Latency & Throughput                  & \multirow{2}{*}{(Mt/sec)/watt} \\
		&        (usec)                          & (Mt/sec) &                             \\
		\hline\noalign{\smallskip} 
		Myriad1 &   17.98                             &  151.8                            &    593                        \\
		\hline\noalign{\smallskip} 
		Myriad2 &    3.04                             &      505.4                       &        1286                     \\
		\hline\noalign{\smallskip} 
		Exynos  &    7.5                             &        47.4                    &      7.93                     \\
		\hline\noalign{\smallskip} 
		GPGPU    &   1.94                              &        2576.3                     &     85.87                       
	\end{tabular}
\end{table}

In this subsection we summarize the conclusions we draw from the demonstration of the methodology that is presented in the previous subsections. The trade-offs we demonstrated in the experimental results can be used to draw conclusions about the relation between the customization options and the evaluation metrics. 

\paragraph{Time-based streaming aggregation conclusions}
\textbf{Observation 1}: Streaming aggregation should be customized differently, not only between I.MX.6 Quad and Myriad architectures, but also between Myriad1 and Myriad2.

For example, in Myriad1, in the first experiment, in the implementation that provides the lowest latency, data transfer is based on hardware buffers. On the contrary, in Myriad2 it is based on busy waiting mechanism. In the implementation that provides the highest throughput, the queue is 256B in Myriad1, while it is 512B in Myriad2.
\\
\\
\textbf{Observation 2}: There is a threshold in the queue size, below which latency is very high. Very large queue sizes may also negatively affect latency.

We notice that in both Myriad and I.MX.6, latency is very high for small queue sizes, which is due to the high overhead of  constantly fetching data for refilling the queues with new tuples. In these cases, the thread that executes the merge-sort phase, often finds the queues to be empty. As the queue size increases latency drops drastically. However, in Myriad1 and Myriad2 experiments, we notice that as the queue size increases, latency tends to increase, as well (Fig. \ref{m1-lat-qs} and Fig. \ref{m2-lat-qs}). The reason is the fact that the larger the queue, the more cycles it takes to complete a DMA transfer of data from the DDR to the local memory and start refilling the queue with new tuples. Thus, the tuples that entered the update phase before a new DMA transfer and exit the output phase after it, they have higher latency than the rest ones. In contrast with Myriad, on I.MX.6 we can use much bigger queues, since the available memory is much larger. However, beyond a specific queue size, throughput and latency on I.MX.6 do not seem to be significantly affected any more (Fig. \ref{imx6-lat-qs}).
\\
\\
\textbf{Observation 3}: Throughput is mainly affected by either the data transfer mechanism (in Myriad2) or by the signaling mechanism (in Myriad1).

In general, in Myriad1 and Myriad2, throughput drops when the queue size becomes smaller, due to overhead of the DMA transfers, which is added more frequently when the queues are small (e.g. Fig. \ref{m1-won-thr-mem} and Fig. \ref{m2-won-thr-mem}). However, latency becomes lower in that case, as stated earlier. In Myriad2, throughput is mainly determined by whether \textit{memcpy} or DMA data transfer mechanism is used. Indeed, data transfer options seem to have major impact on throughput (Fig. \ref{m2-won-thr-mem} and Fig. \ref{m2-woff-thr-mem}). On the other hand, in Myriad1 the utilization of hardware buffer or of busy waiting scheme is the dominant factor that affects throughput (Fig. \ref{m1-won-thr-mem} and Fig. \ref{m1-woff-thr-mem}). In Myriad2 signaling design options have much lower impact in comparison with data transfer options. On the contrary, in Myriad1, data transfer mechanism has relatively small effect on throughput in comparison with the signaling mechanism (\textit{memcpy} however is slightly more efficient). In I.MX.6, the utilization of freelists to avoid the frequent system calls improves throughput and latency results. However, the main factor that improves performance is the utilization of semaphores instead of monitors (Fig. \ref{imx6-thr-mem}).
\\
\\
\textbf{Observation 4:} Latency is affected by the synchronization mechanism. Different mechanism should be used in Myriad1 than in Myriad2. 

The synchronization mechanism is the main design option that affects latency and energy in both Myriad architectures. Busy waiting mechanism provides lower latency in Myriad2 and slightly lower energy consumption. On the contrary, the utilization of hardware buffers in Myriad1 is more efficient it terms of latency. The data transfer mechanism has much lower impact in both architectures in terms of latency and energy. 
\\
\\
\textbf{Observation 5:} The frequency by which data movements are performed from global to local memory affects energy consumption in Myriad. We notice that larger queue sizes are more energy efficient in both Myriad1 and Myriad2, due to the lower rate by which data are fetched in the local memory (e.g. Fig. \ref{m1-won-lat-en} and Fig. \ref{m2-won-lat-en}). On I.MX.6 Quad, energy consumption is determined mainly the by synchronization scheme that it is used. 

Finally, an interesting observation is the fact that the memory allocation of the input queues affects neither the performance nor the energy consumption in Myriad significantly. The reason is the fact that both Myriad architectures provide cache memory and the rate of cache misses for accessing the queues by the PE that performs the merge-sort operation is relatively small. On the other hand, the allocation of the windows data structure in global memory has major impact in both performance and energy consumption. For instance, in Myriad2, by allocating the windows data structure in global memory, latency increases about 9\%, throughput drops by 7\% and energy consumption increases by 20\% in comparison with the allocation in local memory.

The above observations can be used to draw more general conclusions on how the streaming aggregation should be customized on embedded platforms. 
When the optimization target is performance, the following considerations should be taken into account: 
\begin{itemize}
	\item 
	The queue size should be large enough to decrease the rate by which data transfers are instructed. Frequent small data transfers lower performance. However, for implementations that are very sensitive to latency, it should be noted that too large queue sizes may increase latency.
	\item
	Window \textit{size} and \textit{advance} values affect a lot the maximum size of the windows data structure and therefore the memory allocation design options and the performance. Platforms with very small local memory may be not suitable for implementing streaming aggregation, since they would limit the window configuration values that can be used, if allocation of the data structure in global memory is not a option, due to very strict performance requirements.
	\item
	Platform-specific options for efficient communication between cores (such as the hardware buffers on Myriad) should be evaluated, when the streaming aggregation is implemented at low level. In some cases (such as in Myriad1) they can provide increased performance. 
\end{itemize}

On the other hand, if the main goal is energy efficiency, the following issues should be considered:
\begin{itemize}
	\item
	The queues should be as large as possible to avoid the energy consumption overhead of frequent small data transfers.
	\item
	For window \textit{size} and \textit{advance} values apply the same that are stated earlier: Window configuration that forces the allocation of the windows data structure in global memory has negative impact in energy consumption.
	\item
	Finally, developers should try to evaluate features that set the PEs in a low-energy mode when they are forced to wait (such as the hardware buffers in Myriad1).
\end{itemize}

\paragraph{Count-based streaming aggregation conclusions}
\textbf{Observation 1:} Both throughput and latency in Myriad implementations are affected by the memory allocation of the processed tuples. In Exynos implementations, they are mainly affected by the data accessing method by the device.

In general, throughput is apparently affected by the window size. Apart from that, design choices such as the allocation of the window in local memory and R/W buffers in OpenCL implementations, yield increased throughput. 

In contrast with throughput, smaller window sizes provide lower latency. Implementations in which tuples are processed in local memories in Myriad and utilize R/W buffers in mobile GPU provide the lowest latency. 
\\
\\
\textbf{Observation 2:} Energy consumption is mainly affected by the memory allocation and the window size. 

Energy consumption in Myriad is affected by both the type of memory in which tuples are processed and the size of the window (Fig. \ref{s2-m2-lat}). In Exynos, window size has the highest impact in energy (Fig. \ref{s2-ex-lat}). Since the rate of data transfers is increased when smaller windows are used, energy consumption is also increased. 
\\
\\
To summarize, when the optimization target is performance, DMA transfers and R/W OpenCL buffers provide higher throughput than the rest of the design choices. Large windows yield increased throughput, while smaller ones provide low latency. Finally, window sizes that allow processing in local memory benefit both performance and energy.

The methodology we propose in this work provides a systematic approach to the efficient customization of the streaming aggregation on embedded platforms. Instead of trying to tune the application and hardware parameters arbitrary to achieve the desired results, the proposed methodology provides a set of customization solutions from which developers can select the one that is more suitable according the design constraints. 

Finally, it is important to state that the methodology is not fundamentally limited to streaming aggregation. The design space could be adapted to be applicable to other streaming operators, as well (such as join, filter etc.) and to embedded platforms with various other features. New attributes can be integrated in the design space for exploration as new decision trees, leaves or categories. The application and hardware constraints should be updated accordingly to retain the coherency of the customized implementations. 
%
\subsection{Conclusion}
We proposed a customization methodology for the implementation of streaming aggregation in modern embedded devices. The methodology was demonstrated in 4 different embedded architectures, 2 aggregation scenarios and a real-world data set. The customized implementations provided by the methodology can be utilized by developers to perform trade-offs between several parameters, taking into consideration the design constraints that are imposed by both the application requirements and the embedded architecture. In the future, we intend to extend the design space by integrating more streaming aggregation operators and evaluate the approach in embedded platforms with various features. 


\newpage
\section{Energy Model on CPU for Lock-free Data-structures in Dynamic Environments}









\newcommand{\todo}[1]{~\noindent{\color{red}\rule[-.1cm]{.3cm}{.3cm}~{\color{red}{#1}}}~}

\newcommand{\ema}[1]{\ensuremath{#1}\xspace}

\newcommand{\inte}[2][0]{\ema{\left\llbracket #1,#2 \right\rrbracket}}

\newcommand{\pro}[1]{\ema{\mathds{P}\left(#1\right)}}
\newcommand{\expe}[1]{\ema{\mathds{E}\left(#1\right)}}

\newcommand{\indi}[2]{\ema{\mathds{1}_{#1}\left( #2 \right)}}

\newcommand{\expu}[1]{\ema{e^{#1}}}
\newcommand{\expi}[1]{\ema{\operatorname{exp}\left( #1\right)}}

\newcommand{\pinf}{\ema{+\infty}}

\makeatletter
\newsavebox\myboxA
\newsavebox\myboxB
\newlength\mylenA

\newcommand*\xoverline[2][0.75]{%
    \sbox{\myboxA}{$\m@th#2$}%
    \setbox\myboxB\null
    \ht\myboxB=\ht\myboxA%
    \dp\myboxB=\dp\myboxA%
    \wd\myboxB=#1\wd\myboxA
    \sbox\myboxB{$\m@th\overline{\copy\myboxB}$}
    \setlength\mylenA{\the\wd\myboxA}
    \addtolength\mylenA{-\the\wd\myboxB}%
    \ifdim\wd\myboxB<\wd\myboxA%
       \rlap{\hskip 0.5\mylenA\usebox\myboxB}{\usebox\myboxA}%
    \else
        \hskip -0.5\mylenA\rlap{\usebox\myboxA}{\hskip 0.5\mylenA\usebox\myboxB}%
    \fi}
\makeatother

\newcommand{\barov}[1]{\xoverline[1]{#1}}

\newcommand{\kth}[1]{\ema{#1^{\mathrm{th}}}}
\newcommand{\kst}[1]{\ema{#1^{\mathrm{st}}}}
\newcommand{\knd}[1]{\ema{#1^{\mathrm{nd}}}}
\newcommand{\krd}[1]{\ema{#1^{\mathrm{rd}}}}

\newtheorem{theorem}{Theorem}
\newtheorem{proposition}{Proposition}
\newtheorem{lemma}{Lemma}
\newtheorem{corollary}{Corollary}
\newtheorem{property}{Property}


\theoremstyle{definition}

\newtheorem{definition}{Definition}

\newtheorem{remark}{Remark}




\newcommand{\figsidebyside}[9]{
\begin{figure}[#1]
\centering
\begin{minipage}{#2\textwidth}
\begin{center}
\includegraphics[width=\textwidth]{#3}
\end{center}
\captionof{figure}{#4 \label{fig:#5}}
\end{minipage}\hfill%
\begin{minipage}{#6\textwidth}
\begin{center}
\includegraphics[width=\textwidth]{#7}
\end{center}
\captionof{figure}{#8 \label{fig:#9}}
\end{minipage}
\end{figure}
}

\makeatletter
\g@addto@macro\normalsize{%
  \setlength\abovedisplayskip{2pt}
  \setlength\belowdisplayskip{2pt}
  \setlength\abovedisplayshortskip{2pt}
  \setlength\belowdisplayshortskip{2pt}
}
\makeatother

\newcommand{\vspp}[1]{\pp{\vspace*{#1cm}}}

\newcommand{\sta}[1]{\ema{\mathcal{S}_{#1}}}

\newcommand{\eve}[1]{\ema{E_{\mathrm{#1}}}}
\newcommand{\thex}[2]{\ema{X_{[#1,#2[}}}
\newcommand{\kdur}[2]{\ema{a_{#1,#2}}}
\newcommand{\kaft}[2]{\ema{b_{#1,#2}}}

\newcommand{\kkdur}[2]{\ema{a_{#1,#2}}}
\newcommand{\kkaft}[1]{\ema{b_{#1}}}

\newcommand{\rwi}[2]{\ema{\mathit{ct}_{\mathrm{#1}}\left(#2\right)}}
\newcommand{\rw}[1]{\rwi{}{#1}}
\newcommand{\rwl}{\mathit{ct}_{\mathrm{lo}}}

\newcommand{\rwh}[1]{\ema{\xoverline[.88]{\mathit{ct}}\left(#1\right)}}


\newcommand{\ct}{\ema{P}}

\newcommand{\trl}{\ema{\ct_{\mathit{rl}}}}
\newcommand{\tps}{\ema{\ct_{\mathit{ps}}}}
\newcommand{\atrl}{\ema{\xoverline[.7]{\ct_{\mathit{rl}}}}}
\newcommand{\atps}{\ema{\xoverline[.7]{\ct_{\mathit{ps}}}}}

\newcommand{\tcom}{\ema{\ct_{\mathit{com}}}}

\newcommand{\fa}[1]{\ema{f_{#1}}}

\newcommand{\atrlf}{\ema{\trl^{(0)}}}
\newcommand{\trls}{\ema{\trl^{(s)}}}
\newcommand{\atrli}{\ema{\trl^{(\mathrm{i})}}}

\newcommand{\scas}{\ema{\mathit{cc}}}
\newcommand{\fcas}{\ema{\mathit{cc}}}
\newcommand{\mem}{\ema{\mathit{rc}}}
\newcommand{\total}{\ema{T_{total}}}
\newcommand{\calrl}{\ema{\mathit{cw}}}
\newcommand{\calpar}{\ema{W_{par}}}
\newcommand{\successperiod}{\ema{SUC_{period}}}
\newcommand{\retry}{\ema{RL}}
\newcommand{\diff}{\ema{Diff}}
\newcommand{\rint}[4]{\ema{\int_{#1}^{#2} #3 \, \mathit{d#4}}}

\newcommand\vespv{\ema{s}}
\newcommand\vesp[1]{\ema{\vespv_{#1}}}

\newcommand{\avexp}[1]{\ema{\barov{e}\left(#1\right)}}
\newcommand{\reexp}[1]{\ema{e\left(#1\right)}}
\newcommand{\difavexp}[1]{\ema{\barov{e}'\left(#1\right)}}

\newcommand{\expa}{\ema{e}}

\newcommand{\trlo}{\ema{\trl^{(0)}}}


\newcommand{\exppl}{(\texttt{+})}
\newcommand{\expmi}{(\texttt{-})}

\newcommand{\cw}{\ema{\mathit{cw}}}
\newcommand{\pw}{\ema{\mathit{pw}}}

\newcommand{\rc}{\ema{\mathit{rc}}}
\newcommand{\cc}{\ema{\mathit{cc}}}
\newcommand{\rlw}{\ema{\mathit{rlw}}}
\newcommand{\rlwp}{\ema{\rlw^{\expmi}}}
\newcommand{\psiz}{\ema{\pw^{\exppl}}}
\newcommand{\rlsiz}{\ema{\rlw^{\exppl}}}

\newcommand{\thru}{\ema{T}}

\newcommand{\ctot}{\ema{P}}
\newcommand{\thr}[1]{\ema{\mathcal{T}_{#1}}}

\newcommand{\shiftf}{\ema{\mathit{delay}}}
\newcommand{\shift}[1]{\ema{\mathit{delay}\left(#1\right)}}
\newcommand{\shifti}[2]{\ema{\mathit{delay}_{#1}\left(#2\right)}}

\newcommand{\supe}[1]{\ema{\mathit{sp}\left(#1\right)}}
\newcommand{\avsupe}[1]{\ema{\barov{\mathit{sp}}\left(#1\right)}}
\newcommand{\avsupef}{\ema{\barov{\mathit{sp}}}}

\def\ffuf{\ema{\operatorname{f_1}}}
\def\ffus{\ema{\operatorname{f_2}}}
\def\ffusi{\ema{\ffus^{-1}}}

\newcommand{\fuf}[1]{\ema{\ffuf{\left( #1\right)}}}
\newcommand{\fufi}[1]{\ema{\ffuf^{-1}\left( #1\right)}}
\newcommand{\fus}[1]{\ema{\ffus\left( #1\right)}}
\newcommand{\fusi}[1]{\ema{\ffusi\left( #1\right)}}

\newcommand{\rl}{retry loop\xspace}
\newcommand{\rls}{retry loops\xspace}
\newcommand{\Rls}{Retry loops\xspace}
\newcommand{\RLs}{Retry Loops\xspace}
\newcommand{\wl}{work loop\xspace}
\newcommand{\wls}{work loops\xspace}
\newcommand{\Wls}{Work loops\xspace}
\newcommand{\WLs}{Work Loops\xspace}
\newcommand{\supw}{success period\xspace}
\newcommand{\Supw}{Success period\xspace}
\newcommand{\SUpw}{Success Period\xspace}
\newcommand{\supws}{success periods\xspace}
\newcommand{\Supws}{Success periods\xspace}
\newcommand{\SUpws}{Success Periods\xspace}
\newcommand{\cww}{critical work\xspace}
\newcommand{\Cww}{Critical work\xspace}
\newcommand{\cwws}{critical works\xspace}
\newcommand{\Cwws}{Critical works\xspace}
\newcommand{\pww}{parallel work\xspace}
\newcommand{\Pww}{Parallel work\xspace}
\newcommand{\pwws}{parallel works\xspace}
\newcommand{\Pwws}{Parallel works\xspace}

\newcommand{\tbf}{reset time\xspace}
\newcommand{\Tbf}{Reset time\xspace}
\newcommand{\TBF}{Reset Time\xspace}

\newcommand{\re}{retry\xspace}
\newcommand{\res}{retries\xspace}
\newcommand{\RES}{Retries\xspace}
\newcommand{\RE}{Retry\xspace}

\newcommand{\wc}{worst-case\xspace}

\newcommand{\avba}{average-based\xspace}
\newcommand{\Avba}{Average-based\xspace}

\newcommand{\psw}{parallel section\xspace}
\newcommand{\pss}{parallel sections\xspace}
\newcommand{\ds}{data structure\xspace}
\newcommand{\dss}{data structures\xspace}
\newcommand{\DSs}{Data Structures\xspace}
\newcommand{\casexp}{{\it Compare-And-Swap}\xspace}
\newcommand{\cas}{{\it CAS}\xspace}
\newcommand{\cass}{\cas{}'s\xspace}

\newcommand{\rf}{{\it Read}\xspace}
\newcommand{\faa}{{\it Fetch-and-Increment}\xspace}
\newcommand{\acc}{{\it Access}\xspace}

\newcommand{\deq}{deque\xspace}
\newcommand{\deqs}{deques\xspace}
\newcommand{\Deq}{Deque\xspace}
\newcommand{\Deqs}{Deques\xspace}

\newcommand{\delmin}{\FuncSty{DeleteMin}\xspace}
\newcommand{\enqop}{\FuncSty{Enqueue}\xspace}
\newcommand{\deqop}{\FuncSty{Dequeue}\xspace}
\newcommand{\popop}{\FuncSty{Pop}\xspace}
\newcommand{\pushop}{\FuncSty{Push}\xspace}
\newcommand{\incop}{\FuncSty{Increment}\xspace}
\newcommand{\decop}{\FuncSty{Decrement}\xspace}

\newcommand{\casop}[1]{\FuncSty{CAS\textsubscript{#1}}\xspace}

\newcommand{\caca}{wasted \re}
\newcommand{\cacas}{wasted \res}
\newcommand{\cacaf}{\ema{\operatorname{IT}}}

\newcommand{\flo}{\ema{f^{\exppl}}}
\newcommand{\fup}{\ema{f^{\expmi}}}

\newcommand{\ghz}[1]{\ema{#1\,\mathrm{GHz}}}
\newcommand{\megb}[1]{\ema{#1\,\mathrm{MB}}}

\newcommand{\cycles}[1]{\ema{#1\,\mathrm{cycles}}}
\newcommand{\uow}[1]{\ema{#1\,\mathrm{u.o.w.}}}

\newcommand{\aexpi}[1]{\ema{\barov{e_i}\left(#1\right)}}
\newcommand{\aexp}[2]{\ema{\barov{e_{#1}}\left(#2\right)}}

\newcommand{\expansion}[1]{\avexp{#1} }
\newcommand{\expansionp}[1]{\difavexp{#1} }
\newcommand{\stag}[1]{\ema{\mathit{Stage}_{#1}}}
\newcommand{\casn}[1]{\ema{\mathit{CAS}_{#1}}}


\newcommand{\pushl}{\FuncSty{PushLeft}\xspace}
\newcommand{\pushr}{\FuncSty{PushRight}\xspace}
\newcommand{\popl}{\FuncSty{PopLeft}\xspace}
\newcommand{\popr}{\FuncSty{PopRight}\xspace}

\newcommand{\anch}{{\it Anchor}\xspace}

\newcommand{\ie}{\textit{i.e.}\xspace}
\newcommand{\etal}{\textit{et al.}\xspace}
\newcommand{\eg}{\textit{e.g.}\xspace}
\newcommand{\etc}{\textit{etc.}\xspace}
\newcommand{\afort}{\textit{a fortiori}\xspace}
\newcommand{\Afort}{\textit{A fortiori}\xspace}

\newcommand\rr[1]{#1}

\newcommand\pp[1]{}

\newcommand\pr[2]{#2}

\newcommand\tra[1]{}

\newcommand{\ptcom}[1]{\textcolor{Bittersweet}{[\bf PT: #1]}}
\newcommand{\prgcom}[1]{\textcolor{Bittersweet}{[\bf PRG: #1]}}

\newcommand\falseparagraph[1]{\noindent{\bf #1:}\xspace}



\makeatletter
\newcommand{\removelatexerror}{\let\@latex@error\@gobble}
\makeatother

\newcommand{\abstalgo}{
\removelatexerror
\pp{\scriptsize}
\begin{procedure}[H]
\SetKwData{pet}{execution\_time}
\SetKwData{pdo}{done}
\SetKwData{psucc}{success}
\SetKwData{pcur}{current}
\SetKwData{pnew}{new}
\SetKwData{pacp}{AP}
\SetKwData{pttot}{t}
\SetKwFunction{pinit}{Initialization}
\SetKwFunction{ppw}{Parallel\_Work}
\SetKwFunction{pcw}{Critical\_Work}
\SetKwFunction{pread}{Read}
\SetKwFunction{pcas}{CAS}

\SetAlgoLined
\While{! \pdo}{\nllabel{alg:li-bwl}
\ppw{}\;\nllabel{alg:li-ps}
\While{! \psucc}{\nllabel{alg:li-bcs}
\pcur $\leftarrow$ \pread{\pacp}\;\nllabel{alg:li-bbcs}
\pnew $\leftarrow$ \pcw{\pcur}\;
\psucc $\leftarrow$ \pcas{\pacp, \pcur, \pnew}\;\nllabel{alg:li-ecs}
}
}
\caption{AbstractAlgorithm()\label{alg:gen-name}}
\end{procedure}
}

\newcommand\posrem[2]{#2}

\setcounter{tocdepth}{3}

\setcounter{secnumdepth}{3}


\newcommand{\Watiw}{Slack time\xspace}
\newcommand{\WAtiw}{Slack Time\xspace}
\newcommand{\watiw}{slack time\xspace}
\newcommand{\watisw}{slack times\xspace}

\newcommand{\staw}{stage\xspace}
\newcommand{\Staw}{Stage\xspace}
\newcommand{\staws}{stages\xspace}
\newcommand{\Staws}{Stages\xspace}

\newcommand{\wati}[1]{\ema{\mathit{st}\left(#1\right)}\xspace}
\newcommand{\nwati}[1]{\ema{\mathit{nst}\left(#1\right)}\xspace}

\newcommand{\avwati}[1]{\ema{\xoverline[.8]{\mathit{st}}\left(#1\right)}\xspace}

\newcommand{\Compw}{Completion time\xspace}
\newcommand{\COmpw}{Completion time\xspace}
\newcommand{\compw}{completion time\xspace}

\newcommand{\reti}[1]{\ema{\xoverline[.85]{\mathit{rt}}\left(#1\right)}\xspace}

\newcommand{\retiw}{reset time\xspace}
\newcommand{\retiws}{reset times\xspace}
\newcommand{\Retiws}{Reset times\xspace}
\newcommand{\REtiws}{Reset Times\xspace}

\DeclarePairedDelimiter\abs{\lvert}{\rvert}%
\DeclarePairedDelimiter\norm{\lVert}{\rVert}%

\makeatletter
\let\oldabs\abs
\def\abs{\@ifstar{\oldabs}{\oldabs*}}
\let\oldnorm\norm
\def\norm{\@ifstar{\oldnorm}{\oldnorm*}}
\makeatother

\newcommand{\ra}{\ema{\rightarrow}}






\newcommand\lemsl{
\begin{lemma}
\label{lem:unif-min}
Let an integer $n$, a real positive number $a$, and $n$
independent random variables $X_1, X_2, \dots, X_n$, uniformly
distributed within $[0,a[$. Let then $X$ be the random variable
defined by: $X = \min_{i \in \inte[1]{n}} X_i$. The expectation of $X$ is:
\[ \expe{X} = \frac{a}{n+1}. \]
\end{lemma}
\begin{proof}
Let a positive real number $x$ be such that $x<a$. We have
\begin{align*}
\pro{X > x} &=  \pro{\forall i : X_i > x}\\
&= \prod_{i=1}^n \pro{X_i > x}\\
\pro{X > x} &=  \left( \frac{a-x}{a} \right) ^{n}\\
\end{align*}
Therefore, the probability distribution of $X$ is given by:
\[ t \mapsto \frac{n}{a} \left( \frac{a-x}{a} \right) ^{n-1},\]
and its expectation is computed through

\begin{align*}
\expe{X} &= \frac{n}{a} \rint{0}{a}{x \times \left( \frac{a-x}{a} \right) ^{n-1}}{x} \\
&= \frac{n}{a} \rint{0}{a}{(a-u) \times \left( \frac{u}{a} \right) ^{n-1}}{u} \\
&= \frac{n}{a^n} \rint{0}{a}{(a-u) \times u^{n-1} }{u} \\
&= \frac{n}{a^n} \left( a \times \frac{a^n}{n} - \frac{a^{n+1}}{n+1} \right) \\
\expe{X} &= \frac{a}{n+1}.
\end{align*}
\end{proof}
}

\newcommand\proofswitch{
\pr{We have}{It remains} to decide whenever the \ds is under contention or not, and
to find the corresponding solution.
Concerning the frontier between contended and non-contended system, we
can remark that Equations~\ref{eq:little-co} and~\ref{eq:little-nc}
are equivalent if and only if
\begin{equation}
\label{eq:little-fronti}
\frac{\rc + \cw + \scas}{\atrl} =
\frac{\atrl +2}{\atrl +1}  \left( \cw + \avexp{\atrl} \right)
+ 2\scas,
\end{equation}
which leads to Lemma~\ref{lem:lit-swi}.

\begin{lemma}
\label{lem:lit-swi}
The system switches from being
non-contended to being contended at $\atrl = \atrlf$, where
\pr{
\[ \atrlf = \frac{\scas+\cw-\mem}{2 (\cw + 2 \scas)} \left( \sqrt{1+\frac{4 (\mem+\cw+\scas) (\cw+2\scas)}{(\scas+\cw-\mem)^2}} -1 \right). \]
}
{
\[ \atrlf = \frac{-(\scas+\cw-\mem) + \sqrt{\left( \scas+\cw-\mem \right)^2 + 4 (\mem+\cw+\scas) (\cw+2\scas)}}{2 (\cw + 2 \scas)}.\]
}
\end{lemma}
\begin{proof}
We show that:
\begin{itemize}
\item \atrlf is the unique positive solution of Equation~\ref{eq:little-fronti} if the expansion is set to 0,
\item $\atrlf \leq 1$,
\item there is no solution of Equation~\ref{eq:little-fronti} with a non-null expansion.
\end{itemize}

If the expansion is set to 0, then Equation~\ref{eq:little-fronti} can
be turned into the second order equation
\[ \atrl^2 (\cw + 2 \cc) + \atrl \left( \cw + \cc - \mem \right) - (\mem + \cw + \cc) = 0,
\]
that has a single positive solution: \atrlf.

While instantiating the binomial with $\atrl = 1$, we obtain $\cw + 2
(\scas - \mem)$, which is not negative, since $\scas \geq \mem$ in all
the architectures that we are aware of.
As the second order equation has also a negative solution, and $\cw +
2 \cc$ is positive, we have that $1 \geq \atrlf$.
This implies that \atrlf is a solution of the former
Equation~\ref{eq:little-fronti}: the expansion is indeed a
non-decreasing function, thus $0 \leq \avexp{\atrlf} \leq \avexp{1} =
0$. Still we could have other solutions with a non-null expansion.

However, Equation~\ref{eq:little-fronti} can be rewritten as:
\begin{equation}
\label{eq:lit-fro-mon}
\rc + \cw + \scas =
\frac{\atrl +2}{\atrl +1} \times \atrl \times  \left( \cw + \avexp{\atrl} \right)
+ 2\scas.
\end{equation}
The left-hand side of Equation~\ref{eq:lit-fro-mon} is
constant, while the right-hand side is increasing, which discards
any other solution, hence the lemma.
\end{proof}

Thanks to Lemma~\ref{lem:lit-swi}, we can unify the \supw as:
\[
\avsupe{\atrl} = \left\{\begin{array}{ll}
\left( \mem + \cw + \scas \right) / \atrl & \quad \text{if } \atrl \leq \atrlf \\
 \left( \cw + \avexp{\atrl}\right) \times \frac{\atrl + 2}{\atrl +1} + 2\scas  &
 \quad \text{otherwise.}
\end{array}\right.
\]
The unified \supw obeys to the following equation
\begin{equation}
\label{eq:little-res}
\avsupe{\atrl} = \frac{\pw}{\ct - \atrl}.
\end{equation}
}

\newcommand{\prooffp}{
Let us note $\fuf{\atrl} = \avsupe{\atrl} \times \atrl$ and $\fus{\atrl} =
\pw \times \atrl / (\ct-\atrl)$; then Equation~\ref{eq:little-res} is equivalent to
$\fuf{\atrl} = \fus{\atrl}$, and we have some properties on \ffuf and \ffus.

Firstly, since $x \mapsto x(x+2)/(x+1)$ is non-decreasing on
$[0,\pinf[$, as well as the expected expansion, we know that \ffuf is
a non-decreasing function.
Secondly, \ffus is increasing on $[0,\ct[$, and is bijective from
$[0,\ct[$ to $[0,\pinf[$. We can thus rewrite Equation~\ref{eq:little-res} as:
\begin{equation}
\label{eq:fi-po}
 \atrl = \fusi{\fuf{\atrl}}.
\end{equation}
Moreover, $\ffusi \circ \ffuf$ is a non-decreasing function, as a
composition of two non-decreasing functions.
Thirdly, \ffusi can be obtained through $x = \fus{\fusi{x}} = \pw \times \fusi{x} / (\ct -
\fusi{x})$, which leads to
\[ \fusi{x} = \frac{x}{\pw + x} \ct. \]

In addition, we know by construction that if $\atrl > \atrlf$, then
\begin{equation}
\label{eq:little-fip}
 \left( \cw + \avexp{\atrl}\right) \times \frac{\atrl + 2}{\atrl +1} + 2\scas \geq \frac{\mem + \cw + \scas}{\atrl}.
\end{equation}
Indeed, on the one hand,
\[ \lim_{\atrl \rightarrow 0^{+}} \frac{\mem + \cw + \scas}{\atrl} = \pinf, \]
and on the other hand, $(\cw + \avexp{\atrl}) \times (\atrl + 2)/(\atrl +1) + 2\cc$
remains bounded. According to Lemma~\ref{lem:lit-swi}, those two
functions cross only once, hence Equation~\ref{eq:little-fip}.

Since $\avsupe{\atrl} = (\mem + \cw + \scas)/\atrl$ if $\atrl \leq
\atrlf$, we have $\avsupe{\atrl} \geq (\mem + \cw + \scas)/\atrl$
for any \atrl, and then
\[ \fuf{\atrl} \geq \mem + \cw + \scas. \]

Let then
\[ \atrli = \frac{\mem + \cw + \scas}{\pw + \mem + \cw + \scas} \ct.\]
We have seen that $\ffusi \circ \ffuf$ is a non-decreasing function,
hence
\begin{align*}
\fusi{\fuf{\atrli}} &\geq \fusi{\mem + \cw + \scas}\\
&\geq \frac{\mem + \cw + \scas}{\pw +  \mem + \cw + \scas} \times \ct\\
\fusi{\fuf{\atrli}} &\geq \atrli.
\end{align*}
Since \ffusi is bounded, Equation~\ref{eq:fi-po} admits a solution.

We are interested in the solution whose \atrl is minimal since it
corresponds to the first attained solution when the expansion grows,
starting from 0. The current theorem comes then from the application
of the Knaster-Tarski theorem.
}


\newcommand\proofaliexp{
Let us set the timeline so that at the beginning of the \supw, \ie
just after a successful \cas, we are at $t=0$.
Firstly, a success cannot start before $t=t_0$, where
$t_0=\scas+\lceil\cw/\scas\rceil\scas$. The quickest thread indeed
starts a failed \cas at $t=0$ and comes back from \cww at
$t=\scas+\cw$. It has then to wait for the current \cas to finish
before being able to obtain the cache line.
At $t=t_0$, $\trl - t_0/\scas +1$ threads are competing for the
data. Among them, 1 thread will lead to a successful \cas, while the
$\trl - t_0/\scas$ other threads will end up with a failed \cas.
If a failed \cas occurs, then at $t=t_0+\scas$, the same number of
threads compete, but now there is one more potential success and one
less potential failure.  In the worst case, it will continue until all
competing threads will lead to a successful \cas.

Let $\tcom = \trl - t_0/\scas +1$ the number of threads that are
competing at each round, and let, for all $i \in \inte[1]{\tcom}$,
$p_i = i/\tcom$ the probability to draw a thread that will execute a
successful \cas.

The expected number of failed \cass that occurs after the first thread comes back is then given by
\[
\expe{F} = \prsu{1} \times 0 + (1-\prsu{1}) \prsu{2} \times 1 + \dots + \\
(1-\prsu{1})(1-\prsu{2})\times\dots\times (1-\prsu{\tcom -1}) \times \prsu{\tcom} \times (\tcom -1).
\]
More formally,
\begin{align*}
\expe{F} &= \sum_{i=1}^{\tcom} \prod_{j=1}^{i-1} (1-\prsu{j}) \prsu{i} \times (i-1)\\
&= \sum_{i=1}^{\tcom} \prod_{j=1}^{i-1} (1-\frac{j}{\tcom}) \frac{i}{\tcom} \times (i-1)\\
&= \sum_{i=1}^{\tcom} \frac{1}{\left(\tcom\right)^i} \prod_{j=1}^{i-1} (\tcom-j) i(i-1)\\
\expe{F} &= \sum_{i=1}^{\tcom} \frac{i(i-1)}{\left(\tcom\right)^i} \frac{(\tcom-1)!}{(\tcom-i)!}\\
\end{align*}
}


\def\mmtr{M}
\def\mmab{C}
\def\mma{D}
\def\mmpe{Q}
\def\mmtrir{R}
\def\mmtril{T}

\newcommand{\mtr}[2]{\ema{\mmtr_{#1,#2}}}
\newcommand{\mab}[2]{\ema{\mmab_{#1,#2}}}
\newcommand{\ma}[2]{\ema{\mma_{#1,#2}}}
\newcommand{\mpe}[2]{\ema{\mmpe_{#1,#2}}}
\newcommand{\mtrir}[2]{\ema{\mmtrir_{#1,#2}}}
\newcommand{\mtril}[2]{\ema{\mmtril_{#1,#2}}}

\newcommand{\eigv}{\ema{v}}
\newcommand{\eig}[1]{\ema{\eigv_{#1}}}

\newcommand\transmat{
We consider here that the system is in a given state, and\rr{ we} compute
the probability that the system will next reach any other
state. Without loss of generality, we\rr{ can} choose the origin of time
such that the current \supw begins at $t=0$.


Let us first look at the core cases, \ie let $i \in \intmid \cup \inthig$ and
$k \in \inte[0]{\ct-i-1}$; we assume that the system is currently in
state \sta{i}, and we are interested in the probability that the
system will switch to \sta{i+k} at the end of the current state. In
other words, we want to find the probability that, given that the
current \supw started when $i$ threads were in the \rl, the next
\supw will begin while $i+k$ threads are in the \rl.

\rr{
\begin{figure}
\begin{center}
\begin{tikzpicture}[%
it/.style={%
    rectangle,
    text width=11em,
    text centered,
    minimum height=\pr{2}{3}em,
    draw=black!50,
    scale=.9,
  }
]
\coordinate (O) at (0,0);
\pha{pcas}{O}{\cas}{\pr{2}{5}}{\green}
\pha{sla}{pcas}{\wati{i}}{\pr{5}{8}}{\grey}
\pha{acc}{sla}{\cas}{\pr{2}{4}}{red}
\pha{cri}{acc}{\cw}{\pr{3}{4}}{\maroon}
\pha{exp}{cri}{\reexp{i}}{\pr{4}{6}}{\grey}
\pha{fcas}{exp}{\cas}{\pr{2}{5}}{\green}
\coordinate (intsl) at ($(sla.north west)+(0,1*\marup)$); 
\coordinate (intsr) at ($(sla.north east)+(0,1*\marup)$);
\coordinate (inte) at ($(fcas.north east)+(0,1*\marup)$); 
\draw [decorate,decoration={brace,amplitude=10pt}]
(intsl)  --  (intsr)
node [black,midway,yshift=15pt] {\scriptsize 0 new thread};

\draw [decorate,decoration={brace,amplitude=10pt}]
(intsr) -- (inte)
node [black,midway,yshift=15pt] {\scriptsize $k+1$ new threads};
\arcod{($(acc.north west)!.2! (acc.north east) + (0,.7*\marup)$)}{.4*\marup}
\arcod{($(exp.north west)!.1! (exp.north east) + (0,.7*\marup)$)}{.4*\marup}
\arcod{($(exp.north west)!.9! (exp.north east) + (0,.7*\marup)$)}{.4*\marup}

\coordinate (extsl) at ($(sla.south west)+(0,-\marup)$);
\coordinate (extsr) at ($(sla.south east)+(0,-\marup)$);
\draw [decorate,decoration={brace,mirror,amplitude=10pt},text width=11em, align=center] (extsl) -- (extsr)
node (caca) [black,midway,yshift=-18pt,execute at begin node=\setlength{\baselineskip}{8pt}] {\scriptsize at least 1\\new thread};

\coordinate (Ob) at ($(sla.south west)!.2!(sla.south) + (0,-3*\marup)$);
\pha{accb}{Ob}{\rf}{\pr{2}{4}}{yellow}
\pha{crib}{accb}{\cw}{\pr{3}{4}}{\maroon}
\pha{expb}{crib}{\reexp{i}}{\pr{4}{6}}{\grey}
\pha{fcasb}{expb}{\cas}{\pr{2}{5}}{\green}
\arcou{(accb.south west)}{3.3*\marup}
\arcou{($(crib.south west)!.2! (crib.south east) + (0,-.7*\marup)$)}{.4*\marup}
\arcou{($(expb.south west)!.9! (expb.south east) + (0,-.7*\marup)$)}{.4*\marup}
\draw[very thick, draw=blue] (accb.south west) -- (fcasb.south east) -- (fcasb.north east) -- (accb.north west);
\coordinate (extnl) at ($(accb.south west)+(0,-\marup)$); 
\coordinate (extnr) at ($(fcasb.south east)+(0,-\marup)$);
\draw [decorate,decoration={brace,mirror,amplitude=10pt},text width=11em, align=center]
(extnl) -- (extnr)
node [black,midway,yshift=-15pt] {\scriptsize $k$ new threads};

\draw[dotted, draw=black] (intsl) -- (extsl);
\draw[dotted, draw=black] (intsr) -- (extsr);
\draw[dotted, draw=black] (inte) -- (fcas.south east);
\draw[dotted, draw=black] (accb.north west) -- (extnl);
\draw[dotted, draw=black] (fcasb.north east) -- (extnr);
\node[text width=1.2cm, align=center] (inttext) at ($(pcas.west) + (-2*\marup,0)$) {Internal\\ execution};
\node[anchor=east] (eint) at ($(inttext.east) + (0.2,1.2*\marup)$) {\eve{int}};
\node[anchor=east] (eext) at ($(inttext.east) + (0.2,-2*\marup)$){\eve{ext}};

\path[->,out=90,in=-180] ($(inttext.north west)!.3!(inttext.north)$) edge (eint);
\path[->,out=-90,in=-180] ($(inttext.south west)!.3!(inttext.south)$) edge (eext);
\end{tikzpicture}
\end{center}
\caption{Possible executions\label{fig:ex-eint-eext}}
\end{figure}
}

As the successful thread will exit the \rl at the end of the current
\supw, there is at least one thread that enters the \rl during the
current \supw. Two non-overlapping events can then occur (see
Figure~\ref{fig:ex-eint-eext}): either the first thread exiting the
\psw starts within $[0,\wati{i}[$, \ie in the \watiw of the internal
    execution, and this event is written \eve{ext}, or the first
    thread entering the \rl starts after $t=\wati{i}$, and this event
    is denoted by \eve{int}.
Therefore, we have $\pro{\sta{i} \rightarrow \sta{i+k}} =
\pro{\eve{ext}} + \pro{\eve{int}}$.

First note that \eve{ext} cannot happen when the \supw is highly
contended; in this case, the \watiw is indeed null, and we conclude
$\pro{\eve{ext}} = 0$. In addition, we have seen in
Section~\ref{sec:mark-expa} that external threads, \ie threads that
are in the \psw at the beginning of the \supw, do not participate to
the game of expansion, so they cannot be successful. Under
high-contention, \eve{int} happens, and the successful \cas{} that
ends the \supw is operated by an internal thread, \ie a thread that
was already in the \rl when the \supw began.

Under medium contention, \eve{ext} can occur. In this case, an
external thread accesses the \ds before any internal thread does. We
have also seen that the expansion is null in medium contention level,
thus the external thread will execute its \cww, and especially its
\cas without being delayed; this implies that the first external
thread that accesses the \ds will end the current \supw with the end
of its \cas. If however \eve{int} occurs, an internal thread succeeds,
but is not necessarily the first thread that accessed the \ds during
the \supw.


The two possible events are pictured in Figure~\ref{fig:ex-eint-eext},
where the blue arrows represent the threads that exit the \psw. Recall,
we aim at computing the probability to start the next \supw with $i+k$
threads inside the \rl. We formalize the idea drawn in the figure by
using \thex{a}{b}, which is defined as a random variable
indicating the number of threads exiting the \psw during the time
interval $[a,b[$.
The probability of having \eve{int} is then given by\vspp{0}
\pr{\begin{equation*}
\pro{\eve{int}} = \pro{\thex{0}{\wati{i}}=0 \quad | \quad \trl=i \text{ at } t=0^{+}}
\times \pro{\thex{\wati{i}}{\wati{i}+\rw{i}}=k+1 \quad | \quad \trl=i \text{ at } t=\wati{i}^{+}}.
\end{equation*}}{
\[ \pro{\eve{int}} = \pro{\thex{0}{\wati{i}}=0 \quad | \quad \trl=i \text{ at } t=0^{+}} \\
\times \pro{\thex{\wati{i}}{\wati{i}+\rw{i}}=k+1 \quad | \quad \trl=i \text{ at } t=\wati{i}^{+}}.\]
}


Concerning \eve{ext}, we know that if $i \in \inthig$, then $\pro{\eve{ext}} = 0$.
Otherwise, if we denote by $t_3$ the starting time of the first
thread that exits the \psw, we obtain
\begin{align*}
\pro{\eve{ext}} =& \pro{\thex{0}{\wati{i}}>0 \quad | \quad \trl=i \text{ at } t=0^{+}} \\
&\times \pro{\thex{t_3}{t_3+\rc+\cw+\scas}=k \quad | \quad \trl=i+1 \text{ at } t=t_3^{+}}
\end{align*}
To simplify the reasoning, and given that the costs of \rf and \cas
are approximately the same, we approximate $t_3+\rc+\cw+\scas$ with $t_3+\scas+\cw+\scas$, leading to
\begin{align*}
\pro{\eve{ext}} =& \pro{\thex{0}{\wati{i}}>0 \quad | \quad \trl=i \text{ at } t=0^{+}} \\
&\times \pro{\thex{t_3}{t_3+\rw{i+1}}=k \quad | \quad \trl=i+1 \text{ at } t=t_3^{+}}
\end{align*}

According to the exponential distribution, given a thread that is in
the \psw at $t=a$, the probability to exit the \psw within $[a,b[$ is:
\[ \rint{a}{b}{\lambda \expu{-\lambda (t-a)}}{t} = \rint{0}{b-a}{\lambda \expu{-\lambda u}}{u}.  \]
It is also the probability, given a thread that is in the \psw at
$t=0$, to exit the \rl within $[a,b-a[$. This implies:
\begin{align*}
\pro{\eve{int}} =& \pro{\thex{0}{\wati{i}}=0 \quad | \quad \trl=i \text{ at } t=0^{+}} \\
&\times \pro{\thex{0}{\rw{i}}=k+1 \quad | \quad \trl=i \text{ at } t=0^{+}}
\end{align*}
and
\begin{align*}
\pro{\eve{ext}} =& \pro{\thex{0}{\wati{i}}>0 \quad | \quad \trl=i \text{ at } t=0^{+}} \\
&\times \pro{\thex{0}{\rw{i}}=k \quad | \quad \trl=i+1 \text{ at } t=0^{+}}.
\end{align*}

To lighten the notations, let us define
\begin{equation}
\label{eq:def-ab}
\left\{ \begin{array}{l}
\kkdur{i}{k} = \pro{\thex{0}{\rw{i}}=k \quad | \quad \trl=i \text{ at } t=0}\\
\kkaft{i} = \pro{\thex{0}{\wati{i}} = 0 \quad | \quad \trl=i \text{ at } t=\rw{i}^{+}}.
\end{array} \right.
\end{equation}

In addition, given a thread that is in the \psw at $t=0$, the
probability to exit the \psw within $[0,b-a[$ is $\rint{0}{b-a}{\lambda
      \expu{-\lambda u}}{u}$. By counting the number of threads that
    need to exit the \psw, we obtain:
\begin{equation}
\label{eq:exp-ab}
\left\{ \begin{array}{l}
\kkdur{i}{k} = \binom{\ct-i}{k} \left( 1 - \expu{-\lambda \rw{i}} \right) ^{k} \left( \expu{-\lambda \rw{i}} \right) ^{\ct-i-k}\\
\kkaft{i} = \left( \expi{-\lambda \wati{i}} \right) ^{\ct-i}.
\end{array} \right.
\end{equation}

Altogether, we have that
\[ \pro{\sta{i} \rightarrow \sta{i+k}} = \kkaft{i} \times \kkdur{i}{k+1} + (1-\kkaft{i}) \times \kkdur{i+1}{k}. \]

\medskip

The situation is slightly different if $k=-1$; in this case, no thread
should exit the \psw during the \watiw and no thread should exit during
the \re of the first thread that accessed the \ds during the
\supw neither. This shows that
\[ \pro{\sta{i} \rightarrow \sta{i-1}} = \kkaft{i} \times \kkdur{i}{0}. \]

When the \supw is not contended, \ie if $i=0$, the \watiw of the
execution that ignores external threads can be seen as infinite, hence
we can define $\kkaft{0} = 0$ (the probability that a thread exits its
\psw during an infinite interval of time is $1$). As for the
\kkdur{i}{k}'s, they can be defined in the same way as earlier.

We have obtained the full transition matrix $\left( \mtr{i}{j}
\right)_{(i,j) \in \inte{\ct-1}^2}$, which is a triangular matrix,
augmented with a subdiagonal:
\begin{equation*}
\left\{
\begin{array}{lll}
\mtr{i}{i+k} &=  \kkaft{i} \kkdur{i}{k+1} + (1-\kkaft{i}) \kkdur{i+1}{k} & \text{ if }
k \in \inte[0]{\ct-i-1}\\
\mtr{i}{i-1} &= \kkaft{i} \times \kkdur{i}{0} & \text{ if }
i > 0\\
\mtr{i}{j} &= 0 & \text{ otherwise}\\
\end{array}
\right.
\end{equation*}

\medskip

\begin{lemma}
$\mmtr$ is a right stochastic matrix.
\end{lemma}
\begin{proof}
First note that, by definition of \kkdur{i}{k}, for all $i \in \inte[0]{\ct-1}$,
\[ \sum_{k=0}^{\ct-i} \kkdur{i}{k} = 1. \]
If $i$ threads are indeed inside the \rl at $t=0$, then, within
$[0,\wati{i}[$, at least $0$ thread, and at most $\ct-i$ threads (inclusive) will exit
their \psw.

We have first
\[
\sum_{j=0}^{\ct-1} \mtr{0}{j} = \sum_{k=0}^{\ct-1} \kkdur{0+1}{k} = 1.
\]

In the same way, for all $i \in \inte[1]{\ct-1}$,
\begin{align*}
\sum_{j=0}^{\ct-1} \mtr{i}{j} =&\; \sum_{k=-1}^{\ct-1-i} \mtr{i}{i+k}\\
=&\; \kkaft{i} \times \kkdur{i}{0}  + \sum_{k=0}^{\ct-1-i} \kkaft{i} \kkdur{i}{k+1} + (1-\kkaft{i}) \kkdur{i+1}{k}\\
=&\; \kkaft{i} \times \sum_{k=-1}^{\ct-1-i} \kkdur{i}{k+1}  +   (1-\kkaft{i}) \sum_{k=0}^{\ct-1-i}  \kkdur{i+1}{k}\\
\sum_{j=0}^{\ct-1} \mtr{i}{j} =&\;1.
\end{align*}

\end{proof}

\begin{lemma}
The transition matrix has a unique stationary distribution, which is
the unique left eigenvector of the transition matrix with eigenvalue 1
and sum of its elements equal to 1.
\end{lemma}
\begin{proof}
Note that the Markov chain is
irreducible and aperiodic. Let $X \geq \ct -1$, $i \in \inte{\ct-1}$ and
$j \in \inte[i]{\ct-1}$.
\begin{align*}
\pro{\sta{j} \rightarrow \sta{i} \text{ in X steps}}  \geq &
\pro{\sta{j} \rightarrow \sta{j-1} \rightarrow \dots \rightarrow \sta{i}}\\
& \times \pro{\sta{i} \rightarrow \sta{i}}^{X-(j-i)}\\
\pro{\sta{j} \rightarrow \sta{i} \text{ in X steps}} > & 0
\end{align*}
As
\[ \pro{\sta{i} \rightarrow \sta{j} \text{ in X steps}} \geq \pro{\sta{i} \rightarrow \sta{j}} > 0, \]
the Markov chain is irreducible.
Since \sta{1} is clearly aperiodic, and the chain is irreducible, the chain is aperiodic as well.

This implies that the Markov chain has a unique stationary
distribution, which is the unique left eigenvector of the transition
matrix with eigenvalue 1 and sum of its elements equal to 1.
\end{proof}
}


\newcommand{\statdis}{
\begin{theorem}
Given the transition matrix, the stationary distribution can be found
in $(\ct+1)\ct -1$ operations.
\end{theorem}
\begin{proof}
As the Markov chain is irreducible, the stationary
distribution does not contend any zero. The space of the left
eigenvectors with unit eigenvalue is uni-dimensional; therefore, for
any \eig{0}, there exists a vector $\eigv = (\eig{0} \; \eig{1} \; \dots \;
\eig{\ct-1})$, such that \eigv spans this space.

Let $\eig{0}$ a real number; necessarily, \eigv fulfills $\eigv \cdot
\mmtr = \eigv$, hence for all $i \in \inte{\ct-2}$
\[ \sum_{k=0}^{i+1} \eig{k} \mtr{k}{i} = \eig{i}, \]
which leads to, for all $i \in \inte{\ct-2}$:
\[ \eig{i+1} = \frac{1}{\mtr{i+1}{i}} \left( (1-\mtr{i}{i})\eig{i} - \sum_{k=0}^{i-1} \eig{k} \mtr{k}{i} \right). \]
So we obtain the $\eig{1}, \dots, \eig{\ct-1}$ iteratively (we know
that $\mtr{i+1}{i} = \kkaft{i+1} \times \kkdur{i+1}{0}$, which is not
null), with $2\times i + 1$ operations needed to compute \eig{i+1}.

The elements of the stationary distribution should sum to one, so we
start from any \eig{0}, compute the whole vector, and then normalize
each element by their sum, hence the theorem.
\end{proof}

}


\newcommand{\watithput}{
In order to compute the final throughput, we have to compute the
expectation of the \watiw, when the system goes from state \sta{i} to
any other state, that we note \expe{\wati{\sta{i} \rightarrow
    \sta{\star}}}.
Also, we will be able to exhibit a vector $\vespv = (\vesp{0} ,
\vesp{1} , \dots , \vesp{\ct-1})$ of expected \supw, where \vesp{i} is
the expectation of the execution time of the \supw if $i$ threads are
in the \rl when the \supw begins:
\begin{equation*}
\left\{\begin{array}{ll}
\vesp{i} = \expe{\wati{\sta{i} \rightarrow \sta{\star}}} + \scas + \cw + \reexp{i} + \scas&\;\text{if }
i \notin \intnoc\\
\vesp{i} = \expe{\wati{\sta{i} \rightarrow \sta{\star}}} + \mem + \cw + \scas&\;\text{otherwise.}
\end{array}\right.
\end{equation*}

Finally, the expected throughput (inverse of the \supw) is calculated through
\pr{$\thru = (\eigv \cdot \vespv)^{-1},$}
{\[ \thru = \frac{1}{\eigv \cdot \vespv},\]}
where \eigv is the stationary distribution of the Markov chain.


 We know already that if
  $i \in \inthig$, then $\expe{\wati{\sta{i} \rightarrow \sta{i+k}}} =
  0$.

In the other extreme case, \ie if $i \in \intnoc$, we rely on the
following lemma.

\begin{lemma}
\label{lem:exp-min}
Let an integer $n$, a real number $\lambda$, and $n$ independent random
variables $X_1, X_2, \dots, X_n$, following an exponential
distribution of mean $\lambda^{-1}$. Let then $X$ be the random variable
defined by: $X = \min_{i \in \inte[1]{n}} X_i$. The expectation of $X$
is:
\[ \expe{X} = \frac{1}{\lambda n}. \]
\end{lemma}
\begin{proof}
We have
\begin{align*}
\pro{X > x} &=  \pro{\forall i : X_i > x}\\
&= \prod_{i=1}^n \pro{X_i > x}\\
&=  \left( \int_{x}^{\pinf}  \lambda \expu{-\lambda t} \right) ^{n}\\
\pro{X > x} &=  \expu{-\lambda n x}\\
\end{align*}
Therefore, the probability distribution of $X$ is given by:
\[ t \mapsto \lambda n \expu{-\lambda n t},\]
and its expectation is computed through

\begin{align*}
\expe{X} &= \rint{0}{\pinf}{ \lambda n t \expu{-\lambda n t}}{t}\\
&= \left[ \expu{-\lambda n t} t \right]_{\pinf}^0  + \rint{0}{\pinf}{\expu{-\lambda n t}}{t}\\
&= \left[ \frac{1}{\lambda n} \expu{-\lambda n t} \right]_{\pinf}^0 \\
\expe{X} &= \frac{1}{\lambda n}\\
\end{align*}
\end{proof}

This proves that
\[ \expe{\wati{\sta{0} \rightarrow \sta{\star}}} = \frac{1}{\pw \times \ct}. \]

Let now $i \in \intmid$, and $k \in \inte[-1]{\ct -i -1}$; we are
interested in \expe{\wati{\sta{i} \rightarrow \sta{i+k}}}.
%
%
The \watiw is less immediate, and we use the following reasoning.
First note that the probability distribution of the first thread
exiting the \psw is given by $t \mapsto \lambda (\ct-i) \expu{-\lambda (\ct-i)
  t}$. If this thread comes back during $]0,\wati{i}[$, the time that
    passed since the beginning of the \supw is the \watiw, otherwise,
    it is \wati{i}.

\begin{align*}
\expe{\wati{\sta{i} \rightarrow \sta{\star}}}
&= \rint{0}{\wati{i}}{\lambda (\ct-i) \expu{-\lambda (\ct-i) t} t}{t}
+ \rint{\wati{i}}{\pinf}{\lambda (\ct-i) \expu{-\lambda (\ct-i) t} \wati{i}}{t}\\
&= \left[ \expu{-\lambda (\ct-i) t} t \right]_{\wati{i}}^{0}
+ \left[ \frac{1}{\lambda (\ct-i)} \expu{-\lambda (\ct-i) t} \right]_{\wati{i}}^{0}
+ \wati{i} \left[ \expu{-\lambda (\ct-i) t} \right]_{\pinf}^{\wati{i}} \\
\expe{\wati{\sta{i} \rightarrow \sta{\star}}}&= - \wati{i}\expu{-\lambda (\ct-i) \wati{i}}
+ \frac{1- \expu{-\lambda (\ct-i) \wati{i}}}{\lambda (\ct-i)}
+ \wati{i} \left( \expu{-\lambda (\ct-i) \wati{i}} \right)
\end{align*}

We conclude that
\[ \expe{\wati{\sta{i} \rightarrow \sta{\star}}} = \frac{1 - \expu{-\frac{(\ct-i) \wati{i}}{\pw}}}{\ct-i} \pw. \]

Putting all together, we obtain
\begin{equation*}
\left\{\begin{array}{ll}
\expe{\wati{\sta{i} \rightarrow \sta{\star}}} = \frac{1 - \expu{-\frac{(\ct-i) \wati{i}}{\pw}}}{\ct-i} \pw
& \quad\text{if } i \in \intnoc \cup \intmid\\
\expe{\wati{\sta{i} \rightarrow \sta{\star}}} = 0 & \quad \text{if } i \in \inthig.
\end{array}\right.
\end{equation*}
}


\newcommand\failedres{
Another metric to estimate the quality of the model is the number of
failed \res per successful \re. We compute it by counting the number
of failed \res within the current \supw, where a \re is billed to a
given \supw if its failed \cas occurs during this \supw. We denote by
$\expe{\fa{i}}$ the expected number of failed \cas during a \supw that
begins with $i$ threads, where $i \in \inte{\ct-1}$.

If the \supw is not contended, \ie if $i \in \intnoc$, no failure will
occur since the first \cas of the \supw will be a success; hence
$\expe{\fa{i}} = 0 = i$.

If the \supw is medium contended, \ie if $i \in \intmid$, every thread
that is in the \rl in the beginning of the \supw will execute at least
one \cas during this \supw, and exactly two if the thread is the
successful one. We know indeed that, even if a thread exits its \psw
during the \watiw, and is then successful, the failed \cass will occur
before the thread entering the \rl executes its successful \cas. As
any thread that exits its \psw during the \supw either is successful at
its first \cas, or does not operate the \cas during the \supw, we
conclude that: $\expe{\fa{i}} = i$.

If the \supw is highly contended, \ie if $i \in \inthig$, then we know
that we have an uninterrupted sequence of failed \cass, from the
beginning of the \supw to the last ending successful \cas. The
expected number of failed \cass is then directly related to the
expected duration of the \supw. Recalling that the expansion is given
in Theorem~\ref{th:mark-expa}, we obtain:
\[ \expe{\fa{i}} = 1 + \frac{\cw + \reexp{i}}{\scas}. \]

}


\newcommand{\wholemarkov}{
\subsubsection{Transition Matrix}

\transmat

\subsubsection{Stationary Distribution}

\statdis

\subsubsection{\Watiw and Throughput}

\watithput

\subsubsection{Number of Failed \RES}
\label{sec:nbf}

\failedres
}


\newcommand\treib{
The lock-free stack by Treiber~\cite{lf-stack} is a fundamental \ds
that provides \popop and \pushop operations. To \popop an element, the
top pointer is read and the next pointer of the initial element is
obtained. The latter pointer will be the new value of the \cas that
linearizes the operation. So, accessing the next pointer of the
topmost element represents \cw as it takes place between the \rf and
the \cas.
We initialize the stack by pushing elements with or without
a stride from a contiguous chunk of memory. By this way, we are able
to introduce both costly or not costly cache misses. We also vary the
number of elements popped at the same time to obtain different \cw;
the results, with different \cw values are illustrated
in Figure~\ref{fig:stack}.}

\newcommand\synth{
We first evaluate our models using a set of synthetic tests
that have been constructed to abstract
different possible design patterns of lock-free data structures (value of \cw)
and different application contexts (value of \pw).
The \cww is either constant, or follows a Poisson distribution; in
Figure~\ref{fig:synt}, its mean value \cw is indicated at the top of
the graphs.

A steep decrease in throughput, as \pw gets low, can be observed for the cases with
low \cw, that mainly originates due to expansion.
When \cw is high, performance continues to increase when \pw
decreases, though slightly. The expansion is indeed low but the
\watiw, which appears as a more dominant factor, decreases as the
number of threads inside the retry loop increases.

When looking into the differences between the constructive and the \avba
approach: the \avba approach estimations come out
to be less accurate for mid-contention cases as it only differentiates
between contended and non-contended modes. In addition, it fails to capture
the failing retries when measured throughput starts to deviate
from the theoretical upper bound, as \pw gets lower. In contrast, the
constructive approach provides high accuracy in all metrics for almost
every case.

We have also run the same synthetic tests with a \pww that follows a
Poisson distribution (Figure~\ref{fig:synt-poisson}) or is constant
(Figure~\ref{fig:synt-const}), in order to observe the impact of the
distribution nature of the \pww. Compared to the exponential
distribution, a better throughput is achieved with a Poisson
distribution on the \pww. The throughput becomes even better with a
constant \pww, since the slack time is minimized due to the
synchronization between the threads, as explained
in~\cite{EXCESS:D2.3}.
}

\newcommand\synthtreib{
Here, we consider lock-free algorithms that strictly follow the
pattern in~\ref{alg:gen-name} and provide predictions using both the \avba
and the constructive approach together with the theoretical upper bound.}


\newcommand\finemm{%
One quantum of the collection
phase is the collection of the list of one thread, while three nodes
are reclaimed during one quantum of the reclamation phase. The
traditional MM scheme was parameterized by a threshold based on the number
of the removed nodes; the fine-grain MM scheme
is parameterized by the number of quanta that are executed at each
call.

We apply different MM schemes on the \deqop operation of the
Michael-Scott queue, and plot the results in Figure~\ref{fig:mm_perf}.
We initialize the queue with enough elements. Threads execute
\deqop, which returns an element, then call the MM scheme.
On the left side, we compare a pure queue (without MM), a queue with
the traditional MM (complete reclamation once in a while) and a queue with
fine-grain MM (according to the numbers of quanta that are executed
at each call). Note that the
performance of the traditional MM is also subject to the tuning of the
threshold parameter. We have tested and kept only the best parameter
on the studied domain.
First, unsurprisingly, we can observe that the pure queue
outperforms the others as its \cw is lower (no need to maintain the
list of nodes that a thread is accessing).
Second, as the fine-grain MM is called after each completed \deqop,
adding a constant work, the MM can be seen as a part of the \pww. We
highlight this idea on the second experiment (on the right side). We first
measure the work done in a quantum. It follows that, for each value
of the granularity parameter, we are able to estimate the effective
\pww as the sum of the initial \pw and the work added by the
fine-grain MM. Finally, we run the queue with the fine-grain MM, and
plot the measured throughput, according to the effective parallel
work, together with our two approaches instantiated with the effective
\pw. The graph shows the validity of the model estimations for all
values of the granularity parameter.
}


\newcommand\adaptsine{
Numerous scientific applications are built upon a pattern of
alternating phases, that are communication- or
computation-intensive. If the application involves \dss, it is
expected that the rate of the modifications to the \dss is high
in the data-oriented phases, and conversely.
These phases could be clearly separated, but the application can also
move gradually between phases. The rate of modification to a \ds
will anyway oscillate periodically between two extreme values.
We place ourselves in this context, and evaluate the two MMs
accordingly. The \pww still follows an exponential distribution of
mean \pw, but \pw varies in a sinusoidal manner with time, in order to
emulate the numerical phases. More precisely, \pw is a step
approximation of a sine function. Thus, two additional parameters
rule the experiment: the period of the oscillating function represents
the length of the phases, and the number of steps within a period
depicts how continuous are the phase changes.
}


\newcommand\fulldeq{
We consider the \deq designed in~\cite{deq}. \pushl and
\pushr (resp. \popl and \popr) operations are exactly the same, except
that they operate on the different ends of the \deq.
The status flags, which depict the state of the \deq, and the
pointers to the leftmost element and the rightmost element are
together kept in a single double-word variable, so-called
\anch, which could be modified by a double-word \cas
atomically.

A \popl operation linearizes and even completes in one \staw that ends
with a double-word \cas that just sets the left pointer of the anchor
to the second element from left.

A \pushl operation takes three \staws to complete. In the first \staw,
the operation is linearized by setting the left pointer of the \anch
to the new element and at the same time changing the status flags to
``left unstable''\pr{.}{, to indicate the status of the incomplete but
linearized \pushl operation.} In the second \staw, the left pointer of
the leftmost element is redirected to the recently pushed element.
In the third \staw, a \cas is executed on \anch to bring the \deq
status flags into ``stable state''. Every operation can help an incomplete
\pushl or \pushr until the \deq comes into the stable state; in this
state, the other operations can attempt to linearize anew.

As noticed, the first and the third \staw execute a \cas on the same
variable (\anch) so it is possible to delay the third \staw of the
\supw by executing a \cas in the first \staw. This implies
that the expansion in \staw one should also be considered when the
delay in the third \staw is considered, and the other way around. This
can be done by summing expansion estimates of the \staws that run the
\cas on the same variable and using this expansion value in all these
\staws. Again, it just requires simple modifications in the expansion
formula by keeping assumptions unchanged.

We first run pop-only and push-only experiments where dedicated
threads operate on both ends of the \deq, in a half-half
manner. We provide predictions by plugging the slightly modified
expansion estimate, as explained above, into the \avba approach. Then,
we take one step further and mix the operations, assigning the threads
inequally among push and pop operations.
And, we obtain estimates for them by simply taking the weighted
average (depending on the number of threads running each operation) of
the \supw of pop-only and push-only experiments, with
the corresponding \pw value.


\begin{figure}[h!]
\begin{center}
\includegraphics[width=.8\textwidth]{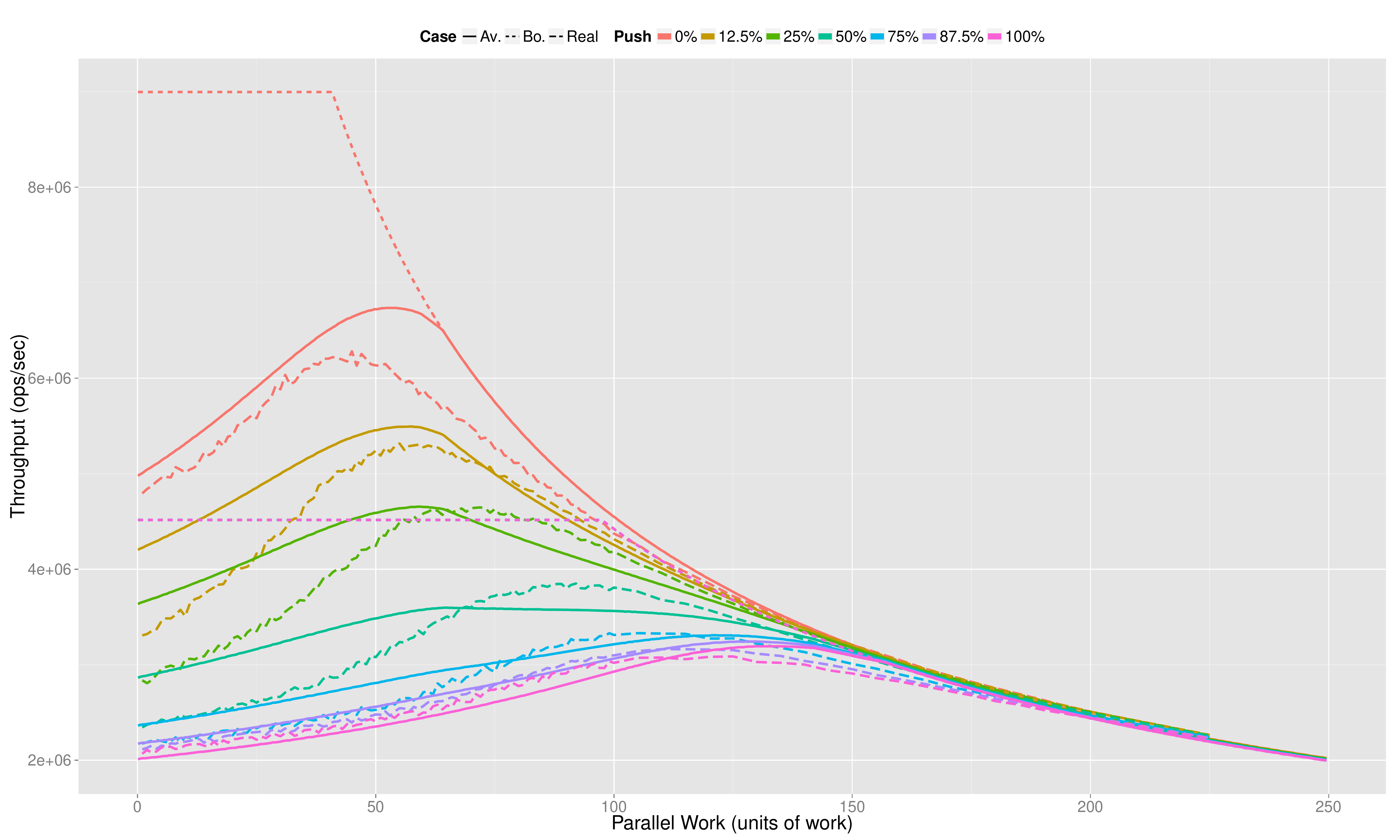}
\end{center}
\caption{Operations on \deq\label{fig:deq}}
\end{figure}

In Figure~\ref{fig:deq}, results are illustrated; they are
satisfactory for the push-only and pop-only cases.
For the mixed-case experiments, the results are mixed: our analysis
follows the trend and becomes less accurate
when the \pw gets lower, as experimental
curves tend toward push-only \supw. This, presumably, happens because
the first \staw of a \pushl (or \pushr) operation is shorter than the
first \staw of a \popl (or \popr) operation. This brings indeed an
advantage to push operations, under contention: they have higher chances
to linearize before pop operations after the \ds comes into the stable
state. It\rr{ also} provides an interesting observation which highlights
the lock-free nature of operations: it is improbable to complete a pop
operation if numerous threads try to push, due to the
difference of work inside the first \staw of their \rl.
}


\newcommand{\fullenq}{
As a first step, we consider the \enqop operation of the MS queue to
validate our approach.  This operation requires two
pointer updates leading to two \staws, each ending with a \cas. The
first \staw, that linearizes the operation, updates the next pointer
of the last element to the newly enqueued element. In the next and
last \staw, the queue's head pointer is updated to point to the recently
enqueued element, which could be done by a helping thread, that brings
the data structure into a stable state.

\rr{
\begin{figure}[h!]
\begin{center}
\includegraphics[width=.8\textwidth]{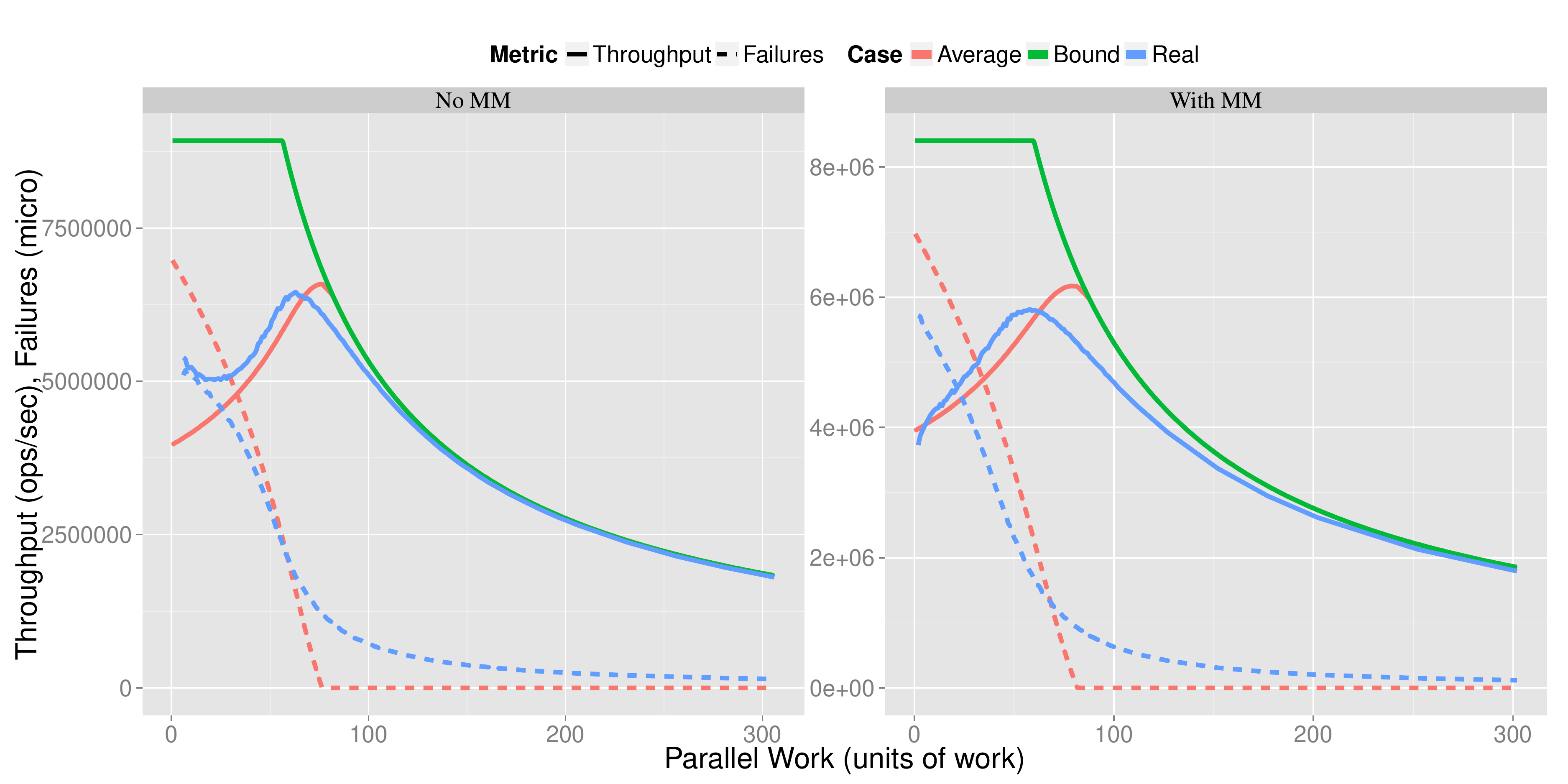}
\end{center}
\caption{Enqueue on MS Queue \label{fig:enqueue}}
\end{figure}

We estimate the expansion in the \supw as described above and
throughput as explained in Section~\ref{sec:avba}. The results for the
\enqop experiments where all threads execute \enqop are presented in
Figure~\ref{fig:enqueue}.
}
}


\newcommand{\fulladsexp}{
Consider an operation such that, the \supw (ignoring the slack time) is composed
of $S$ \staws (denoted by $\stag{1}, \dots, \stag{S}$) where each stage represents a step
towards the completion of the operation.
Let \casn{i} denote the \cas operation at the end of the \stag{i}.
From a system-wide perspective, $\{ \casn{1}, \dots , \casn{S}\}$ is the set of \cass that
have to be successfully and consecutively executed to complete an operation, assuming
all threads are executing the same operation.
This design enforces that \casn{i} can be successful only if the last successful \cas
is a \casn{i-1}. And, \casn{1} can be successful only if the last successful \cas
is a \casn{S}. In other words, another operation can not linearize before the
completion of the linearized but incomplete operation.

Now, let $e_i$ denote the expected expansion of \casn{i}.  If the \ds
is in the stable state (\ie is in \stag{1}, where a new operation can be
linearized), then we have to consider the probability, for all threads
except one, to expand the successful \casn{1} which linearizes the
operation. After the linearization, this operation will be completed
in the remaining stages where again the successful \cass at the end
of the stages are subject to the same expansion possibility by the
threads in the \rl, as they might be still trying to help for the
completion of the previously completed operation.

Similar to~\cite{EXCESS:D2.3}, our assumption here is that any
thread that is in the \rl, can launch \casn{i}, with probability $h$,
that might expand the successful \casn{i}. We consider, the starting
point of a failing \casn{i} is a random variable which is distributed
uniformly within the \rl, which is composed of expanded \staws of the
operation. This is because an obsolete thread can launch a \casn{i},
regardless of the \staw in which the \ds is in (equally, regardless of the
last successful \cas). Due to the uniformity assumption, the expansion
for the successful \cass in all stages, would be equal.  Similar to
the~\cite{EXCESS:D2.3}, we estimate the expansion $e_i$ by considering the
impact of a thread that is added to the \rl. Let the cost function
$\mathit{delay}_i$ provide the amount of delay that the additional thread
introduces, depending on the point where the starting point of its
\casn{i} hits. By using these cost functions, we can formulate the
total expansion increase that each new thread introduces and derive
the differential equation below to calculate the expected total
expansion in a \supw , where $\avexp{\atrl}=\sum^{S}_{i=1}
\aexpi{\atrl}$. Note that, we assume that the expansion starts as soon
as strictly more than 1 thread are in the retry loop, in
expectation.

\begin{lemma}
\label{lem.1}
The expansion of a \cas operation is the solution of the following
system of equations, where $\rlw = \sum^{S}_{i=1} \rlw_i =
\sum^{S}_{i=1}(\rc_i + \cw_i + \cc_i)$:
\[ \left\{
\begin{array}{lcl}
\expansionp{\atrl} &=& \fcas \times \dfrac{S \times \frac{\fcas}{2} + \expansion{\atrl}}{ \rlw + \expansion{\atrl}}\\
\expansion{\trlo} &=& 0
\end{array} \right., \text{ where \trlo is the point where expansion begins.}
\]
\end{lemma}

\begin{proof}

We compute $\expansion{\atrl + h}$, where $h\leq1$, by assuming that
there are already \atrl threads in the \rl, and that a new thread
attempts to \cas during the \re, within a probability $h$. For
simplicity, we denote $a^i_j = (\sum_{j=1}^{i-1} \rlw_j + e_j(\atrl))
+ \rc_i + \cw_i$.

\begin{align*}
\expansion{\atrl + h}
&= \expansion{\atrl} + h\times
\sum^{S}_{i=1} \rint{0}{\rlsiz}{\frac {\shifti{i}{t_i}}{\rlsiz}}{t_i} \\
&= \expansion{\atrl}
     + h \times \sum^{S}_{i=1} \Big( \rint{0}{a^i_j - \cc}{\frac{\shifti{i}{t_i}}{\rlsiz}}{t_i}  +
      \rint{a^i_j - \fcas}{a^i_j}{\frac{\shifti{i}{t_i}}{\rlsiz}}{t_i}\\
& \quad\quad\quad\quad\quad\quad\quad\quad\quad\quad  + \rint{a^i_j}{a^i_j + \aexpi{\atrl}}{\frac{\shifti{i}{t_i}}{\rlsiz}}{t_i}
  + \rint{a^i_j + \aexpi{\atrl}}{\rlsiz}{\frac{\shifti{i}{t_i}}{\rlsiz}}{t_i}\Big)\\
  &= \expansion{\atrl} + h \times \sum^{S}_{i=1} \Big(
    \rint{a^i_j-\fcas}{a^i_j}{\frac{t_i}{\rlsiz}}{t_i}
      + \rint{a^i_j}{a^i_j + \aexpi{\atrl}}{\frac{\fcas}{\rlsiz}}{t_i} \Big)\\
\expansion{\atrl + h} &= \expansion{\atrl} + h \times \frac{ (\sum^{S}_{i=1} \frac{\fcas^2}{2}) + \expansion{\atrl}\times\fcas}{\rlsiz}
\end{align*}

This leads to
\[ \quad\frac{\expansion{\trl + h}- \expansion{\atrl}}{ h} = \frac{ S \times \frac{\fcas^2}{2} + \expansion{\atrl}\times\fcas}{\rlsiz}.\]
When making $h$ tend to $0$, we finally obtain
\[ \expansionp{\atrl} = \fcas \times \frac{S \times \frac{\fcas}{2} + \expansion{\atrl}}{ \rlw + \expansion{\atrl}}. \qedhere\]
\end{proof}

In addition, if a set $S_k$ of \cass are operating on the same
variable $var_k$, then $\casn{i} \in S_k$ can be expanded by the
$\casn{j} \in S_k$. In this case, we can obtain $\aexp{k}{\atrl}$ by
using the reasoning above. The calculation simply ends up as follows:
Consider the problem as if no \cas shares a variable and denote
expansion in \stag{i} with $\aexpi{\atrl}^{(\mathit{old})}$. Then, $\aexp{k}{\atrl}
= \sum_{\cas_i \in S_k} \aexpi{\atrl}^{(\mathit{old})}$.
}


\newcommand{\fulladswati}{
We assume here the slack time can only occur after the completion of
an operation (\ie before stage 1), as the other stages are expected to
start immediately due to the thread that completes the previous
stage. Similar to Section~\ref{sec:litt-slack}, we consider that, at
any time, the threads that are running the retry loop have the same
probability to be anywhere in their current retry.  Thus, a thread can
be in any stage just after the successful CAS that completes the
operation. So, we need to consider the thread which is closest to the
end of its current stage when the operation is completed. We denote
the execution time of the expanded retry loop with \rlsiz and the
number of \staws with $S$. For a thread executing \stag{i} when the
operation completes, the time before accessing the \ds is then
uniformly distributed between 0 and $\rlsiz_i$.

Here, we take another assumption and consider all stages can be
completed in the same amount of time (\ie for all (i, j) in $\{1,
\dots ,S\}^2$, $\rlsiz_i = \rlsiz_j = \rlsiz/S$). This assumption
does not diverge much from the reality and provides a reasonable
approximation. With these assumption and using
Lemma~\ref{lem:unif-min}, we conclude that:

\begin{equation}
\label{eq:slack-multiple}
\avwati{\atrl} = \frac{\rlsiz}{S \times (\atrl +1)}.
\end{equation}

}


\newcommand{\sumads}{
Here, we consider \dss that apply immediate helping, where threads
help for the completion of a recently linearized operation until the
\ds comes into a stable state in which a new operation can be
linearized. The crucial observation is that the \ds goes through
multiple \staws in a round robin fashion.  The first \staw is the one
where the operation is linearized. The remaining ones are the \staws
in which other threads, that execute another operation, might help for
the completion of the linearized operation, before attempting to
linearize their own operations.  Thus, the \supw (ignoring the \watiw)
can be seen as the sum of the execution time of these \staws, each
ending with a \cas that updates a pointer. The \cas in the first stage
might be expanded by the threads that are competing for the
linearization of their operation, and consequent \cass might be
expanded by the helper threads, which are still trying to help an
already completed operation.  Also, there might be slack time before
the start of the first \staw as the other \staws will start
immediately due to the thread that has completed the previous \staw.

Although it is hard to stochastically reconstruct the executions
with Markov chains, our \avba approach
provides the flexibility required to estimate the performance by
plugging the expected \supw, given the number of threads inside the
\rl, into the Little's Law.  As the impacting factors are similar, we
estimate the \supw in the same vein as in Section~\ref{sec:avba}; with
a minor adaptation of the expansion formula
and by slightly adapting the slack time estimation based on the
same arguments.
}


\newcommand\bothmm{
\falseparagraph{Fine-grain Memory Management Scheme}
We divide the routine (and further the phases) of the traditional MM
mechanism into quanta (equally-sized chunks).%
\finemm

\falseparagraph{Adaptive Memory Management Scheme}
We build the adaptive MM scheme on top of the fine-grain MM mechanism by
adding a monitoring routine that tracks the number of failed \rls,
employing a sliding windows. Given a granularity parameter and a
number of failed \rls, we are able to estimate the parallel work and
the throughput, hence we can decide a change in the granularity
parameter to reach the peak performance. Note that one can avoid
memory explosion by specifying a threshold like the traditional
implementation in case the application provides a durable low
contention; in the worst case, it performs like the traditional MM.
}


\newcommand\fullbosy{
In Figure~\ref{fig:bo-synt}, we compare, on a synthetic workload, this
constant back-off strategy against widely known strategies, namely
exponential and linear, where the back-off amount increases
exponentially or linearly after each failing retry loop starting from
a \cycles{115} step size.
}


\newcommand\figsynthcst{
\begin{figure}[b!]
\begin{center}
\includegraphics[width=.85\textwidth]{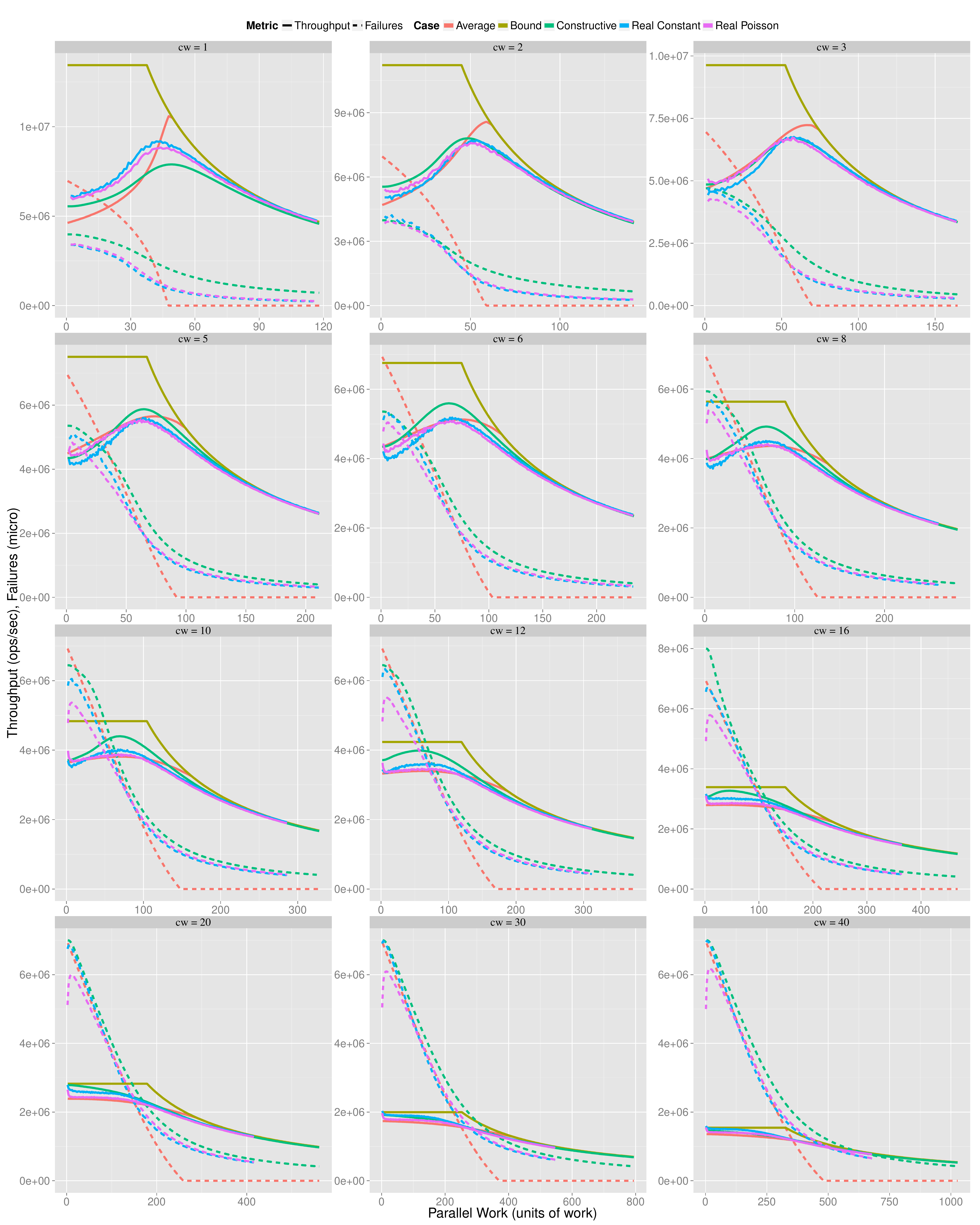}
\end{center}
\caption{Synthetic program with exponentially distributed \pww\label{fig:synt}}
\end{figure}}

\newcommand\figsynthpoi{
\begin{figure}[b!]
\begin{center}
\includegraphics[width=.85\textwidth]{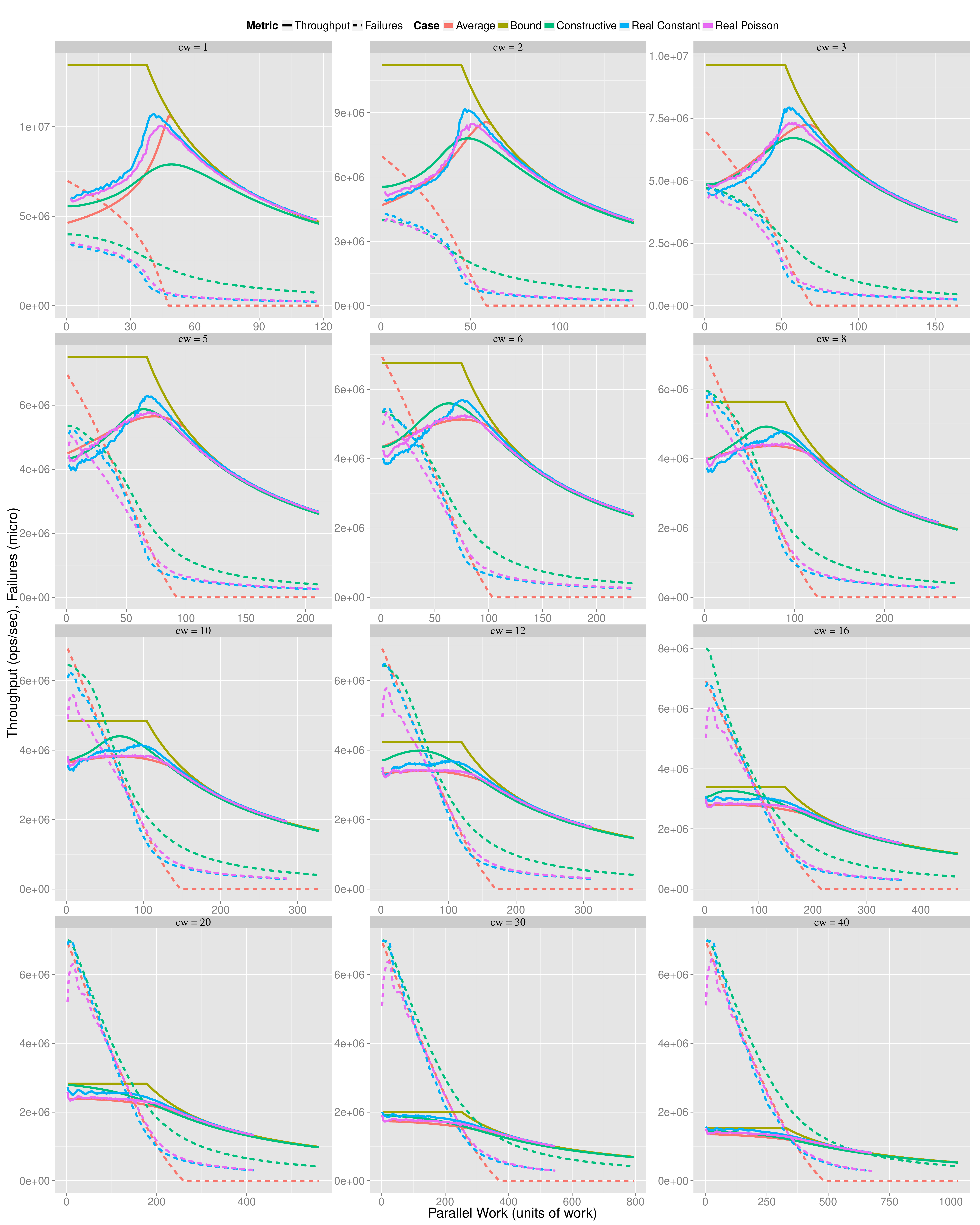}
\end{center}
\caption{Synthetic program with \pww following Poisson\label{fig:synt-poisson}}
\end{figure}}

\newcommand\figsynthconst{
\begin{figure}[b!]
\begin{center}
\includegraphics[width=.85\textwidth]{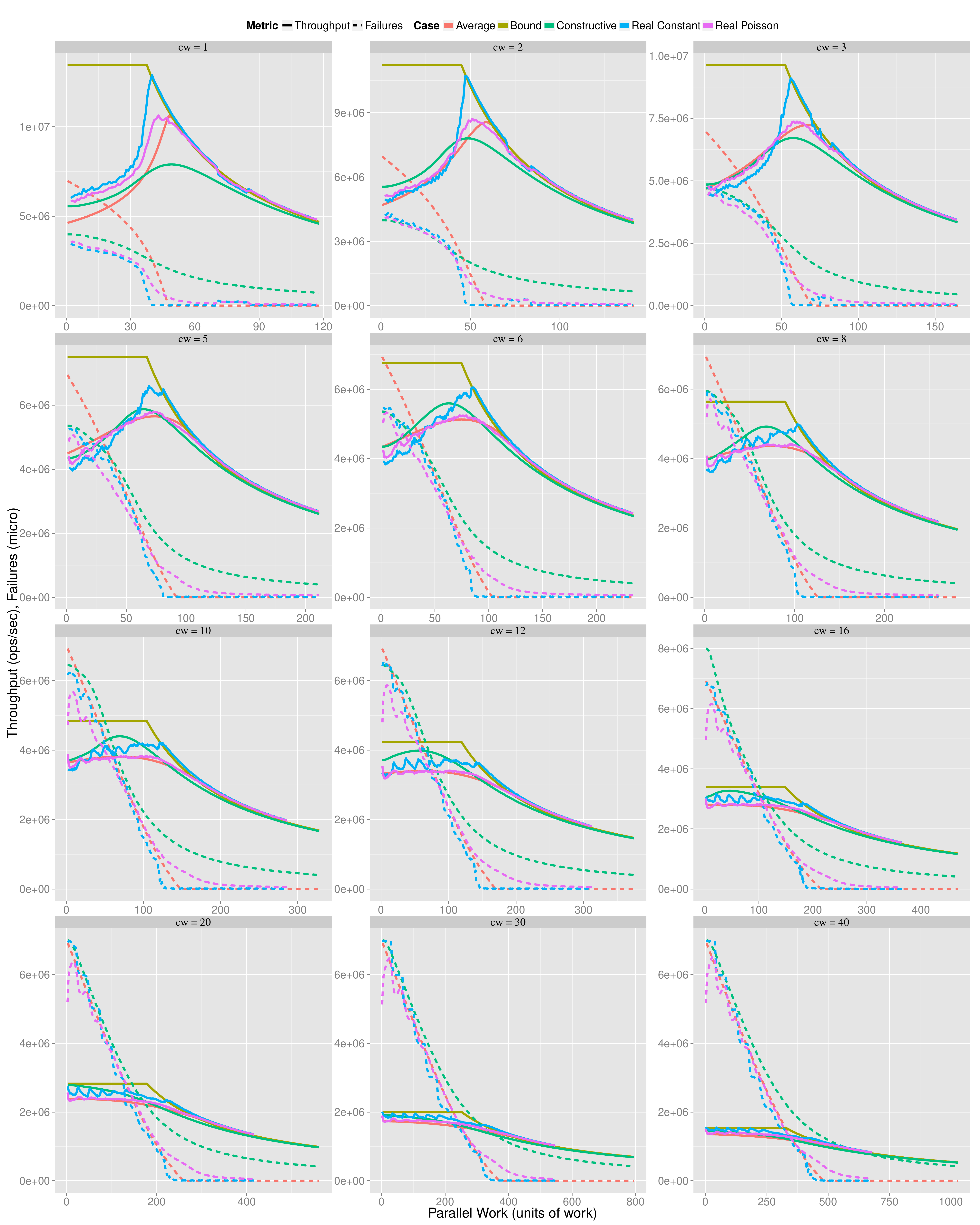}
\end{center}
\caption{Synthetic program with Constant \pww \label{fig:synt-const}}
\end{figure}}

\newcommand\figtreib{
\begin{figure}[h!]
\begin{center}
\includegraphics[width=\textwidth]{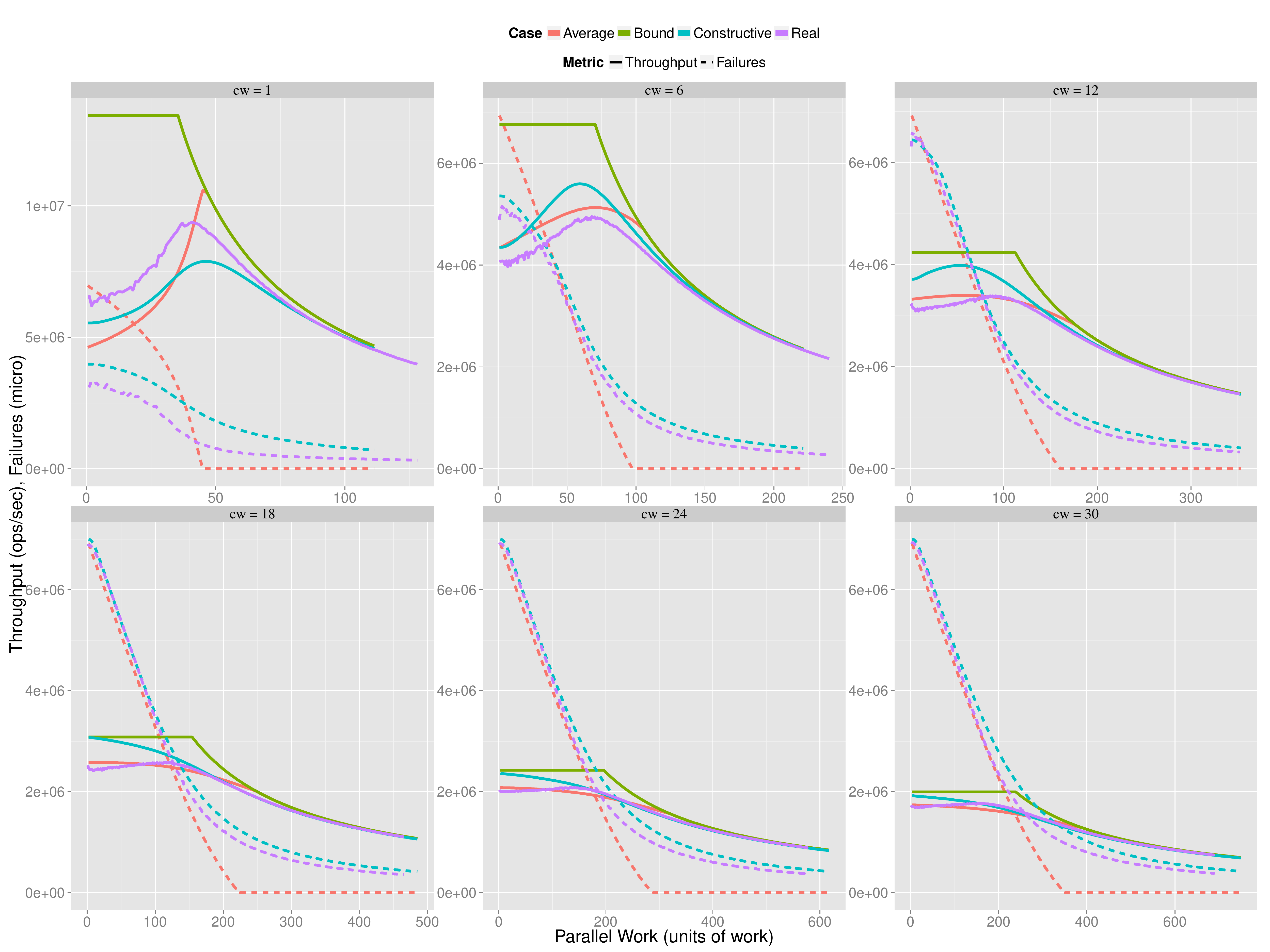}
\end{center}
\caption{Treiber's Stack\label{fig:stack}}
\end{figure}}

\newcommand\figcompmm{
\begin{figure}[b!]
\begin{center}
\pr{\includegraphics[width=.8\textwidth]{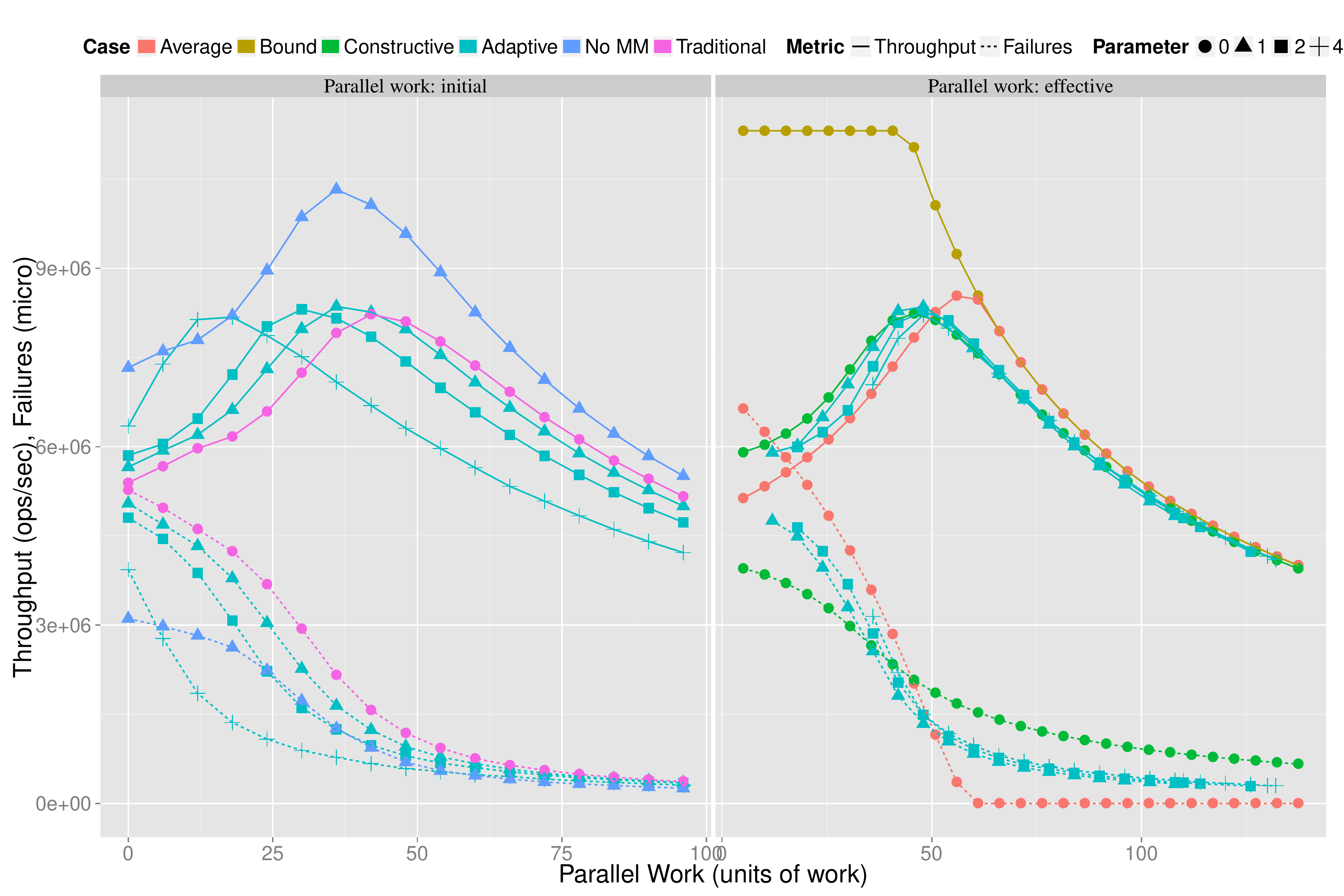}}{\includegraphics[width=.9\textwidth]{dequeue_xp_rr_disc}}
\end{center}
\caption{Performance of memory management mechanisms\label{fig:mm_perf}}
\end{figure}}


\subsection{Introduction}





%

Here, we consider the modeling and the analysis of the performance
of lock-free data structures. Then, we combine the perfomance analysis with our power
model that is introduced in D2.1~\cite{EXCESS:D2.1} and D2.3~\cite{EXCESS:D2.3} to estimate
the energy efficiency of lock-free data structures that are used in various settings.

Lock-free data structures are based on retry loops and are called by 
application-specific routines. In contrast to the model and analysis 
provided in D2.3, we consider here the lock-free data structures in 
dynamic environments. The size of each of the retry loops,
and the size of the application routines invoked in between, are not constant
but may change dynamically.

During the last two decades, lock-free \dss have received a lot of
attention in the literature, and have been accepted in industrial
applications, \eg in the Intel's Threading Building Blocks
Framework~\cite{itbbf}, the Java concurrency package~\cite{jav-conc}
and the Microsoft .NET Framework~\cite{mic-net-f}.
\rr{Lock-free implementations provide indeed a way out of several
limitations of their lock-based counterparts, in robustness,
availability and programming flexibility. Last but not least, the
advent of multi-core processors has pushed lock-freedom on top of the
toolbox for achieving scalable synchronization.}

Naturally, the development of lock-free \dss was accompanied by
studies on the performance of such \dss, in order to characterize their
scalability.
%
%
Having no guarantee on the execution time of an individual operation,
the time complexity analyses of lock-free algorithms have turned
towards amortized analyses.
The so-called amortized analyses are thus interested in the \wc
behavior over a sequence of operations, which can be seen as a \wc
bound on the average time per operation.
In order to cover various contention environments, the time complexity
of the algorithms is often parametrized by different contention
measures, such as point~\cite{point-contention}, interval~\cite{interval-contention}
or step~\cite{step-contention} contention.
Nonetheless these investigations are targeting worst-case asymptotic
behaviors.  There is a lack of analytical results in the literature
capable of describing the execution of lock-free algorithms on top of a
hardware platform, and providing predictions that are close to what is
observed in practice.
Asymptotic bounds are particularly useful to rank different
algorithms, since they rely on a strong theoretical background, but
the presence of potentially high constants might produce misleading
results. Yet, an absolute prediction of the performance can be of
great importance by constituting the first step for further
optimizations. 

The common measure of performance for \dss is throughput, defined
as the number of operations on the \ds per unit
of time.
To this end, this performance measure is usually obtained by
considering an algorithm that strings together a pure sequence of
calls to an operation on the \ds. However, when used in a more
realistic context, the calls to the operations are mixed with
application-specific code (that we call here \pww). For instance, in a
work-stealing environment designed with \deqs, a thread basically runs
one of the following actions: pushing a new-generated task in its
\deq, popping a task from a \deq or executing a task. The
modifications on the \deqs are thus interleaved with \deq-independent
work. There exist some papers that consider in their experiments local
computations between calls to operations during their respective
evaluations, but the amount of local computations follows a given
distribution varying from paper to paper, \eg constant
\cite{lf-queue-michael}, uniform \cite{scalable-stack-uniform},
exponential \cite{Val94}.


%

In this work, we derive a general approach for unknown distributions of 
the size of the application-specific code, as well as a tighter method 
when it follows an exponential distribution.

As for modeling the \ds itself, we use as a basis the universal construction
described by Herlihy in~\cite{herli-univ-const}, where it is shown
that any abstract data type can get such a lock-free implementation, which
relies on one \rl.
%
Moreover, we have particularly focused our experiments on \dss
that present a low level of disjoint-access
parallelism~\cite{disjoint-access} (stack, queue, shared counter,
\deq). Coming back to amortized analyses, the time complexity of an
operation is often expressed as a contention-free time complexity
added with a contention overhead. In this work, we want to model and
analyze the impact of contention, whether nonexistent, mediocre or
high. So that the contention overhead is not hidden, we focus on \dss
with low contention-free complexity, that can also provide very high
contention without bringing hundreds of threads into play.


We propose two different approaches that analyze the performance of
such \dss.
On the one hand, we derive an \avba approach invoking queuing theory,
which provides the throughput of a lock-free algorithm without any
knowledge about the distribution of the \pww. This approach
is flexible but allows only a coarse-grained analysis, and hence a
partial knowledge of the contention that stresses the \ds.
On the other hand, we exhibit a detailed picture of the execution of
the algorithm when the \pww is instantiated with an exponential
distribution, through a second complementary approach. We prove that
the multi-threaded execution follows a Markovian process and a Markov
chain analysis allows us to pursue and reconstruct the execution, and
to compute a more accurate throughput.

We finally show several ways to use our analyses and we evaluate the
validity of our ideas by experimental results. Those two analysis
approaches give a good understanding of the phenomena that drive the
performance of a lock-free \ds, at a high-level for the \avba
approach, and at a detailed level for the constructive
method. 
We also
emphasize that there exist several concrete paths to apply our
analyses.
To this end, based on the knowledge about the application at hand, we
implement two back-off strategies. We show the applicability of these strategies by tuning a Delaunay triangulation
application~\cite{caspar} and a streaming pipeline component which is
fed with trade exchange workloads~\cite{taq-se}.
We also design a new adaptive memory management mechanism for
lock-free data structures in dynamic environments which surpasses the
traditional scheme and which is such that the loss in performance,
when compared to a static data structure without memory management, is
largely leveraged. This memory management mechanism is based on the analyses presented in this work. 

Lastly, we show how these results can be used to obtain the energy consumption 
of the lock-free data structures.

\rr{
The rest is organized as follows: we start by presenting
related work in Section~\ref{sec:related}, then we define the
algorithm and the platform that we consider, together with concepts
that are common to our both approaches in Section~\ref{sec:preli}. The
\avba approach is described in Section~\ref{sec:avba}, while the
constructive analysis is exposed in Section~\ref{sec:cons}, both
methods are evaluated in the experiment part that is presented in
Section~\ref{sec:xp} and the energy model with the evaluations is given in 
Section~\ref{sec:markov-energy}.
}

\vspp{-.4}
\subsection{Previous Work}
\label{sec:related}

In D2.3, performance impacting factors are
illustrated for a subset of the lock-free structures that we consider
in this work. In the former paper, the analysis is built upon
properties that arise only when the sizes of the \cww and the \pww are
constant. There, we show that the execution is not memoryless due to the
natural synchrony provided by the \rls; at the end of the line, we
prove that the execution is cyclic and use this property to bound the
rate of failed \res. This work is complementary to that work, not only
because of the difference in the analysis tools but also because they
altogether exhibit the impact of the size distributions of the \pww
on the performance of lock-free \dss. Moreover,\rr{ owing to our
assumptions on the size of the parallel and \cwws,} the results of this paper
can be applied to a larger variety of \dss running on a larger variety of
environments.


\vspp{-.15}
\subsection{Preliminaries}
\label{sec:preli}

We describe in this subsection the structure of the algorithm that is covered by
our model. We explain how to analyze the execution of an instance of
such an algorithm when executed by several threads, by slicing this
execution into a sequence of adjacent \supws, where a \supw is an interval of
time during which exactly one operation returns. Each of the \supws
is further split into two by the first access to the \ds in the
considered \rl.
%
This execution pattern reflects fundamental phases of both analyses,
whose first steps and general direction are outlined at the end of the
subsection.



\vspp{-.3}
\subsubsection{System Settings}


All threads call Procedure~\ref{alg:gen-name} (see
Figure~\ref{alg:gen-nb}) when they are spawned. So each thread follows
a simple though expressive pattern: a sequence of calls to an
operation on the \ds, interleaved with some \pww during which the
thread does not try to modify the \ds.
For instance, it can represent a work-stealing algorithm, as
described in the introduction.
%
%

The algorithm is decomposed in two main sections: the {\it \psw},
represented on line~\ref{alg:li-ps}, and the {\it \rl} (which
represents one operation on the shared \ds) from
line~\ref{alg:li-bcs} to line~\ref{alg:li-ecs}. A {\it \re} starts at
line~\ref{alg:li-bbcs} and ends at line~\ref{alg:li-ecs}. The outer
loop that goes from line~\ref{alg:li-bwl} to line~\ref{alg:li-ecs} is
designated as the {\it \wl}\rr{.

}\pp{. }
In each \re, a thread tries to modify the \ds and does not exit the
\rl until it has successfully modified the \ds.
\rr{It firstly reads the
access point \DataSty{AP} of the \ds, then, according to the value
that has been read, and possibly to other previous computations that
occurred in the past, the thread prepares, during the \cww,
%
%
the new desired value as an access point of
the \ds. Finally, it atomically tries to perform the change through a
call to the \cas primitive. If it succeeds, \ie if the access point
has not been changed by another thread between the first \rf and the
\cas, then it goes to the next \psw, otherwise it repeats the
process. }%
The \rl is composed of at least one \re (and the first
iteration of the \rl is strictly speaking not a \re, but a try).


We denote by \cc the execution time of a \cas{} when the executing
thread does not own the cache line in exclusive mode, in a setting
where all threads share a last level cache. Typically, there
exists a thread that touches the data between two requests of the same
thread, therefore this cost is paid at every occurrence of a \cas.
As for the \rf{}s, \rc holds for the execution time of a cache
miss. When a thread executes a failed \cas, it immediately reads the
same cache line (at the beginning of the next \re), so the cache line
is not missing, and the execution time of the \rf is considered as
null. However, when the thread comes back from the \psw, a cache miss
is paid.
To conclude with the parameters related to the platform, we dispose of
\ct cores, where the \cas (resp. the \rf) latency is identical for all
cores, \ie \cc (resp. \rc) is constant.

The algorithm is parametrized by two execution times. In the general
case, the execution time of an occurrence of the \psw
(application-specific section) is a random variable that follows an
unknown probability distribution. In the same way, the execution time
of the \cww (specific to a \ds) can vary while following an unknown
probability distribution. The only provided information is the mean
value of those two execution times: \cw for the \cww, and \pw for the
\pww. These values will be given in units of work, where $\uow{1} =
\cycles{50}$.\vspp{-.4}


\subsubsection{Execution Description}
\label{sec:fra-exe}

It has been underlined in~\cite{EXCESS:D2.3} that there are two
main conflicts that degrade the performance of the \dss which do not
offer a great degree of disjoint-access parallelism: logical and
hardware conflicts.

{\it Logical conflicts} occur when there are more than one thread
in the retry loop at a given time (happens typically when the
number of threads is high or when the parallel section is small).
At any time, considering only the threads
that are in the \rl, there is indeed at most one thread whose \re will
be successful (\ie whose ending \cas will succeed), which implies the
execution of more \res for the failing threads.
%
In addition, after a thread executes successfully its final
\cas, the other threads of the \rl have first to finish their current
\re before starting a potentially successful \re, since they are not
informed yet that their current \re is doomed to failure. This creates
some ``holes'' in the execution where all threads are executing useless
work.

The threads will also experience {\it hardware conflicts}: if several
threads are requesting for the same data, so that they can operate a
\cas on it, a single thread will be satisfied. All the other threads
will have to wait until the current \cas is finished, and give a new
try when this \cas is done. While waiting for the ownership of the
cache line, the requesting threads cannot perform any useful
work. This waiting time is referred to as {\it expansion}.

\def\herec{.5}
\def\wirec{2}
\def\grey{black!20}
\def\greywh{black!4}

\def\maroon{blue!20!black!40!red!}
\def\green{black!20!green}

\newcommand{\pha}[5]{%
\node[it,text width=#4em,right= 0 of #2,fill=#5] (#1) {#3};
}

\newcommand{\supfig}{
\begin{center}
\begin{tikzpicture}[%
it/.style={%
    rectangle,
    text width=11em,
    text centered,
    minimum height=3em,
    draw=black!50
  }
]
\coordinate (O) at (0,0);
\pha{pcas}{O}{successful\\\cas}{\pr{4.5}{5}}{\green}
\pha{sla}{pcas}{useless\\work}{\pr{4}{8}}{\grey}
\pha{acc}{sla}{\acc}{\pr{3.5}{4}}{orange}
\pha{cri}{acc}{\cw}{\pr{2}{4}}{\maroon}
\pha{exp}{cri}{expansion}{\pr{4.5}{6}}{\grey}
\pha{fcas}{exp}{successful\\\cas}{\pr{4.5}{5}}{\green}
\draw [decorate,decoration={brace,mirror,amplitude=10pt}]
(sla.south west) -- (sla.south east)
node [black,midway,yshift=-20pt] {\watiw};
\draw [decorate,decoration={brace,mirror,amplitude=10pt}]
(acc.south west) -- (fcas.south east)
node [black,midway,yshift=-20pt] {\compw};
\draw [decorate,decoration={brace,amplitude=10pt}]
(sla.north west) -- (fcas.north east)
node [black,midway,yshift=20pt] (supw) {\supw};

\node (anot) at (.8,1) {{\scriptsize can be null}};

\draw[very thin] ($(anot.south east)!.8!(anot.east)$) -- ++(.4,0) -- ($(sla.north)+(0,-.1)$);
\draw[very thin] ($(anot.north east)!.8!(anot.east)$) -- ++(.4,0) -- ($(exp.north)+(0,-.1)$);


\end{tikzpicture}
\end{center}
}


\rr{\begin{figure}[t!]
\abstalgo
\caption{Thread procedure}\label{alg:gen-nb}
\end{figure}}

\rr{\begin{figure}[t!]
\supfig
\caption{Success Period}\label{fig:seq}
\end{figure}}

\pp{\begin{figure}[t!]
\centering
\begin{minipage}{.4\textwidth}
\abstalgo
\captionof{figure}{Thread procedure\label{alg:gen-nb}}
\end{minipage}\hfill%
\begin{minipage}{.6\textwidth}
\supfig
\captionof{figure}{\Supw\label{fig:seq}}
\end{minipage}
\end{figure}}

We now refine the description of the execution of the algorithm. The
timeline is initially decomposed into a sequence of \supws that will
define the throughput. A \supw is an interval of time of the
execution that
%
(i) starts after a successful \cas,
(ii) contains a single successful \cas,
(iii) finishes after this successful \cas.
\pr{To}{As explained in the previous subsection, to} be successful in its \re,
a thread has first to access the \ds, then modify it locally, and
finally execute a \cas, while no other thread performs changes on the
\ds. That is why each \supw is further cut into two main phases (see
Figure~\ref{fig:seq}). During the first phase, whose duration is
called the {\it \watiw}, no thread is accessing the \ds. The second
phase, characterized by the {\it \compw}, starts with the first access
to the \ds (by any thread). Note that this \acc could be either a \rf
(if the concerned thread just exited the \psw) or a failed \cas (if the thread
was already in the \rl).
The next successful \cas will come at least
after \cw (one thread has to traverse the \cww anyway), that is why we
split the latter phase into: \cw, then expansion, and finally a
successful \cas.


\subsubsection{Our Approaches}
\label{sec:fra-app}

In this work, we propose two different approaches to compute the
throughput of a lock-free algorithm, which we name as \avba and
constructive. The \avba approach relies on queuing theory and is
focused on the average behavior of the algorithm: the throughput is
obtained through the computation of the expectation of the \supw at a
random time.
%
%
As for the constructive approach, it describes precisely the instants
of accesses and modifications to the \ds in each \supw: in this way,
we are able to deconstruct and reconstruct the execution, according to
observed events. The constructive approach leads to a more accurate
prediction at the expense of requiring more information about the
algorithm: the distribution functions of the critical and \pwws have
indeed to be instantiated.


In both cases, we partition the domain space into different levels of
contention (or {\it modes}); these partitions are independent across
approaches, even if we expect similarities, but in each case, cover
the whole domain space (all values of \cww, \pww and number of
threads).
\medskip

\subsubsection{\Avba Analysis}
\label{sec:fra-app-asy}

We distinguish two main modes in which the algorithm can run:
contended and non-contended. In the non-contended mode, \ie when the
\pww is large or the number of threads is low,
concurrent operations are not likely to collide. So every \rl will
count a single \re, and atomic primitives will not delay each
other. In the contended mode, any operation is likely to experience
unsuccessful \res before succeeding (logical conflicts), and a \re
will last longer than in the non-contended mode because of the
collision of atomic primitives (hardware conflicts).

Once all the parameters are given, 
the analysis is centered around the calculation of a single variable
\atrl, which represents the expectation of the number of threads
inside the \rl at a random instant. Based on this variable, we are
able to express the expected expansion \avexp{\atrl} at a random
time. As a next step, we show how this expansion can be used to
estimate the expected \watiw \avwati{\atrl} and the expected \compw
\rwh{\atrl}, and at the end, the expected time of a \supw
\avsupe{\atrl}.

\subsubsection{Constructive Method}
\label{sec:fra-app-con}


The previous \avba reasoning is founded on expected values at a random time,
while in the constructive approach, we study each \supw individually,
based on the number of threads at the beginning of the considered
\supw. So we are able to exhibit more clearly the instants of
occurrences of the different accesses and modifications to the \ds,
and thus to predict the throughput more accurately.

We rely on the same set of values used in the \avba approach, but
these values are now associated with a given
\supw.
%
%
Thus the number of threads inside the \rl \trl, as well as the \watiw
and the \compw are evaluated at the beginning of each \supw.
We denote these times in the same way as in the first approach, but
remove the bar on top since these values are not expectations any
more.

The different contention modes do not characterize here the
steady-state of the \ds as in the previous approach but are
associated with the current \supw. Accordingly, the contention can
oscillate through different modes in the course of the execution.
First, a \supw is not
contended when $\trl=0$, \ie when there is no thread in the \rl after
a successful \cas.
In this case, the first thread that exits the \psw
will be successful, and the \acc of the sequence will be a \rf.
Second, the contention of a \supw is high when at any time during
the \supw, there exists a thread that is executing a \cas. In other
words, at the end of each \cas, there is at least one thread that is
waiting for the cache line to operate a \cas on it. This implies that
the first access of the \supw is a \cas and occurs immediately after
the preceding successful \cas: the \watiw is null.
Third, the medium contention mode takes place when $\trl>0$, while at the
same time, there are not enough requesting threads to fill the whole
\supw with \cass (which implies a non-null \watiw). Since these
requesting threads have synchronized in the previous \supw, \cass do
not collide in the current \supw, and because of that, the expansion
is null.

\subsection{Average-based Approach}
\label{sec:avba}

We propose in this section our coarse-grained analysis to predict the
performance of lock-free \dss.
Our approach utilizes
fundamental queuing theory techniques, describing the average
behavior of the algorithm.
%
In turn, we need only a minimal knowledge about the algorithm: the mean
execution time values \cw and \pw.
%
As explained in Section~\ref{sec:fra-app-asy}, the system runs in one of
the two possible modes: either contended or uncontended.


\subsubsection{Contended System}

We first consider a system that is contended.
When the system is contended, we use Little's law to obtain, at a
random time, the expectation of the \supw, which is the interval of
time between the last and the next successful \cass
(see Figure~\ref{fig:seq}).

The stable system that we observe is the \psw: threads are entering it
(after exiting a successful \rl) at an average rate, stay inside,
then leave (while entering a new \rl).
The average number of threads inside the \psw is $\atps = \ct - \atrl$,
each thread stays for an average duration of \pw, and in average, one
thread is exiting the \rl every \supw \avsupe{\atrl}, by definition of
the \supw.
According to Little's law~\cite{littles-law}, we have:\vspp{-.3}
\pr{
\begin{equation}
\label{eq:little-gen}
\atps = \pw \times 1 / \avsupe{\atrl}, \text{ \ie} \quad  \avsupe{\atrl} = \pw / (\ct - \atrl)
\end{equation}
}
{
\[ \atps = \pw \times \frac{1}{\avsupe{\atrl}}, \text{ \ie} \]
\begin{equation}
\label{eq:little-gen}
\frac{1}{\pw} \times \avsupe{\atrl} = \frac{1}{\ct - \atrl}
\end{equation}
}

\pr{We decompose a \supw into two parts: \watiw and \compw (as explained
in Section~\ref{sec:fra-exe}). 
We express the expectation of the \supw time as}%
{As explained in Section~\ref{sec:fra-exe}, we further decompose a \supw
into two parts, separated by the first access to the \ds after a
successful \cas. We can then write the average \supw as the sum of:
(i) the expected time before some thread starts its
  \acc (the \watiw), and
(ii) the expected \compw.
We compute these two expectations independently and gather them into
the \supw thanks to:}
\begin{equation}
\label{eq:little-sp}
\avsupe{\atrl} = \avwati{\atrl} + \rwh{\atrl}.
\end{equation}

When the \ds is contended, a thread is likely to be
successful after some failed \res. Therefore a thread that is
successful was  already in the \rl when the previous
successful \cas occurred.
\rr{This implies that the \acc to the \ds will
be due to a failed \cas, instead of a \rf.}%
The time before a thread starts its \acc is then the time before a
thread finishes its current \cww since there is a thread
currently executing a \cas.

\smallskip

\subsubsection{Expected \COmpw}
\label{sec:little-expa}

Since the \ds is contended, numerous threads are inside the \rl, and,
due to hardware conflicts, a \re can experience expansion:
the more threads inside the \rl, the longer time between a \cas
request and the actual execution of this \cas. The expectation of the
\compw can be written as:\vspp{-.2}
\begin{equation}
\label{eq:little-ret}
\rwh{\atrl} = \scas + \cw + \avexp{\atrl} + \scas ,
\end{equation}
where \avexp{\atrl} is the expectation of expansion when there are \atrl
threads inside the \rl, in expectation.
This expansion can be computed in the same way as
in~\cite{EXCESS:D2.3}, through the following differential equation:\vspp{.2}
\pr{\begin{minipage}{.4\textwidth}
\begin{equation*}
\left\{
\begin{array}{lcl}
\difavexp{\atrl} &=& \fcas \times \dfrac{\frac{\fcas}{2} + \avexp{\atrl}}{ \scas +\calrl + \scas + \avexp{\atrl}},\\
\avexp{1} &=& 0
\end{array} \right.
\end{equation*}\end{minipage}\hfill\begin{minipage}{.5\textwidth}
by assuming that the expansion starts as soon as strictly more than 1
thread are in the \rl, in expectation.\end{minipage}}{
\begin{equation*}
\left\{
\begin{array}{lcl}
\difavexp{\atrl} &=& \fcas \times \dfrac{\frac{\fcas}{2} + \avexp{\atrl}}{ \scas +\calrl + \scas + \avexp{\atrl}}\\
\avexp{1} &=& 0
\end{array} \right.,
\end{equation*}
by assuming that the expansion starts as soon as strictly more than 1
thread are in the \rl, in expectation.
}
%

\vspace*{.3cm}
\subsubsection{Expected \WAtiw}
\label{sec:litt-slack}

Concerning the \watiw, we consider that, at any time, the threads that
are running the \rl have the same probability to be anywhere in their
current \re. However, when a thread is currently executing a \cas, the
other threads cannot execute as well a \cas. The other threads are
thus in their \cww or expansion. For every thread, the time
before accessing the \ds is then uniformly distributed
between $0$ and $\cw+\avexp{\atrl}$.
\pr{Using a well-known formula on the expectation of the minimum of
  uniformly distributed random variables, we show
  in Appendix~\ref{app:lemsl} that\vspace*{-.4cm} 
}
{

According to Lemma~\ref{lem:unif-min}, we conclude that
}
\begin{equation}
\label{eq:slack-cont}
\pp{\hspace{5cm}}\avwati{\atrl} = \left( \cw + \avexp{\atrl} \right) / (\atrl +1).
\end{equation}

\rr{\lemsl}

\vspace*{.1cm}
\subsubsection{Expected \SUpw}

We just have to combine Equations~\ref{eq:little-sp},
\ref{eq:little-ret}, and~\ref{eq:slack-cont} to obtain the general
expression of the expected \supw: \pr{$\avsupe{\atrl} = \left( 1 + 1/(\atrl +1) \right) \left( \cw + \avexp{\atrl} \right) + 2\scas$,}%
{\[\avsupe{\atrl} = \left( 1 + \frac{1}{\atrl +1} \right) \left( \cw + \avexp{\atrl} \right) + 2\scas, \]}
which leads, according to Equation~\ref{eq:little-gen}, to\vspp{-.2}
\begin{equation}
\label{eq:little-co}
\frac{1}{\pw} \times \left(
\frac{\atrl +2}{\atrl +1}  \left( \cw + \avexp{\atrl} \right)
+ 2\scas
\right) = \frac{1}{\ct - \atrl}.
\end{equation}

\subsubsection{Non-contended System}

When the system is not contended, logical conflicts are not likely to
happen, hence each thread succeeds in its \rl at its first {\it
  \re}. \Afort, no hardware conflict occurs. Each thread still
performs one success every \wl, and the \supw is given by
\pr{%
$\avsupe{\atrl} = (\pw + \mem+\cw+\scas)/\ct$.
Moreover, a thread spends average \pw units of time
in the \rl within each \wl. As this holds for
every thread, we deduce
$\ct - \atrl = \atps = \pw/(\pw + \mem+\cw+\scas) \times \ct$.
Combining the two previous equations, we obtain
\begin{equation}
\label{eq:little-nc}
\frac{\avsupe{\atrl}}{\pw}  = \frac{1}{\ct - \atrl},
\text{ where } \avsupe{\atrl} = \frac{\mem+\cw+\scas}{\atrl}.
\end{equation}
}%
{
\begin{equation}
\label{eq:little-avsp}
\avsupe{\atrl} = \frac{\pw + \mem+\cw+\scas}{\ct}.
\end{equation}

Moreover, a thread spends in average $\mem+\cw+\scas$ units of time
in the \rl within each \wl. As this holds for
every thread, we can obtain the following expression for the total
average number of threads inside the \rl:
\begin{equation}
\label{eq:little-trl}
\atrl  = \frac{\mem+\cw+\scas}{\pw + \mem+\cw+\scas}  \times \ct = \frac{\mem+\cw+\scas}{\avsupe{\atrl}}
\end{equation}

Equation~\ref{eq:little-avsp} also gives $ \mem+\cw+\scas = \ct \times
\avsupe{\atrl} - \pw$, hence, thanks to Equation~\ref{eq:little-trl},
\begin{equation}
\label{eq:little-nc}
\atrl = \frac{\ct \times \avsupe{\atrl}-\pw}{\avsupe{\atrl}}, \text{ \ie} \quad
\frac{\avsupe{\atrl}}{\pw}  = \frac{1}{\ct - \atrl},
\end{equation}
where $\avsupe{\atrl} = \frac{\mem+\cw+\scas}{\atrl}$.}

\subsubsection{Unified Solving}

\pp{Based on the fact that Equations~\ref{eq:little-co}
  and~\ref{eq:little-nc} are equivalent on a single point, we can
  exhibit an expression of the \supw that is valid on the whole domain
  (see Appendix~\ref{app:switch}). In turn
  Theorem~\ref{th:fixed-point}, proved in
  Appendix~\ref{app:fixed-point}, explains how to compute the
  throughput estimate.}

\rr{
\proofswitch

We show in the following theorem how to compute the throughput estimate; the proof
manipulates equations in order to be
able to use the fixed-point Knaster-Tarski theorem.
} 

\begin{theorem}
\label{th:fixed-point}
The throughput can be obtained iteratively through a fixed-point
search, as $\thru = \left( \avsupe{\lim_{n \rightarrow \pinf} u_n} \right) ^{-1}$, where\vspp{-.2}
\[ \left\{\begin{array}{ll}
u_0 = \frac{\mem + \cw + \scas}{\pw + \mem + \cw + \scas} \ct &\\
u_{n+1} = \frac{u_n \avsupe{u_n}}{\pw + u_n \avsupe{u_n}} \times \ct & \quad \text{for all } n \geq 0.
\end{array}\right.
\]
\end{theorem}
\rr{ 
\begin{proof}
\prooffp
\end{proof}
} 

\subsection{Constructive Approach}
\label{sec:cons}


In this section, we instantiate the probability distribution of the
\pww with an exponential distribution. We have therefore a
better knowledge of the behavior of the algorithm, particularly in
medium contention cases, which allows us to follow a fine-grained
approach that studies individually each successful operation together
with every \cas occurrence. We provide an elegant and efficient
solution that relies on a Markov chain analysis.

\vspace{-.3cm}
\subsubsection{Process}
\label{sec:mark-proc}

%


We have seen in Section~\ref{sec:fra-app-con} that the \supw can run
in one of the three modes: no contention, medium contention or high
contention.
\pr{We }{The main idea is to }start from a configuration with a given
number of threads \trl \rr{just }after a successful \cas, and to describe
what will happen until the next successful \cas: what will be the mode
of the next \supw, and\rr{ even} more precisely, which will be the
number of threads at the beginning of the next \supw.

As a basis, we consider the execution that would occur without any
other thread exiting the \psw (then entering the \rl); we call this
execution the {\it internal execution}. This execution follows the
\supw pattern described in Figure~\ref{fig:seq} (with an infinite
\watiw if the system is not contended).
On top of this basic \supw, we inject the threads that can exit the
\psw, which has a double impact. On the one hand, they increase the
number of threads inside the \rl for the next \supw. On the other
hand, if the first thread that exits the \psw starts its \re during the
\watiw of the \supw of the internal execution, then this thread will
succeed its \acc, which is a \rf, and will shrink the actual \watiw
of the current \supw.

According to the distribution probability of the arrival of the new
threads, we can compute the probability for the next \supw to
start with any number of threads. The expression of this stochastic
sequence of \supws in terms of Markov chains results in the throughput
{estimate}.

\vspp{-.4}
\subsubsection{Expansion}
\label{sec:mark-expa}

The expansion, as before, represents the additional time in the
execution time of a \re, due to the serialization of atomic
primitives. However, in contrary to Section~\ref{sec:little-expa}, we
compute here this additional time in the current \supw, according to
the number of threads \trl inside the \rl at the beginning of the
\supw.
The expansion only appears when the \supw is highly contended, \ie
when we can find a continuous sequence of \cass all through the
\supw. We assume that for the rest of the section.

The expansion is highly correlated with the way the cache coherence
protocol handles the exchange of cache lines between threads. We rely
on the experiments of the research report associated
with~\cite{ali-same}, which show that if several threads request for
the same cache line in order to operate a \cas, while another thread
is currently executing a \cas, they all have an equal probability to
obtain the cache line when the current \cas is over.

\def\herec{.5}
\def\wirec{1}
\def\wicw{2.5}
\newcommand{\dcasf}[2]{
\draw[pattern=north west lines, draw=none] (0,-#2*\herec) rectangle (#1*\wirec,-#2*\herec+\herec);
\draw[fill=red] (#1*\wirec,-#2*\herec) rectangle ++(\wirec,\herec) node[midway, align=center] {\cas};
}
\newcommand{\dcass}[2]{
\draw[fill=\green] (#1*\wirec,-#2*\herec) rectangle ++(\wirec,\herec) node[midway, align=center] {\cas};
}
\newcommand{\dcw}[2]{
\draw[fill=\maroon] (#1*\wirec,-#2*\herec) rectangle ++(\wicw,\herec) node[midway, align=center] {\cw};
}
\newcommand{\dpw}[3]{
\draw[fill=\grey] (#1*\wirec,-#2*\herec) rectangle ++(#3*\wirec,\herec) node[midway, align=center] {\pw};
}
\newcommand{\dwt}[3]{
\draw[densely dashed] (#1*\wirec,-#2*\herec+.5*\herec) -- ++(#3*\wirec,0);
}

\def\marup{1}
\def\decx{.2}
\def\decy{.35}

\newcommand\arcod[2]{
\draw[draw=blue,very thick,fill=blue] #1 -- ++ (0,-#2) -- ++(\decx,\decy) -- ++(-\decx,0);
}
\newcommand\arcou[2]{
\draw[draw=blue, very thick,fill=blue] #1 --  ++ (0,#2) -- ++(\decx,-\decy) -- ++(-\decx,0);}

\pp{
\begin{figure}
\begin{center}
\hspace*{-1cm}\begin{minipage}{.5\textwidth}
\begin{center}
\begin{tikzpicture}[scale=.97]
\clip (-1*\wirec,1*\herec) rectangle (8*\wirec,-9.5*\herec);

\dcass{0}{0}\dpw{1}{0}{8}
\dcasf{1}{3}\dcw{2}{3}
\dcasf{2}{2}\dcw{3}{2}
\dcasf{3}{5}\dcw{4}{5}
\dcasf{4}{4}\dcw{5}{4}
\dcasf{5}{1}\dcw{6}{1}
\dcasf{6}{6}\dcw{7}{6}

\dwt{2+\wicw}{3}{3.5}
\dwt{3+\wicw}{2}{1.5}\dcass{3+1.5+\wicw}{2}
\dwt{4+\wicw}{5}{1.5}
\dwt{5+\wicw}{4}{.5}

\draw[blue,densely dotted, thick] (5*\wirec,0) -- ++(0,-6.4*\herec) -- ++(-.5*\wirec,0) -- ++(0,-.5*\herec) ;
\draw[blue,densely dotted, thick] (6*\wirec,0) -- ++(0,-6.9*\herec);
\draw[blue,densely dotted, thick] (7*\wirec,0) -- ++(0,-6.4*\herec) -- ++(.5*\wirec,0) -- ++(0,-.5*\herec) ;
\draw[blue,densely dotted, thick] (1*\wirec,0) -- ++(0,-6.9*\herec);

\node[text width=1cm, align=center] at (4.5*\wirec,-8*\herec) { {\color{red} $\trl-4$}\\ {\color{black} vs}\\ {\color{green}1}};
\node[text width=1cm, align=center] at (6*\wirec,-8*\herec) { {\color{red} $\trl-5$}\\ {\color{black} vs}\\ {\color{green}2}};
\node[text width=1cm, align=center] at (7.5*\wirec,-8*\herec) { {\color{red} $\trl-6$}\\ {\color{black} vs}\\ {\color{green}3}};

\node[text width = 3cm, align=center] at (1.2*\wirec,-7.8*\herec) {\trl threads inside\\ the \rl};

\node[draw=black,rounded corners=4,scale=.8] at (-.5*\wirec,0.5*\herec ) {\tiny Thread 1};
\node[draw=black,rounded corners=4,scale=.8] at (-.5*\wirec,-0.5*\herec) {\tiny Thread 2};
\node[draw=black,rounded corners=4,scale=.8] at (-.5*\wirec,-1.5*\herec) {\tiny Thread 3};
\node[draw=black,rounded corners=4,scale=.8] at (-.5*\wirec,-2.5*\herec) {\tiny Thread 4};
\node[draw=black,rounded corners=4,scale=.8] at (-.5*\wirec,-3.5*\herec) {\tiny Thread 5};
\node[draw=black,rounded corners=4,scale=.8] at (-.5*\wirec,-4.5*\herec) {\tiny Thread 6};
\node[draw=black,rounded corners=4,scale=.8] at (-.5*\wirec,-5.5*\herec) {\tiny Thread 7};

\end{tikzpicture}
\end{center}
\captionof{figure}{Highly-contended execution\label{fig:mark-expa}}
\end{minipage}%
\hspace*{1cm}\begin{minipage}{.48\textwidth}
\begin{center}
\begin{tikzpicture}[%
it/.style={%
    rectangle,
    text width=11em,
    text centered,
    minimum height=\pr{2}{3}em,
    draw=black!50,
    scale=.9,
  }
]
\coordinate (O) at (0,0);
\pha{pcas}{O}{\cas}{\pr{2}{5}}{\green}
\pha{sla}{pcas}{\wati{i}}{\pr{5}{8}}{\grey}
\pha{acc}{sla}{\cas}{\pr{2}{4}}{red}
\pha{cri}{acc}{\cw}{\pr{3}{4}}{\maroon}
\pha{exp}{cri}{\reexp{i}}{\pr{4}{6}}{\grey}
\pha{fcas}{exp}{\cas}{\pr{2}{5}}{\green}
\coordinate (intsl) at ($(sla.north west)+(0,.4*\marup)$); 
\coordinate (intsr) at ($(sla.north east)+(0,.4*\marup)$);
\coordinate (inte) at ($(fcas.north east)+(0,.4*\marup)$); 
\draw [decorate,decoration={brace,amplitude=10pt}]
(intsl)  --  (intsr)
node [black,midway,yshift=15pt] {\scriptsize 0 new thread};

\draw [decorate,decoration={brace,amplitude=10pt}]
(intsr) -- (inte)
node [black,midway,yshift=15pt] {\scriptsize $k+1$ new threads};
\arcod{($(acc.north west)!.2! (acc.north east) + (0,.5*\marup)$)}{.4*\marup}
\arcod{($(exp.north west)!.1! (exp.north east) + (0,.5*\marup)$)}{.4*\marup}
\arcod{($(exp.north west)!.9! (exp.north east) + (0,.5*\marup)$)}{.4*\marup}

\coordinate (extsl) at ($(sla.south west)+(0,-.4*\marup)$);
\coordinate (extsr) at ($(sla.south east)+(0,-.4*\marup)$);
\draw [decorate,decoration={brace,mirror,amplitude=10pt},text width=11em, align=center] (extsl) -- (extsr)
node (caca) [black,midway,yshift=-13pt,xshift=26pt] {\scriptsize at least 1 new thread};

\coordinate (Ob) at ($(sla.south west)!.2!(sla.south) + (0,-1.5*\marup)$);
\pha{accb}{Ob}{\rf}{\pr{2}{4}}{yellow}
\pha{crib}{accb}{\cw}{\pr{3}{4}}{\maroon}
\pha{expb}{crib}{\reexp{i}}{\pr{4}{6}}{\grey}
\pha{fcasb}{expb}{\cas}{\pr{2}{5}}{\green}
\arcou{(accb.south west)}{1.65*\marup}
\arcou{($(crib.south west)!.2! (crib.south east) + (0,-.6*\marup)$)}{.4*\marup}
\arcou{($(expb.south west)!.9! (expb.south east) + (0,-.6*\marup)$)}{.4*\marup}
\draw[very thick, draw=blue] (accb.south west) -- (fcasb.south east) -- (fcasb.north east) -- (accb.north west);
\coordinate (extnl) at ($(accb.south west)+(0,-.5*\marup)$); 
\coordinate (extnr) at ($(fcasb.south east)+(0,-.5*\marup)$);
\draw [decorate,decoration={brace,mirror,amplitude=10pt},text width=11em, align=center]
(extnl) -- (extnr)
node [black,midway,yshift=-15pt] {\scriptsize $k$ new threads};

\draw[dotted, draw=black] (intsl) -- (extsl);
\draw[dotted, draw=black] (intsr) -- (extsr);
\draw[dotted, draw=black] (inte) -- (fcas.south east);
\draw[dotted, draw=black] (accb.north west) -- (extnl);
\draw[dotted, draw=black] (fcasb.north east) -- (extnr);
\node[text width=1.2cm, align=center] (inttext) at ($(pcas.west) + (-1*\marup,0)$) {Internal\\ execution};
\node[anchor=east] (eint) at ($(inttext.east) + (0.2,1.2*\marup)$) {\eve{int}};
\node[anchor=east] (eext) at ($(inttext.east) + (0.2,-2*\marup)$){\eve{ext}};

\path[->,out=90,in=-180] ($(inttext.north west)!.3!(inttext.north)$) edge (eint);
\path[->,out=-90,in=-180] ($(inttext.south west)!.3!(inttext.south)$) edge (eext);
\end{tikzpicture}
\end{center}
\captionof{figure}{Possible executions\label{fig:ex-eint-eext}}
\end{minipage}
\end{center}
\end{figure}
}

\rr{
\begin{figure}
\begin{center}
\begin{tikzpicture}[scale=1.6]
\clip (-1.2*\wirec,1*\herec) rectangle (8*\wirec,-9*\herec);

\dcass{0}{0}\dpw{1}{0}{8}
\dcasf{1}{3}\dcw{2}{3}
\dcasf{2}{2}\dcw{3}{2}
\dcasf{3}{5}\dcw{4}{5}
\dcasf{4}{4}\dcw{5}{4}
\dcasf{5}{1}\dcw{6}{1}
\dcasf{6}{6}\dcw{7}{6}

\dwt{2+\wicw}{3}{3.5}
\dwt{3+\wicw}{2}{1.5}\dcass{3+1.5+\wicw}{2}
\dwt{4+\wicw}{5}{1.5}
\dwt{5+\wicw}{4}{.5}

\draw[blue,densely dotted, thick] (5*\wirec,0) -- ++(0,-6.4*\herec) -- ++(-.5*\wirec,0) -- ++(0,-.5*\herec) ;
\draw[blue,densely dotted, thick] (6*\wirec,0) -- ++(0,-6.9*\herec);
\draw[blue,densely dotted, thick] (7*\wirec,0) -- ++(0,-6.4*\herec) -- ++(.5*\wirec,0) -- ++(0,-.5*\herec) ;
\draw[blue,densely dotted, thick] (1*\wirec,0) -- ++(0,-6.9*\herec);

\node[text width=2cm, align=center] at (4.5*\wirec,-8*\herec) { {\color{red} $\trl-4$}\\ {\color{black} vs}\\ {\color{green}1}};
\node[text width=2cm, align=center] at (6*\wirec,-8*\herec) { {\color{red} $\trl-5$}\\ {\color{black} vs}\\ {\color{green}2}};
\node[text width=2cm, align=center] at (7.5*\wirec,-8*\herec) { {\color{red} $\trl-6$}\\ {\color{black} vs}\\ {\color{green}3}};

\node[text width = 3cm, align=center] at (1.2*\wirec,-7.8*\herec) {\trl threads inside\\ the \rl};

\node[draw=black,rounded corners=4,scale=.8] at (-.6*\wirec,0.5*\herec ) {\small Thread 1};
\node[draw=black,rounded corners=4,scale=.8] at (-.6*\wirec,-0.5*\herec) {\small Thread 2};
\node[draw=black,rounded corners=4,scale=.8] at (-.6*\wirec,-1.5*\herec) {\small Thread 3};
\node[draw=black,rounded corners=4,scale=.8] at (-.6*\wirec,-2.5*\herec) {\small Thread 4};
\node[draw=black,rounded corners=4,scale=.8] at (-.6*\wirec,-3.5*\herec) {\small Thread 5};
\node[draw=black,rounded corners=4,scale=.8] at (-.6*\wirec,-4.5*\herec) {\small Thread 6};
\node[draw=black,rounded corners=4,scale=.8] at (-.6*\wirec,-5.5*\herec) {\small Thread 7};

\end{tikzpicture}
\end{center}
\caption{Highly-contended execution\label{fig:mark-expa}}
\end{figure}
}

We draw an illustrative example in Figure~\ref{fig:mark-expa}. The
green \cass are successful while the red \cass fail. To lighten
the picture, we hide what happened for the threads before they
experience a failed \cas. The horizontal dash lines represent the time
where a thread wants to access the data in order to operate a \cas
but has to wait because another thread owns the data in exclusive
mode.
We can observe in this example that the first thread that accesses
the \ds is not the thread whose operation returns.

We are given that \trl threads are inside the \rl at the end of the
previous successful \cas, and we only consider those threads. When
such a thread executes a \cas for the first time, this \cas is
unsuccessful. The thread was in the \rl when the successful \cas has
been executed, so it has read a value that is not up-to-date
anymore. However, this failed \cas will bring the current version of
the value (to compare-and-swap) to the thread, a value that will be
up-to-date until a successful \cas occurs.

So we have firstly a sequence of failed \cass until the first thread
that operated its \cas within the current \supw finishes its \cww. At
this point, there exists a thread that is executing a \cas. When this
\cas is finished, some threads compete to obtain the cache line. We
have two bags of competing threads: in the first bag, the thread that
just ended its \cww is alone, while in the second bag, there are all
the threads that were in the \rl at the beginning of the \supw, and
did not operate a \cas yet. The other, non-competing, threads are
running their \cww and do not yet want to access the data.

As described before, every thread has the same probability to become the next
owner of the cache line. If a thread from the first bag is drawn, then
the \cas will be successful and the \supw ends. Otherwise, the \cas is
a failure, and we iterate at the end of this failed \cas. However, the
thread that just failed its \cas is now executing its \cww,
and does not request for a new \cas until this work has been done,
thus it is not anymore in the second bag. In addition, the thread that
had executed its \cas after the thread of the first bag is now back
from its \cww and falls into the first bag. The process iterates until
a thread is drawn from the first bag.

As a remark, note that we do not consider threads that are not in the
\rl at the beginning of the \supw since even if they come back from
the \psw during the \supw, their \rf will be delayed and their \cas is
likely to occur after the end of the \supw.


Theorem~\ref{th:mark-expa} \pp{, proved in Appendix~\ref{app:aliexp}, }
gives the explicit formula for the expansion\pr{.}{, based on the previous
explanations.}

\newcommand{\prsu}[1]{\ema{p_{#1}}}

\begin{theorem}
\label{th:mark-expa}
The expected time between the end of the \cww of the first
thread that operates a \cas in the \supw and the beginning of a
successful \cas is given by:\vspp{0}
\[ \reexp{\trl} = \lceil\cw/\scas\rceil\scas - \cw +
\sum_{i=1}^{\tcom} \frac{i(i-1)}{\left(\tcom\right)^i} \frac{(\tcom-1)!}{(\tcom-i)!} \times \scas, \pp{\quad\text{where }\tcom = \trl - \lceil\cw/\scas\rceil +1.}\]
\rr{where $\tcom = \trl - \lceil\cw/\scas\rceil +1$.}
\end{theorem}
\rr{ 
\begin{proof}
\proofaliexp
\end{proof}
}  

\subsubsection{Formalization}

The \pww follows an exponential distribution, whose mean is
\pw. More precisely, if a thread starts a \psw at the instant $t_1$,
the probability distribution of the execution time of the \psw is
\pr{$ t \mapsto \lambda \expu{-\lambda (t-t_1)} \indi{[t_1,\pinf[}{t}$, where $\lambda = 1/\pw. $}
{\[ t \mapsto \lambda \expu{-\lambda (t-t_1)} \indi{[t_1,\pinf[}{t}, \text{ where } \lambda = \frac{1}{\pw}. \]}
This probability distribution is memoryless, which implies that the
threads that are executing their \psw cannot be differentiated: at a
given instant, the probability distribution of the remaining execution
time is the same for all threads in the \psw, regardless of when the \psw began. For all
threads, it is defined by:
\pr{ $ t \mapsto \lambda \expi{-\lambda t}$, where $\lambda = 1/\pw. $}
{\[ t \mapsto \lambda \expu{-\lambda t}, \text{ where } \lambda = \frac{1}{\pw}. \]}

For the behavior in the \rl, we rely on the same approximation as in
the previous section, \ie when a successful thread exits its
\rl, the remaining execution time of the \re of every other thread
that is still in the \rl is uniformly distributed between $0$ and the
execution time of a whole \re. We have seen that the expectation of
this remaining time is the size of the execution time of a \re divided
by the number of threads inside the \rl plus one. Here, we assume that
a thread will start a \re at this time.
This implies another kind of memoryless property: the behavior of a
thread that is in the \rl does not depend on the moment that it
entered its \rl.

To tackle the problem of estimating the throughput of such a system,
we use an approach based on Markov chains. We study the behavior of
the system over time, step by step: a state of the Markov chain
represents the state of the system when the current \supw began (\ie
just after a successful \cas) and (thus) the system changes state at
the end of every successful \cas.
According to the current state, we are able to compute the probability
to reach any other state at the beginning of the next \supw.
In addition, the two memoryless properties render the description of a
state easy to achieve: the number of threads inside the \rl when the
current success begins, indeed fully characterizes the system.

We recall that \trl is the number of threads inside the \rl when the
\supw begins.
The Markov chain is strongly connected with \trl, since it composed of
\ct states $\sta{0}, \sta{1}, \dots, \sta{\ct-1}$, where, for all $i
\in \inte{\ct-1}$, the \supw is in state \sta{i} iff $\trl=i$. For all
$(i,j) \in \inte{\ct-1}^2$, $\pro{\sta{i} \rightarrow \sta{j}}$
denotes the probability that a success characterized by \sta{j}
follows a success in state \sta{i}.
$\wati{\sta{i} \rightarrow \sta{j}}$ denotes the
\watiw that passed while the system has gone from state \sta{i} to
state \sta{j}.
This \watiw can be expressed based on the \watiw \wati{i} of the
internal execution, \ie the execution that involves only the $i$
threads of the \rl and ignores the other threads (see
Section~\ref{sec:mark-proc}). \rr{Recall that we consider that the \watiw of
the internal execution with $0$ thread is infinite, since no thread will
access the \ds.}
In the same way, we denote by \rw{i} the \compw of the internal
execution, hence $\rw{i} = \cc + \cw + \reexp{i} + \cc$.

\newcommand{\intgen}[1]{\ema{\mathcal{I}_{\mathrm{#1}}}}
\newcommand{\intnoc}{\intgen{noc}}
\newcommand{\intmid}{\intgen{mid}}
\newcommand{\inthig}{\intgen{hi}}

\newcommand{\ihig}{\ema{i_{\mathrm{hi}}}}

We have seen that the level of contention (mode) is determined by
\trl, hence the interval \inte{\ct-1} can be partitioned into
\pr{$\inte{\ct-1} = \intnoc \cup \intmid \cup \inthig$,}
{\[ \inte{\ct-1} = \intnoc \cup \intmid \cup \inthig, \]}
where the partitions correspond to the different contention levels.
So, by definition, $\intnoc = \{ 0 \}$, and for all $i \in \intnoc
\cup \intmid$, $\reexp{i} = 0$ (see Section~\ref{sec:fra-app-con}).

The \supw is highly-contended, \ie we have a continuous sequence of
\cass in the \supw, if the sum of the execution time of all the
\cass that need to be operated exceeds the \cww. Hence
$\inthig = \inte[\ihig]{\ct-1}$, where
\pr{$\ihig = \min \{ i \in \inte[1]{\ct-1} \; | \; i \times \scas > \cw  \}$.}
{\[ \ihig = \min \{ i \in \inte[1]{\ct-1} \; | \; i \times \scas > \cw  \}. \]}
In addition, as the sequence of \cass is continuous when the
contention is high, the \watiw is null when the \supw is highly
contended, \ie, for all $i \in \inthig$, $\wati{i} = 0$, and \afort,
$\wati{\sta{i} \rightarrow \sta{\star}} = 0$.

Otherwise, the \supw is in medium contention, hence $\intmid =
\inte[1]{\ihig-1}$. Moreover, if $i \in \intmid$, $\wati{i} > 0$, and
$\reexp{i} = 0$, because the \cass synchronized during the previous
\supw and will not collide any more in the current \supw.





\pp{Everything is now in place to be able to obtain the stationary
  distribution of the Markov chain, and in turn to compute the
  throughput and the failure rate estimates. The reasoning that leads
  to the computation of the probability of going from state \sta{i} to
  state \sta{i+k} can be roughly summarized by
  Figure~\ref{fig:ex-eint-eext}, where we start from an internal
  execution with $i$ threads inside the \rl and the blue arrows
  represent the threads that exit the \psw. Two non-overlapping events
  can then potentially occur: either (event \eve{ext}) the first thread exiting the \psw
  starts within $[0,\wati{i}[$, \ie in the \watiw of
      the internal execution, or (event \eve{int}) the first thread entering the \rl
      starts after $t=\wati{i}$. The details can be
      found in Appendix~\ref{app:markov}.
}


\rr{\wholemarkov}

\subsection{Experiments}
\label{sec:xp}

To validate our analysis results, we use two main types of lock-free
algorithms.  In the first place, we consider a set of algorithms that follow
the pattern in~\ref{alg:gen-name}. This set of algorithms includes: (i) synthetic
designs, that cover the design space of possible lock-free data structures;
(ii) several fundamental designs of \ds operations such as lock-free
stacks~\cite{lf-stack} (\popop, \pushop),
queues~\cite{lf-queue-michael} (\deqop), counters~\cite{count-moir}
(\incop, \decop).
As a second step, we consider more advanced lock-free operations that
involve helping mechanisms, and show how to use our analysis in this
context.  Finally, in order to highlight the benefits of the analysis
framework, we show how it can be applied to i) determine a beneficial back-off
strategy and ii) optimize the memory management scheme used by a \ds,
in the context of an application.


We also give insights about the strengths of our two approaches.
\pr{The }{On the one hand, the }constructive approach exhibits better predictions
due to the tight estimation of the failing retries. On the other hand, the
\avba approach is applicable to a broader spectrum of algorithmic designs
as it leaves room to abstract complicated algoritmic designs, which do
not follow the pattern of~\ref{alg:gen-name}.

%
%

\vspp{-.2}
\subsubsection{Setting}

We have conducted experiments on an Intel ccNUMA workstation
system. The system is composed of two sockets equipped with
Intel Xeon E5-2687W v2 CPUs\pp{.}\rr{ with frequency band
  \ghz{1.2\text{-}3.4.} The physical cores have private L1, L2 caches
  and they share an L3 cache, which is \megb{25}.}  In a socket, the
ring interconnect provides L3 cache accesses and core-to-core
communication. \rr{Due to the bi-directionality of the ring
  interconnect, uncontended latencies for intra-socket communication
  between cores do not show significant variability.}\pp{Threads are
  pinned to a single socket to minimize non-uniformity in \rf and \cas
  latencies. }%
\pp{The methodology in~\cite{david-emp-atom} is used to measure the
  \cas and \rf latencies, while the \pww is implemented
  by a for-loop of {\it Pause} instructions. }%
\rr{Our model assumes uniformity in the \cas and \rf latencies on the
  shared cache line. Thus, threads are pinned to a single socket to
  minimize non-uniformity in \rf and \cas latencies.  In the
  experiments, we vary the number of threads between 4 and 8 since the
  maximum number of threads that can be used in the experiments are
  bounded by the number of physical cores that reside in one socket. }%
We show the experimental results with 8 threads.

In all figures, the y-axis shows both the throughput\rr{ values}, \ie
number of operations completed per second, and the ratio of failing to
successful retries (multiplied by $10^6$, for readability),
while the mean of the exponentially distributed \pww \pw is
represented on the x-axis. 
%
The number of failures per success in the \avba
approach is computed as $\atrl-1$ and \pr{in the constructive approach
by stochastically counting the failing \cass inside a \supw
(see Appendix~\ref{sec:nbf}).}{ is described in Section~\ref{sec:nbf}
for the constructive approach.}


\rr{\figsynthcst}

We have also added a straightforward upper bound as a baseline
approach,\rr{ which is} defined as the minimum of $1/(\rc+\cw+\cc)$ (two
successful \res cannot overlap) and $\ctot/(\pw+\rc+\cw+\cc)$ (a
thread can succeed only once in each \wl).

\vspp{-.2}
\subsubsection{Basic Data Structures}

\pr{Firstly, we consider lock-free algorithms that strictly follow the
  pattern in~\ref{alg:gen-name}. We provide predictions, on the one
  hand, on a set of synthetic tests that have been constructed to
  abstract different possible design patterns of lock-free data
  structures (value of \cw) and different application contexts (value
  of \pw), and, on the other hand, on the well-known Treiber's
  stack. The results are depicted in Appendix~\ref{app:xp-basic}.}
{
\synthtreib


\figsynthpoi
\figsynthconst

\subsubsection{Synthetic Tests}

\synth

\subsubsection{Treiber's Stack}

\figtreib

\treib

}

%


\subsection{Towards Advanced Data Structure Designs}

Advanced lock-free operations generally require multiple pointer
updates that cannot be done with a single \cas. 
One way to design such operations, in a lock-free manner, is to
use helping mechanisms: an inconsistency will be fixed eventually by
some thread.
Here we consider two \dss that apply immediate helping, the queue
from~\cite{lf-queue-michael} and the deque designed in~\cite{deq}. In the queue
experiment (Figure~\ref{fig:enqueue}), we run the \enqop operation on the
queue with and without memory management; in the \deq experiment, each
thread is dedicated to an end of the \deq (equally distributed), while
we vary the proportion of push operations (colors in Figure~\ref{fig:deq}).

\pp{\figsidebyside{h!}{.5}{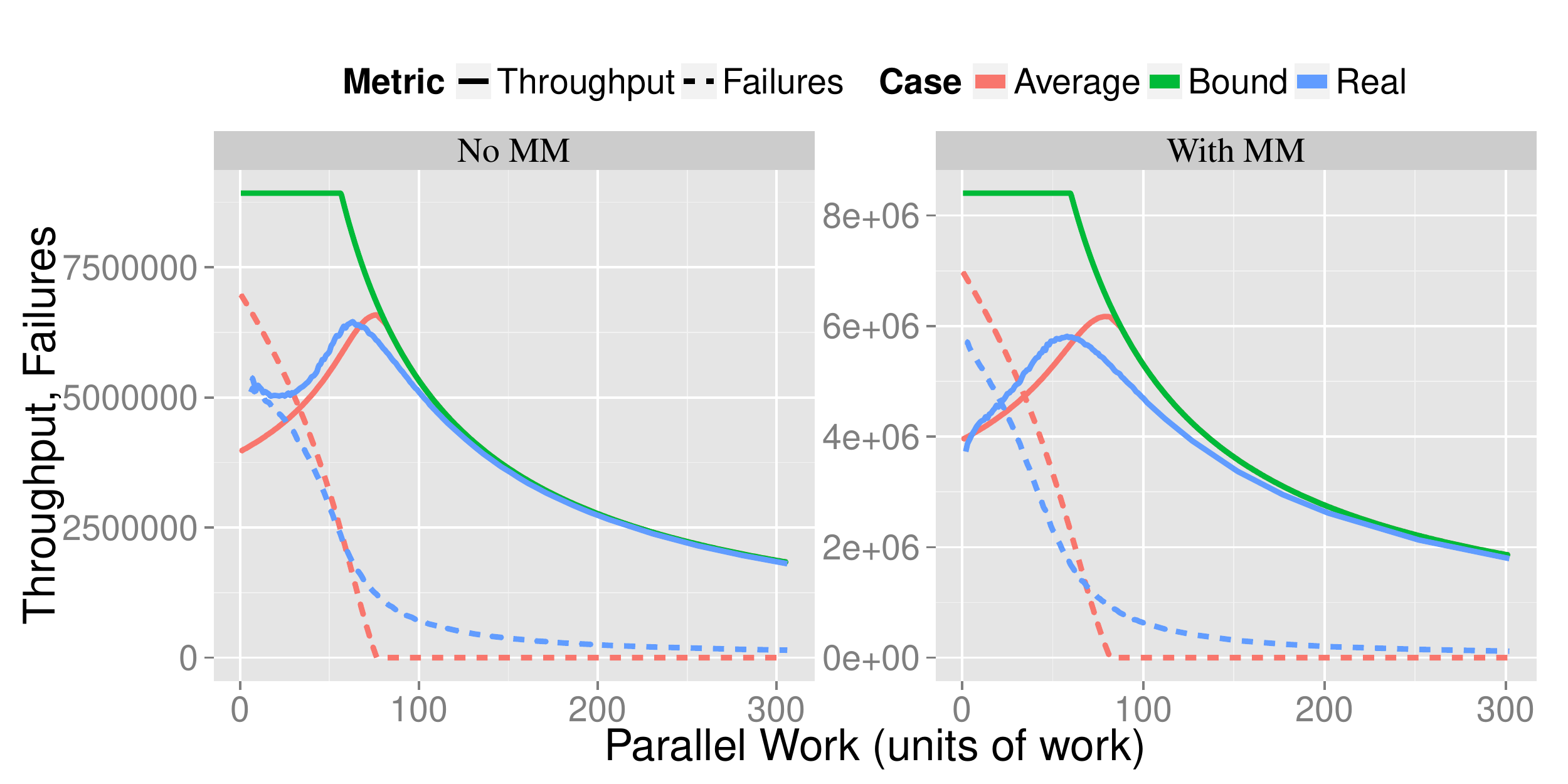}{Enqueue on MS queue}{enqueue}{.5}{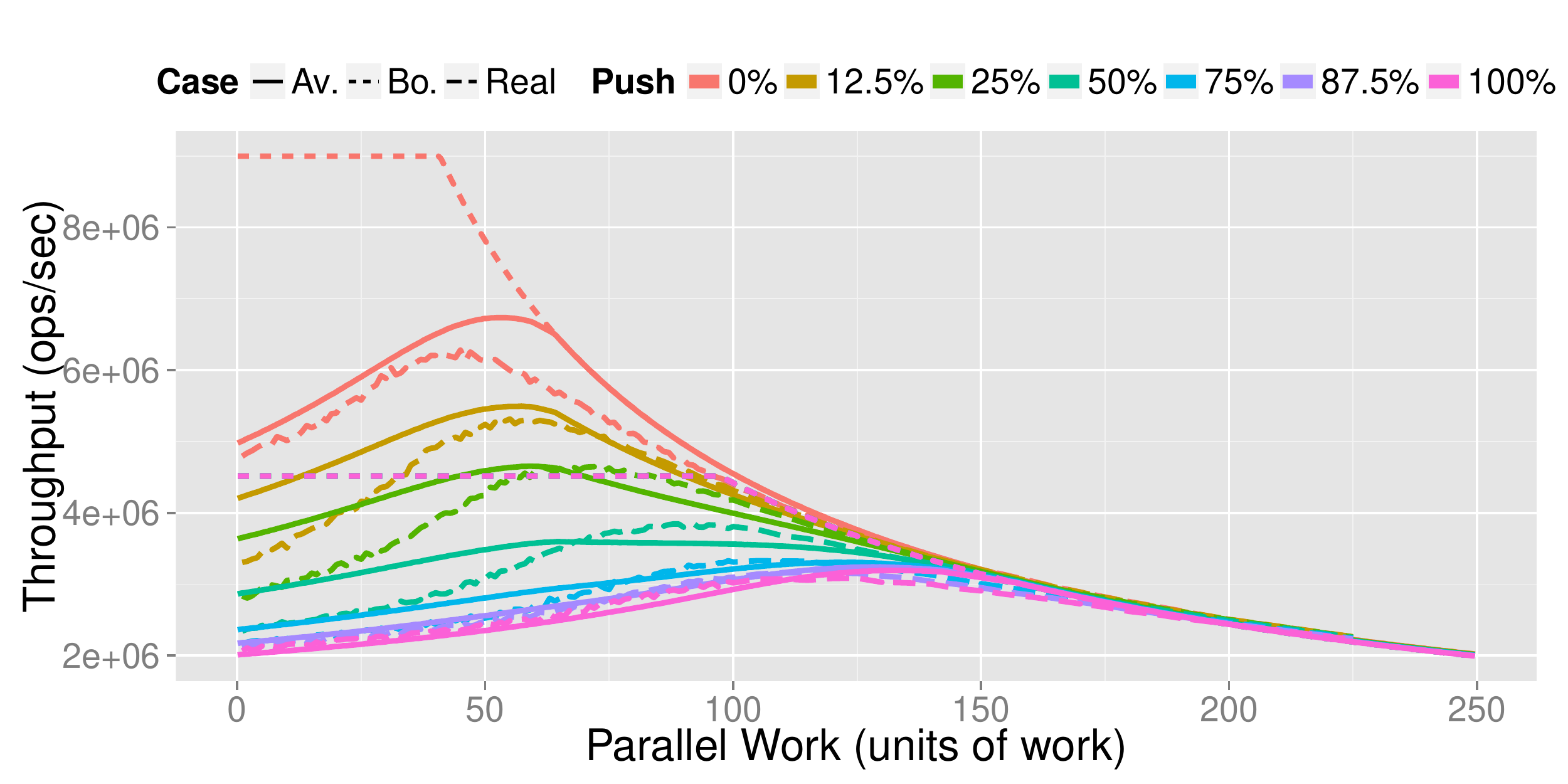}{Operations on \deq}{deq}}

\rr{
\sumads

\subsubsection{Expected Expansion for the Advanced Data Structures}
\fulladsexp

\subsubsection{Expected Slack Time for the Advanced Data Structures}
\fulladswati

\subsubsection{Enqueue on Michael-Scott Queue}
\fullenq

\subsubsection{\Deq}
\label{sec:xp-deq}
\fulldeq}

\subsection{Applications}

\pp{\figsidebyside{b!}{.5}{back-offs-new}{Performance impact of our back-off tunings}{bos}{.5}{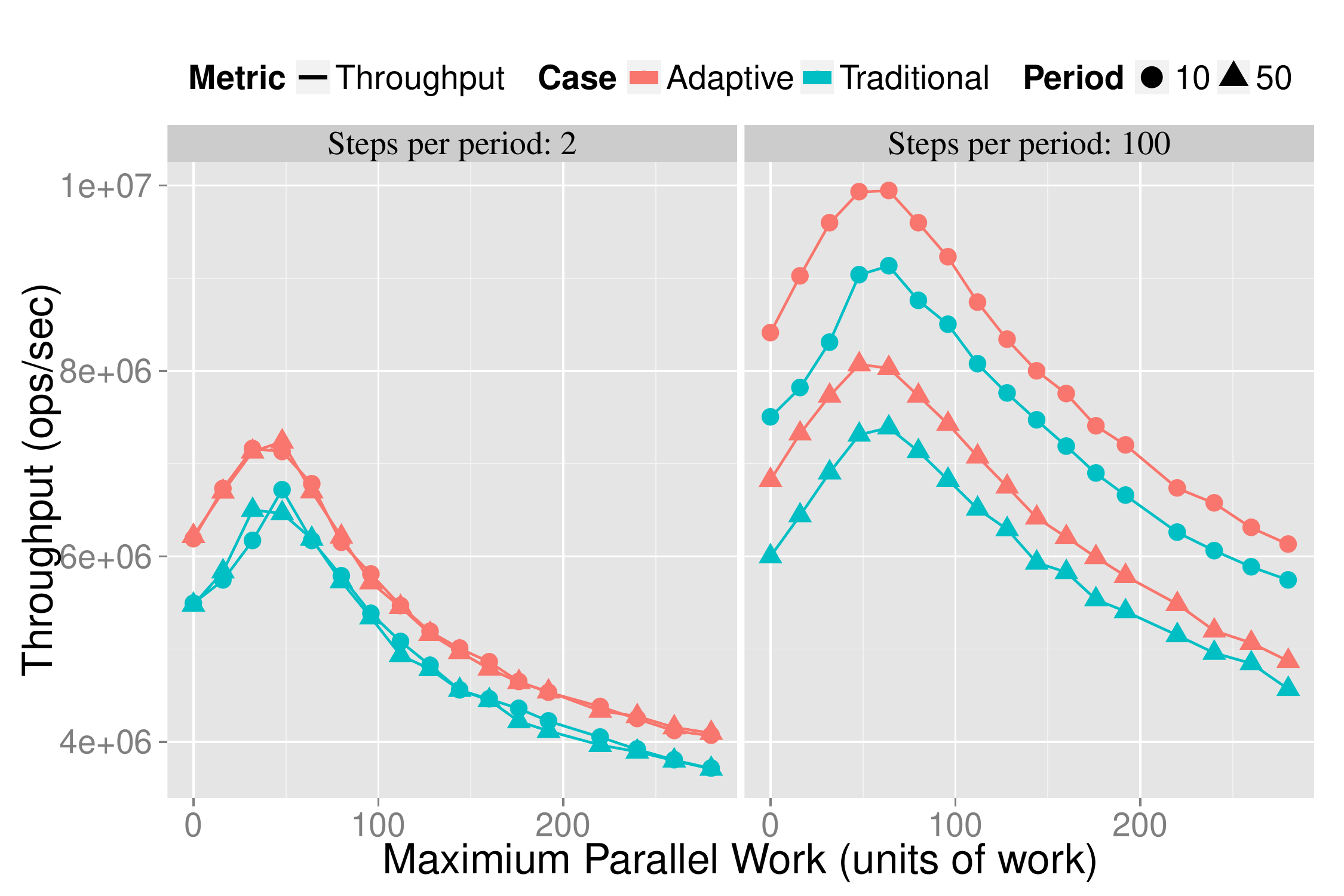}{Adaptive MM with varying mean \pw}{mm_adaptive}}

\subsubsection{Back-off Optimizations}

When the \pww is known, we can deduce from our analysis a simple and
efficient back-off strategy: as we are able to estimate the value for
which the throughput is maximum, we just have to back-off for the time
difference between the peak \pw and the actual \pw.
\pr{In Appendix~\ref{app:bo}, we compare this back-off strategy
against widely known strategies, namely exponential and linear, on a
synthetic workload. }%
{\fullbosy}%
In Figure~\ref{fig:bos}, we apply our constant
back-off on a Delaunay triangulation application~\cite{caspar},
provided with several workloads. The application uses a stack in two
phases, whose first phase pushes elements on top of the stack without
delay. We are able to estimate a corresponding back-off time, and we
plot the results by normalizing the execution time of our back-offed
implementation with the execution time of the initial implementation.


\rr{\begin{figure}[t!]
\begin{center}
\includegraphics[width=.9\textwidth]{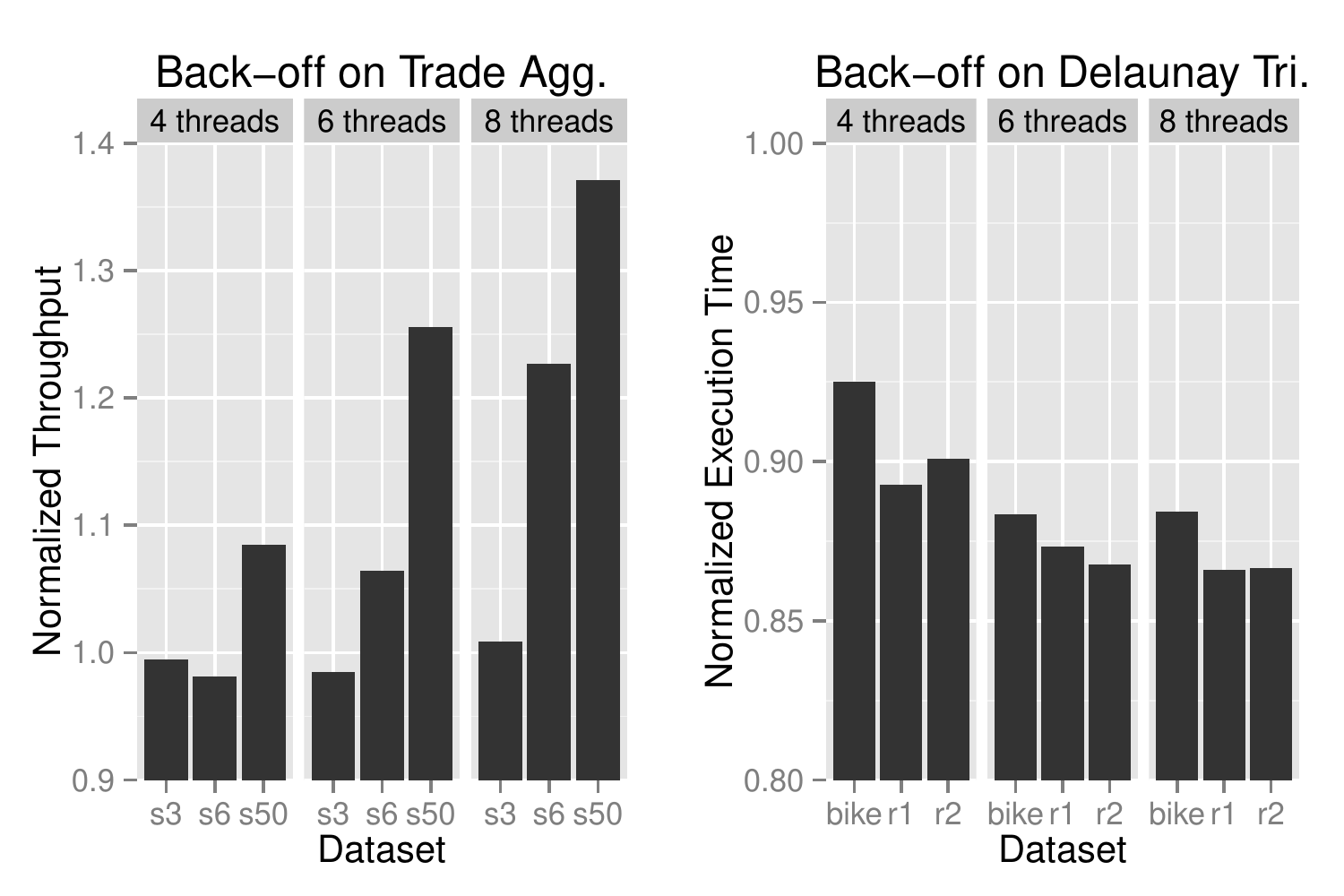}
\end{center}
\caption{Performance impact of our back-off tunings \label{fig:bos}}
\end{figure}}

A measure or an estimate of \pw is not always available (and could
change over time, see next section), therefore we propose also an
adaptive strategy: we incorporate in the \ds a monitoring routine that tracks
the number of failed \res, employing a sliding window. As our analysis
computes an estimate of the number of failed \res as a function of
\pw, we are able to estimate the current \pw, and hence the
corresponding back-off time like previously.

We test our adaptive back-off mechanism on a workload originated
from~\cite{taq-se}, where global operators of exchanges for financial
markets gather data of trades with a microsecond accuracy. We assume
that the data comes from several streams, each
of them being associated with a thread. All threads enqueue the
elements that they receive in a concurrent queue, so that they can be
later aggregated. We extract from the original data a trade stream
distribution that we use to generate similar streams that reach the
same thread; varying the number of streams to the same thread leads to
different workloads. The results, represented as the normalized
throughput (compared to the initial throughput) of trades that are enqueued when the
adaptive back-off is used, are plotted in Figure~\ref{fig:bos}.
For any number of threads, the queue is not
contended on workload s3, hence our improvement is either small or
slightly negative. On the contrary, the workload s50 contends the
queue and we achieve very significant improvement.


\rr{\begin{figure}[h!]
\begin{center}
\includegraphics[width=.7\textwidth]{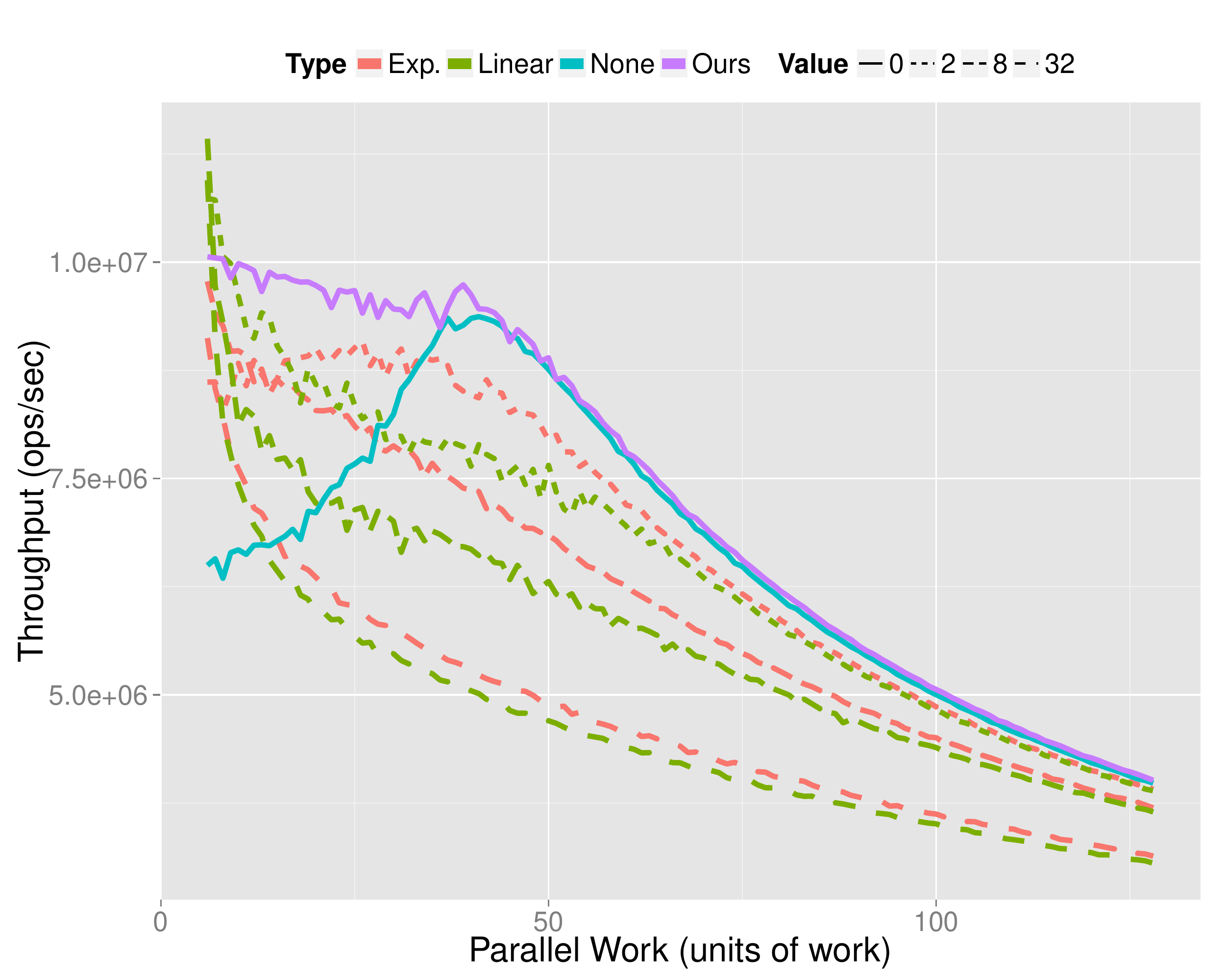}
\end{center}
\caption{Back-off Tuning on Treiber's Stack\label{fig:bo-synt}}
\end{figure}}


\subsubsection{Memory Management Optimization}

%
%
%
%

Memory Management (MM) is an inseparable part of dynamic concurrent
\dss. In contrary to lock-based implementations, a node that has been
{\it removed} from a lock-free \ds can still be accessed by other threads,
\eg if they have been delayed. Collective decisions are thus required
in order to {\it reclaim} a node in a safe manner.
A well-known solution to deal with this problem is the hazard pointers
technique~\cite{Mic04b}. \pp{In an implementation of such design each
  thread lists the nodes that it accesses and the nodes that it has
  removed. When the number of nodes it has removed reaches a threshold,
  it reclaims its listed removed nodes that are not listed as
  accessed by any thread.}



\newcommand\nodhp[1]{\ema{\mathcal{N}_{#1}}}
\newcommand\delhp[1]{\ema{\mathcal{D}_{#1}}}

\rr{A traditional design to implement this technique works as follows.
Each thread \thr{i}, maintains two lists of nodes: \nodhp{i} contains
the nodes that \thr{i} is currently accessing, and \delhp{i} stores
the nodes that have been removed from the \ds by \thr{i}. Once a
threshold on the size of \delhp{i} is reached, \thr{i} calls a routine
that: (i) collects the nodes that are accessed by any other thread, \ie
\nodhp{j} for $j \neq i$ (collection phase), and (ii) for each element
in \delhp{i}, checks whether someone is accessing the element, \ie
whether it belongs to $\cup_{j \neq i} \nodhp{j}$, and if not,
reclaims it (reclamation phase).}

The main goal of our adaptive MM scheme is to distribute this
extra-work in a way that the loss in performance is largely leveraged,
knowing that additional work can be an advantage under high-contention
(see previous section).
The optimization is based on two main
modifications.
\pr{First, we divide the reclamation phase of the traditional MM scheme
into quanta (equally-sized chunks), whose finer granularity allows for
accurate back-off times. Second, we track continuously the contention
level in the same way as our adaptive back-off. See Appendix~\ref{app:mm}.}%
{First, the granularity has to be finer, since the
additional quantum that the back-off mechanism uses, has to
be rather small (hundreds of cycles for a queue). Second, we need to
track the contention level on the \ds in order to be able to inject the work
at a proper execution point.}

\rr{\figcompmm}

\rr{\begin{figure}[t!]
\begin{center}
\includegraphics[width=.9\textwidth]{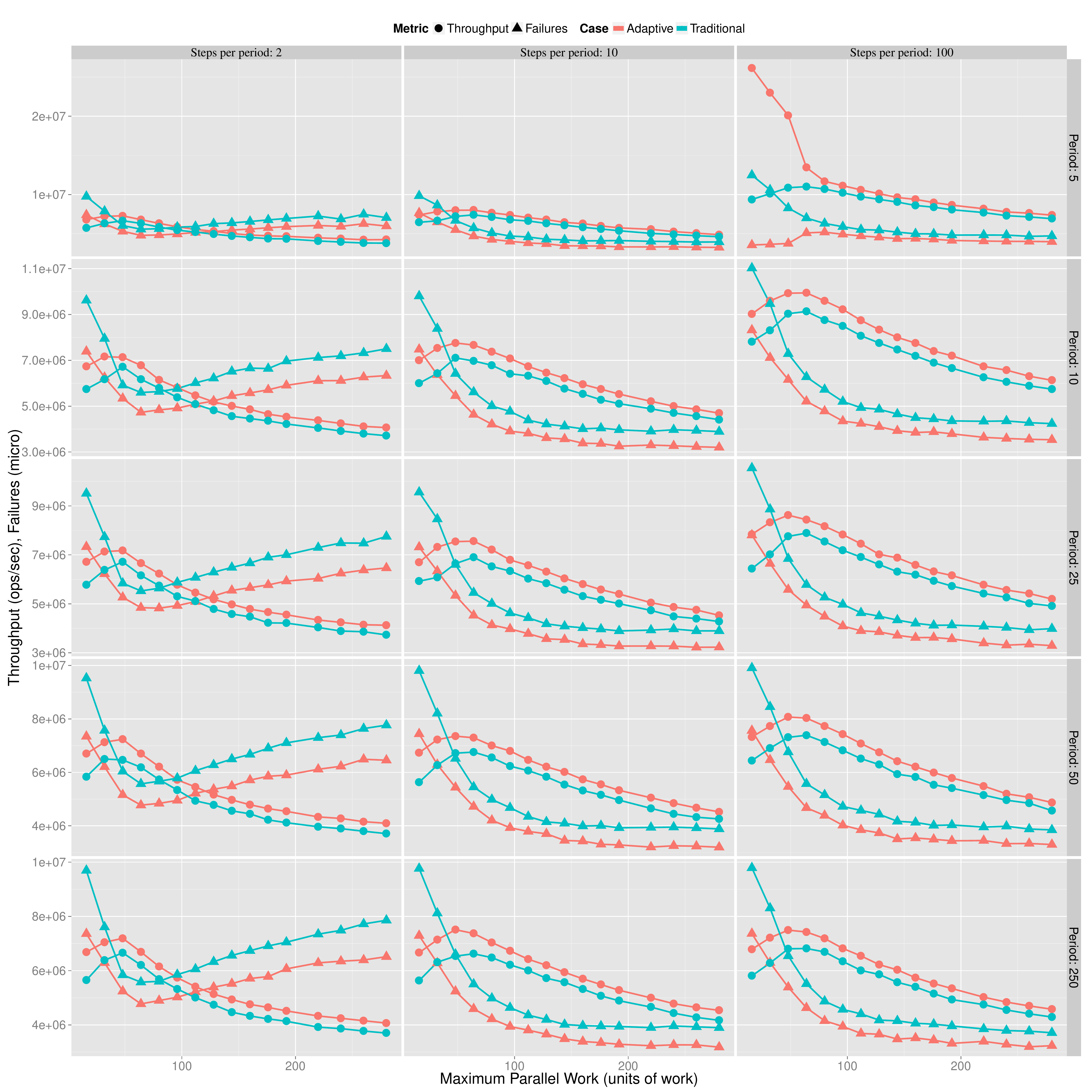}
\end{center}
\caption{Adaptive MM with varying mean \pw\label{fig:mm_adaptive}}
\end{figure}}

\rr{\bothmm}




\pr{%
%
We emulate the behavior of many scientific applications, that are
  built upon a pattern of alternating phases, that are
  communication-intensive (synchronization phase) or
  computation-intensive. Here we assume a synchronization ensured
  through a shared \ds, hence the communication-intensive phases
  correspond to a high access rate to the \ds, while \ds is accessed
  at a low rate during a computation-intensive phase. The \pww still
  follows an exponential distribution of mean \pw, but \pw varies in a
  sinusoidal manner with time.
To study also the impact of the continuity of the change in \pw, \pw
is set as a step approximation of a sine function. Thus, two
additional parameters rule the experiment: the period of the
oscillating function represents the length of the phases, and the
number of steps within a period depicts how continuous are the phase
changes.
}%
{\adaptsine}

In Figure~\ref{fig:mm_adaptive}, we compare our approach with the
traditional implementation for different periods of the sine function,
on the \deqop of the Michael-Scott queue~\cite{lf-queue-michael}.
The adaptive MM, that relies on the analysis presented
in this work, outperforms the traditional MM
because it provides an advantage both under low contention due to the
costless (since delayed) invocation of the MM and under high
contention due to the back-off effect.

\pp{\newpage\appendix

  \setlength\abovedisplayskip{10pt}
  \setlength\belowdisplayskip{10pt}
  \setlength\abovedisplayshortskip{10pt}
  \setlength\belowdisplayshortskip{10pt}

\subsubsection{Average-based Approach}
\label{app:avba}

\subsubsection{Helping Lemma}
\label{app:lemsl}
\lemsl

\subsubsection{Switch between non-contended and contended}
\label{app:switch}

\proofswitch

\subsubsection{Proof of Theorem~\ref{th:fixed-point}}
\label{app:fixed-point}

\prooffp

\subsection{Constructive Approach}
\label{app:markov}

\subsubsection{Proof of Theorem~\ref{th:mark-expa}}
\label{app:aliexp}
\proofaliexp

\wholemarkov
\figsynthcst

\subsection{Experiments}

\subsubsection{Basic Data Structures}
\label{app:xp-basic}
\synthtreib


\figsynthpoi
\figsynthconst

\paragraph{Synthetic Tests}

\synth

\paragraph{Treiber's Stack}

\figtreib

\treib

\subsubsection{Towards Advanced Data Structure Designs}
\label{app:ads}

\sumads

\paragraph{Expected Expansion for the Advanced Data Structures}
\fulladsexp

\paragraph{Expected Slack Time for the Advanced Data Structures}
\fulladswati

\paragraph{Enqueue on Michael-Scott Queue}
\fullenq

\paragraph{\Deq}
\label{sec:xp-deq}
\fulldeq

\subsubsection{Applications}

\paragraph{Back-off Optimizations}
\label{app:bo}

\begin{figure}[h!]
\begin{center}
\includegraphics[width=.6\textwidth]{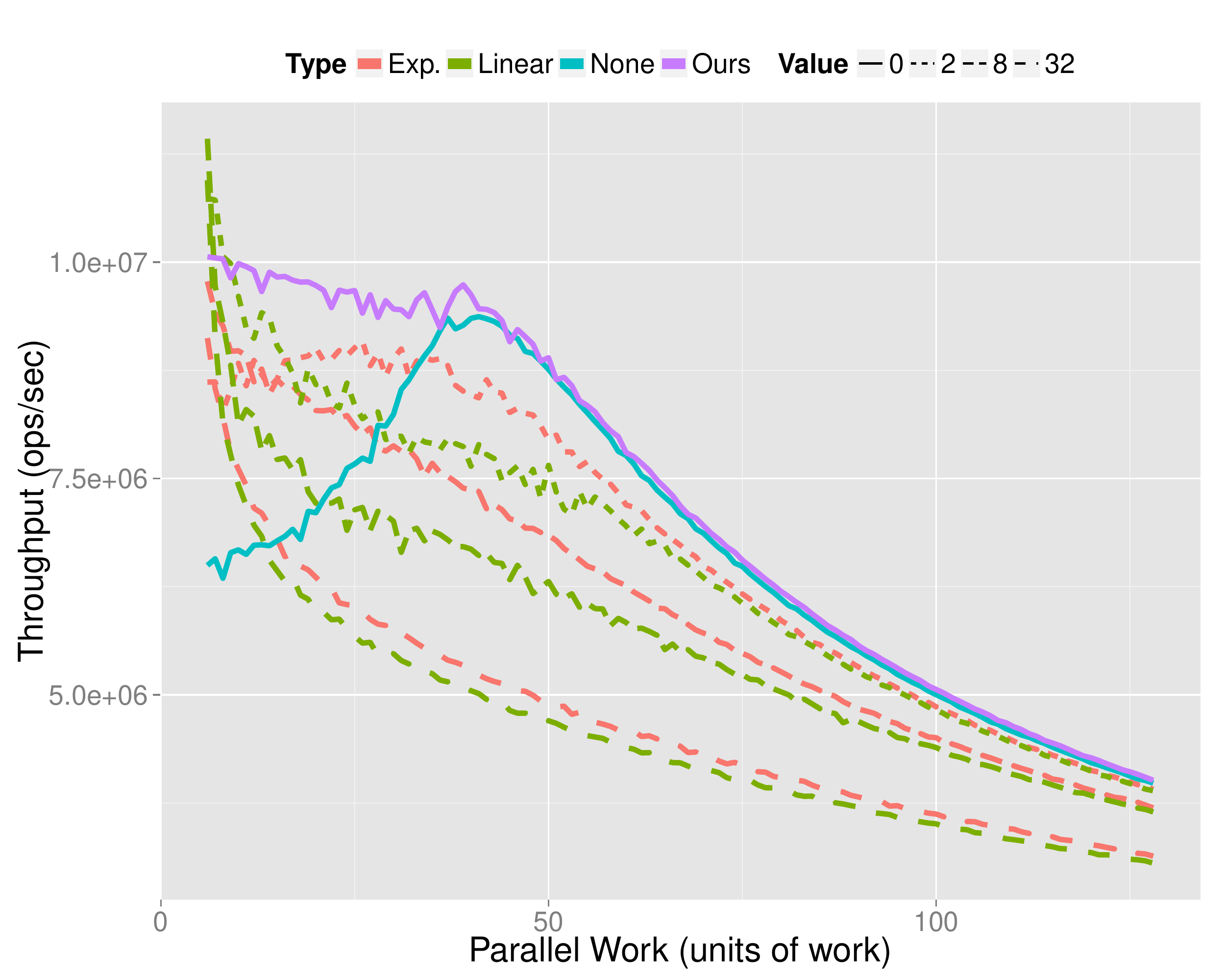}
\end{center}
\caption{Back-off Tuning on Treiber's Stack\label{fig:bo-synt}}
\end{figure}

\fullbosy

\figcompmm

\paragraph{Memory Management Optimization}
\label{app:mm}

\begin{figure}[b!]
\begin{center}
\includegraphics[width=.7\textwidth]{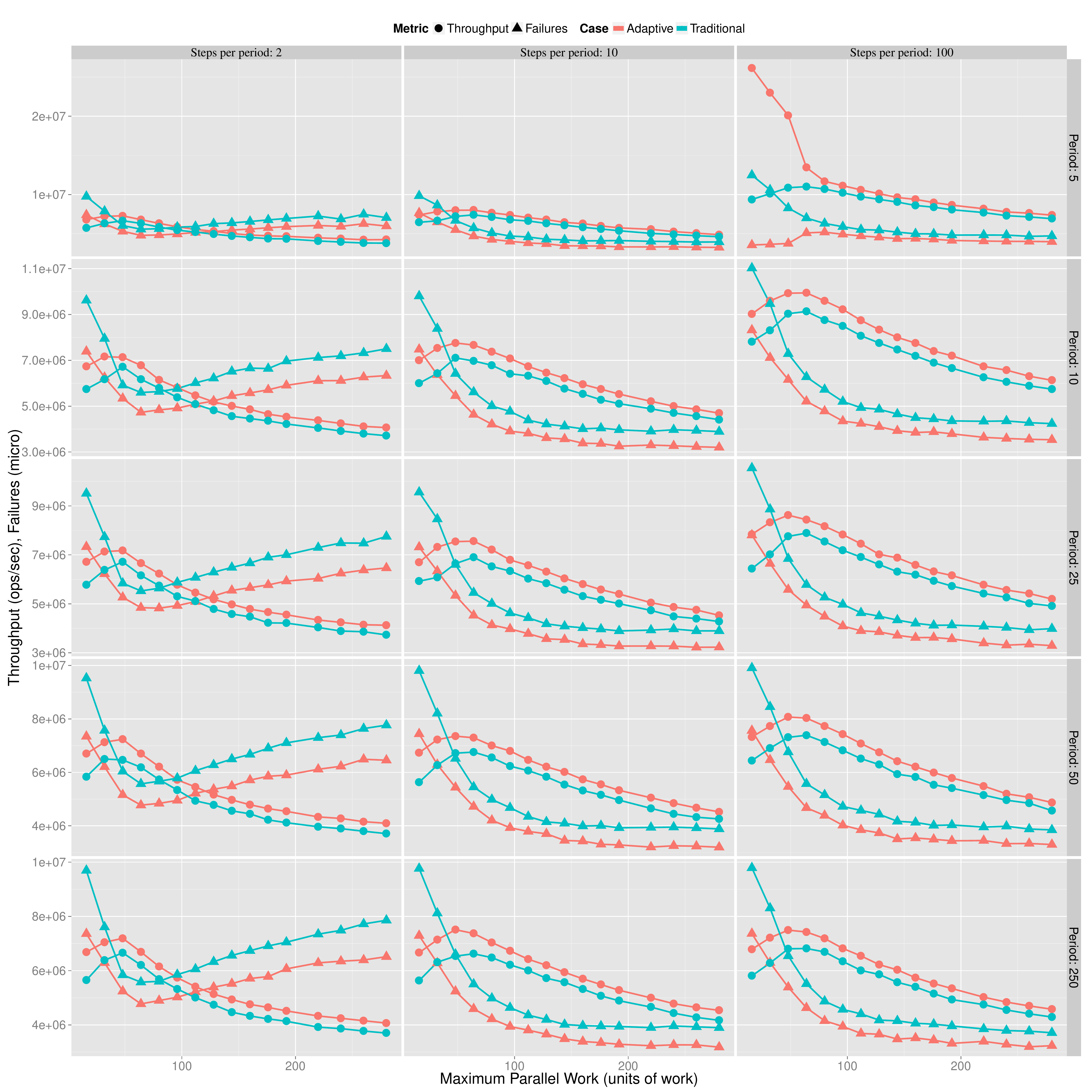}
\end{center}
\caption{Extensive experiments on adaptive MM\label{fig:ad-ext-mm}}
\end{figure}

~\\\bothmm

We show here extended experiments on Memory Management
optimization. Results are depicted in Figure~\ref{fig:ad-ext-mm}.


}


\subsection{Energy Modelling and Empirical Evaluation}
\label{sec:markov-energy}

We introduced our power model and the power impacting factors in 
D2.1~\cite{EXCESS:D2.1}. Then, we combined it with our initial performance 
model in D2.3~\cite{EXCESS:D2.3} to obtain the average power consumption in the 
static parallel programs that uses the fundamental lock-free data structures 
(\ie the size of the parallel work that is executed in between data structure 
operations is constant). 

Here, we take one step further and aim to obtain the energy efficiency of a 
wider range of lock-free data structure implementations that are used in the 
dynamic environments (\ie the parallel work that is executed in between data 
structure operations follows a probability distribution). The performance analysis,
that is presented above, can be used to predict the 
performance of such data structures in such environments. For the energy consumption
estimations, we apply the 
methodology that was provided in D2.3, where we combine the power model 
presented in D2.1 with the performance estimations to obtain the average power 
consumption estimations.

\begin{figure}[t!]
\begin{center}
\includegraphics[width=\textwidth]{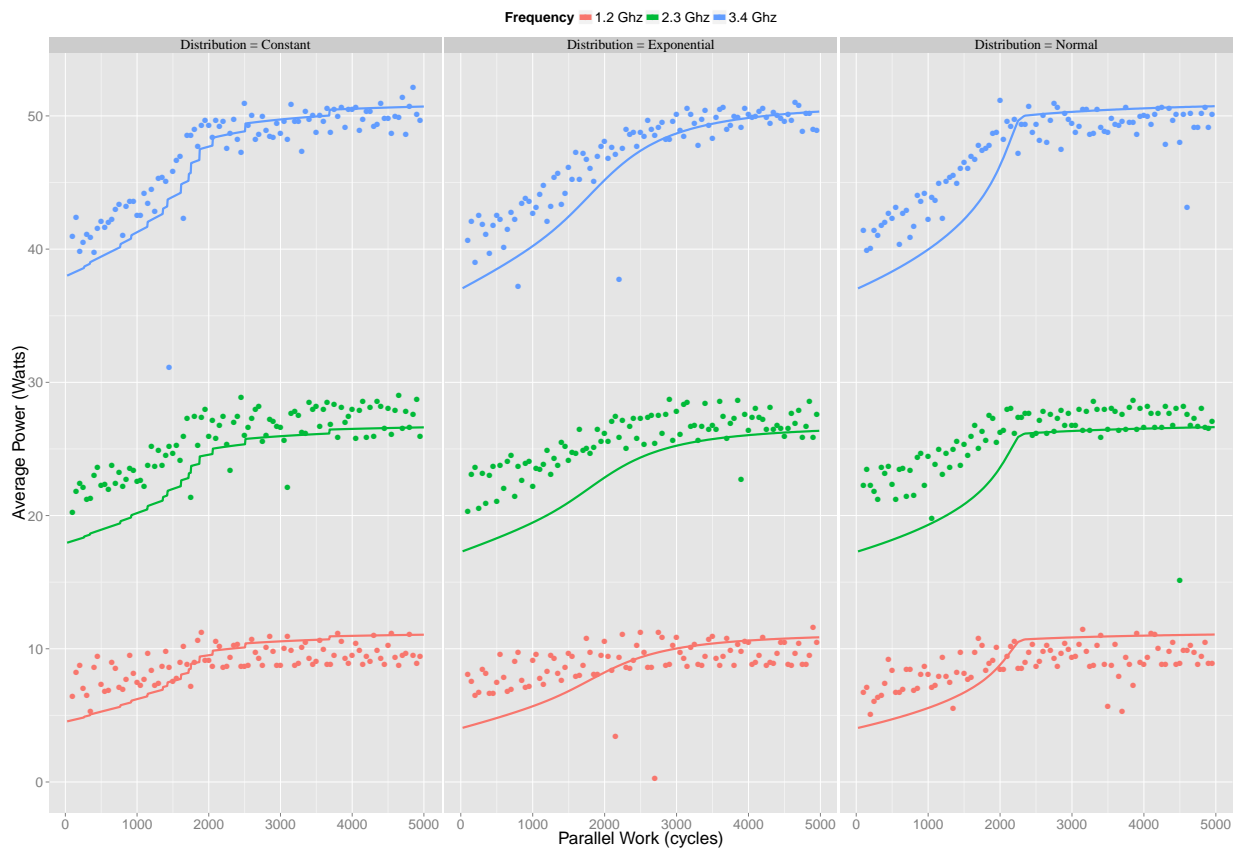}
\end{center}
\caption{Average Power Consumption for Treiber's Stack (Pop operation)\label{fig:en_stack}}
\end{figure}

\begin{figure}[h!]
\begin{center}
\includegraphics[width=\textwidth]{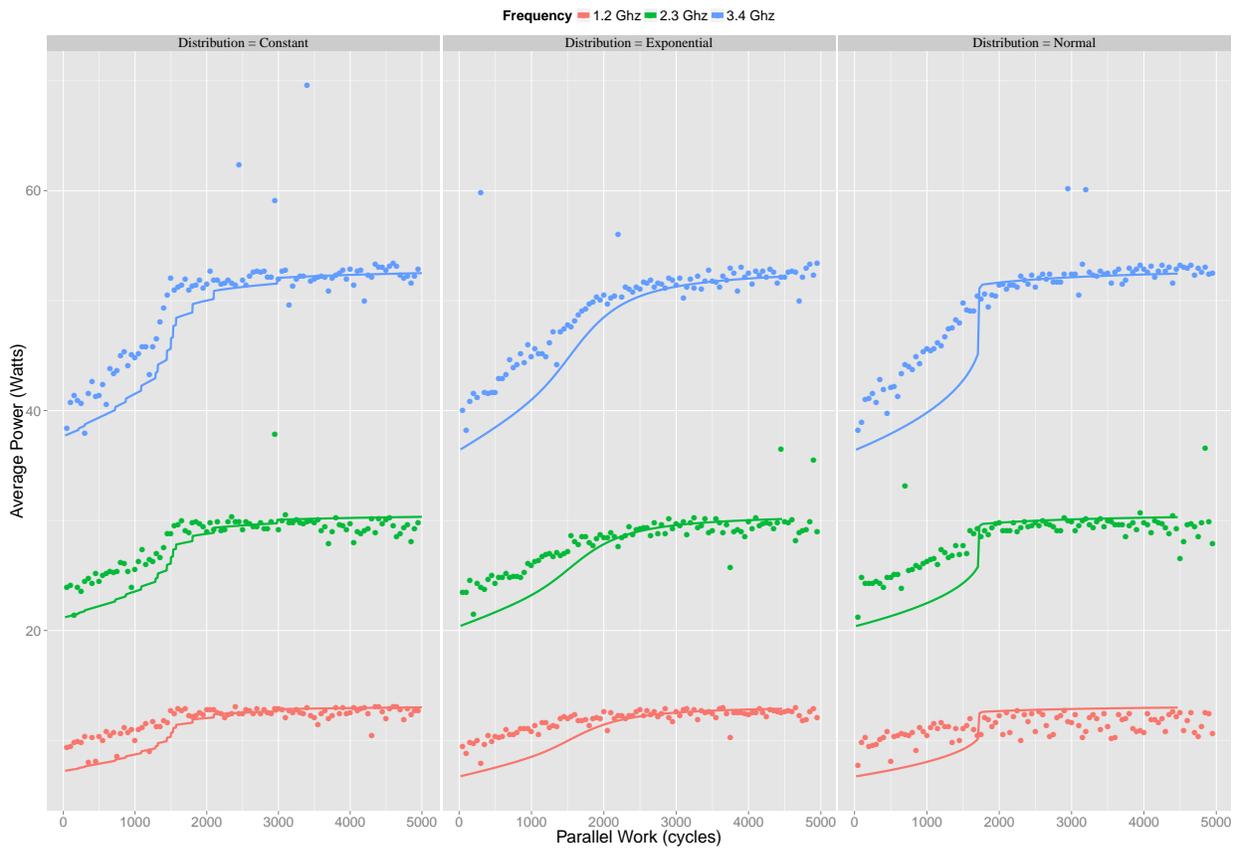}
\end{center}
\caption{Average Power Consumption for Shared Counter (Increment operation)\label{fig:en_sc}}
\end{figure}

In D2.1, we decompose the power into two orthogonal bases, each base
having three dimensions. On the one hand, we define the model base by
separating the power into static, active and dynamic power. On the
other hand, the measurement base corresponds to the components that
actually dissipate the power,\ie CPU, memory and uncore, in
accordance with RAPL energy counters. We recall that we are interested
only in the dynamic component of power, since we determine the static
power and the activation power, that do not depend on the \ds
implementation or the application that uses the concurrent \ds. Our
performance model does not cover the cases where the inter-socket
communication takes place. Here, we do not present the dynamic memory
and uncore power evaluations because they are insignificant (\ie close
to 0 for all cases) when there are not memory accesses (parallel work
is composed of multiplication instructions) or inter-socket
communication (threads are pinned to the same socket).

\begin{figure}[t!]
\begin{center}
\includegraphics[width=\textwidth]{Enq-energy.pdf}
\end{center}
\caption{Average Power Consumption for MS Queue (Enqueue operation)\label{fig:en_enqueue}}
\end{figure}

\begin{figure}[h!]
\begin{center}
\includegraphics[width=\textwidth]{Deq-energy.pdf}
\end{center}
\caption{Average Power Consumption for MS Queue (Dequeue operation)\label{fig:en_dequeue}}
\end{figure}

In D2.3, we explain in detail the methodology to obtain the energy consumption estimations
that span the whole parallel work and the number of threads domain. Here, we use
the same power model that relies on the variation of dynamic components of the power in between  
the execution of the data structure operations and parallel work. 
Different from D2.3, here we back this power model with a more extensive
performance analysis, presented above, in order to find the ratio of time that the 
parallel programs spend executing the data structure operations. Thanks to our performance 
analysis, 
we are able to estimate the energy consumption of lock-free data structures 
in dynamic environments where the size of parallel work, denoted by \pw, is either
constant or follows a probability distribution.    

We present the results for a set of fundamental lock-free \ds operations,
namely for Micheal and Scott Queue (Enqueue and Dequeue operations),
Treiber's Stack (Pop operation) and Shared Counter (Increment
operation). In the figures, x-axis provides the mean of \pw which follows a 
probability distribution. Lines and points represent predictions and
actual measurements, respectively. The performance estimations, for the different probability
distributions, are conducted by making use of different approaches. We used, 
respectively, the approach 
presented in D2.3, the constructive approach in Section~\ref{sec:cons}, the \avba approach in 
Section~\ref{sec:avba} for the cases where \pw is constant, follows exponential 
distribution and follows normal distribution.

In the figures, we observe a similar behaviour. Dynamic CPU power
decreases when \pw decreases. We know that \pw is a key aspect that
influences the contention on the \ds, equally with the
the ratio of time that threads spend in the data structure operation.
Also, we can observe, though slightly, the difference between the different 
probability distributions of \pw. For instance, the variation of the 
average power occurs more smoothly when \pw follows exponential distribution, 
similar to what is estimated by our model.

\clearpage


\subsection{Conclusion}
\label{sec:conc}

In this work we have presented two analyses for calculating the
performance of lock-free \dss in dynamic environments. The first
analysis has its roots in queuing theory, and gives the flexibility to
cover a large spectrum of configurations. The second analysis makes
use of Markov chains to exhibit a stochastic execution; it gives
better results, but it is restricted to simpler \dss and exponentially
distributed \pww.
\rr{We have evaluated the quality of the prediction on basic \dss like
stacks, as well as more advanced \dss like optimized queues and
\deqs. Our results can be directly used by algorithmicians to gain a
better understanding of the performance behavior of different designs,
and by experimentalists to rank implementations within a fair
framework. }%
%
We have\rr{ also} shown how to use our results to tune applications using
lock-free codes. These tuning methods include: (i) the calculation of
simple and efficient back-off strategies whose applicability is
illustrated in application contexts; (ii) a new adaptative memory
management mechanism that acclimates to a changing environment.

Moreover, we have integrated the performance estimations with the power 
model to estimate the energy efficiency of lock-free data structure by using
only a restricted amount of information about the application at hand.
%
%
%
%

The main differences between the \dss of this work and linked lists,
skip lists and trees occur when the size of the data structure
grows. With large sizes, the performance is dominated by the traversal
cost that is ruled by the cache parameters. The reduction in the size
of the data structure decreases the traversal cost which in turn
increases the probability of encountering an on-going \cas operation
that delays the threads which traverse the link.
The expansion, which can additionally be supported unfavorably by
helping mechanisms, appears then as the main performance degrading
factor.
%
%
While the analysis becomes easier for high degrees of parallelism
(large \ds size), being able to describe the behavior of lock-free
data structures as the degree of parallelism changes constitutes the
main challenge of our future work.


\label{sec:chalmers-lock-free}


\newpage
\section{A General and Validated Energy Complexity Model for Multi-threaded Algorithms}
\label{sec:uit-energy-model}
In this Deliverable D2.4, we report the ICE (Ideal Cache Energy) complexity model for analyzing the energy complexity of a wide range of multi-threaded algorithms \cite{TranH16_ICPADS}. Compared to the EPEM model reported in D2.3, this model proposed using Ideal Cache memory model to compute I/O complexity of the algorithms. Besides a case study of SpMV to demonstrate how to apply the ICE model to find energy complexity of parallel algorithms which is described in Deliverable D2.3, Deliverable D2.4 also reports a case study to apply the ICE model to Dense Matrix Multiplication (matmul) (cf. Section \ref{Matmul-energy}). The model is then validated with both data-intensive (i.e., SpMV) and computation-intensive (i.e., matmul) algorithms according to three aspects: different algorithms, different input types/sizes and different platforms (cf. Section \ref{validation}). In order to make the reading flow easy to follow, we include in this report a complete study of ICE model along with latest results.
\subsection{Introduction}
As described in Deliverable D2.3, understanding the energy complexity of algorithms is crucially important to improve the energy efficiency of algorithms and reduce the energy consumption of computing systems \cite{HaTUAGRTW16, Lagraviere2016}. One of the main approaches to understand the energy complexity of algorithms is to devise energy models.
 
Significant efforts have been devoted to developing power and energy models in literature \cite{Alonso2014, Choi2013, Choi2014, Korthikanti2009, Korthikanti2010, 7108419, Mishra:2015, Snowdon:2009}. However, there are no analytic models for multithreaded algorithms that are both applicable to a wide range of algorithms and comprehensively validated yet (cf. Table \ref{table:energy-model-summary}). The existing {\em parallel} energy models are either theoretical studies without validation or only applicable for specific algorithms. Modeling energy consumption of {\em parallel} algorithms is difficult since the energy models must take into account the complexity of both parallel algorithms and parallel platforms. The algorithm complexity results from parallel computation, concurrent memory accesses and inter-process communication. The platform complexity results from multicore architectures with deep memory hierarchy.

The existing models and their classification are summarized in Table \ref{table:energy-model-summary}. 
To the best of our knowledge, the proposed ICE (Ideal Cache Energy) complexity model is the first energy model that covers all three aspects: i) ability to analyze the energy complexity of parallel algorithms (i.e. Energy complexity analysis for parallel algorithms), ii) applicability to a wide range of algorithms (i.e., Algorithm generality), and iii) model validation (i.e., Validation). Section \ref{related-work} describes how the ICE model complements the other currently used models.
\begin{table*}
\caption{Energy Model Summary}
\label{table:energy-model-summary}
\begin{center}
\begin{tabular}{lllll}
\hline\noalign{\smallskip}
\textbf{Study} 	 &\textbf{Energy complexity} 	&\textbf{Algorithm}  	 &\textbf{Validation}  	 \\
	 &\textbf{analysis for} 	&\textbf{ generality}  		&  	\\
	 &\textbf{parallel algorithms} && 				\\
\noalign{\smallskip}\hline\noalign{\smallskip}
LEO \cite{Mishra:2015} 	&No 	&General	&Yes 	\\
POET \cite{7108419} 	&No 	&General	&Yes	\\
Koala \cite{Snowdon:2009} 	&No 	&General	&Yes	\\
Roofline \cite{Choi2013,Choi2014}	&No	&General	&Yes	\\  
Energy scalability \cite{Korthikanti2009, Korthikanti2010}	&Yes	&General	&No	\\ 
Sequential energy complexity \cite{Roy2013}   	&No	&General	&Yes\\
Alonso  et al. \cite{Alonso2014}	&Yes	&Algorithm-specific	&Yes	\\
Malossi  et al. \cite{Malossi2015}	&Yes	&Algorithm-specific	&Yes	\\
\textbf{ICE model (this study)}       	&\textbf{Yes}	&\textbf{General}	&\textbf{Yes}	\\
\noalign{\smallskip}\hline\noalign{\smallskip}
\end{tabular}
\end{center}
\end{table*}

The energy complexity model ICE proposed in this study is for general multithreaded algorithms and validated on three aspects: different algorithms for a given problem, different input types and different platforms. The proposed model is an analytic model which characterizes both algorithms (e.g., representing algorithms by their {\em work}, {\em span} and {\em I/O} complexity) and platforms (e.g., representing platforms by their static and dynamic energy of memory accesses and computational operations). By considering {\em work}, {\em span} and {\em I/O} complexity, the new ICE model is applicable to any multithreaded algorithms. 

The new ICE model is designed for analyzing the energy {\em complexity} of algorithms and therefore the model does not provide the estimation of absolute energy consumption. The goal of the ICE model is to answer energy complexity question: "Given two parallel algorithms A and B for a given problem, which algorithm consumes less energy analytically?". Hence, the details of underlying systems (e.g., runtime and architectures) are abstracted away to keep ICE model simple and suitable for complexity analysis. O-notation represents an {\em asymptotic upper-bound} on energy complexity.

In this work, the following contributions have been made.
\begin{itemize}
\item Devising a new general energy model ICE for analyzing the energy complexity of a wide range of multithreaded algorithms based on their {\em work}, {\em span} and {\em I/O} complexity (cf. Section \ref{energy-model}). The new ICE model abstracts away possible {\em multicore platforms} by their static and dynamic energy of computational operations and memory access. The new ICE model complements previous energy models such as energy roofline models \cite{Choi2013, Choi2014} that abstract away possible {\em algorithms} to analyze the energy consumption of different multicore platforms.
\item Conducting two case studies (i.e., SpMV and matmul) to demonstrate how to apply the ICE model to find energy complexity of parallel algorithms. The selected parallel algorithms for SpMV are three algorithms: Compressed Sparse Column(CSC), Compressed Sparse Block(CSB) and Compressed Sparse Row(CSR)(cf. Section \ref{SpMV-energy}). The selected parallel algorithms for matmul are two algorithms: a basic matmul algorithm and a cache-oblivious algorithm (cf. Section \ref{Matmul-energy}). 
\item Validating the ICE energy complexity model with both data-intensive (i.e., SpMV) and computation-intensive (i.e., matmul) algorithms according to three aspects: different algorithms, different input types and different platforms. The results show the precise prediction on which validated SpMV algorithm (i.e., CSB or CSC) consumes more energy when using different matrix input types from Florida matrix collection \cite{Davis:2011} (cf. Section \ref{SpMV-validation}). The results also show the precise prediction on which validated matmul algorithm (i.e., basic or cache-oblivious) consumes more energy (cf. Section \ref{Matmul-validation}). The model platform-related parameters for 11 platforms, including x86, ARM and GPU, are provided to facilitate the deployment of the ICE model.  
\end{itemize}

\subsection{Related Work - Overview of energy models}
\label{related-work}
 
\comm{
We present the summary of existing modeling studies in Table \ref{table:energy-model-summary}. The characteristics of each approaches are extracted as the list of categories, including: whether the models support parallel algorithms (i.e., Parallel Algorithm Support), whether the model is applicable to general algorithms (i.e., Applicability), whether the model is validated (i.e., Validation), the communication model (i.e., Communication model), whether there is pre-run overhead before estimating energy consumption of applications (i.e., Pre-run Overhead) and how the model represents applications (i.e., App-properties). This summary is not an exhaustive survey on the topic of energy models. However, we believe the Table \ref{table:energy-model-details} represents the most current studies on energy models.
}

We also included the related work of the most well-known energy models in this report to show why we need the new proposed ICE model. Energy models for finding energy-optimized system configurations for a given application have been recently reported [12, 16, 19]. Imes et al. \cite{7108419} used controller theory and linear programming to find energy-optimized configurations for an application with soft real-time constraints at runtime. 
Mishra et al. \cite{Mishra:2015} used hierarchical Bayesian model in machine learning to find  energy-optimized configurations. 
Snowdon et al. \cite{Snowdon:2009} developed a power management framework called Koala which models the energy consumption of the platform and monitors an application' energy behavior. 
Although the energy models for finding energy-optimized system configurations have resulted in energy saving in practice, they focus on characterizing system platforms rather than applications and therefore are not appropriate for analyzing the energy complexity of application algorithms. 

Another direction of energy modeling study is to predict the energy consumption of applications by analyzing applications without actual execution on real platforms which we classify as analytic models. 

Among energy and power models for different architectures \cite{Choi2013, Choi2014, Leng:2013, MCPAT,  TranH16_SAMOS, Tran2016}, energy roofline models \cite{Choi2013, Choi2014} are some of the comprehensive energy models that abstract away possible algorithms in order to analyze and characterize different multicore platforms in terms of energy consumption. 
Our new energy model, which abstracts away possible multicore platform and characterize the energy complexity of algorithms based on their {\em work, span} and {\em I/O} complexity, complements the energy roofline models.  

Validated energy models for {\em specific} algorithms have been reported recently \cite{Alonso2014, Malossi2015}. Alonso et al. \cite{Alonso2014} provided an accurate energy model for three key dense matrix factorizations. Malossi et al. \cite{Malossi2015} focused on basic linear-algebra kernels and characterized the kernels by the number of arithmetic operations, memory accesses, reduction and barrier steps. Although the energy models for specific algorithms are accurate for the target algorithms, they are not applicable for other algorithms and therefore cannot be used as general energy complexity models for parallel algorithms.

The {\em energy scalability} of a parallel algorithm has been investigated by Korthikanti et al. \cite{Korthikanti2009, Korthikanti2010}. 
Unlike the energy scalability studies that have not been validated on real platforms, our new energy complexity model is validated on HPC and accelerator platforms, confirming its usability and accuracy.

The energy complexity of {\em sequential} algorithms on a {\em uniprocessor} machine with {\em several memory banks} has been studied by Roy et al. \cite{Roy2013}. Our energy complexity studies complement Roy et al.'s studies by investigating the energy complexity of {\em parallel} algorithms on a {\em multiprocessor} machine with {\em a shared memory bank} and private caches, a machine model that has been widely adopted to study parallel algorithms \cite{Frigo:2006, Arge:2008, Korthikanti2010}.

\subsection{ICE Shared Memory Machine Model}
\label{EPEM-model}
Generally speaking, the energy consumption of a parallel algorithm is the sum of i) static energy (or leakage) $E_{static}$, ii) dynamic energy of computation $E_{comp}$ and iii) dynamic energy of memory accesses $E_{mem}$. The static energy $E_{static}$ is proportional to the execution time of the algorithm while the dynamic energy of computation and the dynamic energy of memory accesses are proportional to the number of computational operations and the number of memory accesses of the algorithm, respectively \cite{Korthikanti2010}. As a result, in the new ICE complexity model, the energy complexity of a multithreaded algorithm is analyzed based on its {\em span complexity} \cite{CormenLRS:2009} (for the static energy), {\em work complexity} \cite{CormenLRS:2009} (for the dynamic energy of computation) and {\em I/O complexity} (for the dynamic energy of memory accesses) (cf. Section \ref{energy-model}). This section describes shared-memory machine models supporting I/O complexity analysis for parallel algorithms.

The first memory model we consider is parallel external memory (PEM) model \cite{Arge:2008}, an extension of the Parallel Random Access Machine (PRAM) model that includes a two-level memory hierarchy. 
In the PEM model, there are $n$ cores (or processors) each of which has its own {\em private} cache of size $Z$ (in bytes) and shares the main memory with the other cores (cf. Figure \ref{fig:machine-model}). 
When $n$ cores access $n$ distinct blocks from the shared memory {\em simultaneously}, the I/O complexity in the PEM model is $O(1)$ instead of $O(n)$.
\begin{figure}[!t] \centering
\resizebox{0.5\columnwidth}{!}{ \includegraphics{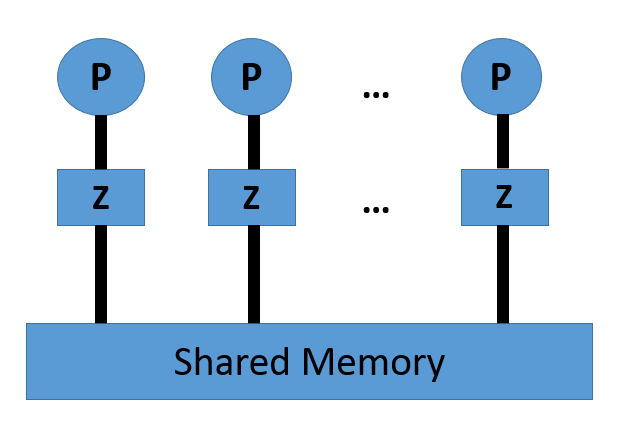}}
\caption{A Shared Memory Machine Model with Private Caches}
\label{fig:machine-model}
\end{figure} 
Although the PEM model is appropriate for analyzing the I/O complexity of parallel algorithms in terms of time performance \cite{Arge:2008}, we have found that the PEM model is not appropriate for analyzing parallel algorithms in terms of the dynamic energy of memory accesses. In fact, even when the $n$ cores can access data from the main memory simultaneously, the {\em dynamic} energy consumption of the access is proportional to the number $n$ of accessing cores (because of the load-store unit activated within each accessing core and the energy compositionality of parallel computations \cite{HaTUTGRWA:2014, chris-eehco14}), rather than a constant as implied by the PEM model.

As a result, we consider the ideal distributed cache (IDC) model \cite{Frigo:2006} to analyze I/O complexity of multithreaded algorithms in terms of dynamic energy consumption. 
Since the cache complexity of $m$ misses is $O(m)$ regardless of whether or not the cache misses are incurred simultaneously by the cores, the IDC model reflects the aforementioned dynamic energy consumption of memory accesses by the cores. 

However, the IDC model is mainly designed for analyzing the cache complexity of divide-and-conquer algorithms, making it difficult to apply to general multi-threaded algorithms targeted by our new ICE model. Constraining the new ICE model to the IDC model would limit the applicability of the ICE model to a wide range of multithreaded algorithms.

In order to make our new ICE complexity model applicable to a wide range of multithreaded algorithms, we show that the cache complexity analysis using the traditional (sequential) ideal cache (IC) model \cite{FrigoLPR:1999} can be used to find an upper bound  on the cache complexity of the same algorithm using the IDC model (cf. Lemma \ref{ePEM}). As the sequential execution of multithreaded algorithms is a valid execution regardless of whether they are divide-or-conquer algorithms, the ability to analyze the cache complexity of multithreaded algorithms via their sequential execution in the ICE complexity model improves the usability of the ICE model.   
 
Let $Q_1(Alg,B,Z)$ and $Q_P(Alg, B, Z)$ be the cache complexity of a parallel algorithm $Alg$ analyzed in the (uniprocessor) ideal cache (IC) model  \cite{FrigoLPR:1999} with block size $B$ and cache size $Z$ (i.e, running $Alg$ with a single core) and the cache complexity analyzed in the (multicore) IDC model with $P$ cores each of which has a private cache of size $Z$ and block size $B$, respectively. We have the following lemma:
\begin{lemma}
\label{ePEM}
The cache complexity $Q_P(Alg, B, Z)$ of a parallel algorithm $Alg$ analyzed in the ideal distributed cache (IDC) model with $P$ cores is bounded from above by the product of $P$ and the cache complexity $Q_1(Alg, B, Z)$ of the same algorithm analyzed in the ideal cache (IC) model. Namely,
\begin{equation}
Q_P(Alg, B, Z) \leq P*Q_1(Alg,B,Z)
\end{equation}
\end{lemma}
\begin{proof} 
(Sketch)
Let $Q_P^{i}(Alg, B, Z)$ be the number of cache misses incurred by core $i$ during the parallel execution of algorithm $Alg$ in the IDC model. Because caches do not interfere with each other in the IDC model, the number of cache misses incurred by core $i$ when executing algorithm $Alg$ in parallel by $P$ cores is not greater than the number of cache misses incurred by core $i$ when executing the whole algorithm $Alg$ only by core $i$. That is,
\begin{equation}
\label{eq:idc_ic_1}
Q_P^{i}{(Alg, B, Z)} \leq Q_1(Alg,B,Z)
\end{equation}
or
\begin{equation}
\label{eq:idc_ic_2}
\sum_{i=1}^P Q_P^{i}{(Alg, B, Z)} \leq P*Q_1(Alg,B,Z)
\end{equation}

On the other hand, since the number of cache misses incurred by algorithm $Alg$ when it is executed by $P$ cores in the IDC model is the sum of the numbers of cache misses incurred by each core during the $Alg$ execution, we have 
\begin{equation}
\label{eq:idc_ic_3}
Q_P(Alg, B, Z) = \sum_{i=1}^P Q_P^{i}{(Alg, B, Z)}
\end{equation}
From Equations \ref{eq:idc_ic_2} and \ref{eq:idc_ic_3}, we have
\begin{equation}
\label{eq:idc_ic_4}
Q_P(Alg, B, Z) \leq P*Q_1(Alg,B,Z)
\end{equation}
\end{proof}

We also make the following assumptions regarding platforms.
\begin{itemize}
\item Algorithms are executed with the best configuration (e.g., maximum number of cores, maximum frequency) following the race-to-halt strategy.
\item The I/O parallelism is bounded from above by the computation parallelism. Namely, each core can issue a memory request only if its previous memory requests have been served. Therefore, the work and span (i.e., critical path) of an algorithm represent the parallelism for both I/O and computation \cite{CormenLRS:2009}. 
\end{itemize}  

\subsection{Energy Complexity in ICE model}
\label{energy-model}
This section describes two energy complexity models, a platform-supporting energy complexity model considering both platform and algorithm characteristics and platform-independent energy complexity model considering only algorithm characteristics. The platform-supporting model is used when platform parameters in the model are available while platform-independent model analyses energy complexity of algorithms without considering platform characteristics.
\subsubsection{Platform-supporting Energy Complexity Model} 
This section describes a methodology to find energy complexity of algorithms. The energy complexity model considers three groups of parameters: machine-dependent, algorithm-dependent and input-dependent parameters. The reason to consider all three parameter-categories is that only operational intensity \cite{Williams2009} is insufficient to capture the characteristics of algorithms. Two algorithms with the same values of operational intensity might consume different levels of energy. The reasons are their differences in data accessing patterns leading to performance scalability gap among them. For example, although the sequential version and parallel version of an algorithm may have the same operational intensity, they may have different energy consumption since the parallel version would have less static energy consumption because of shorter execution time.

The energy consumption of a parallel algorithm is the sum of i) static energy (or leakage) $E_{static}$, ii) dynamic energy of computation $E_{comp}$ and iii) dynamic energy of memory accesses $E_{mem}$: $E=E_{static}+E_{comp}+E_{mem}$ \cite{Choi2013, Korthikanti2009, Korthikanti2010}. The static energy $E_{static}$ is the product of the execution time of the algorithm and the static power of the whole platform. The dynamic energy of computation and the dynamic energy of memory accesses are proportional to the number of computational operations $Work$ and the number of memory accesses $I/O$, respectively. Pipelining technique in modern architectures enables overlapping computation with memory accesses \cite{HaTUTGRWA:2014}. Since computation time and memory-access time can be overlapped, the execution time of the algorithm is assumed to be the maximum of computation time and memory-access time \cite{Choi2013}. Therefore, the energy consumption of algorithms is computed by Equation \ref{eq:BigET}, 
where the values of ICE parameters, including $\epsilon_{op}$, $\epsilon_{I/O}$, $\pi_{op}$, and $\pi_{I/O}$ are described in Table \ref{table:ModelParameters} and computed by the Equation \ref{eq:epsilon_op}, \ref{eq:epsilon_IO}, \ref{eq:pi_op}, and \ref{eq:pi_IO}, respectively.  
\begin{equation} \label{eq:BigET}
	E= \epsilon_{op} \times Work + \epsilon_{I/O} \times I/O + P^{sta} \times max(T^{comp} ,T^{mem})
\end{equation}
\comm{
\begin{equation} \label{eq:BigE0}
	E= \epsilon_{op} \times Work + \epsilon_{I/O} \times I/O + max(\pi_{op} \times Span ,\pi_{I/O} \times \frac{I/O \times Span}{Work})  
\end{equation}
}

\begin{table}
\caption{ICE Model Parameter Description}
\label{table:ModelParameters}
\begin{center}
\begin{tabular}{ll}
\hline\noalign{\smallskip}
Machine & Description \\
\noalign{\smallskip}\hline\noalign{\smallskip}
$\epsilon_{op}$                  & dynamic energy of one operation (average) \\ 
$\epsilon_{I/O}$                  & dynamic energy of a random memory access (1 core)\\ 
$\pi_{op}$                  & static energy when performing one operation  \\ 
$\pi_{I/O}$                  & static energy of a random memory access \\ 
\noalign{\smallskip}\hline\noalign{\smallskip}
\hline\noalign{\smallskip}
Algorithm & Description \\
\noalign{\smallskip}\hline\noalign{\smallskip}
$Work$                & Number of work in flops of the algorithm \cite{CormenLRS:2009}        \\ 
$Span$                 & The critical path of the algorithm \cite{CormenLRS:2009}       \\ 
$I/O$                  & Number of cache line transfer of the algorithm \cite{CormenLRS:2009}       \\ 
\noalign{\smallskip}\hline\noalign{\smallskip}
\end{tabular}
\end{center}
\end{table} 
 
\begin{equation} \label{eq:epsilon_op}
	\epsilon_{op}= P^{op} \times \frac{F}{Freq}
\end{equation}
\begin{equation} \label{eq:epsilon_IO}	
	\epsilon_{I/O}= P^{I/O} \times \frac{M}{Freq}
\end{equation}
\begin{equation} \label{eq:pi_op}
	\pi_{op}= P^{sta} \times \frac{F}{Freq}
\end{equation}
\begin{equation} \label{eq:pi_IO}
	\pi_{I/O}= P^{sta} \times \frac{M}{Freq}
\end{equation}

The dynamic energy of one operation by one core $\epsilon_{op}$ is the product of the consumed power of one operation by one active core $P^{op}$ and the time to perform one operation. Equation \ref{eq:epsilon_op} shows how $\epsilon_{op}$ relates to frequency $Freq$ and the number of cycles per operation $F$. Similarly, the dynamic energy of a random access by one core $\epsilon_{I/O}$ is the product of the consumed power by one active core performing one I/O (i.e., cache-line transfer) $P^{I/O}$ and the time to perform one cache line transfer computed as $M/Freq$, where $M$ is the number of cycles per cache line transfer (cf. Equation \ref{eq:epsilon_IO}). The static energy of operations $\pi_{op}$ is the product of the whole platform static power $P^{sta}$ and time per operation. The static energy of one I/O $\pi_{I/O}$ is the product of the whole platform static power and time per I/O, shown by Equation \ref{eq:pi_op} and \ref{eq:pi_IO}.

In order to compute {\em work}, {\em span} and {\em I/O} complexity of the algorithms, the input parameters also need to be considered. For example, SpMV algorithms consider input parameters listed in Table \ref{table:AlgorithmInputParameters}. Cache size is captured in the ICE model by the {\em I/O complexity} of the algorithm. Note that in the ICE machine model (Section \ref{EPEM-model}), cache size $Z$ is a constant and may disappear in the {\em I/O complexity} (e.g., O-notation).   

The details of how to obtain the ICE parameters of recent platforms are discussed in Section \ref{experiment-set-up}. The actual values of ICE platform parameters for 11 recent platforms are presented in Table \ref{table:platform-parameter-values}. 
\comm{
\begin{table}
\caption{Platform Parameter Description}
\label{table:PlatformParameters}
\begin{center}
\begin{tabular}{ll}
\hline\noalign{\smallskip}
Machine & Description \\
\noalign{\smallskip}\hline\noalign{\smallskip}
$P^{sta}$                & Static power of a whole platform         \\ 
$P^{op}$                  & Dynamic power of an operation        \\ 
$P^{I/O}$                  & Power to transfer one cache line \\
$N$                & Maximum number of cores in the platform      \\ 
$M$              & Number of cycles per cache line transfer \\ 
$F$              & Number of cycles per operation      \\
$Freq$                & Platform frequency       \\
$Z$              & Cache size of a single processor        
 \\
$B$              & Cache block size         \\ 
\noalign{\smallskip}\hline\noalign{\smallskip}
\end{tabular}
\end{center}
\end{table}
}
\begin{table*}
\caption{Platform parameter summary. The parameters of the first nine platforms are derived from \cite{Choi2014} and the parameters of the two new platforms are found in this study.}
\label{table:platform-parameter-values}
\begin{center}
\begin{tabular}{llllll}
\hline\noalign{\smallskip}
Platform & Processor &  $\epsilon_{op}$(nJ) &$\pi_{op}$(nJ) &$\epsilon_{I/O}$(nJ) & $\pi_{I/O}$(nJ) \\
\noalign{\smallskip}\hline\noalign{\smallskip}
Nehalem i7-950	&Intel i7-950	&0.670 	&2.455	&50.88	&408.80\\
Ivy Bridge i3-3217U	&Intel  i3-3217U	&0.024 	&0.591	&26.75	&58.99\\
Bobcat CPU 	&AMD  E2-1800	&0.199 	&3.980	&27.84	&387.47\\
Fermi GTX 580	&NVIDIA GF100	&0.213 	&0.622	&32.83	&45.66\\
Kepler GTX 680	&NVIDIA GK104	&0.263 	&0.452	&27.97	&26.90\\
Kepler GTX Titan	&NVIDIA GK110	&0.094 	&0.077	&17.09	&32.94\\
XeonPhi KNC	&Intel 5110P	&0.012 	&0.178	&8.70	&63.65\\
Cortex-A9	&TI OMAP 4460	&0.302 &1.152	&51.84	&174.00\\
Arndale Cortex-A15	&Samsung Exynos 5	&0.275 	&1.385	&24.70	&89.34\\
\noalign{\smallskip}\hline\noalign{\smallskip}
Xeon 	&2xIntel E5-2650l v3	&0.263	&0.108	&8.86	&23.29\\
Xeon-Phi	&Intel 31S1P	&0.006	&0.078	&25.02	&64.40\\
\noalign{\smallskip}\hline\noalign{\smallskip}
\end{tabular}
\end{center}
\end{table*} 

The computation time of parallel algorithms is proportional to the span complexity of the algorithm, which is $T^{comp}=\frac{Span \times F}{Freq}$ where $Freq$ is the processor frequency, and $F$ is the number of cycles per operation. The memory-access time of parallel algorithms in the ICE model is proportional to the I/O complexity of the algorithm divided by its I/O parallelism, which is $T^{mem} = \frac{I/O}{I/O-parallelism} \times \frac{M}{Freq}$. As I/O parallelism, which is the average number of I/O ports that the algorithm can utilize per step along the span, is bounded by the computation parallelism $\frac{Work}{Span}$, namely the average number of cores that the algorithm can utilize per step along the span (cf. Section \ref{EPEM-model}), the memory-access time $T^{mem}$ becomes: $T^{mem}=\frac{I/O \times Span \times M}{Work \times Freq}$ where $M$ is the number of cycles per cache line transfer. If an algorithm has $T^{comp}$ greater than $T^{mem}$, the algorithm is a CPU-bound algorithm. Otherwise, it is a memory-bound algorithm. 
\paragraph{CPU-bound Algorithms}
If an algorithm has computation time $T^{comp}$ longer than data-accessing time $T^{mem}$ (i.e., CPU-bound algorithms), the ICE energy complexity model becomes Equation \ref{eq:BigE-cpu-time} which is simplified as Equation \ref{eq:BigE1}.
\begin{equation} \label{eq:BigE-cpu-time}
	E= \epsilon_{op} \times Work + \epsilon_{I/O} \times I/O +  P^{sta} \times \frac{Span \times F}{Freq}\\
\end{equation}
or
\begin{equation} \label{eq:BigE1}
	E= \epsilon_{op} \times Work + \epsilon_{I/O} \times I/O +  \pi_{op} \times Span\\
\end{equation}

\paragraph{Memory-bound Algorithms}
If an algorithm has data-accessing time longer than computation time (i.e., memory-bound algorithms): $T^{mem} \geq T^{comp}$, energy complexity becomes Equation \ref{eq:BigE-mem-time} which is simplified as Equation \ref{eq:BigE2}. 
\begin{equation} \label{eq:BigE-mem-time}
	E= \epsilon_{op} \times Work + \epsilon_{I/O} \times I/O +  P^{sta} \times \frac{I/O \times Span \times M}{Work \times Freq}\\
\end{equation}
or
\begin{equation} \label{eq:BigE2}
	E= \epsilon_{op} \times Work + \epsilon_{I/O} \times I/O +  \pi_{I/O} \times \frac{I/O \times Span}{Work}\\
\end{equation}
\subsubsection{Platform-independent Energy Complexity Model}
This section describes the energy complexity model that is platform-independent and considers only algorithm characteristics. When the platform parameters (i.e., $\epsilon_{op}$, $\epsilon_{I/O}$, $\pi_{op}$, and $\pi_{I/O}$) are unavailable, the energy complexity model is derived from Equation \ref{eq:BigET}, where the platform parameters are constants and can be removed. Assuming $\pi_{max} = max(\pi_{op}, \pi_{I/O})$, after removing platform parameters, the platform-independent energy complexity model are shown in Equation \ref{eq:BigE-non-flatform}.
\begin{equation} \label{eq:BigE-non-flatform}
	E= O(Work+I/O+max(Span, \frac{I/O \times Span}{Work}))  
\end{equation} 

\subsection{A Case Study of Sparse Matrix Multiplication}
\label{SpMV-energy}
SpMV is one of the most common application kernels in Berkeley dwarf list \cite{Asa06}. It computes a vector result $y$ by multiplying a sparse matrix $A$ with a dense vector $x$: $y=Ax$. SpMV is a data-intensive kernel and has irregular memory-access patterns. The data access patterns for SpMV is defined by its sparse matrix format and matrix input types. 
There are several sparse matrix formats and SpMV algorithms in literature. To name a few, they are Coordinate Format (COO), Compressed Sparse Column (CSC), Compressed Sparse Row (CSR), Compressed Sparse Block (CSB), Recursive Sparse Block (RSB), Block Compressed Sparse Row (BCSR) and so on.
Three popular SpMV algorithms, namely CSC, CSB and CSR are chosen to validate the proposed energy complexity model. They have different data-accessing patterns leading to different values of I/O, work and span complexity. Since SpMV is a memory-bound application kernel, Equation \ref{eq:BigE2} is applied. The input matrices of SpMV have different parameters listed in Table \ref{table:AlgorithmInputParameters}.
\begin{table}
\caption{SpMV Input Parameter Description}
\label{table:AlgorithmInputParameters}
\begin{center}
\begin{tabular}{ll}
\hline\noalign{\smallskip}
SpMV Input & Description \\
\noalign{\smallskip}\hline\noalign{\smallskip}
$n$                & Number of rows        \\ 
$nz$                 & Number of nonzero elements        \\ 
$nr$                  & Maximum number of nonzero in a row        \\ 
$nc$                  & Maximum number of nonzero in a column        \\ 
$\beta$                  & Size of a block \\ 
\noalign{\smallskip}\hline\noalign{\smallskip}
\end{tabular}
\end{center}
\end{table}
\subsubsection{Compressed Sparse Row}
CSR is a standard storage format for sparse matrices which reduces the storage of matrix compared to the tuple representation \cite{Kotlyar:1997}. This format enables row-wise compression of $A$ with size $n \times n$ (or  $n \times m$) to store only the non-zero $nz$ elements. Let {\em nz} be the number of non-zero elements in matrix A. 
The {\em work} complexity of CSR SpMV is $\Theta(nz)$ where $nz>=n$ and {\em span} complexity is $O(nr + \log{n})$ \cite{Buluc:2009}, where $nr$ is the maximum number of non-zero elements in a row. The {\em I/O} complexity of CSR in the sequential I/O model of row-major layout is $O(nz)$ \cite{Bender2010} namely, scanning all non-zero elements of matrix $A$ costs $O(\frac{nz}{B})$ I/Os with B is the cache block size. However, randomly accessing vector $x$ causes the total of $O(nz)$ I/Os.
Applying the proposed model on CSR SpMV, their total energy complexity are computed as Equation \ref{eq:BigE-CSR}.
\begin{equation} \label{eq:BigE-CSR}
	E_{CSR}= O(\epsilon_{op} \times nz + \epsilon_{I/O} \times nz + \pi_{I/O} \times (nr+ \log{}n)) 
\end{equation}
\subsubsection{Compressed Sparse Column}
CSC is the similar storage format for sparse matrices as CSR. However, it compresses the sparse matrix in column-wise manner to store the non-zero elements. The {\em work} complexity of CSC SpMV is $\Theta(nz)$ where $nz>=n$ and {\em span} complexity is $O(nc + \log{n})$, where $nc$ is the maximum number of non-zero elements in a column. The {\em I/O} complexity of CSC in the sequential I/O model of column-major layout is $O(nz)$ \cite{Bender2010}. Similar to CSR, scanning all non-zero elements of matrix $A$ in CSC format costs $O(\frac{nz}{B})$ I/Os. However, randomly updating vector $y$ causing the bottle neck with total of $O(nz)$ I/Os.
Applying the proposed model on CSC SpMV, their total energy complexity are computed as Equation \ref{eq:BigE-CSC}.
\begin{equation} \label{eq:BigE-CSC}
	E_{CSC}= O(\epsilon_{op} \times nz + \epsilon_{I/O} \times nz + \pi_{I/O} \times (nc+ \log{}n)) 
\end{equation}
\subsubsection{Compressed Sparse Block}
Given a sparse matrix $A$, while CSR has good performance on SpMV $y=Ax$, CSC has good performance on transpose sparse matrix vector multiplication $y=A^{T}\times x$, Compressed sparse blocks (CSB) format is efficient for computing either $Ax$ or $A^{T}x$. CSB is another storage format for representing sparse matrices by dividing the matrix $A$ and vector $x, y$  to blocks. A block-row contains multiple chunks, each chunks contains consecutive blocks and non-zero elements of each block are stored in Z-Morton-ordered \cite{Buluc:2009}.
From Beluc et al. \cite{Buluc:2009}, CSB SpMV computing a matrix with $nz$ non-zero elements, size $n\times n$ and divided by block size $\beta \times \beta$ has span complexity $O(\beta \times \log{\frac{n}{\beta}}+ \frac{n}{\beta})$ and {\em work} complexity as $\Theta(\frac{n^2}{\beta^2}+nz)$.

{\em I/O} complexity for CSB SpMV is not available in the literature. We do the analysis of CSB manually by following the master method \cite{CormenLRS:2009}. The {\em I/O} complexity is analyzed for the algorithm CSB\_SpMV(A,x,y) from Beluc et al. \cite{Buluc:2009}. The I/O complexity of CSB is similar to {\em work} complexity of CSB $O(\frac{n^2}{\beta^2} + nz)$, only that non-zero accesses in a block is divided by B: $O(\frac{n^2}{\beta^2} + {\frac{nz}{B}})$, where $B$ is cache block size. The reason is that non-zero elements in a block are stored in Z-Morton order which only requires $\frac{nz}{B}$ I/Os. The energy complexity of CSB SPMV is shown in Equation \ref{eq:BigE-CSB}.

From the complexity analysis of SpMV algorithms using different layouts, the complexity of CSR-SpMV, CSC-SpMV and CSB-SpMV are summarized in Table \ref{table:SpMV-complexity}.
\begin{floatEq}
\begin{equation}
\label{eq:BigE-CSB}
	E_{CSB}= O(\epsilon_{op} \times (\frac{n^2}{\beta^2} + nz) + \epsilon_{I/O} \times (\frac{n^2}{\beta^2} + \frac{nz}{B}) + \pi_{I/O} \times \frac{(\frac{n^2}{\beta^2} + \frac{nz}{B})\times (\beta \times \log{\frac{n}{\beta}}+ \frac{n}{\beta})}{(\frac{n^2}{\beta^2} + nz)} )
\end{equation}
\end{floatEq}
\begin{table*}
\caption{SpMV Complexity Analysis}
\label{table:SpMV-complexity}
\begin{center}
\begin{tabular}{llll}
\hline\noalign{\smallskip}
Complexity & CSC-SpMV & CSB-SpMV & CSR-SpMV \\
\noalign{\smallskip}\hline\noalign{\smallskip}
Work  & $\Theta(nz)$ \cite{Buluc:2009} & $\Theta(\frac{n^2}{\beta^2} + nz)$ \cite{Buluc:2009} & $\Theta(nz)$ \cite{Buluc:2009}\\ 
I/O                 & $O(nz)$ \cite{Bender2010} & $O(\frac{n^2}{\beta^2} + {\frac{nz}{B}})$ [this study] & $O(nz)$ \cite{Bender2010} \\ 
Span                  & $O(nc+ \log{}n)$ \cite{Buluc:2009} & $O(\beta \times \log{\frac{n}{\beta}}+ \frac{n}{\beta})$ \cite{Buluc:2009} & $O(nr+ \log{}n)$  \cite{Buluc:2009}\\ 

\noalign{\smallskip}\hline\noalign{\smallskip}
\end{tabular}
\end{center}
\end{table*}

\subsection{A Case Study of Dense Matrix Multiplication}
\label{Matmul-energy}
Besides SpMV, we also apply the ICE model to dense matrix multiplication (matmul). Unlike SpMV, a data-intensive kernel, matmul is a computation-intensive kernel used in high performance computing. It computes output matrix C (size n x p) by multiplying two dense matrices A (size n x m) and B (size m x p): $C=A \times B$. In this work, we implemented two matmul algorithms (i.e., a basic algorithm and a cache-oblivious algorithm \cite{FrigoLPR:1999}) and apply the ICE analysis to find their energy complexity. Both algorithms partition matrix A and C equally to N sub-matrices (e.g., $A_{i}$ with i=(1,2,..,N)), where N is the number of cores in the platform. The partition approach is shown in Figure \ref{fig:matmul-algo}. Each core computes a sub-matrix $C_{i}$: $C_{i}=A_{i} \times B$. Since matmul is a computation-bound application kernel, Equation \ref{eq:BigE1} is applied.
\begin{figure}[!t] \centering
\resizebox{0.5\columnwidth}{!}{ \includegraphics{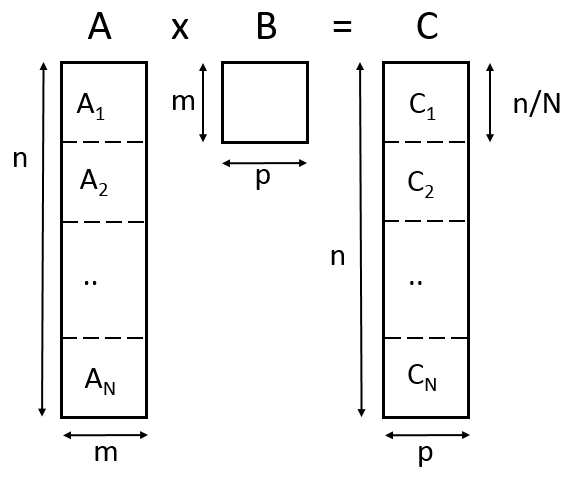}}
\caption{Partition approach for parallel matmul algorithms. Each sub-matrix $A_{i}$ has size $\frac{n}{N} \times m$ and each sub-matrix $C_{i}$ has size $\frac{n}{N} \times p$.}
\label{fig:matmul-algo}
\end{figure} 
\subsubsection{Basic Matmul Algorithm}
The basic matmul algorithm is described in Listing \ref{lst:matmul-simple}. Its work complexity is $\Theta(2nmp)$ \cite{Yelick2004} and span complexity is $\Theta(\frac{2nmp}{N})$ because the computational work is divided equally to N cores due to matrix partition approach. When matrix size of matrix B is bigger than the platform cache size, the basic algorithm loads matrix B n times (i.e., once for computing each row of C), results in $\frac{nmp}{B}$ cache block transfers, where $B$ is cache block size. In total, I/O complexity of the basic matmul algorithm is $\Theta(\frac{nm+nmp+np}{B})$. Applying the ICE model on this algorithm, the total energy complexity are computed as Equation \ref{eq:BigE-Naive}.
\begin{lstlisting}[caption={Simple Matmul},label={lst:matmul-simple}]
for i = 1 to n
	for j = 1 to p
		for k = 1 to m
			C (i,j) = C (i,j) + A(i,k) * B(k,j)
\end{lstlisting}
\begin{equation} \label{eq:BigE-Naive}
	E_{basic}= O(\epsilon_{op} \times 2nmp + \epsilon_{I/O} \times \frac{nm+nmp+np}{B} + \pi_{op} \times \frac{2nmp}{N}) 
\end{equation}
\subsubsection{Cache-oblivious Matmul Algorithm}
The cache-oblivious matmul (CO-matmul) algorithm \cite{FrigoLPR:1999} is a divide-and-conquer algorithm. It has work complexity the same as the basic matmul algorithm $\Theta(2nmp)$. Its span complexity is also $\Theta(\frac{2nmp}{N})$ because of the used matrix partition approach shown in Figure \ref{fig:matmul-algo}. The I/O complexity of CO-matmul, however, is different from the basic algorithm: $\Theta(n+m+p+\frac{nm+mp+np}{B} +\frac{nmp}{B\sqrt[2]{Z}})$ \cite{FrigoLPR:1999}. Applying the ICE model to CO-matmul, the total energy complexity are computed as Equation \ref{eq:BigE-CO}.
\begin{floatEq}
\begin{equation} \label{eq:BigE-CO}
	E_{CO}= O(\epsilon_{op} \times 2nmp + \epsilon_{I/O} \times (n+m+p+\frac{nm+mp+np}{B}+\frac{nmp}{B\sqrt[2]{Z}})+\pi_{op} \times \frac{2nmp}{N})
\end{equation}
\end{floatEq}
\begin{table*}
\caption{Matmul Complexity Analysis}
\label{table:Matmul-complexity}
\begin{center}
\begin{tabular}{llll}
\hline\noalign{\smallskip}
Complexity & Cache-oblivious Algorithm & Basic Algorithm\\
\noalign{\smallskip}\hline\noalign{\smallskip}
Work  & $\Theta(2nmp)$ \cite{FrigoLPR:1999}  & $\Theta(2nmp)$  \cite{Yelick2004}  \\ 
I/O   & $\Theta(n+m+p+\frac{nm+mp+np}{B} +\frac{nmp}{B\sqrt[2]{Z}})$ \cite{FrigoLPR:1999}  & $\Theta(\frac{nm+nmp+np}{B})$ [this study]\\ 
Span  & $\Theta(\frac{2nmp}{N})$ [this study] & $\Theta(\frac{2nmp}{N})$ [this study]\\ 
\noalign{\smallskip}\hline\noalign{\smallskip}
\end{tabular}
\end{center}
\end{table*}

\subsection{Validation of ICE Model}
\label{validation}
This section describes the experimental study to validate the ICE model, including: introducing the two experimental platforms and how to obtain their parameters for the ICE model, describing input types, and discussing the validation results of SpMV and matmul.
\subsubsection{Experiment Set-up}
\label{experiment-set-up}
For the validation of the ICE model, we conduct the experiments on two HPC platforms: one platform with two Intel Xeon E5-2650l v3 processors and one platform with Xeon Phi 31S1P processor. The Intel Xeon platform has two processors Xeon E5-2650l v3 with $2\times12$ cores, each processor has the frequency 1.8 GHz. The Intel Xeon Phi platform has one processor Xeon Phi 31S1P with $57$ cores and its frequency is 1.1 GHz. To measure energy consumption of the platforms, we read the PCM MSR counters for Intel Xeon and MIC power reader for Xeon Phi.
\subsubsection{Identifying Platform Parameters}
We apply the energy roofline approach \cite{Choi2013, Choi2014} to find the platform parameters for the two new experimental platforms, namely Intel Xeon E5-2650l v3 and Xeon Phi 31S1P. Moreover, the energy roofline study \cite{Choi2014} has also provided a list of other platforms including CPU, GPU, embedded platforms with their parameters considered in the Roofline model. Thanks to authors Choi et al. \cite{Choi2014}, we extract the required values of ICE parameters for nine platforms presented in their study as follows: $\epsilon_{op}=\epsilon_{d}$, $\epsilon_{I/O}=\epsilon_{mem}\times B$, $\pi_{op}=\pi_{1} \times \tau_{d}$, $\pi_{I/O}=\pi_{1} \times \tau_{mem}$, where $B$ is cache block size, $\epsilon_{d}$, $\epsilon_{d}$, $\tau_{d}$, $\tau_{mem}$ are defined by \cite{Choi2014} as energy per flop, energy per byte, time per flop and time per byte, respectively. 

The ICE parameter values of the two new HPC platforms (i.e., Xeon and Xeon-Phi 31S1P) used to validate the ICE model are obtained by using the same approach as energy roofline study \cite{Choi2013}. We create micro-benchmarks for the two platforms and measure their energy consumption and performance. The ICE parameter values of each platform are obtained from energy and performance data by regression techniques. Along with the two HPC platforms used in this validation, we provide parameters required in the ICE model for a total of 11 platforms. Their platform parameters are listed in Table \ref{table:platform-parameter-values} for further uses.
\subsubsection{SpMV Implementation}
\label{SpMV-Implementation}
We want to conduct complexity analysis and experimental study with two SpMV algorithms, namely CSB and CSC. Parallel CSB and sequential CSC implementations are available thanks to the study from Bulu\c{c} et al. \cite{Buluc:2009}. Since the optimization steps of available parallel SpMV kernels (e.g., pOSKI \cite{pOSKI}, LAMA\cite{Forster2011}) might affect the work complexity of the algorithms, we decided to implement a simple parallel CSC using Cilk and pthread. To validate the correctness of our parallel CSC implementation, we compare the vector result $y$ from $y=A*x$ of CSC and CSB implementation. The comparison shows the equality of the two vector results $y$. Moreover, we compare the performance of the our parallel CSC code with Matlab parallel CSC-SpMV kernel. Matlab also uses CSC layout as the format for their sparse matrix \cite{Gilbert:1992} and is used as baseline comparison for SpMV studies \cite{Buluc:2009}. Our CSC implementation has out-performed Matlab parallel CSC kernel when computing the same targeted input matrices. Figure \ref{fig:Matlab-comparison} shows the performance comparison of our CSC SpMV implementation and Matlab CSC SpMV kernel. The experimental study of SpMV energy consumption is then conducted with CSB SpMV implementation from Bulu\c{c} et al. \cite{Buluc:2009} and our CSC SpMV parallel implementation.  
\begin{figure}[!t] \centering
\resizebox{0.8\columnwidth}{!}{ \includegraphics{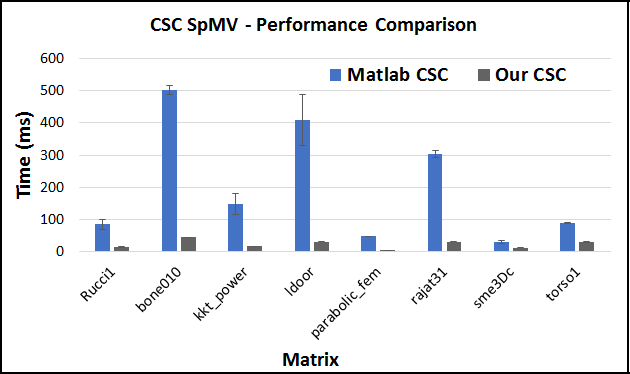}}
\caption{Performance (time) comparison of two parallel CSC SpMV implementations. For a set of different input matrices, the parallel CSC SpMV using Cilk out-performs Matlab parallel CSC.}
\label{fig:Matlab-comparison}
\end{figure}
\subsubsection{SpMV Matrix Input Types}
We conducted the experiments with nine different matrix-input types from Florida sparse matrix collection \cite{Davis:2011}. Each matrix input has different properties listed in Table \ref{table:AlgorithmInputParameters}, including size of the matrix $n\times m$, the maximum number of non-zero of the sparse matrix $nz$, the maximum number of non-zero elements in one column $nc$. Table \ref{table:matrix-type} lists the matrix types used in this experimental validation with their properties. 
\begin{table}
\caption{Sparse matrix input types. The maximum number of non-zero elements in a column $nc$ is derived from \cite{Buluc:2009}.}
\label{table:matrix-type}
\begin{center}
\begin{tabular}{lllll}
\hline\noalign{\smallskip}
\textbf{Matrix type}	&\textbf{n}	&\textbf{m}	&\textbf{nz}	&\textbf{nc}	\\
\noalign{\smallskip}\hline\noalign{\smallskip}
bone010	&986703	&986703	&47851783	&63	\\
kkt\_power	&2063494	&2063494	&12771361	&90	\\
ldoor	&952203	&952203	&42493817	&77	\\
parabolic\_fem	&525825	&525825	&3674625	&7	\\
pds-100	&156243	&517577	&1096002	&7	\\
rajat31	&4690002	&4690002	&20316253	&1200	\\
Rucci1	&1977885	&109900	&7791168	&108	\\
sme3Dc	&42930	&42930	&3148656	&405	\\
torso1	&116158	&116158	&8516500	&1200	\\
\noalign{\smallskip}\hline\noalign{\smallskip}
\end{tabular}
\end{center}
\end{table}
\subsubsection{Validating ICE Using Different SpMV Algorithms}
From the model-estimated data, CSB SpMV consumes less energy than CSC SpMV on both platforms. Even though CSB has higher work complexity than CSC, CSB SpMV has less I/O complexity than CSC SpMV. Firstly, the dynamic energy cost of one I/O is much greater than the energy cost of one operation (i.e., $\epsilon_{I/O}>>\epsilon_{op}$) on both platforms. Secondly, CSB has better parallelism than CSC, computed by $\frac{Work}{Span}$, which results in shorter execution time. Both reasons contribute to the less energy consumption of CSB SpMV. 
\begin{figure}[!t] \centering
\resizebox{0.9\columnwidth}{!}{ \includegraphics{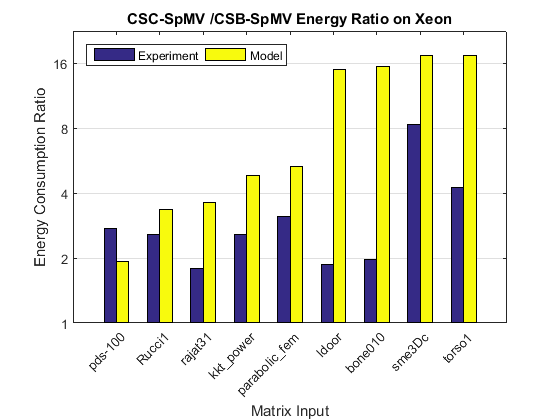}}
\caption{Energy consumption comparison between CSC-SpMV and CSB-SpMV on the Intel Xeon platform, computed by $\frac{E_{CSC}}{E_{CSB}}$. Both the ICE model estimation and experimental measurement on Intel Xeon platform show the consistent results that $\frac{E_{CSC}}{E_{CSB}}$ is greater than 1, meaning CSC SpMV algorithm consumes more energy than the CSB SpMV algorithm on different input matrices.}
\label{fig:CSBvsCSC-Xeon}
\end{figure} 
\begin{figure}[!t] \centering
\resizebox{0.9\columnwidth}{!}{ \includegraphics{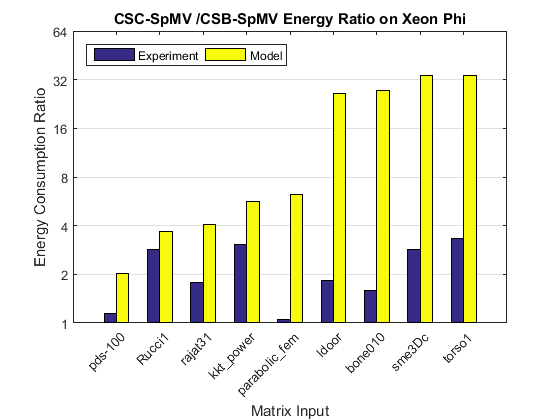}}
\caption{Energy consumption comparison between CSC-SpMV and CSB-SpMV on the Intel Xeon Phi platform, computed by $\frac{E_{CSC}}{E_{CSB}}$. Both the ICE model estimation and experimental measurement on Intel Xeon Phi platform show the consistent results that $\frac{E_{CSC}}{E_{CSB}}$ is greater than 1, meaning CSC SpMV algorithm consumes more energy than the CSB SpMV algorithm on different input matrices.}
\label{fig:CSBvsCSC-XeonPhi}
\end{figure} 
The measurement data confirms that CSB SpMV algorithm consumes less energy than CSC SpMV algorithm, shown by the energy consumption ratio between CSC-SpMV and CSB-SpMV greater than 1 in the Figure \ref{fig:CSBvsCSC-Xeon} and \ref{fig:CSBvsCSC-XeonPhi}. For all input matrices, the ICE model has confirmed that CSB SpMV consumes less energy than CSC SpMV algorithm. 
\paragraph{Validating ICE Using Different Input Types}
\begin{table}
\caption{Comparison of Energy Consumption of Different Matrix Input Types.}
\label{table:Input-Energy-Comparison}
\centering
\resizebox{\textwidth}{!}{%
\begin{tabular}{lllllllll}
\hline\noalign{\smallskip}
Algorithm &CSB &CSB &CSC  &CSC  &CSB  &CSB &CSC  &CSC  \\
\noalign{\smallskip}\hline\noalign{\smallskip}
Platform &Xeon &Xeon &Xeon &Xeon  &Xeon-Phi &Xeon-Phi &Xeon-Phi &Xeon-Phi  \\
\noalign{\smallskip}\hline\noalign{\smallskip}
Model/Exprmt & model &  exprmt &model &exprmt  &model &exprmt &model &exprmt  \\
\noalign{\smallskip}\hline\noalign{\smallskip}
Increasing &sme3Dc	&pds-100	&pds-100	&pds-100	&sme3Dc	&pds-100	&pds-100	&parabolic\\

Energy&torso1	&parabolic	&sme3Dc	&parabolic	&torso1	&parabolic	&sme3Dc	&pds-100\\

Consumption&pds-100	&sme3Dc	&parabolic	&sme3Dc	&pds-100	&Rucci1	&parabolic	&Rucci1\\

Order&parabolic	&Rucci1	&Rucci1	&Rucci1	&parabolic	&sme3Dc	&Rucci1	&sme3Dc\\

&Rucci1	&kkt &torso1	&kkt	&ldoor	&kktr	&torso1	&rajat31\\

&kkt	&torso1	&kkt	&torso1	&bone010	&torso1	&kkt	&kkt\\

&ldoor	&rajat31	&rajat31	&rajat31	&Rucci1	&rajat31	&rajat31	&ldoor\\

&bone010	&ldoor	&ldoor	&ldoor	&kkt	&ldoor	&ldoor	&torso1\\

&rajat31	&bone010	&bone010	&bone010	&rajat31	&bone010	&bone010	&bone010\\
\noalign{\smallskip}\hline\noalign{\smallskip}
\end{tabular}
}
\end{table}
\begin{table*}
\caption{CSC Energy Comparison of Different Input Matrix Types on Xeon}
\label{table:Input-Comparison-CSC}
\begin{center}
\begin{tabular}{llllllllll}
\hline\noalign{\smallskip}
Correctness	&pds-100	&parabolic	&sme3Dc	&Rucci1	&kkt	&torso1	&rajat31	&ldoor	&bone010\\
pds-100	&x	&1	&1	&1	&1	&1	&1	&1	&1\\
parabolic	&	&x	&0	&1	&1	&1	&1	&1	&1\\
sme3Dc		&&	&x	&1	&1	&1	&1	&1	&1\\
Rucci1		&&&		&x	&1	&1	&1	&1	&1\\
kkt			&&&&		&x	&0	&1	&1	&1\\
torso1		&&&&&				&x	&1	&1	&1\\
rajat31		&&&&&&					&x	&1	&1\\
ldoor		&&&&&&&						&x	&1\\
bone010		&&&&&&&&							&x\\
\noalign{\smallskip}\hline\noalign{\smallskip}
\end{tabular}
\end{center}
\end{table*}
\begin{table}
\caption{Comparison accuracy of SpMV energy consumption computing different input matrix types}
\label{table:Input-Comparison-Accuracy}
\begin{center}
\begin{tabular}{llll}
\hline\noalign{\smallskip}
Algorithm &CSB  &CSC\\
\noalign{\smallskip}\hline\noalign{\smallskip}
Xeon &75\%  &94\%\\
Xeon Phi &63.8\%  &80.5\%\\
\noalign{\smallskip}\hline\noalign{\smallskip}
\end{tabular}
\end{center}
\end{table}
To validate the ICE model regarding input types, the experiments have been conducted with nine matrix types listed in Table \ref{table:matrix-type}. The model can capture the energy-consumption relation among different inputs. The increasing order of energy consumption of different matrix-input types are shown in Table \ref{table:Input-Energy-Comparison}, from both model estimation and experimental study.

For instance, in order to validate the comparison of energy consumption for different input types, a validated table as Table \ref{table:Input-Comparison-CSC} is created for CSC SpMV on Xeon to compare model prediction and experimental measurement. For nine input types, there are $\frac{9\times9}{2}-9=36$ input relations. If the relation is correct, meaning both experimental data and model data are the same, the relation value in the table of two inputs is 1. Otherwise, the relation value is 0. From Table \ref{table:Input-Comparison-CSC}, there are 34 out of 36 relations are the same for both model and experiment, which gives 94\% accuracy on the relation of the energy consumption of different inputs. Similarly, the input validation for CSC and CSB on both Xeon and Xeon Phi platforms is provided in Table \ref{table:Input-Comparison-Accuracy}.
\paragraph{Validating The Applicability of ICE on Different Platforms}
\begin{figure*}[!t] \centering
\resizebox{0.9\textwidth}{!}{\includegraphics{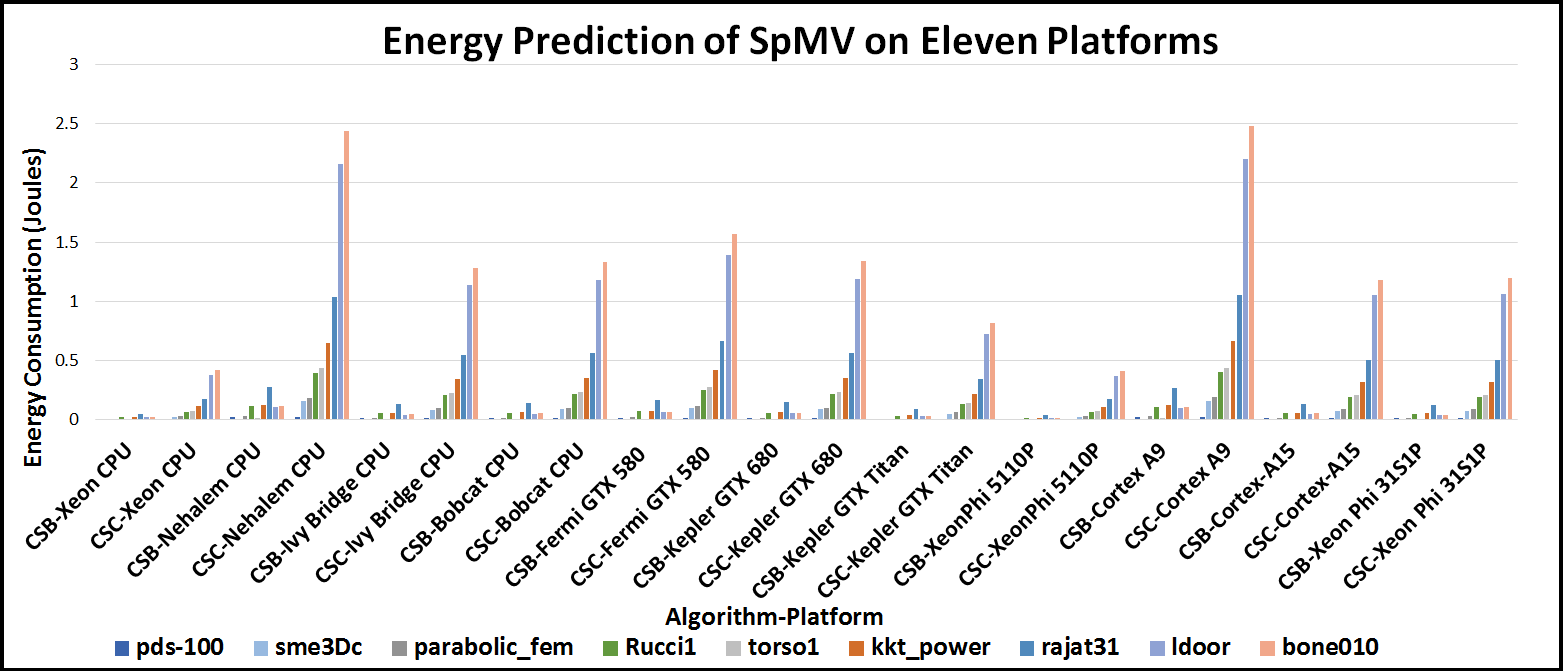}}
\caption{Energy Comparison of CSB and CSC SpMV on eleven different platforms.}
\label{fig:ModelPrediction_Platforms}
\end{figure*}
The energy comparison of CSB and CSC SpMV is concluded for eleven platforms listed in Table \ref{table:platform-parameter-values}. Like two Xeon and Xeon Phi 31S1P platforms used in experiments, Figure \ref{fig:ModelPrediction_Platforms} shows the prediction that CSB SpMV consumes less energy than CSC SpMV, on all platforms listed in Table \ref{table:platform-parameter-values}. This confirms the applicability of ICE model to compare the energy consumption of algorithms on different platforms with different input types.
\paragraph{Validating the Platform-independent Energy Complexity Model}
From Equation \ref{eq:BigE-CSC} and \ref{eq:BigE-CSB}, the platform-independent energy complexity for CSC and CSB SpMV are derived as Equation \ref{eq:BigE-CSC-plat-ind} and \ref{eq:BigE-CSB-plat-ind}, respectively.
\begin{equation} \label{eq:BigE-CSC-plat-ind}
	E_{CSC}= O(2 \times nz + (nc+ \log{}n)) 
\end{equation}
\begin{equation}
\label{eq:BigE-CSB-plat-ind}
	E_{CSB}= O(2 \times \frac{n^2}{\beta^2} + nz \times (1+\frac{1}{B}) + \beta \times \log{\frac{n}{\beta}}+ \frac{n}{\beta}) 
\end{equation}
We validate the platform-independent energy complexity of CSC and CSB SpMV. The platform-independent energy complexity also shows the accurate comparison of CSC and CSB SpMV computing different matrix types shown in Figure \ref{fig:model-comparison}. Both platform-independent and platform-supporting models show that CSC-SpMV algorithm consumes more energy than CSB-algorithm. 
However, the difference gap between the energy complexity of CSC and CSB using the platform-independent model is not clear for all input types except "ldoor" and "bone010" while in the platform-supporting model, the difference gap is clearer and consistent with the experiment results in terms of which algorithm consumes less energy for different input types. 
Comparing energy consumption of different input types requires more detailed information of the platforms. Therefore, the platform-independent model is only applicable to predict which algorithm consumes more energy.
\begin{figure}[!t] \centering
\resizebox{0.5\columnwidth}{!}{\includegraphics{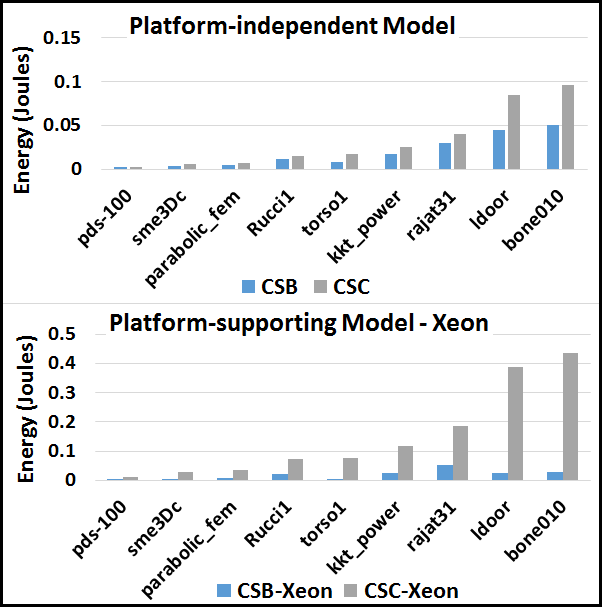}}
\caption{Comparison of platform-dependent and platform-supporting energy complexity model. Both models show that CSC SpMV consumes more energy than CSB SpMV.}
\label{fig:model-comparison}
\end{figure}
\label{SpMV-validation}
\subsubsection{Validating ICE With Matmul Algorithms}
The validation of ICE model with Matmul algorithm is another new result of this study in Deliverable D2.4 as compared to Deliverable D2.3. This makes the validation of the ICE model more complete with both data-intensive and computation-intensive algorithms.
From the model-estimated data, Basic-Matmul consumes more energy than CO-Matmul on both platforms. Even though both algorithms have the same work and span complexity, Basic-Matmul has more I/O complexity than CO-Matmul, which results in greater energy consumption of Basic-Matmul compared to CO-Matmul algorithm.
\begin{figure}[!t] \centering
\resizebox{0.8\columnwidth}{!}{ \includegraphics{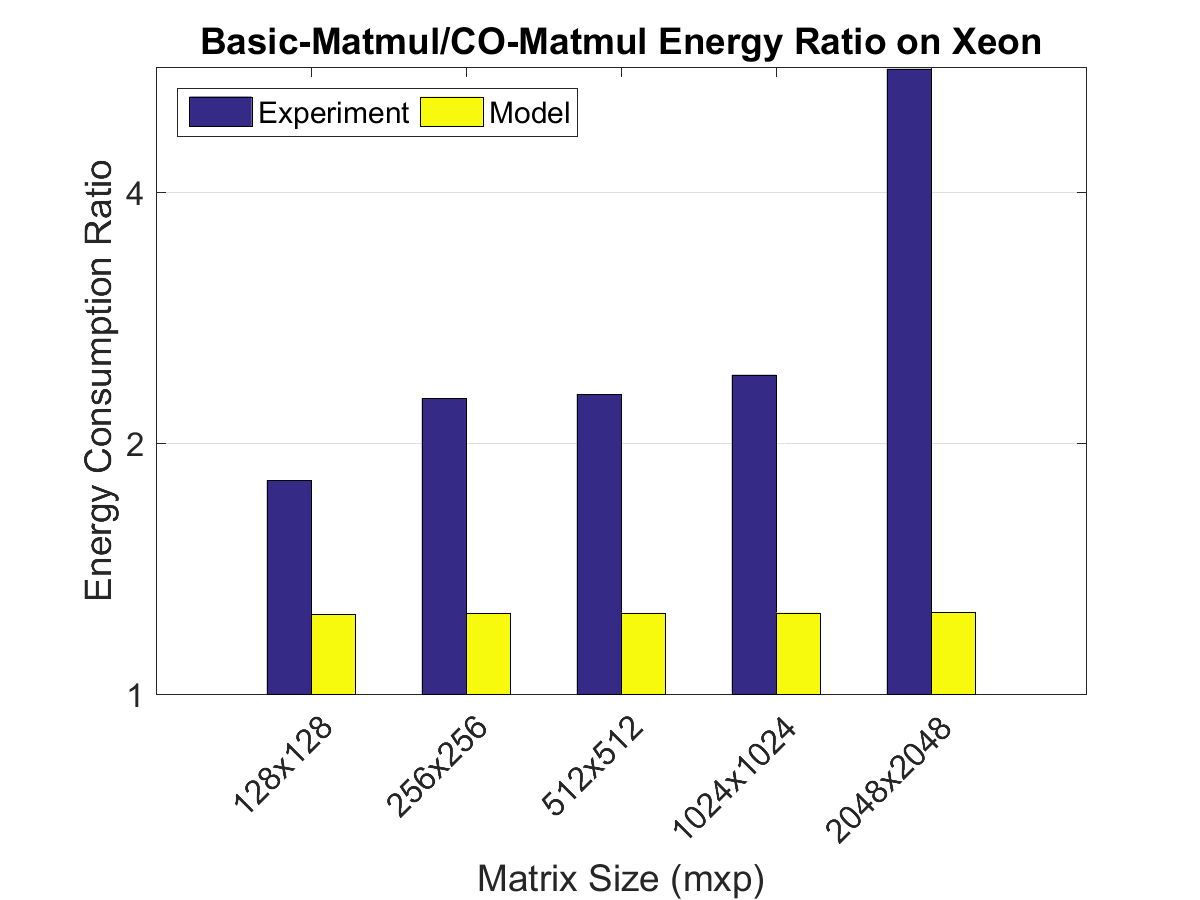}}
\caption{Energy consumption comparison between Basic-Matmul and CO-Matmul on the Intel Xeon platform, computed by $\frac{E_{Basic}}{E_{CO}}$. Both the ICE model estimation and experimental measurement on Intel Xeon platform show that $\frac{E_{Basic}}{E_{CO}}$ is greater than 1, meaning Basic-Matmul algorithm consumes more energy than the CO-Matmul algorithm.}
\label{fig:NaiveCO-Xeon}
\end{figure} 
\begin{figure}[!t] \centering
\resizebox{0.8\columnwidth}{!}{ \includegraphics{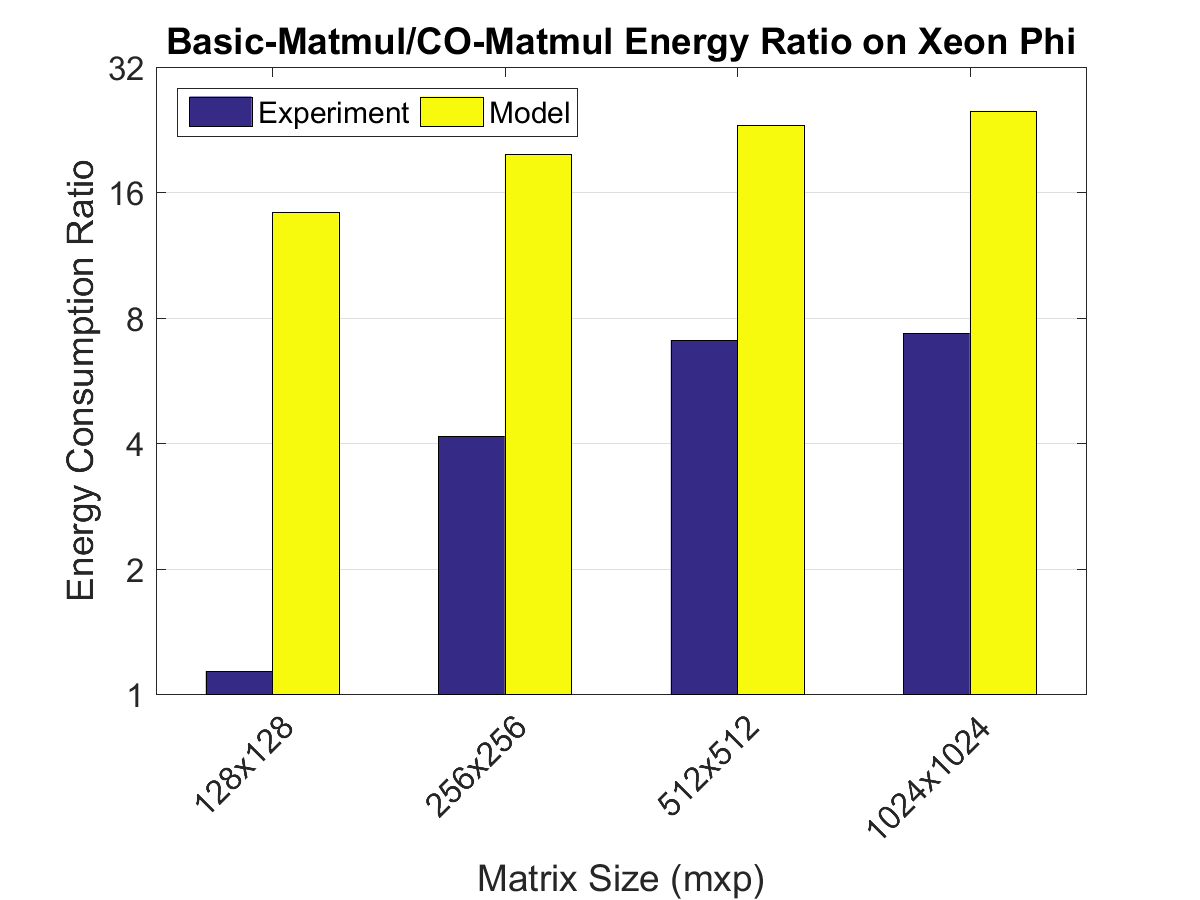}}
\caption{Energy consumption comparison between Basic-Matmul and CO-Matmul on the Intel Xeon Phi platform, computed by $\frac{E_{Basic}}{E_{CO}}$. Both the ICE model estimation and experimental measurement on Intel Xeon Phi platform show that $\frac{E_{Basic}}{E_{CO}}$ is greater than 1, meaning Basic-Matmul algorithm consumes more energy than the CO-Matmul algorithm.}
\label{fig:NaiveCO-XeonPhi}
\end{figure} 
The measurement data confirms that Basic-Matmul algorithm consumes more energy than CO-Matmul algorithm, shown by the energy consumption ratio between Basic-Matmul and CO-Matmul greater than 1 in the Figure \ref{fig:NaiveCO-Xeon} and \ref{fig:NaiveCO-XeonPhi}. For all input matrices, the ICE model has confirmed that Basic-Matmul consumes more energy than CO-Matmul algorithm.

\label{Matmul-validation}


\clearpage


\section{Conclusions} \label{sec:Conclusion}
In this Deliverable D2.4, we have reported our current results on the new energy/power models modeling the trade-off of energy efficiency and performance of data structures and algorithms; as well as the final prototype of libraries and programming abstractions.

\begin{itemize}
\item We have presented a detailed description of GreenBST, an energy-efficient concurrent search tree that is briefly described in D2.3. We have evaluated GreenBST with new state-of-the-art concurrent search trees and showed that GreenBST is portable and has a better energy efficiency and throughput than the state-of-the-art. We have developed GreenBST for Myriad2 and have experimentally evaluated our implementation.
\item  We developed a methodology for the customization of streaming aggregation implemented in modern low power embedded devices. We further compared the proposed embedded system implementations of the streaming aggregation operator with the corresponding HPC and GPGPU implementations in terms of performance per watt. 
\item We have introduced two new frameworks that can be used to the capture the performance of a wide set of lock-free data structures in dynamic environments. Then, we have integrated these performance analyses to our previous power model to obtain energy efficiency. 
\item  We have validated the ICE model, a new energy complexity model for multithreaded algorithms with both data-intensive and computation-intensive kernels. This new energy complexity model is general for parallel (multithreaded) algorithms. The ICE model derives the energy complexity of a given algorithm from its {\em work}, {\em span} and {\em I/O} complexity. We also showed that {\em I/O} complexity in energy complexity is computed based on the Ideal Cache memory model.  
\end{itemize}

\newpage

\bibliographystyle{plain}

\bibliography{../WP6-bibtex/longhead,../D2.1/D2.1_related_papers,../WP6-bibtex/excess,../WP6-bibtex/related-papers,../WP6-bibtex/peppher-related,UiT-Vi-D2-4,bib-chalmers/biblio-disc,bib-chalmers/longhead,bib-chalmers/biblio-aggr,UiT-Ibrahim}
\newpage

\begin{appendices}
\section{The tree library}\label{sec:tree-library}
We have developed concurrent search tree libraries that contain the implementation of the concurrent search tree algorithms described in Table~\ref{tbl:algos}.

\subsection{Getting the source and compilation.} The libraries are provided in a separate directory for easy access and maintenance. The repository address 
is http://gitlab.excess-project.eu/ibrahim/tree-libraries. A makefile for each of the libraries is also provided to aid compilations. 
The libraries have been tested on Linux and Mac OS X platforms.

\subsection{Running and outputs.} By default, the provided makefile will build the standalone benchmark 
version of the libraries which will accept these following parameters:

\noindent\fbox{%
    \parbox{0.95\textwidth}{%
   \texttt{-r <NUM>    : Allowable range for each element (0..NUM)\\
-u <0..100> : Update ratio. 0 = Only search; 100 = Only updates\\
-i <NUM>    : Initial tree size (initial pre-filled element count)\\
-t <NUM>    : DeltaNode ($\UB$) size (ONLY USED IN DELTATREE FAMILIES)\\
-n <NUM>    : Number of benchmark threads\\
-s <NUM>    : Random seed. (0 = using time as seed, Default)
}}}
\\
\\
The benchmark outputs are formatted in this sequence: 

\noindent\fbox{%
    \parbox{0.95\textwidth}{%
    \texttt{
0: range, insert ratio, delete ratio, \#threads, \#attempted insert, \#attempted delete, \#attempted search, \#effective insert, \#effective delete, \#effective search, time (in msec.)
}}}

NOTE: {\tt 0:} characters are just unique token for easy tagging (e.g., for using {\tt grep}).

\noindent\fbox{%
    \parbox{0.95\textwidth}{%
\ttfamily
\$ ./DeltaTree -h \newline
DeltaTree v0.1\newline
===============\newline
Use -h switch for help.\newline
\newline
Accepted parameters\newline
-r <NUM>    : Range size\newline
-u <0..100> : Update ratio. 0 = Only search; 100 = Only updates\newline
-i <NUM>    : Initial tree size (inital pre-filled element count)\newline
-t <NUM>    : DeltaNode size\newline
-n <NUM>    : Number of threads\newline
-s <NUM>    : Random seed. 0 = using time as seed\newline
-d <0..1>   : Density (in float)\newline
-v <0 or 1> : Valgrind mode (less stats). 0 = False; 1 = True\newline
-h          : This help\newline
\newline
Benchmark output format: \newline
"0: range, insert ratio, delete ratio, \#threads, attempted insert, attempted delete, attempted search, effective insert, effective delete, effective search, time (in msec)"
}}

\noindent\fbox{%
    \parbox{0.95\textwidth}{%
\ttfamily
\$ ./DeltaTree -r 5000000 -u 10 -i 1024000 -n 10 -s 0\newline
DeltaTree v0.1\newline
===============\newline
Use -h switch for help.\newline
\newline
Parameters:\newline
- Range size r:		 5000000\newline
- DeltaNode size t:	 127\newline
- Update rate u:	 10\% \newline
- Number of threads n:	 10\newline
- Initial tree size i:	 1024000\newline
- Random seed s:	 0\newline
- Density d:		 0.500000\newline
- Valgrind mode v:	 0\newline
\newline
Finished building initial DeltaTree\newline
The node size is: 25 bytes\newline
Now pre-filling 1024000 random elements...\newline
...Done!\newline
\newline
Finished init a DeltaTree using DeltaNode size 127, with initial 1024000 members\newline
\#TS: 1421050928, 511389\newline
Starting benchmark...\newline
Pinning to core 0... Success\newline
Pinning to core 3... Success\newline
Pinning to core 1... Success\newline
Pinning to core 8... Success\newline
Pinning to core 9... Success\newline
Pinning to core 10... Success\newline
Pinning to core 2... Success\newline
Pinning to core 11... Success\newline
Pinning to core 4... Success\newline
Pinning to core 12... Success\newline
\newline
0: 5000000, 5.00, 5.00, 10, 249410, 248857, 4501733, 195052, 53720, 1000568, 476\newline
\newline
Active (alloc'd) triangle:258187(266398), Min Depth:12, Max Depth:30 \newline
Node Count:1165332, Node Count(MAX): 1217838, Rebalance (Insert) Done: 234, Rebalance (Delete) Done: 0, Merging Done: 1\newline
Insert Count:195052, Delete Count:53720, Failed Insert:54358, Failed Delete:195137 \newline
Entering top: 0, Waiting at the top:0
}}

NOTE: {\tt \#TS:} is the benchmark start timestamp. 

\subsection{Pluggable library.} To use any component as a library, each library provides a (.h) header file and a simple, 
uniform API in C. These available and callable APIs are:

\noindent\fbox{%
    \parbox{0.95\textwidth}{%
\textsc{Structure:}\\
\\
\texttt{<libname>\_t} : Structure variable declaration.\\
\\
\textsc{Functions:}\\
\\
\texttt{<libname>\_t* <libname>\_alloc()} :	Function to allocate the defined structure, returns the allocated (empty) structure.\\
\\
\texttt{void* <libname>\_free(<libname>\_t* map)} :	Function to release all memory used by the structure, returns NULL on success.\\
\\
\texttt{int <libname>\_insert(<libname>\_t* map, void* key, void* data)} : Function to insert a key and a linked pointer (data), returns 1 on success and 0 otherwise.\\
\\
\texttt{int <libname>\_contains(<libname>\_t* map, void* key)} : Function to check whether a key is available in the structure, returns 1 if yes and 0 otherwise.\\
\\
\texttt{void *<libname>\_get(<libname>\_t* map, void* key)} : Function to get the linked data given its key, returns the pointer of the data of the corresponding key and 0 if the 
key is not found.\\
\\
\texttt{int <libname>\_delete(<libname>\_t* map, void* key)} : Function to delete an element that matches the given key, returns 1 on success and 0 otherwise. 
}}

As an example, the concurrent B-tree library provides the \texttt{cbtree.h} file that can be linked into
any C source code and provides the callable \texttt{cbtree\_t* cbtree\_alloc()} function. 
Note that the valid \texttt{<libname>} is \texttt{dtree} for DeltaTree, \texttt{gbst} for GreenBST, and \texttt{cbtree} for CBTree.
It is also possible
to use the MAP selector header (\texttt{map\_select.h}) plus defining which tree type to use 
so that  MAP\_\textless{operator}\textgreater functions are used instead as specific tree function as the below example:

\begin{lstlisting}[frame=single, language=C]
#define MAP_USE_CBTREE 
#include "map_select.h"

int main(void)
{
	long numData = 10;
	long i;
	char *str;
	puts("Starting...");
	MAP_T* cbtreePtr = MAP_ALLOC(void, void); 
	assert(cbtreePtr);
	for (i = 0; i < numData; i++) { 
		str = calloc(1, sizeof(char)); *str = 'a'+(i%254);
		MAP INSERT(cbtreePtr, i+1, str); 
	}
	for (i = 0; i < numData; i++) {
		printf("%ld: %c\n", i+1, 
			*((char*)MAP_FIND(cbtreePtr, i+1))); 
	}
	for (i = 0; i < numData; i++) {
		printf("%ld: %d\n", i+1, 
			MAP_CONTAINS(cbtreePtr, i+1));
	}
	for (i = 0; i < numData; i++) {
		MAP_REMOVE(cbtreePtr, i+1);
	}
	for (i = 0; i < numData; i++) {
		printf("%ld: %d\n", i+1, 
			MAP_CONTAINS(cbtreePtr, i+1));
	}
	MAP_FREE(cbtreePtr)
	puts("Done."); 
	return 0;
}

\end{lstlisting}

\subsection{Intel PCM integration.} All of the libraries provide support for Intel PCM measurement. 
To enable Intel PCM measurement metrics,
the compiler must be invoked using \texttt{-DUSE\_PCM} parameter during the libraries's compilation 
and all the Intel PCM compiled object files must be linked to the output executables.
  





\end{appendices}

\section*{Glossary}
\begin{flushleft}
\begin{tabular}{lp{12cm}}
\textbf{BRU}    &  Branch Repeat Unit (on SHAVE processor) \\
\textbf{CAS}    &  Compare-and-Swap instruction \\
\textbf{CMX}    &  Connection MatriX on-chip (shared) memory unit, 128KB (Movidius Myriad) \\
\textbf{CMU}    &  Compare-Move Unit (on SHAVE processor) \\
\textbf{Component} & 1. [hardware component] part of a chip's or motherboard's 
  circuitry; \ 2. [software component] encapsulated and annotated reusable
  software entity with contractually specified interface and
  explicit context dependences only, subject to third-party (software) composition.\\
\textbf{Composition}    & 1. [software composition] Binding a call to a 
  specific callee (e.g., implementation variant of a component) and allocating
  resources for its execution; \ 2. [task composition] Defining a macrotask and
  its use of execution resources 
  by internally scheduling its constituent tasks in serial,
  in parallel or a combination thereof. \\
  
\textbf{CPU}    &  Central (general-purpose) Processing Unit\\

\textbf{uncore}    &  including the ring interconnect, shared cache, integrated memory controller, home agent, power control unit, integrated I/O module, config Agent, caching agent and Intel QPI link interface \\ 
\textbf{CTH}    &  Chalmers University of Technology \\
\textbf{DAQ}    &  Data Acquisition Unit \\
\textbf{DCU}    &  Debug Control Unit (on SHAVE processor) \\
\textbf{DDR}    &  Double Data Rate Random Access Memory \\
\textbf{DMA}    &  Direct (remote) Memory Access \\
\textbf{DRAM}   &  Dynamic Random Access Memory \\
\textbf{DSP}    &  Digital Signal Processor \\
\textbf{DVFS}   &  Dynamic Voltage and Frequency Scaling \\
\textbf{ECC}    &  Error-Correcting Coding \\
\textbf{EXCESS} &  Execution Models for Energy-Efficient Computing Systems\\
\textbf{GPU}    &  Graphics Processing Unit\\
\textbf{HPC}    &  High Performance Computing\\
\textbf{IAU}    &  Integer Arithmetic Unit (on SHAVE processor) \\
\textbf{IDC}    &  Instruction Decoding Unit (on SHAVE processor) \\
\textbf{IRF}    &  Integer Register File (on SHAVE processor) \\
\textbf{LEON}    &  SPARCv8 RISC processor in the Myriad1 chip\\
\textbf{LIU}    &  Link\"oping University \\
\textbf{LLC}    &  Last-level cache\\
\textbf{LSU}    &  Load-Store Unit (on SHAVE processor) \\
\textbf{Microbenchmark} & Simple loop or kernel developed to measure one or few properties of the underlying architecture or system software\\
\textbf{PAPI}   &  Performance Application Programming Interface\\
\end{tabular}
\end{flushleft}

\newpage 

\begin{flushleft}
\begin{tabular}{lp{12cm}}
\textbf{PEU}    &  Predicated Execution Unit (on SHAVE processor) \\
\textbf{Pinning} &  [thread pinning] Restricting the operating system's CPU scheduler in order to map a thread to a fixed CPU core \\
\textbf{QPI}    &  Quick Path Interconnect\\
\textbf{RAPL}   &  Running Average Power Limit energy consumption counters (Intel)\\
\textbf{RCL}   &  Remote Core Locking (synchronization algorithm)\\
\textbf{SAU}    &  Scalar Arithmetic Unit (on SHAVE processor) \\
\textbf{SHAVE}  &  Streaming Hybrid Architecture Vector Engine (Movidius) \\
\textbf{SoC}    &  System on Chip \\
\textbf{SRF}    &  Scalar Register File (on SHAVE processor) \\
\textbf{SRAM}   &  Static Random Access Memory \\
\textbf{TAS}    &  Test-and-Set instruction\\
\textbf{TMU}    &  Texture Management Unit (on SHAVE processor) \\
\textbf{USB}    &  Universal Serial Bus \\
\textbf{VAU}    &  Vector Arithmetic Unit (on SHAVE processor) \\
\textbf{Vdram}  &  DRAM Supply Voltage \\
\textbf{Vin}    &  Input voltage level  \\
\textbf{Vio}    &  Input/Output voltage level  \\
\textbf{VLIW}   &  Very Long Instruction Word (processor) \\
\textbf{VLLIW}  &  Variable Length VLIW (processor) \\
\textbf{VRF}    &  Vector Register File (on SHAVE processor) \\
\textbf{Wattsup}&  Watts Up .NET power meter \\
\textbf{WP1}   &  Work Package 1 (here: of EXCESS) \\
\textbf{WP2}   &  Work Package 2 (here: of EXCESS) \\
\end{tabular}
\end{flushleft}

\end{document}